\RequirePackage{tikz}

\documentclass[sn-mathphys]{sn-jnl}

\usepackage{amsmath}
\usepackage{amssymb} 
\usepackage{amsthm} 
\usepackage{bm} 
\usepackage{extarrows} 
\usepackage{enumitem} 
\usepackage[T1]{fontenc} 

\usepackage{stmaryrd} 
\usepackage{todonotes} 
\presetkeys{todonotes}{inline,fancyline}{} 
\usepackage{subcaption} 

\jyear{2023}%

\newcommand{\old}[1]{\textcolor{black}{#1}}
\newcommand{\new}[1]{\textcolor{black}{#1}}

\renewcommand{\mod}[1]{\llbracket #1 \rrbracket_{\mathrm{mod}}}
\newcommand{\sem}[1]{\llbracket #1 \rrbracket_{\mathrm{sem}}}

\newcommand{\sat}{\mathbin{\leq}}

\newcommand{\quotient}{\backslash\!\backslash}

\usetikzlibrary{positioning}
\usetikzlibrary{arrows}
\usetikzlibrary{backgrounds}
\usetikzlibrary{calc}
\usetikzlibrary{automata}

\tikzset{connector/.style={circle, fill=black, inner sep=2pt}} 
\tikzset{main node/.style={circle,draw,font=\sffamily\small,minimum size=12pt, prefix after command= {\pgfextra{\tikzset{every label/.style={font=\sffamily}}}}}} 
\tikzset{action/.style={font=\sffamily}} 

\theoremstyle{thmstyleone}%
\newtheorem{theorem}{Theorem}
\newtheorem{lemma}{Lemma}
\newtheorem{corollary}{Corollary}
\newtheorem{definition}{Definition}

\begin{document}

\title[Timed I/O Automata]{Timed I/O Automata: It Is Never Too Late to Complete Your Timed Specification Theory}


\author*[1]{\fnm{Martijn A.} \sur{Goorden}}\email{mgoorden@cs.aau.dk}

\author[1]{\fnm{Kim G.} \sur{Larsen}}\email{kgl@cs.aau.dk}

\author[2]{\fnm{Axel} \sur{Legay}}\email{axel.legay@uclouvain.be}

\author[1]{\fnm{Florian} \sur{Lorber}}\email{florber@cs.aau.dk}

\author[1]{\fnm{Ulrik} \sur{Nyman}}\email{ulrik@cs.aau.dk}

\author[3]{\fnm{Andrzej} \sur{W\k{a}sowski}}\email{wasowski@itu.dk}

\affil*[1]{\orgdiv{Department of Computer Science}, \orgname{Aalborg University}, \orgaddress{\street{Selma Lagerl\"{o}fs Vej 300}, \city{Aalborg \"{O}st}, \postcode{9220}, \country{Danmark}}}

\affil[2]{\orgdiv{Department of Computer Science Engineering}, \orgname{UC Louvain}, \orgaddress{\street{Place Sainte Barbe 2}, \city{Louvain-la-Neuve}, \postcode{1348}, \country{Belgium}}}

\affil[3]{\orgdiv{Department of Computer Science}, \orgname{IT University}, \orgaddress{\street{Rued Langgaards Vej 7}, \city{Copenhagen S}, \postcode{2300}, \country{Denmark}}}


\abstract{A specification theory combines notions of specifications and implementations with a satisfaction relation, a refinement relation and a set of operators supporting stepwise design.
We develop a complete specification framework for real-time systems using Timed I/O Automata as the specification formalism, with the semantics expressed in terms of Timed I/O Transition Systems.
We provide constructs for refinement, consistency checking, logical and structural composition, and quotient of specifications---all indispensable ingredients of a compositional design methodology.
\new{The theory is backed by rigorous proofs and is being implemented in the open-source tool ECDAR.}}




\maketitle

\section{Introduction}
\label{sec:intro}

Modularity is a highly desired property of a software architecture.
Modular software systems are decomposed into components, often designed by independent teams, working under a common agreement on what the interface of each component should be.  Consequently, \emph{compositional reasoning}~\cite{henzinger_embedded_2006,benveniste_contracts_2018}, the mathematical foundations of reasoning about interfaces, is an active research area.  Besides supporting compositional development, it enables compositional reasoning about the system (verification) and allows well-grounded reuse.

In a logical interpretation, interfaces are specifications, while components that implement an interface are understood as models/implementations.  Specification theories should support various features including (1) a \emph {refinement} that allows to compare specifications and to replace a specification by another one in a design, (2) a \emph {logical conjunction} that expresses combining the requirements of two or more specifications, (3) a \emph {structural composition}, which allows to combine specifications, and (4) a \emph{quotient operator} that, being a dual to structural composition, allows decomposing the design by groups of requirements. The latter is crucial to perform incremental design. Also, the operations have to be related by compositional reasoning theorems, guaranteeing both incremental design and independent implementability~\cite{alfaro_interface-based_2005}.

Building good specification theories is the subject of intensive studies~\cite{chakabarti_resource_2003,alfaro_interface_2001,benveniste_contracts_2018}. Interface automata are one such successful direction~\cite{alfaro_interface_2001,alfaro_interface-based_2005,larsen_modal_1989,milner_communication_1989}. In this framework, an interface is represented by an input/output automaton~\cite{lynch_introduction_1988},  i.e., an automaton whose transitions are typed with \emph{input} and \emph{output}. The semantics of such an automaton is given by a two-player game: the \emph{input} player represents the environment, and the \emph{output} player represents the component itself. Contrary to the input/output model proposed by Lynch and Tuttle~\cite{lynch_introduction_1988}, this semantic offers an optimistic treatment of composition: two interfaces can be composed if there exists at least one environment in which they can interact together in a safe way. A timed extension of the theory of interface automata has been motivated by the fact that real time can be a crucial parameter in practice, for example in embedded systems~\cite{alfaro_timed_2002}. 

There have been several other attempts to propose an interface theory for timed systems (see~\cite{alfaro_timed_2002,alfaro_accelerated_2007,bertrand_refinement_2009,bertrand_compositional_2009,cerans_timed_1993,thiele_real-time_2006,lee_handbook_2007,hampus_formally_2022} for some examples). 
Our model shall definitely be viewed as an extension of the timed input/output automaton model proposed by Lynch and Tuttle~\cite{lynch_introduction_1988} and Kaynar et al.~\cite{kaynar_timed_2003}. 
The majors differences are in the game-based treatment of interactions and the addition of quotient and conjunction operators.

In~\cite{alfaro_timed_2002}, de Alfaro et al. suggested {\em timed interfaces}, a model that is similar to the one of TIOTSs. 
Our definition of composition builds on the one proposed in there. 
However, the work in~\cite{alfaro_timed_2002} is incomplete. 
Indeed, there is no notion of implementation and refinement. 
Moreover, conjunction and quotient are not studied. 
Finally, the theory has only been implemented in a prototype tool~\cite{alfaro_accelerated_2007} which does not handle continuous time, while our contribution takes advantages of the powerful Difference Bound Matrices (DBM) representation~\cite{bellman_dynamic_1957,dill_timing_1990,behrmann_uppaal_2002}.

In~\cite{larsen_modal_1989}, Larsen proposes {\em modal automata}, which are deterministic automata equipped with transitions of the following two types: \emph{may} and \emph{must}. 
The components that implement such interfaces are simple labeled transition systems. 
Roughly, a must transition is available in every component that implements the modal specification, while a may transition need not be. 
Recently, \cite{bertrand_refinement_2009,bertrand_compositional_2009} a timed extension of \emph{modal} automata was proposed, which embeds all the operations presented in the present paper. 
However, modalities are orthogonal to inputs and outputs, and it is well-known~\cite{larsen_modal_2007} that, contrary to the game-semantic approach, they cannot be used to distinguish between the behaviors of the component and those of the environment. 

\emph{Component interface specification and consistency.} We represent    specifications by timed input/output transition systems~\cite{kaynar_timed_2003}, i.e., timed  transitions systems whose sets  of discrete transitions are split into input and output transitions. Contrary to~de Alfaro and colleagues \cite{alfaro_timed_2002} and unlike Kaynar et al.\ \cite{kaynar_timed_2003} we distinguish between implementations and  specifications.  This is done by assuming that the  former have fixed timing behavior and they can always advance either by producing an output or delaying. We also provide a game-based methodology to decide whether a specification is consistent, i.e., whether it has at least one implementation. The  latter reduces to deciding existence of a strategy that despite the behavior of the environment will avoid states that cannot possibly satisfy the implementation requirements.

\emph{Refinement and logical conjunction.} A specification $S_1$ {\em refines} a specification $S_2$ iff it is possible to replace $S_2$ with $S_1$ in every environment and obtain a system satisfying the same high-level specification (the substitutability principle). In the input/output setting, checking refinement reduces to deciding an alternating timed simulation between the two specifications~\cite{alfaro_interface_2001}. In our timed extension, checking such simulation can be done with a slight modification of the theory proposed by Bulychev and coauthors~\cite{Bulychev_efficient_2009}. As implementations are specifications, refinement coincides with the satisfaction relation.
Our refinement operator has the {\em model inclusion property}, i.e., $S_1$ refines $S_2$ iff the set of implementations satisfied by $S_1$ is included in the set of implementations satisfied by $S_2$. We also propose a {\em logical conjunction} operator between specifications. Given two specifications, the operator will compute a specification whose implementations are satisfied by both operands. The operation may introduce error states that do not satisfy the implementation requirement. Those states are pruned by synthesizing a strategy for the component to avoid reaching them. \new{As we assume that we want to avoid reaching error states with any possible environment, hence this pruning is called \emph{adversarial pruning}}. We also show that conjunction coincides with shared refinement, i.e., it corresponds to the greatest specification that refines both $S_1$ and $S_2$.

\emph{Structural composition.} Following Timed Interfaces~\cite{alfaro_timed_2002}, specifications interact by synchronizing on inputs and outputs. However, like in I/O \hbox{Automata\,\cite{kaynar_timed_2003,lynch_introduction_1988},} we restrict ourselves to input-enabled systems. This makes it impossible to reach an immediate {\em deadlock state}, where a component proposes an output that cannot be captured by the other component. However, unlike in I/O Automata, input-enabledness shall not be seen as a way to avoid error states. Indeed, such error states can \new{still} be designated by the designer as states which do not warrant desirable temporal properties. \new{When composing specifications together, one would like to simplify the composition as much as possible before continuing the compositional analysis. We show that adversarial pruning does not distribute over the parallel composition operator. Therefore, we introduce the notion of \emph{cooperative pruning}.} \new{Finally, we show that} our composition operator is associative and that refinement is a precongruence with respect to it.

\emph{Quotient.} We propose a quotient operator dual to composition. Intuitively,  given a global specification $T$ of a composite system as well as the specification of an already realized component $S$, the \emph{quotient} will return the most liberal specification $X$ for the missing component, i.e., $X$ is the largest specification  such that $S$ in parallel with $X$ refines $T$.

\new{\emph{Extension over the earlier versions of this paper.} This journal paper is an extended and revised version of earlier conference papers~\cite{david_timed_2010,david_methodologies_2010} and the journal paper~\cite{david_real-time_2015}. In this journal paper, we clarify the notion and effect of pruning by introducing adversarial pruning and cooperative pruning, we show that adversarial pruning (previously just called pruning~\cite{david_timed_2010,david_methodologies_2010,david_real-time_2015}) does not distribute over the parallel composition so we no longer want to and require pruning after each composition, we corrected several definitions, including the one of the quotient, we removed the notion of strictly undesirable states, we included proofs for all theorems, and we updated the section on tool support. In the rest of the paper, we try to indicate changes to the theory with respect to these original works as much as possible.}

\new{\emph{Structure of the paper.} The paper is organized as follows. We continue first by providing a motivating example in Section~\ref{sec:example}. Parts of this example are used later in the paper to illustrate the theory. Section~\ref{sec:specandref} introduces the general framework of timed input/output transition systems and timed input/output automata, the notions of specification and implementation, and the concept of refinement. Section~\ref{sec:consandconj} continues by introducing consistency, the conjunction operator, and adversarial pruning. Then, in Section~\ref{sec:compandcomp} we introduce and discuss parallel composition and in Section~\ref{sec:quotient} the quotient operator. Finally, Section~\ref{sec:conclusion} concludes the paper. }

\section{Motivating Example}\label{sec:example}

Universities operate under increasing pressure and competition. One of the popular factors used in determining the level of national funding is that of societal impact, which is approximated by the number of \new{news articles published based on research outcomes}.  Clearly one would expect that the number (and size) of grants given to a university has a (positive) influence on the number of news articles.

Figure~\ref{fig:university} gives the insight as to the organization of a very small \textsf{University} comprising three components \textsf{Administration}, \textsf{Machine} and \textsf{Researcher}. The \textsf{Administration} is responsible for interaction with society in terms of acquiring  grants (\textsf{grant}) and  \new{writing news articles (\textsf{news})}. However, the other components are necessary for \new{news articles to be obtained}. The \textsf{Researcher} will produce the crucial publications (\textsf{pub}) within given time intervals, provided timely stimuli in terms of  coffee (\textsf{cof}) or tea (\textsf{tea}). Here coffee is clearly preferred over tea. The beverage is provided by a \textsf{Machine}, which given a coin (\textsf{coin}) will provide either coffee or tea within some time interval, or even the possibility of free tea after some time.

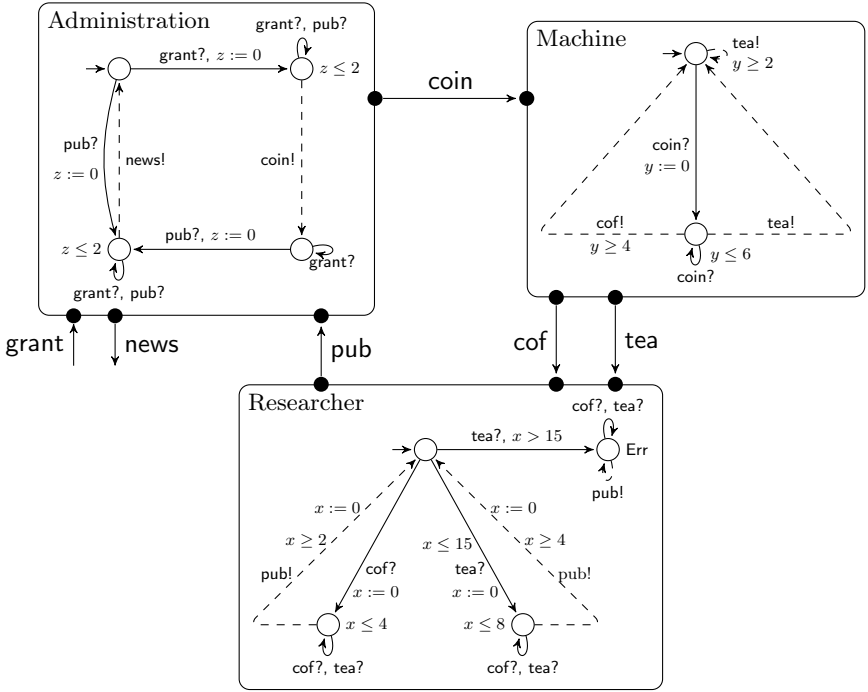
\begin{figure}
  \centering

	\begin{tikzpicture}[->,>=stealth',shorten >=1pt, font=\small, scale=1, transform shape]
		\node[draw, shape=rectangle, rounded corners] (adm) {\begin{tikzpicture}[->,>=stealth',shorten >=1pt, align=left,node distance=3cm, scale=.7]
				\node[main node, initial, initial text={}, initial where=left] (1) {};
				\node[main node, label={right:$z \leq 2$}] (2) [right = of 1] {};
				\node[main node] (3) [below = of 2] {};
				\node[main node, label={left:$z \leq 2$}] (4) [below = of 1] {};

				\path[every node/.style={action}]
				(1) edge node[above] {grant?, $z := 0$} (2)
				(2) edge[dashed] node[left] {coin!} (3)
				(3) edge node[above] {pub?, $z := 0$} (4)
				(4) edge[dashed] node[right] {news!} (1)
				;
				\path[bend right = 15, every node/.style={action}]
				(1) edge node[left,pos=.4] {pub?} node[left,pos=.6] {$z:= 0$} (4)
				;
				\path[every node/.style={action}]
				(2) edge[loop above] node[above] {grant?, pub?} (2)
				(3) edge[loop right] node[below] {grant?} (3)
				(4) edge[loop below] node[below] {grant?, pub?} (4)
				;
		\end{tikzpicture}};
		\node[anchor=north west] () at (adm.north west) {Administration};

		\node[draw, shape=rectangle, rounded corners] (machine) [right = 2cm of adm] {\begin{tikzpicture}[->,>=stealth',shorten >=1pt,align=left,node distance=3cm, scale=.7]
				\node[main node, initial, initial text={}, initial where=left] (1) {};
				\node[main node, label={below right:$y \leq 6$}] (2) [below = of 1] {};

				\path[every node/.style={action}]
				(1) edge node[left,pos=.5] {coin?} node[left,pos=.65] {$y := 0$} (2)
				;
				\draw[dashed, rounded corners, every node/.style={action}]
				(2) -- node[above] {cof!} node[below] {$y \geq 4$} ++(180:3) -- (1)
				;
				\draw[dashed, rounded corners, every node/.style={action}]
				(2) -- node[above] {tea!} ++(0:3) -- (1)
				;
				\path[every node/.style={action}]
				(1) edge[loop right, dashed] node[right] {tea!\\$y\geq 2$} (1)
				(2) edge[loop below] node[below] {coin?} (2)
				;
		\end{tikzpicture}};
		\node[anchor=north west] () at (machine.north west) {Machine};

		\coordinate (help) at ($(adm.east)!0.5!(machine.west)$);
		\node[draw, shape=rectangle, rounded corners] (researcher) [below = 3cm of help] {\begin{tikzpicture}[->,>=stealth',shorten >=1pt,align=left,node distance=3cm, scale=.7]
				\node[main node, initial, initial text={}, initial where=left] (1) {};
				\node (mid) [below = of 1] {};
				\node[main node, label={right:$x \leq 4$}] (2) [left = 1.5cm of mid] {};
				\node[main node, label={left:$x \leq 8$}] (3) [right = 1.5cm of mid] {};
				\node[main node, label={right:Err}] (4) [right = of 1] {};

				\path[every node/.style={action}]
				(1) edge node[right, pos=.7] {cof?} node[right, pos=.85] {$x := 0$} (2)
				(1) edge node[left, pos=0.55] {$x \leq 15$} node[left, pos=.7] {tea?} node[left, pos=.85] {$x := 0$} (3)
				(1) edge node[above] {tea?, $x > 15$} (4)
				;
				\draw[dashed, rounded corners, every node/.style={action}]
				(2) -- ++(180:1.5) -- node[left, pos=.3] {pub!} node[left, pos=.5] {$x \geq 2$} node[left, pos=.7] {$x := 0$} (1)
				; 
				\draw[dashed, rounded corners]
				(3) -- ++(0:1.5)     -- node[right, pos=.3] {pub!} node[right, pos=.5] {$x \geq 4$} node[right, pos=.7] {$x := 0$} (1)
				;
				\path[every node/.style={action}]
				(2) edge[loop below] node[below] {cof?, tea?} (2)
				(3) edge[loop below] node[below] {cof?, tea?} (3)
				(4) edge[loop above] node[above] {cof?, tea?} (4)
				(4) edge[loop below, dashed] node[below] {pub!} (4)
				;
		\end{tikzpicture}};
		\node[anchor=north west] () at (researcher.north west) {Researcher};

		\node[connector] (coin1) at (adm.20) {};
		\node[connector] (coin2) at (coin1 -| {machine.west}) {};
		\node[connector] (cof1) at (machine.-135) {};
		\node[connector] (cof2) at (cof1 |- {researcher.north}) {};
		\node[connector] (tea1) at (machine.-120) {};
		\node[connector] (tea2) at (tea1 |- {researcher.north}) {};
		\node[connector] (pub1) at (researcher.130) {};
		\node[connector] (pub2) at (pub1 |- {adm.south}) {};
		\node[connector] (grant1) at (adm.-130) {};
		\node[connector] (news1) at (adm.-120) {};

		\draw[->,>=stealth',shorten >=1pt, every node/.style={action}]
			(coin1) edge node[above] {coin} (coin2)
			(cof1) edge node[left] {cof} (cof2)
			(tea1) edge node[right] {tea} (tea2)
			(pub1) edge node[right] {pub} (pub2)
			(news1) edge node[right] {news} ++(-90:0.7)
		;
		\draw[<-,>=stealth',shorten >=1pt, every node/.style={action}]
			(grant1) edge node[left] {grant} ++(-90:0.7)
		;
	\end{tikzpicture}
  \caption{Specifications  for and  interconnections between  the three
    main components of a modern \textsf{University}: \textsf{Administration}, \textsf{Machine} and \textsf{Researcher}.}
  \label{fig:university}
\end{figure}

In Figure~\ref{fig:university} the three components are specifications, each allowing for a multitude of incomparable, actual implementations differing with respect to exact timing behavior (e.g., at what time are publications actually produced by the \textsf{Researcher} given a coffee) and exact output produced (e.g., does the \textsf{Machine} offer tea or coffee given a coin).

As a first  property, we may want to check that the composition of the three components comprising our \textsf{University} is compatible: we notice that the specification of the \textsf{Researcher} contains an \textsf{Err} state, essentially not providing any guarantees as to what  behavior to expect if tea is offered at a late stage. Now, compatibility checking amounts simply to deciding whether the user of the  \textsf{University} (i.e., the society) has such a strategy for using it that the \textsf{Researcher} will avoid ever entering this error state.

\begin{figure}
	\centering

	\begin{tikzpicture}[->,>=stealth',shorten >=1pt, font=\small, scale=1, transform shape]
		\node[draw, shape=rectangle, rounded corners] (spec) {\begin{tikzpicture}[->,>=stealth',shorten >=1pt, align=left,node distance=3cm, scale=.7]
				\node[main node, initial, initial text={}, initial where=below] (1) {};
				\node[main node, label={right:$u \leq 20$}] (2) [right = of 1] {};
				\node[main node] (3) [left = of 1] {};
				\node () [above= 1.3cm of 1] {};

				\path[every node/.style={action}]
				(1) edge node[above] {grant?} node[below] {$u > 2$} (3)
				;
				\path[bend left = 15, every node/.style={action}]
				(1) edge node[above, align=center] {grant?, $u \leq 2$\\$u := 0$} (2)
				(2) edge[dashed] node[below] {news!, $u := 0$} (1)
				;
				\path[every node/.style={action}]
				(2) edge[loop above] node[above] {grant?} (2)
				(3) edge[loop above] node[above] {grant?} (3)
				(3) edge[loop below, dashed] node[below] {news!} (3)
				;
		\end{tikzpicture}};
		\node[anchor=north west] () at (spec.north west) {Specification};

		\node[connector] (grant1) at (spec.-140) {};
		\node[connector] (news1) at (spec.-40) {};

		\draw[->,>=stealth',shorten >=1pt]
		(news1) edge node[action, right] {news} ++(-90:0.7)
		;
		\draw[<-,>=stealth',shorten >=1pt]
		(grant1) edge node[action,left] {grant} ++(-90:0.7)
		;
	\end{tikzpicture}
  \caption{Overall specification for a \textsf{University}.}
  \label{fig:university_spec}
\end{figure}
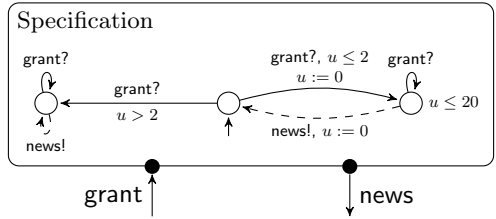

As a second property, we may want to show that the composition of arbitrary implementations conforming to respective component specification is guaranteed to satisfy some overall specification. Here Figure~\ref{fig:university_spec} provides an overall specification (essentially saying that whenever grants are given to the \textsf{University} sufficiently often then news articles are also guaranteed within a certain upper time-bound). 
Checking this property amounts to establishing a refinement between the composition of the three component specifications and the overall specification. We leave the reader in suspense until the concluding section before we reveal whether the refinement actually holds or not!




\section{Specifications and refinements}
\label{sec:specandref}

\old{
\noindent
Throughout the presentation of our specification theory, we continuously switch the mode of discussion between the semantic and syntactic levels. In general, the formal framework is developed for the semantic objects, \emph{Timed I/O Transition Systems} (TIOTSs in short)~\cite{henzinger_timed_1991}, and \new{lifted to the} syntactic constructions for \emph{Timed I/O Automata} (TIOAs), which act as a symbolic and finite representation for TIOTSs. However, it is important to emphasize that the theory for TIOTSs does not rely in any way on the TIOAs representation -- one can build TIOTSs that cannot be represented by TIOAs, and the theory remains sound for them (although we do not know how to manipulate them automatically).
}

\begin{definition}\label{def:tiots}
	A \emph{Timed Input Output Transition System} (TIOTS) is a tuple $S = (Q^S,q_0^S,\mathit{Act}^S,\rightarrow^S)$, where \new{$Q^S$} is \new{usually an infinite} set of states, $q_0\in Q$ the initial state, \new{$\mathit{Act}^S=\mathit{Act}_i^S\uplus\mathit{Act}_o^S$} a finite set of actions partitioned into inputs ($\mathit{Act}_i^S$) and outputs ($\mathit{Act}_o^S$), and $\rightarrow^S \subseteq Q^S \times (\mathit{Act}^S \cup \mathbb{R}_{\geq 0})\times Q^S$ a transition relation satisfying the following conditions: 
	\begin{itemize}[leftmargin=\parindent,align=left,labelwidth=\parindent,labelsep=0pt,itemsep=0pt]
		\item[{[time determinism]}]\ whenever $q\xlongrightarrow{d}{}^{\!\!S} q'$ and $q\xlongrightarrow{d}{}^{\!\!S} q''$, then $q' = q''$
		\item[{[time reflexivity]}]\ $q\xlongrightarrow{0}{}^{\!\!S} q$ for all $q \in Q^S$
		\item[{[time additivity]}]\ for all $q,q'' \in Q^S$ and all $d_1, d_2\in\mathbb{R}_{\geq 0}$ we have $q\xlongrightarrow{d_1+d_2}{}^{\!\!S} q''$ iff $q\xlongrightarrow{d_1}{}^{\!\!S} q'$ and $q'\xlongrightarrow{d_2}{}^{\!\!S} q''$ for some $q'\in Q^S$.
	\end{itemize}
\end{definition}

\new{We write $q\xlongrightarrow{a}{}^{\!\!S} q'$ instead of $(q,a,q')\in\rightarrow^S$ and use $i?$, $o!$, and $d$ to range over inputs, outputs, and $\mathbb{R}_{\geq 0}$, respectively.
When no confusion can arise, for example when only a single specification is given in a definition, we might drop the superscript for readability, like $Q$ instead of $Q^S$ if $S$ is the only given TIOTS.
We write $q\xlongrightarrow{a}$ to indicate that there exists a $q' \in Q$ s.t. $q\xlongrightarrow{a} q'$, and $q\arrownot\xlongrightarrow{a}$ to indicate that there does not exist $q' \in Q$ s.t. $q\xlongrightarrow{a} q'$. 
}%
\old{
In the interest of simplicity, we work with \emph{deterministic} TIOTSs: 
for all $a\in\mathit{Act}\cup\mathbb{R}_{\geq 0}$ whenever $q\xlongrightarrow{a}{}^{\!\!S} q'$ and $q\xlongrightarrow{a}{}^{\!\!S} q''$, we have $q'=q''$ (determinism is required
not only for timed transitions but also for discrete transitions).  In the
rest of the paper, we often drop the adjective `deterministic'.
Finally, the action set $\mathit{Act}^S$ is also called the alphabet.
}

For a TIOTS $S$ and a set of states $X$, we write
\begin{equation*}
	\mathrm{pred}_a^S(X) = \{q \in Q^S \mid \exists q' \in X: q\xlongrightarrow{a}{}^{\!\!S} q'\}
\end{equation*}
for the set of all $a$-predecessors of states in $X$. We write $\mathrm{ipred}^S(X)$ for the set of all input predecessors and $\mathrm{opred}^S(X)$ for all output predecessors of $X$:
\begin{equation*}
	\mathrm{ipred}^S(X) = \bigcup_{a\in\mathit{Act}_i^S} \mathrm{pred}_a^S(X)
\end{equation*}
\begin{equation*}
	\mathrm{opred}^S(X) = \bigcup_{a\in\mathit{Act}_o^S} \mathrm{pred}_a^S(X).
\end{equation*}
Furthermore, $\mathrm{post}_d^S(q)$ is the set of all time successors of a state $q$ that can be reached by delays smaller than $d$:
\begin{equation*}
	\mathrm{post}_d^S(q) = \{q' \in Q^S \mid \exists d'\in[0,d\rangle: q\xlongrightarrow{d'}{}^{\!\!S} q'\}.
\end{equation*}
\noindent \new{Note that $\mathrm{post}_d^S(q)$ is defined using the open interval $[0,d\rangle$, while our previous work used the closed interval $[0,d]$. This subtle difference is needed for its application in adversarial pruning (Section~\ref{sec:consandconj}).}

We shall now introduce a symbolic representation for TIOTSs in terms of Timed I/O Automata (TIOAs). Let $\mathit{Clk}$ be a finite set of {\em clocks}. A {\em clock valuation} over $\mathit{Clk}$ is a mapping $v \in [\mathit{Clk} \mapsto \mathbb{R}_{\geq 0}]$. We write $v + d$ to denote a valuation such that for any clock $r$ we have $(v+d)(r) = \new{v(r)}+d$. Given $d\in \mathbb{R}_{\geq 0}$ \new{and a set of clocks $c$}, we write $v[r \mapsto 0]_{r\in c}$ for a valuation which agrees with $v$ on all values for clocks not in $c$, and returns 0 for all clocks in $c$. \new{So this notation resets the clocks in $c$. For example, $\{x\mapsto 3, y\mapsto 4.5\}[r\mapsto 0]_{r\in\{x\}} \equiv \{x\mapsto 0, y\mapsto 4.5\}$.} 
A guard over $\mathit{Clk}$ is a finite \new{Boolean formula with the usual propositional connectives where clauses are} expressions of the form $x \prec n$, where $x\in \mathit{Clk}$, $\prec\in\{<,\leq,>, \geq, \new{=} \}$, and $n\in\mathbb{N}$. We write $\mathcal{B}(\mathit{Clk})$ for the set of all guards over $\mathit{Clk}$. \new{The notation $\mathbf{T}$ is used for the logical true and $\mathbf{F}$ for the logical false. The reset of a guard $q \in\mathcal{B}(\mathit{Clk})$, denoted by $g[r\mapsto 0]_{r\in c}$, is again a guard where each occurrence of clock $x\in c$ is replaced by $0$. For example $(x < 4 \wedge y > 2)[x \mapsto 0] \equiv 0 < 4 \wedge y > 2 \equiv y > 2$.}

\begin{definition}
	A \emph{Timed Input Output Automaton} (TIOA) is a tuple $A = (\mathit{Loc}, l_0, \mathit{Act}, \mathit{Clk}, E, \mathit{Inv})$ where $\mathit{Loc}$ is a finite set of locations, $l_0\in \mathit{Loc}$ the initial location, $\mathit{Act} = \mathit{Act}_i \uplus \mathit{Act}_o$ is a finite set of actions partitioned into inputs ($\mathit{Act}_i$) and outputs ($\mathit{Act}_o$), $\mathit{Clk}$ a finite set of clocks, $E\subseteq \mathit{Loc} \times \mathit{Act} \times \mathcal{B}(\mathit{Clk}) \times 2^{\mathit{Clk}} \times \mathit{Loc}$ a set of edges, and $\mathit{Inv} : \mathit{Loc} \mapsto \mathcal{B}(\mathit{Clk})$ a location invariant function.
\end{definition}

If $(l,a,\varphi,c,l')\in E$ is an edge, then $l$ is a \new{source} location, $a$ is an action label, $\varphi$ is a \new{guard} over clocks that must be satisfied when the edge is executed, $c$ is a set of clocks to be reset, and $l'$ is a target location. Examples of TIOAs have been shown in the introduction. \new{Note that, contrary to standard definitions of timed automata, guards and invariants are allowed to be a Boolean formula using all usual propositional connectives, including the disjunction. The disjunction connective naturally arises in adversarial pruning (Section~\ref{sec:consandconj}) and the quotient (Section~\ref{sec:quotient}).}

\begin{definition}\label{def:semanticTIOA}
	The \emph{semantic} of a TIOA $A = (\mathit{Loc}, l_0, \mathit{Act}, \mathit{Clk}, E, \mathit{Inv})$ is the TIOTS $\sem{A} = (\mathit{Loc}\times [\mathit{Clk}\mapsto \mathbb{R}_{\geq 0}], (l_0,\bm{0}), \allowbreak \mathit{Act}, \rightarrow)$, where $\bm{0}$ is a constant function mapping all clocks to zero, \new{$\bm{0}\models \mathit{Inv}(l_0)$}, and $\rightarrow$ is the largest transition relation generated by the following rules:
	\begin{itemize}
		\item Each $(l,a,\varphi, c, l') \in E$ gives rise to $(l, v)\xlongrightarrow{a}(l', v')$ for each clock valuation $v\in [\mathit{Clk} \mapsto \mathbb{R}_{\geq 0}]$ such that $v\models \varphi$ and $v' = v[r\mapsto 0]_{r\in c}$ and $v' \models \mathit{Inv}(l')$.
		\item Each location $l\in \mathit{Loc}$ with a valuation $v\in [\mathit{Clk} \mapsto \mathbb{R}_{\geq 0}]$ gives rise to a transition $(l,v)\xlongrightarrow{d}(l,v+d)$ for each delay $d\in\mathbb{R}_{\geq 0}$ such that $v+d \models \mathit{Inv}(l)$ \new{and $\forall d' \in\mathbb{R}_{\geq 0}, d' < d: v+d' \models \mathit{Inv}(l)$}.
	\end{itemize}
\end{definition}

\new{Compared to the definition of the semantic of TIOAs from previous works~\cite{david_timed_2010,david_methodologies_2010,david_real-time_2015}, we additionally require 1) for a delay transition that $\forall d' \in\mathbb{R}_{\geq 0}, d' < d: v+d' \models \mathit{Inv}(l)$, since guards and invariants are relaxed to Boolean formulas using conjunction \emph{and} disjunction, and 2) $\bm{0}\models \mathit{Inv}(l_0)$ to prevent an undesirable edge case for the initial state.}
Note that the TIOTSs induced by TIOAs satisfy the axioms 1--3 of Definition~\ref{def:tiots}. In order to guarantee determinism, the TIOA has to be deterministic: for each action--location pair, \new{if more than} one edge is enabled at the same time, \new{the resets and target locations need to be the same}. This is a standard check. We assume that all TIOAs below are deterministic.

Having introduced a syntactic representation for TIOTSs, we now turn back to the semantic level in order to define the basic concepts of implementation and specification.

\begin{definition}\label{def:specification}
	A TIOTS $S$ is a \emph{specification} if each of its states $q\in Q$ is input-enabled: $\forall i?\in \mathit{Act}_i : \exists q'\in Q$ s.t. $q\xlongrightarrow{i?} q'$. A TIOA $A$ is a \emph{specification automaton} if its semantic $\sem{A}$ is a specification.
\end{definition}

The assumption of input-enabledness, also seen in many interface theories~\cite{lynch_automata_1988,garland_ioa_1998,stark_process-algebraic_2003,vaandrager_relationship_1991,nicola_process_1995}, reflects our belief that an input cannot be prevented from being sent to a system, but it might be unpredictable how the system behaves after receiving it. Input-enbledness encourages explicit modeling of this unpredictability, and compositional reasoning about it; for example, deciding if an unpredictable behavior of one component induces unpredictability of the entire system.

In practice tools can interpret absent input transitions in at least two reasonable ways. First, they can be interpreted as ignored inputs, corresponding to location loops in the automaton. Second, they may be seen as unavailable (`blocking') inputs, which can be achieved by assuming implicit transitions to a designated error state. 

The role of specifications in a specification theory is to abstract, or underspecify, sets of possible implementations. \emph{Implementations} are concrete executable realizations of systems.  We will assume that implementations of timed systems have fixed timing behavior (outputs occur at predictable times) and systems can always advance either by producing an output or delaying. This is formalized using axioms of \emph{output-urgency} and \emph{independent-progress} below.

\begin{definition}\label{def:implementation}
	A specification $P = (Q, q_0, \mathit{Act},\rightarrow)$ is an \emph{implementation} if for each state $q\in Q$ we have
	\begin{itemize}[leftmargin=\parindent,align=left,labelwidth=\parindent,labelsep=0pt,itemsep=0pt]
		\item[{[output urgency]}]\ $\forall q',q''\in Q$, if $q\xlongrightarrow{o!}{}^{\!\!P} q'$ and $q\xlongrightarrow{d}{}^{\!\!P} q''$ for some $o!\in\mathit{Act}_o$ and $d\in\mathbb{R}_{\geq 0}$, then $d = 0$.
		\item[{[independent progress]}]\ either $\forall d \in \mathbb{R}_{\geq 0} : q\xlongrightarrow{d}{}^{\!\!P}$ or $\exists d\in\mathbb{R}_{\geq 0}, \exists o!\in\mathit{Act}_o$ s.t. $q\xlongrightarrow{d}{}^{\!\!P}q'$ and $q'\xlongrightarrow{o!}{}^{\!\!P}$. 
	\end{itemize}
	\new{A specification automaton $A$ is an \emph{implementation automaton} if its semantic $\sem{A}$ is an implementation.}
\end{definition}

Independent progress is one of the central properties in our theory: it states that an implementation cannot ever get stuck in a state where it is up to the environment to induce progress. \new{So in every state there either exists an ability to delay until an output is possible or the state can delay indefinitely.} An implementation cannot wait for an input from the environment without letting time pass. \new{Unfortunately, implementations might contain zeno behavior, for example, a state having an output action as a self-loop might stop time by firing this transition infinitely often. So time should be able to diverge~\cite{alfaro_element_2003}. Yet, to verify whether an implementation has time divergence, we need to analyze it in the context of an environment to form a closed-system. Environments could both ensure or prevent time to diverge, so one cannot determine time divergence by analyzing the system without an environment. In this paper, we focus on specifying components as part of a system. Therefore, we ignore the problem of time divergence for now and postpone it to future work.}

A notion of \emph{refinement} allows to compare two specifications as well as to relate an implementation to a specification. Refinement should \new{be a pre-congruence when we compose several specifications of a system together.} This is formalized with Theorem~\ref{thm:precongruence} in Section~\ref{sec:compandcomp}.

\begin{figure}
	\centering
	\begin{tikzpicture}
		\draw (0,0) -- (0,2.4) -- (1,2.4) -- (1,0);
		
		\node[scale=1.5] () at (0.5,0.4) {$\xlongrightarrow{d}$};
		\node[scale=1.5] () at (0.5,1.2) {$\xlongrightarrow{o!}$};
		\node[scale=1.5] () at (0.5,2.0) {$\xlongleftarrow{i?}$};
		\node[scale=1.5,anchor=south] () at (0.5,2.5) {$\leq$};
		
		\node[scale=2] () at (-0.5,1.2) {$S$};
		\node[scale=2] () at (1.5,1.2) {$T$};
	\end{tikzpicture}
	\caption{\new{Visual representation of the simulation relation defined by refinement.}}
	\label{fig:ref}
\end{figure}
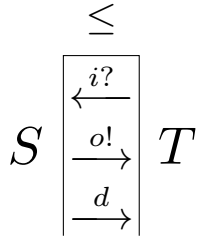

We study these kind of properties in later sections. It is well known from the literature~\cite{alfaro_interface_2001,alfaro_interface-based_2005,Bulychev_efficient_2009} that in order to give these kind of guarantees a refinement should have the flavour of {\em alternating (timed) simulation}~\cite{alur_alternating_1998}. \new{Figure~\ref{fig:ref} shows a visual representation of the direction of the simulation relation captures by refinement. While it is typical to define simulation relations on transitions systems that have equal alphabet, we relaxed that in our definition of refinement below. Then it fits the main theorem of quotient in Section~\ref{sec:quotient} and it matches the usage in practical examples, see for example the university example in this paper.}

\begin{definition}\label{def:refinement}
	Given specifications $S = (Q^S,q_0^S,\mathit{Act}^S,\rightarrow^S)$ and $T = (Q^T,q_0^T,\mathit{Act}^T,\rightarrow^T)$ \new{where $\mathit{Act}_i^S \cap \mathit{Act}_o^T =\emptyset$, $\mathit{Act}_o^S \cap \mathit{Act}_i^T =\emptyset$, $\mathit{Act}_i^S \subseteq \mathit{Act}_i^T$, and $\mathit{Act}_o^T\subseteq \mathit{Act}_o^S$}. \emph{$S$ refines $T$}, denoted by $S\leq T$, iff there exists a binary relation $R\subseteq Q^S\times Q^T$ such that $(q_0^S,q_0^T) \in R$ and for each pair of states $(s,t)\in R$ it holds that
	\begin{enumerate}[itemsep=0pt]
		\item Whenever $t\xlongrightarrow{i?}{}^{\!\!T} t'$ for some $t'\in Q^T$ and $i?\in\mathit{Act}_i^T \cap \mathit{Act}_i^S$, then $s\xlongrightarrow{i?}{}^{\!\!S}s'$ and $(s',t')\in R$ for some $s'\in Q^S$
		\new{\item Whenever $t\xlongrightarrow{i?}{}^{\!\!T} t'$ for some $t'\in Q^T$ and $i?\in\mathit{Act}_i^T \setminus \mathit{Act}_i^S$, then $(s,t')\in R$}
		\item Whenever $s\xlongrightarrow{o!}{}^{\!\!S} s'$ for some $s'\in Q^S$ and $o!\in\mathit{Act}_o^S \cap \mathit{Act}_o^T$, then $t\xlongrightarrow{o!}{}^{\!\!T}t'$ and $(s',t')\in R$ for some $t'\in Q^T$
		\new{\item Whenever $s\xlongrightarrow{o!}{}^{\!\!S} s'$ for some $s'\in Q^S$ and $o!\in\mathit{Act}_o^S \setminus \mathit{Act}_o^T$, then $(s',t)\in R$}
		\item Whenever $s\xlongrightarrow{d}{}^{\!\!S} s'$ for some $s'\in Q^S$ and $d\in \mathbb{R}_{\geq 0}$, then $t\xlongrightarrow{d}{}^{\!\!T}t'$ and $(s',t')\in R$ for some $t'\in Q^T$
	\end{enumerate}
	A specification automaton $A$ \emph{refines} another specification automaton $B$, denoted by $A\leq B$, iff $\sem{A} \leq \sem{B}$.
\end{definition}

It is easy to see that the refinement is reflexive. \new{Refinement is only transitive under specific conditions. These conditions are captured in Lemma~\ref{lemma:refinementtransitive}. A special case satisfying these conditions is when the action sets of all specifications are the same.}  Refinement can be checked for specification automata by reducing the problem to a specific refinement game, and using a symbolic representation to reason about it. Figure~\ref{fig:machine2} shows a coffee machine that is a refinement of the one in Figure~\ref{fig:university}. It has been refined in two ways: one output transition has been completely dropped and one state invariant has been tightened.

\begin{lemma}\label{lemma:refinementtransitive}
	\new{Given specifications $S^i = (Q^i,q_0^i,\mathit{Act}^i,\rightarrow^i)$ with $i\in \{1,2,3\}$. If $S^1\leq S^2$, $S^2\leq S^3$, $\mathit{Act}_i^1 \cap \mathit{Act}_o^3 =\emptyset$, and $\mathit{Act}_o^1 \cap \mathit{Act}_i^3 =\emptyset$, then $S^1 \leq S^3$.}
\end{lemma}
\begin{proof}
	
	\new{
	($\Rightarrow$)
	We first show that the action sets of $S^1$ and $S^3$ satisfy the conditions of refinement. From $S^1\leq S^2$ it follows that $\mathit{Act}_i^1 \subseteq \mathit{Act}_i^2$, and $\mathit{Act}_o^2\subseteq \mathit{Act}_o^1$; similarly, from $S^2\leq S^3$ it follows that $\mathit{Act}_i^2 \subseteq \mathit{Act}_i^3$, and $\mathit{Act}_o^3\subseteq \mathit{Act}_o^2$. Combining this results in $\mathit{Act}_i^1 \subseteq \mathit{Act}_i^3$, and $\mathit{Act}_o^3\subseteq \mathit{Act}_o^1$. Together with the antecedent and Definition~\ref{def:refinement} of refinement we can conclude that action sets of $S^1$ and $S^3$ satisfy the conditions of refinement.
	}

	\new{
	It remains to show that there exists a relation $R^{13}$ witnessing $S^1 \leq S^3$. Let $R^{12}$ and $R^{23}$ the relations witnessing $S^1 \leq S^2$ and $S^2\leq S^3$, respectively. Using a standard co-inductive argument it can be shown that
	\begin{equation*}
		R^{13} = \left\{ (q^1, q^3) \in R^{13} \mid \exists q^2 \in Q^2: (q^1, q^2)\in R^{12} \wedge (q^2, q^3)\in R^{23} \right\}
	\end{equation*}
	witnesses $S^1 \leq S^3$.
	}

\end{proof}

\begin{figure}
  \centering
	
	\begin{tikzpicture}[->,>=stealth',shorten >=1pt, font=\small, scale=1, transform shape]
		\node[draw, shape=rectangle, rounded corners] (machine) [right = 2cm of adm] {\begin{tikzpicture}[->,>=stealth',shorten >=1pt,align=left,node distance=3cm, scale=.7]
				\node[main node, initial, initial text={}, initial where=left] (1) {};
				\node[main node, label={below right:$y \leq 5$}] (2) [below = of 1] {};
				
				\path[every node/.style={action}]
				(1) edge node[left,pos=.5] {coin?} node[left,pos=.65] {$y := 0$} (2)
				;
				\draw[dashed, rounded corners, every node/.style={action}]
				(2) -- node[above] {cof!} node[below] {$y \geq 4$} ++(180:3) -- (1)
				;
				\path[every node/.style={action}]
				(1) edge[loop right, dashed] node[right] {tea!\\$y\geq 2$} (1)
				(2) edge[loop below] node[below] {coin?} (2)
				;
		\end{tikzpicture}};
		\node[anchor=north west] () at (machine.north west) {Machine};
		
		\node[connector] (coin) at (machine.-120) {};
		\node[connector] (cof) at (machine.-90) {};
		\node[connector] (tea) at (machine.-60) {};
		
		\draw[->,>=stealth',shorten >=1pt] 
		(cof) edge node[action, right] {cof} ++(-90:0.7)
		(tea) edge node[action, right] {tea} ++(-90:0.7)
		;
		\draw[<-,>=stealth',shorten >=1pt]
		(coin) edge node[action,left] {coin} ++(-90:0.7)
		;
	\end{tikzpicture}
  \caption{A coffee machine specification that refines the coffee machine in Figure~\ref{fig:university}.}
  \label{fig:machine2}
\end{figure}
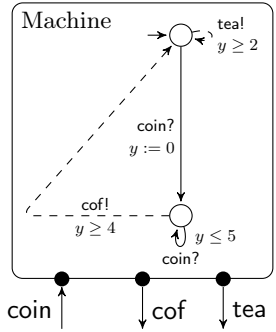

Since our implementations are a subclass of specifications, we simply use \emph{refinement} as an implementation relation.
\begin{definition}\label{def:satisfaction} 
    An implementation $P$ satisfies a specification $S$ iff $P \leq S$. We write $\mod{S}$ for the set of all implementations of $S$, so $\mod{S} = \{ P \mid P \sat S \}$.  
\end{definition}

From a logical perspective, specifications are like formulae, and implementations are their models. This analogy leads us to a classical notion of consistency as existence of models.

\begin{definition}\label{def:consistency}
	A specification $S$ is \emph{consistent} iff there exists an implementation $P$ such that $P\leq S$. \new{A specification automaton $A$ is \emph{consistent} iff its semantic $\sem{A}$ is consistent.}
\end{definition}

All specification automata in Figure~\ref{fig:university} are consistent. An example of an inconsistent specification can be found in Figure~\ref{fig:inconsistent}. Notice that the invariant in the second state ($x \leq 4$) is stronger than the guard ($x \geq 5$) on the \textsf{cof} edge. \new{This location violates the independent progress property.}

\begin{figure}
  \centering

	\begin{tikzpicture}[->,>=stealth',shorten >=1pt, font=\small, scale=1, transform shape]
		\node[draw, shape=rectangle, rounded corners] (machine) [right = 2cm of adm] {\begin{tikzpicture}[->,>=stealth',shorten >=1pt,align=left,node distance=3cm, scale=.7]
				\node[main node, initial, initial text={}, initial where=left] (1) {};
				\node[main node, label={above:$x \leq 4$}] (2) [right = of 1] {};
				\node [above=.7cm of 1] {}; 
				
				\path[bend left=15,every node/.style={action}]
				(1) edge node[above,pos=.5] {coin?} (2)
				(2) edge[dashed] node[below] {cof!, $x \geq 5$} (1)
				;
				\path[every node/.style={action}]
				(2) edge[loop below] node[below] {coin?} (2)
				;
		\end{tikzpicture}};
		\node[anchor=north west] () at (machine.north west) {Inconsistent};
		
		\node[connector] (coin) at (machine.-140) {};
		\node[connector] (cof) at (machine.-90) {};
		\node[connector] (tea) at (machine.-40) {};
		
		\draw[->,>=stealth',shorten >=1pt] 
		(cof) edge node[action, right] {cof} ++(-90:0.7)
		(tea) edge node[action, right] {tea} ++(-90:0.7)
		;
		\draw[<-,>=stealth',shorten >=1pt]
		(coin) edge node[action,left] {coin} ++(-90:0.7)
		;
	\end{tikzpicture}
  \caption{An inconsistent specification.}
  \label{fig:inconsistent}
\end{figure}
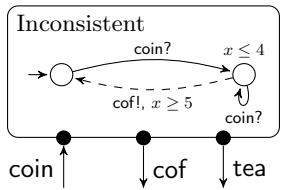

We also define a stricter, more syntactic, notion of consistency, which requires that all states are consistent.

\begin{definition}\label{def:local-consistency}
	A specification $S$ is \emph{locally consistent} iff every state $s\in Q$ allows independent progress. \new{A specification automaton $A$ is \emph{locally consistent} iff its semantic $\sem{A}$ is locally consistent.}
\end{definition}

\begin{theorem}\label{thm:local-consistency} 
    Every locally consistent specification is consistent in the sense of Definition~\ref{def:consistency}.  
\end{theorem}
\noindent The proof of Theorem~\ref{thm:local-consistency} can be found in Appendix~\ref{app:proofs-spec}.

The opposite implication in the theorem does not hold as we shall see later. Local consistency, or independent progress, can be checked for specification automata establishing local consistency for the syntactic representation. Technically it suffices to check for each location that if the supremum of all solutions of every location invariant exists then it satisfies the invariant itself and allows at least one enabled output transition.

Prior specification theories for discrete time~\cite{larsen_modal_1989} and probabilistic~\cite{caillaud_compositional_2009} systems reveal two main requirements for a definition of implementation. These are the same requirements that are typically imposed on a definition of a model as a special case of a logical formula. First, implementations should be consistent specifications (logically, models correspond to some consistent formulae). Second, implementations should be most specified (models cannot be refined by non-models), as opposed to proper specifications, which should be \emph{underspecified}. For example, in propositional logics, a model is represented as a complete consistent term. Any implicant of such a term is also a model (in propositional logics, it is actually equivalent to it).

Our definition of implementation satisfies both requirements, and to the best of our knowledge, it is the first example of a proper notion of implementation for timed specifications. As the refinement is reflexive we get $P\sat P$ for any implementation and thus each implementation is consistent as per Definition~\ref{def:consistency}. Furthermore, each implementation cannot be refined anymore by any underspecified specifications.
\begin{theorem}\label{thm:implementations-minimum}
	Any locally consistent specification $S$ refining an implementation $P$ is an implementation as per Definition~\ref{def:implementation}.
\end{theorem}
\noindent The proof of Theorem~\ref{thm:implementations-minimum} can be found in Appendix~\ref{app:proofs-spec}.

We conclude the section with the first major theorem. Observe that every preorder $\preceq$ is intrinsically complete in the following sense: $S\preceq T$ iff for every smaller element $P \preceq S$ also $P \preceq T$. This means that a refinement of two specifications coincides with inclusion of sets of all the specifications refining each of them. However, since out of all specifications only the implementations correspond to real world objects, another completeness question is more relevant: does the refinement coincide with the inclusion of implementation sets? This property, which does not hold for any preorder in general, turns out to hold for our refinement.
\begin{theorem} \label{thm:completeness} 
    For any two locally consistent specifications $S$, $T$ \new{having the same action set} we have that $S \leq T$ iff $\mod{S} \subseteq \mod{T}$.
\end{theorem}
\begin{proof}
    ($\Rightarrow$)
    Assume existence of relations $R_1$ and $R_2$ witnessing satisfaction of $S$ by \new{the implementation} $P$ and refinement of $T$ by $S$, respectively. Use a standard co-inductive argument \new{and Lemma~\ref{lemma:refinementtransitive}} to show that
    \begin{equation*}
        R = \left\{ (p,t) \in Q^P \times Q^T \mid \exists s\in Q^S: (p,s)\in R_1 \wedge (s,t)\in R_2 \right\}
    \end{equation*}
    is a relation witnessing satisfaction of $T$ by $P$. Also observe that $(p_0, t_0)\in R$.
    
    ($\Leftarrow$)
    In the following we write $p\sat s$ for states $p$ and $s$ meaning that there exists a relation $R'$ witnessing $P\sat S$ that contains $(p,s)$.
    
    We construct a binary relation $R\subseteq Q^S\times Q^T$:
    \begin{equation*}
        R = \left\{ (s,t)  \mid \forall P: p_0 \sat s \implies p_0 \sat t \right\},
    \end{equation*}
    where $p_0$ is the initial state of $P$. We shall argue that $R$ witnesses $S \le T$. Consider a pair $(s,t)\in R$. There are two cases to be considered.
    \begin{itemize}
        \item \new{Consider} any input $i?$. \new{Due to input-enabledness,} there exists $t'\in Q^T$ such that $t\xlongrightarrow{i?}{}^{\!\!T} t'$. We need to show existence of a state $s'\in Q^S$ such that $s\xlongrightarrow{i?}{}^{\!\!S} s'$ and $(s',t')\in R$, \new{so $\forall P: p_0 \sat s' \implies p_0 \sat t'$}.
        
        Due to input-enabledness, for the same $i?$ there exists a state $s'\in Q^S$ such that $s\xlongrightarrow{i?}{}^{\!\!S} s'$. We need to show that $(s',t')\in R$. By Theorem~\ref{thm:local-consistency} \new{applied to $Q^S$} we have that there exists an implementation $P$ and its state $p_0\in Q^P$ such that $p_0\sat s'$ (technically speaking $s$ may not be an initial state of $S$, but we can consider a version of $S$ with initial state changed to $s$ to apply Theorem~\ref{thm:local-consistency}, concluding existence of an implementation).
        
        Consider an arbitrary implementation $Q\sat S$ and its state $q_0\in Q^Q$ such that $q_0\sat s'$. We need to show that also $q_0\sat t'$. \new{We do this by extending $Q$ deterministically to a model of $s$, showing that this is also a model of $t$, and then arguing that the only $i?$ successor state models $t'$.}
        Create an implementation $Q'$ by merging $Q$ and $P$ above and adding a fresh state $q$ with transition $q\xlongrightarrow{i?}{}^{\!\!Q'} q_0$ and transitions $q\xlongrightarrow{j?}{}^{\!\!Q'} p_0$ for all $j?\neq i?$, $j?\in\mathit{Act}_i$\footnote{\new{State $q$ allows independent progress if you combine the construction of $q$ with the second case for action $a$.}}. Now $q \sat s$ as \new{$q\xlongrightarrow{i?}{}^{\!\!Q'} q_0$ with $q_0\sat s'$ and $q\xlongrightarrow{j?}{}^{\!\!Q'} p_0$ with $p_0\sat s'$ for $j?\ne i?$}. By assumption, every implementation of $S$ is also an implementation of $T$, so $q\sat t$ and consequently $q_0\sat t'$ as $q$ is deterministic on $i?$. Summarizing, for any implementation $q_0\sat s'$ we are able to argue that $q_0\sat t'$, thus necessarily $(s',t')\in R$.
        
        \item Consider any action $a$ (which is an output or a delay) for which there exists $s'$ such that $s\xlongrightarrow{a}{}^{\!\!S} s'$. Using a construction similar to the one above it is not hard to see that one can actually construct (and thus postulate existence of) an implementation $P$ containing $p\in Q^P$ such that $p\sat s$ that has a transition $p\xlongrightarrow{a}{}^{\!\!P} p'$. Since also $p\sat t$, we have that there exists $t'\in Q^T$ such that $t\xlongrightarrow{a}{}^{\!\!T} t'$. It remains to argue that $(s',t')\in R$. This is done in the same way as with the first case, by considering any model of $s'$, then by extending it deterministically to a model of $s$, concluding that it is now a model of $t$ and the only $a$-derivative, which is $p'$, must be a model of $t'$. Consequently $(s',t')\in R$.  
    \end{itemize}
    It follows directly from the definition of $R$ with $\sem{S}\subseteq\sem{T}$ that $(s_0,t_0)\in R$.
\end{proof}

The restriction of the theorem to locally consistent specifications is not a serious one. As we shall see \new{in Theorem~\ref{thm:prune}}, any consistent specification can be transformed into a locally consistent one preserving the set of implementations. 

\section{Consistency and conjunction}
\label{sec:consandconj}

An \emph{immediate error} occurs in a state of a specification if the specification disallows progress of time and output transitions in a given state -- such a specification will break if the environment does not send an input. For a specification $S$ we define the set of immediate error states $\mathrm{imerr}$ as follows\footnote{In our previous work~\cite{david_timed_2010,david_methodologies_2010,david_real-time_2015} immediate error states were represented by $\mathrm{err}$, a symbol we have re-purposed in this work, see Definition~\ref{def:error}.}.

\begin{figure}
	\centering
	\begin{tikzpicture}[->,>=stealth',shorten >=1pt,align=left,node distance=3cm, scale=.7]
			\node[main node, initial, initial text={}, initial where=left, label={above:$q_1$}] (1) {};
			\node[main node, label={above:$q_2$}, label={below:$x\leq 0$}] (2) [right = of 1] {};
			\node[main node, label={above:$q_3$}, label={below:$x \leq 0$}] (3) [right = of 2] {};
			
			\path[dashed, every node/.style={action}]
			(1) edge node[above] {$a!$} node[below] {$x := 0$} (2)
			(2) edge node[above] {$a!$} node[below] {$x := 0$} (3)
			;
	\end{tikzpicture}
	\caption{\new{Example of a specification that illustrates difference between immediate error states ($q_3$) and error states ($q_2$ and $q_3$).}}
	\label{fig:inducederrorexample}
\end{figure}
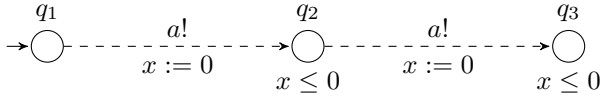

\begin{definition}\label{def:immediateerror}
	Given a specification $S = (Q,q_0,\mathit{Act},\rightarrow)$, the set of \emph{immediate error} states, denoted by $\mathrm{err}$, is defined as
	\begin{align*}
		\mathrm{imerr} = \left\{\vphantom{\xlongrightarrow{d}} q \in Q \mid \right. &(\exists d\in\mathbb{R}_{\geq 0}: q\arrownot\xlongrightarrow{d})\ \wedge\\
		&\left. \forall d\in\mathbb{R}_{\geq 0}\forall o!\in\mathit{Act}_o\forall q'\in Q: q\xlongrightarrow{d} q' \Rightarrow q'\arrownot\xlongrightarrow{o!}\right\}.
	\end{align*}
\end{definition}
\noindent It follows that no immediate error states can occur in implementations, or in locally consistent specifications. \new{In Figure~\ref{fig:inducederrorexample} state $q_3$ is an immediate error state, as it does not allow independent progress. Error states can also be created when output actions are disabled, for example by pruning away immediate error states, see Definition~\ref{def:adversarialpruning} below. Therefore, we extend the definition of immediate error states into error states $\mathrm{err}$ as follows, which was not done in our previous work.}
\begin{definition}\label{def:error}
	\new{Given a specification $S = (Q,q_0,\mathit{Act},\rightarrow)$ and a set of states $X \subseteq Q$, the set of \emph{error} states, denoted by $\mathrm{err}$, is defined as
	\begin{align*}
		\mathrm{err}(X) = \left\{ \vphantom{\xlongrightarrow{d}} q \in Q \mid \right. &(\exists d\in\mathbb{R}_{\geq 0}: q\arrownot\xlongrightarrow{d}) \wedge \forall d\in\mathbb{R}_{\geq 0}\forall o!\in\mathit{Act}_o\forall q'\in Q: \\ 
		&\left. q\xlongrightarrow{d} q' \Rightarrow (q'\arrownot\xlongrightarrow{o!} \vee \forall q'' \in Q: q'\xlongrightarrow{o!} q'' \Rightarrow q''\in X) \right\}.
	\end{align*}}
\end{definition}
\noindent \new{Note that $\mathrm{err}(\emptyset) = \mathrm{imerr}$, thus for any $X$ we have that $\mathrm{imerr} \subseteq \mathrm{err}(X)$. For the example in Figure~\ref{fig:inducederrorexample} we have that $\mathrm{err}(\mathrm{imerr}) = \{q_2, q_3\}$. State $q_2$ is an error state as all outgoing output transitions go to immediate error states and time cannot progress indefinitely. Thus, while $q_2$ allows independent progress in the current form of the specification, disabling all transitions going to immediate error states (something that adversarial pruning will do) will make $q_2$ violate independent progress.}

\begin{figure}
  \centering

	\begin{tikzpicture}[->,>=stealth',shorten >=1pt, font=\small, scale=1, transform shape]
		\node[draw, shape=rectangle, rounded corners] (machine) [right = 2cm of adm] {\begin{tikzpicture}[->,>=stealth',shorten >=1pt,align=left,node distance=3cm, scale=.7]
				\node[main node, initial, initial text={}, initial where=left] (1) {};
				\node[main node, label={below right:$y \leq 6$}] (2) [below = of 1] {};
				\node[main node, label={right:$y \leq 0$}] (3) [below right = 1cm and 1cm of 1] {};
				\node [above=.5cm of 1] {}; 
				
				\path[every node/.style={action}]
				(1) edge node[left,pos=.5] {coin?} node[left,pos=.65] {$y := 0$} (2)
				;
				\draw[dashed, rounded corners, every node/.style={action}]
				(2) -- node[above] {cof!} node[below] {$y \geq 4$} ++(180:3) -- (1)
				;
				\draw[dashed, rounded corners, every node/.style={action}]
				(2) -- node[above] {tea!} ++(0:3) -- node[left,pos=.4] {$y := 0$} (3)
				;
				\path[every node/.style={action}]
				(2) edge[loop below] node[below] {coin?} (2)
				(3) edge[loop above] node[above] {coin?} (3)
				;
		\end{tikzpicture}};
		\node[anchor=north west] () at (machine.north west) {Partially Inconsistent};
		
		\node[connector] (coin) at (machine.-124) {};
		\node[connector] (cof) at (machine.-90) {};
		\node[connector] (tea) at (machine.-55) {};
		
		\draw[->,>=stealth',shorten >=1pt] 
		(cof) edge node[action, right] {cof} ++(-90:0.7)
		(tea) edge node[action, right] {tea} ++(-90:0.7)
		;
		\draw[<-,>=stealth',shorten >=1pt]
		(coin) edge node[action,left] {coin} ++(-90:0.7)
		;
	\end{tikzpicture}
  \caption{A partially inconsistent specification.}
  \label{fig:inconsistentmachine}
\end{figure}

In general, error states in a specification do not necessarily mean that a specification cannot be implemented. Figure~\ref{fig:inconsistentmachine} shows a partially inconsistent specification, a version of the coffee machine that becomes inconsistent if it ever outputs tea. The inconsistency can be possibly avoided by some implementations, who would not implement delay or output transitions leading to it. More precisely an implementation will exist if there is a strategy for the output player in a safety game to avoid $\mathrm{err}$. In order to be able to build on existing formalizations~\cite{cassez_efficident_2005} we will consider a dual reachability game, asking for a strategy of the input player to reach $\mathrm{err}$. We first define a timed predecessor operator~\cite{alfaro_symbolic_2001,maler_synthesis_1995,cassez_efficident_2005}, which gives all the states that can delay into $X$ while avoiding $Y$:
\begin{equation*}
	\mathrm{Pred}_t^S(X,Y) = \left\{q \in Q^S \mid \exists d\in\mathbb{R}_{\geq 0} \wedge \exists q'\in X \mbox{ s.t. } q\xlongrightarrow{d}{}^{\!\!S} q' \wedge \mathrm{post}_d^S(q)\subseteq \overline{Y} \right\}.
\end{equation*}
\new{Since $\mathrm{post}_d^S(q)$ is defined on an open interval, we have that $X \cap Y \subseteq \mathrm{Pred}_t^S(X,Y)$. This means that the input player has priority over the output player when both could do an action from a state. For example, if a certain state has an outgoing input action going to an error state and an outgoing output action to a non-error state, the output action cannot disable the input action, thus the error state is still reachable.} The controllable predecessors operator, denoted by $\pi^S(X)$, which extends the set of states that can reach an error state uncontrollably, is defined by
\begin{equation*}
	\pi^S(X) = \mathrm{err}^S(X) \cup \mathrm{Pred}_t^S(X \cup \mathrm{ipred}^S(X), \mathrm{opred}^S(\overline{X})).
\end{equation*}
The set of all inconsistent states $\new{\mathrm{incons}^S}\subseteq Q^S$ of specification $S$ (i.e.\ the states for which the environment has a winning strategy \new{for reaching an error state}) is defined as the least fixpoint of $\pi^S$: \(\mathrm{incons}^S = \pi^S(\mathrm{incons}^S)\), which is guaranteed to exist by monotonicity of $\pi^S$ and completeness of the powerset lattice due to the theorem of Knaster and Tarski~\cite{tarski_lattice-theoretical_1955}. For transitions systems enjoying finite symbolic representations, automata specifications included, the fixpoint computation converges after a finite number of iterations~\cite{cassez_efficident_2005}.

Now we define the set of consistent states, $\mathrm{cons}^S$, simply as the complement of $\mathrm{incons}^S$, i.e.\ $\mathrm{cons}^S = \overline{\mathrm{incons}^S}$. We obtain it by complementing the result of the above fixpoint computation for $\mathrm{incons}^S$. For the purpose of proofs it is more convenient to formalize the dual operator, say $\Theta^S$, whose \emph{greatest} fixpoint directly and equivalently characterizes $\mathrm{cons}^S$. While least fixpoints are convenient for implementation of on-the-fly algorithms, characterizations with greatest fixpoint are useful in proofs as they allow use of coinduction. Unlike induction on the number of iterations, coinduction is a sound proof principle without assuming finite symbolic representation for the transition system (and thus finite convergence of the fixpoint computation). 
\new{
We define $\Theta^S$ as
\begin{align*}
	\Theta^S(X) = \overline{\mathrm{err}^S(\overline{X})} \cap \left\{\vphantom{\xlongrightarrow{d}} q\in Q^S \mid \forall d \geq 0: \right.&[\forall q'\in Q^S: q\xlongrightarrow{d}{}^{\!\!S} q' \Rightarrow q'\in X\ \wedge \\ 
	&\forall i?\in \mathit{Act}_i^S: \exists q''\in X: q'\xlongrightarrow{i?}{}^{\!\!S} q'']\ \vee \\ 
	& [\exists d'\leq d \wedge \exists q',q''\in X \wedge \exists o!\in\mathit{Act}_o^S: \\
	& q\xlongrightarrow{d'}{}^{\!\!S} q' \wedge q'\xlongrightarrow{o!}{}^{\!\!S} q'' \wedge \\
	& \left. \forall i?\in\mathit{Act}_i^S: \exists q'''\in X: q'\xlongrightarrow{i?}{}^{\!\!S} q'''] \right\},
\end{align*}
so the greatest fixpoint becomes $\mathrm{cons}^S = \Theta^S(\mathrm{cons}^S)$.
}

\begin{theorem}\label{thm:consistency}
    A specification \(S = (Q,s_0,\mathit{Act},\rightarrow)\) is consistent iff $s_0\in\mathrm{cons}$.
\end{theorem}
\noindent The proof of Theorem~\ref{thm:consistency} can be found in Appendix~\ref{app:proofs-conj}. The set of (in)consistent states can be computed for timed games, and thus for specification automata, using controller synthesis algorithms~\cite{cassez_efficident_2005}. 

The inconsistent states can be pruned from a consistent $S$ leading to a locally consistent specification. Adversarial pruning is applied in practice to decrease the size of specifications.
\begin{definition}\label{def:adversarialpruning}
	Given a specification $S = (Q,q_0,\mathit{Act},\rightarrow)$, the result of \emph{adversarial pruning}, denoted by $S^{\Delta}$, is specification $(\mathrm{cons},q_0,\mathit{Act},\rightarrow^{\Delta})$ where $\rightarrow^{\Delta} = \rightarrow \cap (\mathrm{cons} \times (\mathit{Act}\cup\mathbb{R}_{\geq 0}) \times \mathrm{cons})$.
\end{definition}

For specification automata adversarial pruning is realized by applying a controller synthesis algorithm, obtaining a maximum winning strategy, which is then presented as a specification automaton itself. \new{Theorem~\ref{thm:prune} captures the main result of adversarial pruning. It also explains the reason of the name of adversarial pruning: the pruned specification contains all winning strategies independently of an environment, including those that are adversarial. This contrasts with cooperative pruning, which we define in Section~\ref{sec:compandcomp} later in the paper.} 


\begin{theorem}\label{thm:prune} 
    For a consistent specification $S$, $S^{\Delta}$ is locally consistent and $\mod{S} = \mod{S^{\Delta}}$.
\end{theorem}
\noindent The proof of Theorem~\ref{thm:prune} can be found in Appendix~\ref{app:proofs-conj}.

Consistency guarantees realizability of a single specification. It is of further interest whether several specifications can be \emph{simultaneously} met by the same component, without reaching error states of any of them. We formalize this notion by defining a logical conjunction for specifications.

\begin{definition}\label{def:conjunctionTIOTS}
	Given two \new{TIOTSs} $S^i = (Q^i,q_0^i,\mathit{Act}^i,\rightarrow^i), i=1,2$ where \new{$\mathit{Act}_i^1 \cap \mathit{Act}_o^2 =\emptyset \wedge \mathit{Act}_o^1 \cap \mathit{Act}_i^2 =\emptyset$}, the \new{\emph{conjunction}} of $S^1$ and $S^2$, denoted by $S^1 \wedge S^2$, is TIOTS $(Q^1\times Q^2, (q_0^1,q_0^2),\mathit{Act},\rightarrow)$ where $\mathit{Act} = \mathit{Act}_i \uplus \mathit{Act}_o$ with $\mathit{Act}_i = \mathit{Act}_i^1 \cup \mathit{Act}_i^2$ and $\mathit{Act}_o= \mathit{Act}_o^1 \cup \mathit{Act}_o^2$, and $\rightarrow$ is defined as
	\begin{itemize}
		\item $(q_1^1,q_1^2)\xlongrightarrow{a} (q_2^1,q_2^2)$ if $a\in\mathit{Act}^1\cap\mathit{Act}^2$, $q_1^1\xlongrightarrow{a}{}^{\!\! 1} q_2^1$, and $q_1^2\xlongrightarrow{a}{}^{\!\! 2} q_2^2$
		\new{\item $(q_1^1,q^2)\xlongrightarrow{a} (q_2^1,q^2)$ if $a\in\mathit{Act}^1\setminus\mathit{Act}^2$, $q_1^1\xlongrightarrow{a}{}^{\!\! 1} q_2^1$, and $q^2\in Q^2$
		\item $(q^1,q_1^2)\xlongrightarrow{a} (q^1,q_2^2)$ if $a\in\mathit{Act}^2\setminus\mathit{Act}^1$, $q_1^2\xlongrightarrow{a}{}^{\!\! 2} q_2^2$, and $q^1\in Q^1$
		}
		\item $(q_1^1,q^2)\xlongrightarrow{d} (q_2^1,q^2)$ if $d \in \mathbb{R}_{\geq 0}$, $q_1^1\xlongrightarrow{d}{}^{\!\! 1} q_2^1$, and $q_1^2\xlongrightarrow{d}{}^{\!\! 2} q_2^2$
	\end{itemize}
\end{definition}
\noindent \new{Compared to definitions from previous work, we 1) define the old product operator to be directly the conjunction operator, thereby eliminating the build-in adversarial pruning like in previous work, because adversarial pruning does not distribute over parallel composition, see the upcoming discussion in Section~\ref{sec:compandcomp}, and 2) relax the definition to specifications with unequal alphabets.}

In general, a result of the conjunction may be locally inconsistent, or even inconsistent. To guarantee consistency, one could apply a consistency check to the result, checking if $(s_0,t_0)\in\mathrm{cons}^{S\times T}$ and, possibly, adversarially pruning the inconsistent parts.
Clearly conjunction is commutative \new{and associative}. 

\begin{lemma}\label{lem:common-implementation}
    For two specifications $S$, $T$, and their states $s$ and $t$, respectively, if there exists an implementation $P$ and its state $p$ such that simultaneously $p\sat s$ and $p\sat t$ then $(s,t)\in\mathrm{cons}^{S\wedge T}$.  
\end{lemma}
\noindent The proof of Lemma~\ref{lem:common-implementation} can be found in Appendix~\ref{app:proofs-conj}.

\begin{theorem}\label{thm:conjunction}
    For any locally consistent specifications $S$, $T$  and $U$ over the same alphabet:
    \begin{enumerate}
        \item $S \land T \leq S$ and $S \land T \leq T$
        \item $(U\leq S)$ and $(U\leq T)$ implies $U\leq (S \land T)$
        \item $\mod{S \land T} = \mod{S} \cap \mod{T}$
    \end{enumerate}
\end{theorem}
\begin{proof}
We will prove the three items separately.
    \begin{enumerate}
        \item We will prove that $S\land T$ refines $S$ (the other refinement is entirely symmetric).

        
        Let $S\land T = (Q^S \times Q^T, (s_0,t_0),\mathit{Act},\rightarrow)$ constructed according to the definition of conjunction. We abbreviate the set of states of $S\land T$ as $Q^{S \land T}$. It is easy to see that the following relation on states of $S\land T$ and states of $T$ witnesses refinement of \new{$S$} by $S\land T$:
        \begin{equation*} 
            R = \left\{ ((s_1,t), s_2) \in Q^{S \land T} \times Q^S \mid s_1 = s_2 \right\} 
        \end{equation*}
        The argument is standard, and it takes into account that $Q^{S \land T}=\mathrm{cons}^{S\wedge T}$ is a fixpoint of $\Theta$. How $\Theta$ is taken into account is demonstrated in more detail in the proof for the next item.
        
        \item Assume that $U\leq S$ and $U\leq T$. Then $U \leq S\land T$. The first refinement is witnessed by some relation $R_1$, the second refinement by $R_2$. Then the third refinement is witnessed by the following relation $R\subseteq Q^U \times Q^{S\land T}$:
        \begin{equation*}
            R = \left \{ (u,(s,t)) \in Q^U \times \mathrm{cons}^{S\wedge T} \mid (u,s) \in R_1 \land (u,t) \in R_2 \right\}.
        \end{equation*}
        The argument that $R$ is a refinement is standard again, relying on the fact that $\mathrm{cons}^{S\wedge T}$ is a fixed point of $\Theta$. 
        
        Consider an output case when $u\xlongrightarrow{o!}{}^{\!\!U} u'$ for some output $o!$ and the target state $u'$. Then $s\xlongrightarrow{o!}{}^{\!\!S} s'$ and $t\xlongrightarrow{o!}{}^{\!\!T} t'$ for some states $s'$ and $t'$ and $(u',s')\in R_1$ and $(u',t')\in R_2$. This means that $(s,t)\xlongrightarrow{o!}{}^{\!\!S\wedge T} (s', t')$. In order to finish the case we need to argue that $(s',t')\in Q^{S\land T} = \mathrm{cons}^{S\wedge T}$. This follows from Lemma~\ref{lem:common-implementation} since $U$, and thus $u'$, is locally consistent, and by transitivity any implementation satisfying $u'$ would be a common implementation of $s'$ and $t'$.
        
        The case for delay is identical, while the case for inputs is unsurprising (since adversarial pruning in the computation of conjunction never removes input transitions from consistent to inconsistent states -- there are no such transitions).
        
        \item The 3rd statement follows from the above facts. \new{First assume that $U$ is an implementation (and thus also a specification) such that $U\in \mod{S \wedge T}$. This means that $U \leq S\wedge T$. Using statement 1 and Lemma~\ref{lemma:refinementtransitive} we can extend this to $U \leq S \wedge T \leq S$. Therefore, $U \in \mod{S}$. With the same argument we can also show $U \in \mod{T}$, thus $U\in\mod{S} \cap \mod{T}$.}
        
        \new{The reverse of the 3rd statement can be shown by assuming that $U\in \mod{S} \cap \mod{T}$. This implies that $U\leq S$ and $U \leq T$. Now, using statement 2 we have $U \leq S \wedge T$, which concludes that $U \in \mod{S\wedge T}$.}
        
    \end{enumerate}
\end{proof}

We turn our attention to syntactic representations again. 
\begin{definition}\label{def:conjunctionTIOA}
	Given two TIOAs $A^i = (\mathit{Loc}^i, l_0^i, \mathit{Act}^i, \mathit{Clk}^i, E^i, \mathit{Inv}^i), i=1,2$ where \new{$\mathit{Act}_i^1 \cap \mathit{Act}_o^2 =\emptyset \wedge \mathit{Act}_o^1 \cap \mathit{Act}_i^2 =\emptyset$}\footnote{Formulated differently, $\nexists a \in\bigcup_{i\in I}\mathit{Act}^i$ s.t. $a\in Act_i^i \wedge a \in\mathit{Act}_o^j, i,j\in I, i\neq j$ and $I=\{1,2\}$. This is a more direct formulation of the desired property and can be extended easily for the conjunction of more than two TIOAs.}, the \new{\emph{conjunction}} of $A^1$ and $A^2$, denoted by $A^1\wedge  A^2$, is TIOA $(\mathit{Loc}^1\times\mathit{Loc}^2, (l_0^1,l_0^2), \mathit{Act}, \mathit{Clk}^1\uplus\mathit{Clk}^2, E, \mathit{Inv})$ where $\mathit{Act} = \mathit{Act}_i \uplus \mathit{Act}_o$ with $\mathit{Act}_i = \mathit{Act}_i^1 \cup \mathit{Act}_i^2$ and $\mathit{Act}_o= \mathit{Act}_o^1 \cup \mathit{Act}_o^2$, $\mathit{Inv}((l^1,l^2)) = \mathit{Inv}^1(l^1) \wedge \mathit{Inv}^2(l^2)$, and $E$ is defined as
	\begin{itemize}
		\item $((l_1^1,l_1^2),a, \varphi^1\wedge\varphi^2, c^1\cup c^2,(l_2^1,l_2^2)) \in E$ if $a \in \mathit{Act}^1\cap\mathit{Act}^2$, $(l_1^1, a, \varphi^1, c^1,l_2^1) \in E^1$, and $(l_1^2,a, \varphi^2, c^2,l_2^2) \in E^2$
		\new{
		\item $((l_1^1,l^2),a, \varphi^1, c^1,(l_2^1,l^2)) \in E$ if $a \in \mathit{Act}^1\setminus\mathit{Act}^2$, $(l_1^1, a, \varphi^1, c^1,l_2^1) \in E^1$, and $l^2\in\mathit{Loc}^2$
		\item $((l^1,l_1^2),a, \varphi^2, c^2,(l^1,l_2^2)) \in E$ if $a \in \mathit{Act}^2\setminus\mathit{Act}^1$, $(l_1^2, a, \varphi^2, c^2,l_2^2) \in E^2$, and $l^1\in\mathit{Loc}^1$}
	\end{itemize}
\end{definition}

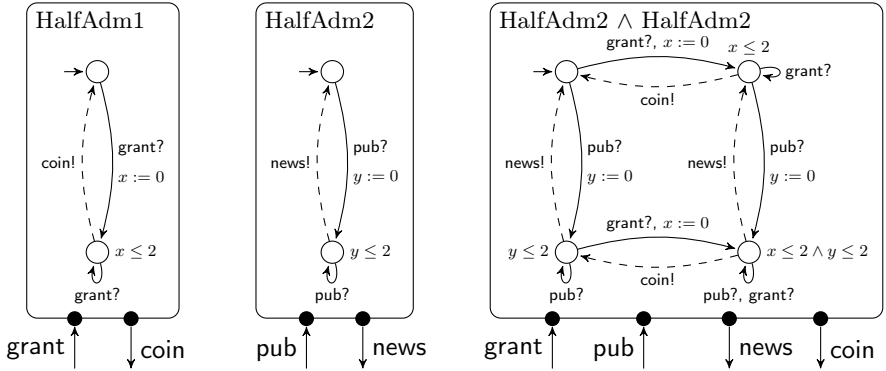
\begin{figure}
	\centering
	\begin{tikzpicture}[->,>=stealth',shorten >=1pt, font=\small, scale=1, transform shape]
		\node[draw, shape=rectangle, rounded corners] (adm1) {\begin{tikzpicture}[->,>=stealth',shorten >=1pt,main node/.style={circle,draw,font=\sffamily\small\bfseries,minimum size=12pt}, align=left,node distance=3cm, scale=.7]
				\node[main node, initial, initial text={}, initial where=left] (1) {};
				\node[main node, label={right:$x \leq 2$}] (2) [below = of 1] {};
				\node [above=.7cm of 1] {}; 
				
				\path[bend left = 15, every node/.style={action}]
				(1) edge node[right,pos=.4] {grant?} node[right,pos=.6] {$x:= 0$} (2)
				(2) edge[dashed] node[left] {coin!} (1)
				;
				\path[every node/.style={action}]
				(2) edge[loop below] node[below] {grant?} (2)
				;
		\end{tikzpicture}};
		\node[anchor=north west] () at (adm1.north west) {HalfAdm1};
		
		\node[draw, shape=rectangle, rounded corners] (adm2) [right = of adm1] {\begin{tikzpicture}[->,>=stealth',shorten >=1pt,main node/.style={circle,draw,font=\sffamily\small\bfseries,minimum size=12pt}, align=left,node distance=3cm, scale=.7]
				\node[main node, initial, initial text={}, initial where=left] (1) {};
				\node[main node, label={right:$y \leq 2$}] (2) [below = of 1] {};
				\node [above=.7cm of 1] {}; 
				
				\path[bend left = 15, every node/.style={action}]
				(1) edge node[right,pos=.4] {pub?} node[right,pos=.6] {$y := 0$} (2)
				(2) edge[dashed] node[left] {news!} (1)
				;
				\path[every node/.style={action}]
				(2) edge[loop below] node[below] {pub?} (2)
				;
		\end{tikzpicture}};
		\node[anchor=north west] () at (adm2.north west) {HalfAdm2};
		
		\node[draw, shape=rectangle, rounded corners] (adm12) [right = of adm2] {\begin{tikzpicture}[->,>=stealth',shorten >=1pt,main node/.style={circle,draw,font=\sffamily\small\bfseries,minimum size=12pt}, align=left,node distance=3cm, scale=.7]
				\node[main node, initial, initial text={}, initial where=left] (1) {};
				\node[main node, label={left:$y \leq 2$}] (2) [below = of 1] {};
				\node[main node, label={above:$x \leq 2$}] (3) [right = of 1] {};
				\node[main node, label={right:$x\leq 2 \wedge y \leq 2$}] (4) [below = of 3] {};
				\node [above=.7cm of 1] {}; 
				
				\path[bend left = 15, every node/.style={action}]
				(1) edge node[right,pos=.4] {pub?} node[right,pos=.6] {$y := 0$} (2)
				(2) edge[dashed] node[left] {news!} (1)
				(3) edge node[right,pos=.4] {pub?} node[right,pos=.6] {$y := 0$} (4)
				(4) edge[dashed] node[left] {news!} (3)
				(1) edge node[above] {grant?, $x := 0$} (3)
				(3) edge[dashed] node[below] {coin!} (1)
				(2) edge node[above] {grant?, $x := 0$} (4)
				(4) edge[dashed] node[below] {coin!} (2)
				;
				\path[every node/.style={action}]
				(2) edge[loop below] node[below] {pub?} (2)
				(3) edge[loop right] node[right] {grant?} (3)
				(4) edge[loop below] node[below] {pub?, grant?} (4)
				;
		\end{tikzpicture}};
		\node[anchor=north west] () at (adm12.north west) {HalfAdm2 $\wedge$ HalfAdm2};
		
		\node[connector] (grant) at (adm1.-100) {};
		\node[connector] (coin) at (adm1.-80) {};
		\node[connector] (pub) at (adm2.-100) {};
		\node[connector] (news) at (adm2.-80) {};
		\node[connector] (grant2) at (adm12.-130) {};
		\node[connector] (coin2) at (adm12.-50) {};
		\node[connector] (pub2) at (adm12.-105) {};
		\node[connector] (news2) at (adm12.-75) {};
		
		\draw[->,>=stealth',shorten >=1pt, every node/.style={action}] 
		(coin) edge node[right] {coin} ++(-90:0.7)
		(news) edge node[right] {news} ++(-90:0.7)
		(coin2) edge node[right] {coin} ++(-90:0.7)
		(news2) edge node[right] {news} ++(-90:0.7)
		;
		\draw[<-,>=stealth',shorten >=1pt, every node/.style={action}]
		(grant) edge node[left] {grant} ++(-90:0.7)
		(pub) edge node[left] {pub} ++(-90:0.7)
		(grant2) edge node[left] {grant} ++(-90:0.7)
		(pub2) edge node[left] {pub} ++(-90:0.7)
		;
	\end{tikzpicture}
	\caption{\new{Example of two specifications each handling one aspect of the administration and their conjunction.}}
	\label{fig:adm1adm2}
\end{figure}

It might appear as if two systems can only advance on an input if both are ready to receive an input, but because of input enableness this is always the case.
\new{An example of a conjunction is shown in Figure~\ref{fig:adm1adm2}. The two aspects of the administration, handing out coins and writing news articles, is split into two specifications. \textsf{HalfAdm1} describes the alternation between \textsf{grant?} and \textsf{coin!}, while \textsf{HalfAdm2} describes the alternation between \textsf{pub?} and \textsf{news!}. Together they form \textsf{HalfAdm1} $\wedge$ \textsf{HalfAdm2}. Observe that this is an alternative and slightly more loose specification of the administration than the one in Figure~\ref{fig:university}. Yet it is the case that \textsf{Administration} refines \textsf{HalfAdm1} $\wedge$ \textsf{HalfAdm2}, while the opposite is not true.}
 
The following theorem lifts all the results from the TIOTSs level to the symbolic representation level\footnote{\new{Where we now include adversarial pruning on both sides instead of just on the left-hand side in previous works.}}:
\begin{theorem}\label{thrm:conjunctionTSandA}
	Given two TIOAs $A^i = (\mathit{Loc}^i, l_0^i, \mathit{Act}^i, \mathit{Clk}^i, E^i, \mathit{Inv}^i), i=1,2$ where $\mathit{Act}_i^1 \cap \mathit{Act}_o^2 =\emptyset \wedge \mathit{Act}_o^1 \cap \mathit{Act}_i^2 =\emptyset$. Then $(\sem{A^1 \wedge A^2})^{\Delta} = (\sem{A^1} \wedge \sem{A^2})\new{^{\Delta}}$.
\end{theorem}

\new{Before we can prove this theorem, we have to introduce several lemmas. The proofs of these lemmas can all be found in Appendix~\ref{app:proofs-conj}. The first lemma shows that the state set of $\sem{A^1 \wedge A^2}$ and $\sem{A^1} \wedge \sem{A^2}$ are the same, including the initial state.}

\begin{lemma}\label{lemma:conjunctionTSandAsamestateset}
	\new{Given two TIOAs $A^i = (\mathit{Loc}^i, l_0^i, \mathit{Act}^i, \mathit{Clk}^i, E^i, \mathit{Inv}^i), i=1,2$ where $\mathit{Act}_i^1 \cap \mathit{Act}_o^2 =\emptyset \wedge \mathit{Act}_o^1 \cap \mathit{Act}_i^2 =\emptyset$. Then $Q^{\sem{A^1 \wedge A^2}} = Q^{\sem{A^1} \wedge \sem{A^2}}$ and $q_0^{\sem{A^1 \wedge A^2}} = q_0^{\sem{A^1} \wedge \sem{A^2}}$.}
\end{lemma}

\new{Lemmas~\ref{lemma:conjunctionTSandAdelay} and~\ref{lemma:conjunctionTSandAsharedaction} show that $\sem{A^1 \wedge A^2}$ and $\sem{A^1} \wedge \sem{A^2}$ mimic each other with delays and shared actions.}
\begin{lemma}\label{lemma:conjunctionTSandAdelay}
	\new{Given two TIOAs $A^i = (\mathit{Loc}^i, l_0^i, \mathit{Act}^i, \mathit{Clk}^i, E^i, \mathit{Inv}^i), i=1,2$ where $\mathit{Act}_i^1 \cap \mathit{Act}_o^2 =\emptyset \wedge \mathit{Act}_o^1 \cap \mathit{Act}_i^2 =\emptyset$. Denote $X = \sem{A^1 \wedge A^2}$ and $Y = \sem{A^1} \wedge \sem{A^2}$, and let $d\in\mathbb{R}_{\geq 0}$ and $q_1,q_2\in Q^X\cap Q^Y$. Then $q_1\xlongrightarrow{d}{}^{\!\! X} q_2$ if and only if $q_1\xlongrightarrow{d}{}^{\!\! Y} q_2$.}
\end{lemma}

\begin{lemma}\label{lemma:conjunctionTSandAsharedaction}
    \new{Given two TIOAs $A^i = (\mathit{Loc}^i, l_0^i, \mathit{Act}^i, \mathit{Clk}^i, E^i, \mathit{Inv}^i), i=1,2$ where $\mathit{Act}_i^1 \cap \mathit{Act}_o^2 =\emptyset \wedge \mathit{Act}_o^1 \cap \mathit{Act}_i^2 =\emptyset$. Denote $X = \sem{A^1 \wedge A^2}$ and $Y = \sem{A^1} \wedge \sem{A^2}$, and let $a\in\mathit{Act}^1\cap\mathit{Act}^2$ and $q_1,q_2\in Q^X\cap Q^Y$. Then $q_1\xlongrightarrow{a}{}^{\!\! X} q_2$ if and only if $q_1\xlongrightarrow{a}{}^{\!\! Y} q_2$.}
\end{lemma}

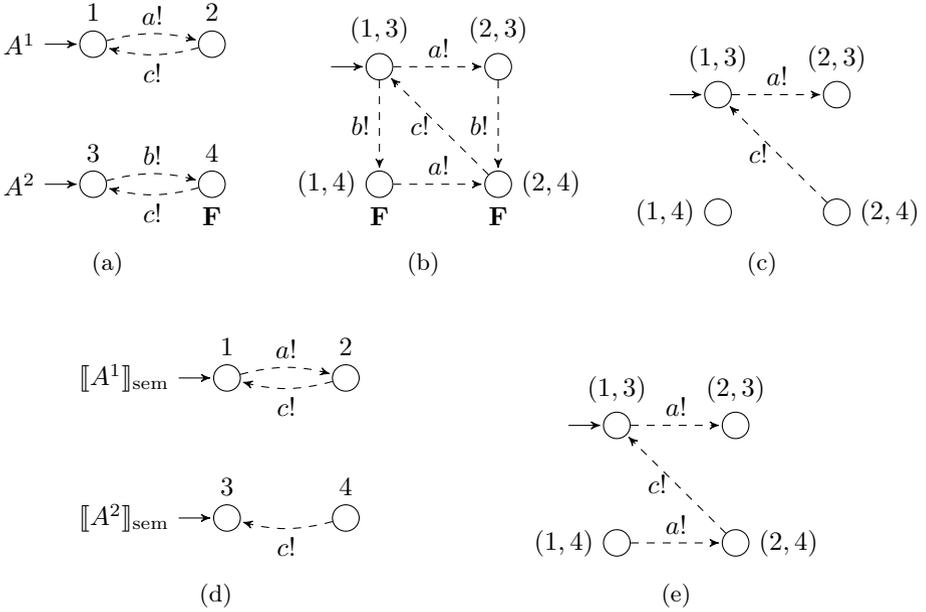
\begin{figure}   
    \begin{subfigure}[b]{0.25\textwidth}
    \centering
    \begin{tikzpicture}[->,>=stealth',shorten >=1pt,main node/.style={circle,draw,font=\sffamily\small\bfseries,minimum size=10pt},align=left,node distance=1.2cm]
		\node[main node, initial, initial text={$A^1$}, initial where=left, label={above:$1$}] (1) {};
		\node[main node, label={above:$2$}] (2) [right = of 1] {};
		
		\path[bend left = 15, dashed]
		(1) edge node[above] {$a!$} (2)
		(2) edge node[below] {$c!$} (1)
		;
		
		\node[main node, initial, initial text={$A^2$}, initial where=left, label={above:$3$}] (3) [below = 1.5 cm of 1] {};
		\node[main node, label={above:$4$}, label={below:$\mathbf{F}$}] (4) [right = of 3] {};
		
		\path[bend left = 15, dashed]
		(3) edge node[above] {$b!$} (4)
		(4) edge node[below] {$c!$} (3)
		;
	\end{tikzpicture}
	\caption{}
	\end{subfigure}
	\hfill
	\begin{subfigure}[b]{0.3\textwidth}
	\centering
	\begin{tikzpicture}[->,>=stealth',shorten >=1pt,main node/.style={circle,draw,font=\sffamily\small\bfseries,minimum size=10pt},align=left,node distance=1.2cm]
		\node[main node, initial, initial text={}, initial where=left, label={above:$(1,3)$}] (1) {};
		\node[main node, label={above:$(2,3)$}] (2) [right = of 1] {};
		\node[main node, label={left:$(1,4)$}, label={below:$\mathbf{F}$}] (3) [below = of 1] {};
		\node[main node, label={right:$(2,4)$}, label={below:$\mathbf{F}$}] (4) [right = of 3] {};
		
		\path[dashed]
		(1) edge node[above] {$a!$} (2)
		(3) edge node[above] {$a!$} (4)
		(1) edge node[left] {$b!$} (3)
		(2) edge node[left] {$b!$} (4)
		(4) edge node[left] {$c!$} (1)
		;
	\end{tikzpicture}
	\caption{}
	\end{subfigure}
	\hfill
	\begin{subfigure}[b]{0.3\textwidth}
	\centering
	\begin{tikzpicture}[->,>=stealth',shorten >=1pt,main node/.style={circle,draw,font=\sffamily\small\bfseries,minimum size=10pt},align=left,node distance=1.2cm]
		\node[main node, initial, initial text={}, initial where=left, label={above:$(1,3)$}] (1) {};
		\node[main node, label={above:$(2,3)$}] (2) [right = of 1] {};
		\node[main node, label={left:$(1,4)$}] (3) [below = of 1] {};
		\node[main node, label={right:$(2,4)$}] (4) [right = of 3] {};
		
		\path[dashed]
		(1) edge node[above] {$a!$} (2)
		(4) edge node[left] {$c!$} (1)
		;
	\end{tikzpicture}
	\caption{}
	\end{subfigure}
	
	\bigskip
	\begin{subfigure}[b]{0.49\textwidth}
	\centering
	\begin{tikzpicture}[->,>=stealth',shorten >=1pt,main node/.style={circle,draw,font=\sffamily\small\bfseries,minimum size=10pt},align=left,node distance=1.2cm]
		\node[main node, initial, initial text={$\llbracket A^1\rrbracket_{\mathrm{sem}}$}, initial where=left, label={above:$1$}] (1) {};
		\node[main node, label={above:$2$}] (2) [right = of 1] {};
		
		\path[bend left = 15, dashed]
		(1) edge node[above] {$a!$} (2)
		(2) edge node[below] {$c!$} (1)
		;
		
		\node[main node, initial, initial text={$\llbracket A^2\rrbracket_{\mathrm{sem}}$}, initial where=left, label={above:$3$}] (3) [below = 1.5 cm of 1] {};
		\node[main node, label={above:$4$}] (4) [right = of 3] {};
		
		\path[bend left = 15, dashed]
		(4) edge node[below] {$c!$} (3)
		;
	\end{tikzpicture}
	\caption{}
	\end{subfigure}
	\hfill
	\begin{subfigure}[b]{0.49\textwidth}
	\centering
	\begin{tikzpicture}[->,>=stealth',shorten >=1pt,main node/.style={circle,draw,font=\sffamily\small\bfseries,minimum size=10pt},align=left,node distance=1.2cm]
		\node[main node, initial, initial text={}, initial where=left, label={above:$(1,3)$}] (1) {};
		\node[main node, label={above:$(2,3)$}] (2) [right = of 1] {};
		\node[main node, label={left:$(1,4)$}] (3) [below = of 1] {};
		\node[main node, label={right:$(2,4)$}] (4) [right = of 3] {};
		
		\path[dashed]
		(1) edge node[above] {$a!$} (2)
		(3) edge node[above] {$a!$} (4)
		(4) edge node[left] {$c!$} (1)
		;
	\end{tikzpicture}
	\caption{}
	\end{subfigure}
    \caption{\new{Example demonstrating additional condition in Lemma~\ref{lemma:conjunctionTSandAnonsharedaction}. In (a) two TIOA $A^1$ and $A^2$ are shown, where location $4$ has a $\mathbf{F}$ invariant. In (b) the conjunction $A^1 \wedge A^2$ is shown. In (c) the semantic representation $\sem{A^1 \wedge A^2}$ is shown (ignoring the delays for simplicity). In (d) the semantic representations $\sem{A^1}$ and $\sem{A^2}$ are shown. And finally, in (e) the conjunction $\sem{A^1} \wedge \sem{A^2}$ is shown.}}
    \label{fig:example_condition_conjunction}
\end{figure}

\new{Lemma~\ref{lemma:conjunctionTSandAnonsharedaction} considers transitions in $\sem{A^1 \wedge A^2}$ and $\sem{A^1} \wedge \sem{A^2}$ labeled by non-shared actions. A special case of this lemma is captured with Corollary~\ref{cor:conjunctionTSandAnonsharedaction}. Compared to Lemma~\ref{lemma:conjunctionTSandAsharedaction}, we can see that we need the additional condition $v_2\models\mathit{Inv}^2(l_2)$ in order to show that transitions can be mimicked. A simple example demonstrating the necessity of this condition is shown in Figure~\ref{fig:example_condition_conjunction}. From two TIOA $A^1$ and $A^2$, the TIOTSs $\sem{A^1 \wedge A^2}$ in (c) and $\sem{A^1} \wedge \sem{A^2}$ in (e) are calculated. As can be seen, $\sem{A^1} \wedge \sem{A^2}$ has an additional transition $(1,4)\xrightarrow{a!} (2,4)$, which is not present in $\sem{A^1 \wedge A^2}$. The reason for this is that the location invariant $\mathit{Inv}(4) = \mathbf{F}$ is processed by the semantic operator before $\sem{A^2}$ is combined with $\sem{A^1}$ by the conjunction operator. Therefore, it is suddenly possible to reach location $(2,4)$ with $a!$ in $\sem{A^1} \wedge \sem{A^2}$. Looking at Lemma~\ref{lemma:conjunctionTSandAnonsharedaction}, we can see that the condition $v_2\models\mathit{Inv}^2(l_2)$ is not satisfied for $q_2=(l_2^1,l_2^2,v_2) = (2,4)$, as $\mathit{Inv}^2(4) = \mathbf{F}$ and no valuation $v_2$ can satisfy a false invariant. So, the additional condition in the lemma `remembers' the original invariant in case we first go to the semantic representation before we perform the conjunction operation.}

\begin{lemma}\label{lemma:conjunctionTSandAnonsharedaction}
    \new{Given two TIOAs $A^i = (\mathit{Loc}^i, l_0^i, \mathit{Act}^i, \mathit{Clk}^i, E^i, \mathit{Inv}^i), i=1,2$ where $\mathit{Act}_i^1 \cap \mathit{Act}_o^2 =\emptyset \wedge \mathit{Act}_o^1 \cap \mathit{Act}_i^2 =\emptyset$. Denote $X = \sem{A^1 \wedge A^2}$ and $Y = \sem{A^1} \wedge \sem{A^2}$, and let $a\in\mathit{Act}^1\setminus\mathit{Act}^2$ and $q_1,q_2\in Q^X\cap Q^Y$, where $q_2=(l_2^1,l_2^2,v_2)$. If $v_2\models\mathit{Inv}^2(l_2)$, then $q_1\xlongrightarrow{a}{}^{\!\! X} q_2$ if and only if $q_1\xlongrightarrow{a}{}^{\!\! Y} q_2$.}
\end{lemma}

\begin{corollary}\label{cor:conjunctionTSandAnonsharedaction}
    \new{Given two TIOAs $A^i = (\mathit{Loc}^i, l_0^i, \mathit{Act}^i, \mathit{Clk}^i, E^i, \mathit{Inv}^i), i=1,2$ where $\mathit{Act}_i^1 \cap \mathit{Act}_o^2 =\emptyset \wedge \mathit{Act}_o^1 \cap \mathit{Act}_i^2 =\emptyset$. Denote $X = \sem{A^1 \wedge A^2}$ and $Y = \sem{A^1} \wedge \sem{A^2}$, and let $a\in\mathit{Act}^1\setminus\mathit{Act}^2$ and $q_1,q_2\in Q^X\cap Q^Y$. If $q_1\xlongrightarrow{a}{}^{\!\! X} q_2$, then $q_1\xlongrightarrow{a}{}^{\!\! Y} q_2$.}
\end{corollary}

\new{The following two lemmas consider the error states and consistent states in $\sem{A^1 \wedge A^2}$ and $\sem{A^1} \wedge \sem{A^2}$, respectively. We can show that both sets are the same for $\sem{A^1 \wedge A^2}$ and $\sem{A^1} \wedge \sem{A^2}$.}
\begin{lemma}\label{lemma:conjunctionTSandAsameerror}
	\new{Given two TIOAs $A^i = (\mathit{Loc}^i, l_0^i, \mathit{Act}^i, \mathit{Clk}^i, E^i, \mathit{Inv}^i), i=1,2$ where $\mathit{Act}_i^1 \cap \mathit{Act}_o^2 =\emptyset \wedge \mathit{Act}_o^1 \cap \mathit{Act}_i^2 =\emptyset$. Let $Q \subseteq \mathit{Loc}^1 \times \mathit{Loc}^2 \times [(\mathit{Clk}^1 \cup \mathit{Clk}^2)\mapsto \mathbb{R}_{\geq 0}]$. Then $\mathrm{err}^{\sem{A^1 \wedge A^2}}(Q) = \mathrm{err}^{\sem{A^1} \wedge \sem{A^2}}(Q)$.}
\end{lemma}

\begin{lemma}\label{lemma:conjunctionTSandAsamecons}
	\new{Given two TIOAs $A^i = (\mathit{Loc}^i, l_0^i, \mathit{Act}^i, \mathit{Clk}^i, E^i, \mathit{Inv}^i), i=1,2$ where $\mathit{Act}_i^1 \cap \mathit{Act}_o^2 =\emptyset \wedge \mathit{Act}_o^1 \cap \mathit{Act}_i^2 =\emptyset$. Then $\mathrm{cons}^{\sem{A^1 \wedge A^2}} = \mathrm{cons}^{\sem{A^1} \wedge \sem{A^2}}$.}
\end{lemma}

\new{Finally, we are ready to proof Theorem~\ref{thrm:conjunctionTSandA}. The reason why adversarial pruning is needed becomes apparent in the second half of the proof where we consider non-shared events. To further illustrate this, consider again the example in Figure~\ref{fig:example_condition_conjunction}, where we show that $\sem{A^1} \wedge \sem{A^2}$ has an additional transition $(1,4)\xrightarrow{a!} (2,4)$, which is not present in $\sem{A^1 \wedge A^2}$. We can `remove' this transition with adversarial pruning by realizing that the target state $(2,4)$ is an inconsistent state (you can see this by noticing that no time delay, including a zero time delay, is possible). }

\begin{proof}[Proof of Theorem~\ref{thrm:conjunctionTSandA}]
	\new{We will prove this theorem by showing that $(\sem{A^1 \wedge A^2})^{\Delta}$ and $(\sem{A^1} \wedge \sem{A^2})^{\Delta}$ have the same set of states, same initial state, same set of actions, and same transition relation.
	}

	\new{It follows from Lemma~\ref{lemma:conjunctionTSandAsamestateset} that $\sem{A^1 \wedge A^2}$ and $\sem{A^1} \wedge \sem{A^2}$ have the same state set and initial state. As $\mathrm{cons}^{\sem{A^1 \wedge A^2}} = \mathrm{cons}^{\sem{A^1} \wedge \sem{A^2}} = \mathrm{cons}$ from Lemma~\ref{lemma:conjunctionTSandAsamecons}, it follows that $(\sem{A^1 \wedge A^2})^{\Delta}$ and $(\sem{A^1} \wedge \sem{A^2})^{\Delta}$ have the same state set and initial state. Also, observe that the semantic of a TIOA and adversarial pruning do not alter the action set. Therefore, it follows directly that $(\sem{A^1 \wedge A^2})^{\Delta}$ and $(\sem{A^1} \wedge \sem{A^2})^{\Delta}$ have the same action set and partitioning into input and output actions.
	}

	\new{It remains to show that $(\sem{A^1 \wedge A^2})^{\Delta}$ and $(\sem{A^1} \wedge \sem{A^2})^{\Delta}$ have the same transition relation. In the remainder of the proof, we will use $v^1$ and $v^2$ to indicate the part of a valuation $v$ of only the clocks of $A^1$ and $A^2$, respectively. Also, for brevity we write $X=(\sem{A^1 \wedge A^2})^{\Delta}$, $Y = (\sem{A^1} \wedge \sem{A^2})^{\Delta}$, and $\mathit{Clk} = \mathit{Clk}^1 \uplus\mathit{Clk}^2$ in the rest of this proof.
	}

	\new{($\Rightarrow$) Assume a transition $q_1^X\xrightarrow{a} q_2^X$ in $X$. From Definition~\ref{def:adversarialpruning} it follows that $q_1^X\xrightarrow{a} q_2^X$ in $\sem{A^1 \wedge A^2}$ and $q_2^X \in \mathrm{cons}$. Following Definition~\ref{def:semanticTIOA} of the semantic, it follows that there exists an edge $(l_1,a,\varphi, c, l_2)\in E^{A^1\wedge A^2}$ with $q_1^X = (l_1,v_1)$, $q_2^X = (l_2,v_2)$, $l_1,l_2 \in \mathit{Loc}^{A^1\wedge A^2}$, $v_1, v_2 \in [\mathit{Clk} \mapsto \mathbb{R}_{\geq 0}]$, $v_1\models \varphi$, $v_2 = v_1[r\mapsto 0]_{r\in c}$, and $v_2\models \mathit{Inv}(l_2)$. Now we consider the three cases of Definition~\ref{def:conjunctionTIOA} of the conjunction for TIOA.
	}
	\begin{itemize}
		\item \new{$a\in\mathit{Act}^1\cap\mathit{Act}^2$. It follows directly from Lemma~\ref{lemma:conjunctionTSandAsharedaction} that $q_1^X \xrightarrow{a} q_2^X$ is a transition in $\sem{A^1} \wedge \sem{A^2}$. Since $q_2^X\in\mathrm{cons}$, it holds that $q_1^X \xrightarrow{a} q_2^X$ is a transition in $Y$.
		}

		\item \new{$a\in\mathit{Act}^1\setminus\mathit{Act}^2$. It follows directly from Corollary~\ref{cor:conjunctionTSandAnonsharedaction} that $q_1^X \xrightarrow{a} q_2^X$ is a transition in $\sem{A^1} \wedge \sem{A^2}$. Since $q_2^X\in\mathrm{cons}$, it holds that $q_1^X \xrightarrow{a} q_2^X$ is a transition in $Y$.
		}

		\item \new{$a\in\mathit{Act}^2\setminus\mathit{Act}^1$. It follows directly from Corollary~\ref{cor:conjunctionTSandAnonsharedaction} (where we switched $A^1$ and $A^2$) that $q_1^X \xrightarrow{a} q_2^X$ is a transition in $\sem{A^1} \wedge \sem{A^2}$. Since $q_2^X\in\mathrm{cons}$, it holds that $q_1^X \xrightarrow{a} q_2^X$ is a transition in $Y$.}
	\end{itemize}
	\new{Now consider that $a$ is a delay $d$. It follows directly from Lemma~\ref{lemma:conjunctionTSandAdelay} that $q_1^X \xrightarrow{d} q_2^X$ is a transition in $\sem{A^1} \wedge \sem{A^2}$. Since $q_2^X\in\mathrm{cons}$, it holds that $q_1^X \xrightarrow{d} q_2^X$ is a transition in $Y$.
	}

	\new{We have shown that when $q_1^X \xrightarrow{a} q_2^X$ is a transition in $X = (\sem{A^1 \wedge A^2})^{\Delta}$, it holds that $q_1^X \xrightarrow{a} q_2^X$ is a transition in $Y = (\sem{A^1} \wedge \sem{A^2})^{\Delta}$. Since the transition is arbitrarily chosen, it holds for all transitions in $X$.
	}

	\new{($\Leftarrow$) Assume a transition $q_1^Y\xrightarrow{a} q_2^Y$ in $Y$. From Definition~\ref{def:adversarialpruning} it follows that $q_1^Y\xrightarrow{a} q_2^Y$ in $\sem{A^1} \wedge \sem{A^2}$ and $q_2^Y \in \mathrm{cons}$. Now we consider the three cases of Definition~\ref{def:conjunctionTIOTS} of the conjunction for TIOTS.
	}
	\begin{itemize}
		\item \new{$a\in\mathit{Act}^1\cap\mathit{Act}^2$. It follows directly from Lemma~\ref{lemma:conjunctionTSandAsharedaction} that $q_1^Y \xrightarrow{a} q_2^Y$ is a transition in $\sem{A^1 \wedge A^2}$. Since $q_2^Y\in\mathrm{cons}$, it holds that $q_1^Y \xrightarrow{a} q_2^Y$ is a transition in $X$.
		}

		\item \new{$a\in\mathit{Act}^1\setminus\mathit{Act}^2$. From time reflexivity of Definition~\ref{def:tiots} we have that $q_2^Y\xlongrightarrow{d}$ with $d = 0$. From Definitions~\ref{def:adversarialpruning} and~\ref{def:conjunctionTIOTS} it follows that $q_2^{\sem{A^1}} \xlongrightarrow{d}$ and $q^{\sem{A^2}} \xlongrightarrow{d}$. Now, from Definition~\ref{def:semanticTIOA} it follows that $v^2 + d \models \mathit{Inv}^2(l^2)$, i.e., $v^2 \models \mathit{Inv}^2(l^2)$.
		}

		\new{It now follows directly from Lemma~\ref{lemma:conjunctionTSandAnonsharedaction} that $q_1^Y \xrightarrow{a} q_2^Y$ is a transition in $\sem{A^1 \wedge A^2}$. Since $q_2^Y\in\mathrm{cons}$, it holds that $q_1^Y \xrightarrow{a} q_2^Y$ is a transition in $X$.
		}

		\item \new{$a\in\mathit{Act}^2\setminus\mathit{Act}^1$. From time reflexivity of Definition~\ref{def:tiots} we have that $q_2^Y\xlongrightarrow{d}$ with $d = 0$. From Definitions~\ref{def:adversarialpruning} and~\ref{def:conjunctionTIOTS} it follows that $q^{\sem{A^1}} \xlongrightarrow{d}$ and $q_2^{\sem{A^2}} \xlongrightarrow{d}$. Now, from Definition~\ref{def:semanticTIOA} it follows that $v^1 + d \models \mathit{Inv}^1(l^1)$, i.e., $v^1 \models \mathit{Inv}^1(l^1)$.
		}

		\new{It now follows directly from Lemma~\ref{lemma:conjunctionTSandAnonsharedaction} (where we switched $A^1$ and $A^2$) that $q_1^Y \xrightarrow{a} q_2^Y$ is a transition in $\sem{A^1 \wedge A^2}$. Since $q_2^Y\in\mathrm{cons}$, it holds that $q_1^Y \xrightarrow{a} q_2^Y$ is a transition in $X$.}
	\end{itemize}
	\new{Now consider that $a$ is a delay $d$. It follows directly from Lemma~\ref{lemma:conjunctionTSandAdelay} that $q_1^Y \xrightarrow{d} q_2^Y$ is a transition in $\sem{A^1\wedge A^2}$. Since $q_2^Y\in\mathrm{cons}$, it holds that $q_1^Y \xrightarrow{d} q_2^Y$ is a transition in $X$.
	}

	\new{We have shown that when $q_1^Y \xrightarrow{a} q_2^Y$ is a transition in $Y = (\sem{A^1} \wedge \sem{A^2})^{\Delta}$, it holds that $q_1^Y \xrightarrow{a} q_2^Y$ is a transition in $X = (\sem{A^1 \wedge A^2})^{\Delta}$. Since the transition is arbitrarily chosen, it holds for all transitions in $Y$.}
\end{proof}

\new{The following corollary describes a special case of Theorem~\ref{thrm:conjunctionTSandA}, which happens to be one of the unproven main theorems in our previous work~\cite{david_real-time_2015}.}
\begin{corollary}
	\new{Given two TIOAs $A^i = (\mathit{Loc}^i, l_0^i, \mathit{Act}^i, \mathit{Clk}^i, E^i, \mathit{Inv}^i), i=1,2$ where $\mathit{Act}_i^1 = \mathit{Act}_i^2  \wedge \mathit{Act}_o^1 =\mathit{Act}_o^2$. Then $\sem{A^1 \wedge A^2} = \sem{A^1} \wedge \sem{A^2}$.}
\end{corollary}
\begin{proof}
	\new{This corollary follows directly as a special case from the proof of Theorem~\ref{thrm:conjunctionTSandA}. The special case only depends on Lemmas~\ref{lemma:conjunctionTSandAsamestateset} and~\ref{lemma:conjunctionTSandAsharedaction}, which do not require adversarial pruning to be applied.}
\end{proof}

\section{\new{Parallel} composition}
\label{sec:compandcomp}


We shall now define {\em structural composition}, also called {\em parallel composition}, between specifications. We follow the optimistic approach of~\cite{alfaro_timed_2002}, i.e., {\em two specifications can be composed if there exists at least one environment in which they can work together}. 
Before going further, we would like to contrast the structural and logical composition.

The main use case for parallel composition is in fact dual to the one for conjunction. Indeed, as observed in the previous section, conjunction is used to reason about internal properties of an implementation set, so if a local inconsistency arises in conjunction we limit the implementation set to avoid it in implementations. A pruned specification can be given to a designer, who chooses a particular implementation satisfying conjoined requirements. A conjunction is consistent if the output player can avoid inconsistencies, and its main theorem states that its set of implementation coincides with the intersection of implementation sets of the conjuncts.

In contrast, parallel composition is used to reason about external use of two (or more) components. We assume an independent implementation scenario, where the two composed components are implemented by independent designers. The designer of any of the components can only assume that the \new{other} composed implementations will adhere to \new{the} original specifications being composed. Consequently if an error occurs in parallel composition of the two specifications, \new{the independent designers receive additional information on how to restrict their specifications to avoid reaching the error states in the composed system.}

We now propose our formal definition for parallel composition, which roughly corresponds to the one defined on timed input/output automata~\cite{kaynar_timed_2003}. We consider two \new{TIOTSs} $S = (Q^S,q_0^S, \mathit{Act}^S,\allowbreak \rightarrow^S)$ and $T = (Q^T,q_0^T,\mathit{Act}^T, \allowbreak\rightarrow^T)$ and we say that they are \emph{composable} iff their output alphabets are disjoint $\mathit{Act}_o^S \cap \mathit{Act}_o^T = \emptyset$. \new{Similarly, we say that two specifications are composable if their semantics are composable.}

\begin{definition}\label{def:parallelcompositionTIOTS}
	Given two specifications $S^i = (Q^i,q_0^i,\mathit{Act}^i,\rightarrow^i), i=1,2$ where $\mathit{Act}_o^1 \cap \mathit{Act}_o^2 =\emptyset$, the \emph{parallel composition} of $S^1$ and $S^2$, denoted by $S^1 \parallel S^2$, is TIOTS $(Q^1\times Q^2, (q_0^1,q_0^2),\mathit{Act},\rightarrow)$ where $\mathit{Act} = \mathit{Act}^1\cup \mathit{Act}^2 = \mathit{Act}_i \uplus \mathit{Act}_o$ with $\mathit{Act}_i = (\mathit{Act}_i^1\setminus\mathit{Act}_o^2) \cup (\mathit{Act}_i^2\setminus\mathit{Act}_o^1)$ and $\mathit{Act}_o= \mathit{Act}_o^1 \cup \mathit{Act}_o^2$, and $\rightarrow$ is defined as
	\begin{itemize}
		\item $(q_1^1,q_1^2)\xlongrightarrow{a} (q_2^1,q_2^2)$ if \new{$a\in\mathit{Act}^1\cap\mathit{Act}^2$}, $q_1^1\xlongrightarrow{a}{}^{\!\! 1} q_2^1$, and $q_1^2\xlongrightarrow{a}{}^{\!\! 2} q_2^2$
		\item $(q_1^1,q^2)\xlongrightarrow{a} (q_2^1,q^2)$ if $a\in\mathit{Act}^1\setminus\mathit{Act}^2$, $q_1^1\xlongrightarrow{a}{}^{\!\! 1} q_2^1$, and $q^2\in Q^2$
		\item $(q^1,q_1^2)\xlongrightarrow{a} (q^1,q_2^2)$ if $a\in\mathit{Act}^2\setminus\mathit{Act}^1$, $q_1^2\xlongrightarrow{a}{}^{\!\! 2} q_2^2$, and $q^1\in Q^1$
		\item $(q_1^1,q_1^2)\xlongrightarrow{d} (q_2^1,q_2^2)$ if $d \in \mathbb{R}_{\geq 0}$, $q_1^1\xlongrightarrow{d}{}^{\!\! 1} q_2^1$, and $q_1^2\xlongrightarrow{d}{}^{\!\! 2} q_2^2$
	\end{itemize}
\end{definition}

Observe that if we compose two locally specifications using the above product rules, then the resulting product is also locally consistent. \new{This is formalized in Lemma~\ref{lemma:parallelcompositionTSconsistent}.} Furthermore, observe that parallel composition is commutative, and that two specifications composed give rise to well-formed specifications. It is also associative in the following sense:
\begin{equation*}
	\mod{(S\parallel T) \parallel U} = \mod{S \parallel (T\parallel U)}
\end{equation*}

\begin{lemma}\label{lemma:parallelcompositionTSconsistent}
	\new{Given two locally consistent specifications $S^i = (Q^i,q_0^i,\mathit{Act}^i,\rightarrow^i), i=1,2$ where $\mathit{Act}_o^1 \cap \mathit{Act}_o^2 =\emptyset$. Then $S^1 \parallel S^2$ is locally consistent.}
\end{lemma}
\noindent The proof of Lemma~\ref{lemma:parallelcompositionTSconsistent} can be found in Appendix~\ref{app:proofs-comp}.

\begin{theorem}\label{thm:precongruence}
  Refinement is a pre-congruence with respect to parallel composition: for any specifications $S^1$, $S^2$, and $T$ such that $S^1 \leq S^2$ and $S^1$ is composable with $T$, we have that $S^2$ is composable with $T$ and $S^1\parallel T \leq S^2\parallel T$. 
\end{theorem}
\noindent The proof of Theorem~\ref{thm:precongruence} can be found in Appendix~\ref{app:proofs-comp}.

\begin{figure}
    \begin{subfigure}[t]{0.44\textwidth}
	    \vskip 0pt
        \centering
        \begin{tikzpicture}[->,>=stealth',shorten >=1pt,main node/.style={circle,draw,font=\sffamily\small\bfseries,minimum size=10pt},node distance=1.2cm]
        	\node[main node, initial, initial text={}, initial where=left, label={above:$1$}] (1) {};
        	\node[main node, label={above:$2$}] (2) [right = of 1] {};
        	\node[main node, label={above:$3$}, label={right:$x \leq 0$}] (3) [right = of 2] {};
        	
        	\path[]
        	(1) edge node[above] {$a!$} (2)
        	(2) edge node[above] {$b?$} (3)
        	;
        	
        	\path[loop below]
            (1) edge node[below] {$b?$} (1)
            (3) edge node[below] {$b?$} (3)
            ;
        \end{tikzpicture}
	    \caption{$S$}
	\end{subfigure}
	\hfill
	\begin{subfigure}[t]{0.44\textwidth}
	    \vskip 0pt
        \centering
        \begin{tikzpicture}[->,>=stealth',shorten >=1pt,main node/.style={circle,draw,font=\sffamily\small\bfseries,minimum size=10pt},node distance=1.2cm]
        	\node[main node, initial, initial text={}, initial where=left, label={above:$4$}] (4) [] {};
        	\node[main node, label={above:$5$}] (5) [right = of 4] {};
        	\node[main node, label={above:$6$}] (6) [right = of 5] {};
        	
        	\path[]
        	(4) edge node[above] {$a?$} (5)
        	(5) edge node[above] {$b!$} (6)
        	;
        	
        	\path[loop below]
            (5) edge node[below] {$a?$} (5)
            (6) edge node[below] {$a?$} (6)
            ;
        \end{tikzpicture}
	    \caption{$T = T^{\Delta}$}
	\end{subfigure}
	\hfill
	\begin{subfigure}[t]{0.1\textwidth}
	    \vskip 0pt
        \centering
        \begin{tikzpicture}[->,>=stealth',shorten >=1pt,main node/.style={circle,draw,font=\sffamily\small\bfseries,minimum size=10pt},node distance=1.2cm]
        	\node[main node, initial, initial text={}, initial where=left, label={above:$1$}] (1) {};
        	
        	\path[loop below]
            (1) edge node[below] {$b?$} (1)
            ;
        \end{tikzpicture}
	    \caption{$S^{\Delta}$}
	\end{subfigure}
	
	\bigskip
	\begin{subfigure}[t]{0.1\textwidth}
	    \vskip 0pt
        \centering
        \begin{tikzpicture}[->,>=stealth',shorten >=1pt,main node/.style={circle,draw,font=\sffamily\small\bfseries,minimum size=10pt},node distance=1.2cm]
        	\node[main node, initial, initial text={}, initial where=left, label={above:$(1,4)$}] (1) {};
        	
        \end{tikzpicture}
	    \caption{$S^{\Delta} \parallel T^{\Delta}$}\label{fig:example_adversarial_pruning_parallel_d}
	\end{subfigure}
	\hfill
	\begin{subfigure}[t]{0.5\textwidth}
	    \vskip 0pt
        \centering
        \begin{tikzpicture}[->,>=stealth',shorten >=1pt,main node/.style={circle,draw,font=\sffamily\small\bfseries,minimum size=10pt},node distance=1.2cm]
        	\node[main node, initial, initial text={}, initial where=left, label={above:$(1,4)$}] (1) {};
        	\node[main node, label={above:$(2,5)$}] (2) [right = of 1] {};
        	\node[main node, label={above:$(3,6)$}, label={right:$x \leq 0$}] (3) [right = of 2] {};
        	
        	\path[]
        	(1) edge node[above] {$a!$} (2)
        	(2) edge node[above] {$b!$} (3)
        	;
        \end{tikzpicture}
	    \caption{$S\parallel T$}
	\end{subfigure}
	\hfill
	\begin{subfigure}[t]{0.28\textwidth}
	    \vskip 0pt
        \centering
        \begin{tikzpicture}[->,>=stealth',shorten >=1pt,main node/.style={circle,draw,font=\sffamily\small\bfseries,minimum size=10pt},node distance=1.2cm]
        	\node[main node, initial, initial text={}, initial where=left, label={above:$(1,4)$}] (1) {};
        	\node[main node, label={above:$(2,5)$}] (2) [right = of 1] {};
        	
        	\path[]
        	(1) edge node[above] {$a!$} (2)
        	;
        \end{tikzpicture}
	    \caption{$(S \parallel T)^{\Delta}$}\label{fig:example_adversarial_pruning_parallel_f}
	\end{subfigure}
    \caption{\new{Example showing that adversarial pruning does not distribute over the parallel composition operator. Observe that the result in (d) differs from the one in (f).}}
    \label{fig:example_adversarial_pruning_parallel}
\end{figure}
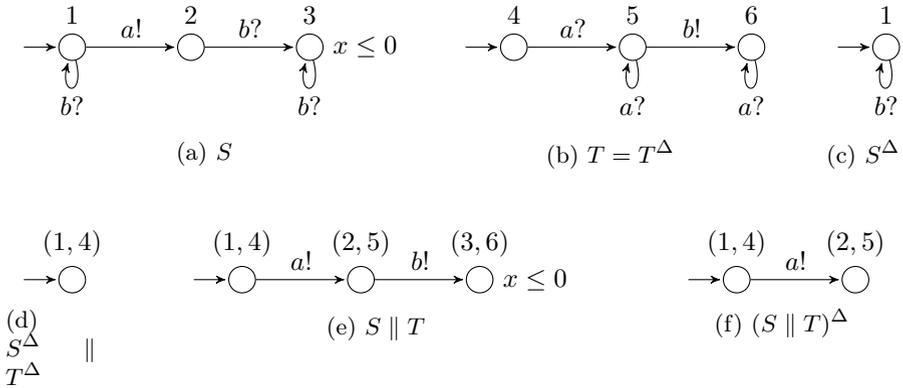

\new{Adversarial pruning does not distribute over the parallel composition operator. Consider two composable specifications $S$ and $T$:  $S^{\Delta} \parallel T^{\Delta} \neq (S \parallel T)^{\Delta}$. An example is shown in Figure~\ref{fig:example_adversarial_pruning_parallel}. Observe that $S^{\Delta} \parallel T^{\Delta}$ (Figure~\ref{fig:example_adversarial_pruning_parallel_d}) does not allow any behavior from the initial state, while $(S\parallel T)^{\Delta}$ (Figure~\ref{fig:example_adversarial_pruning_parallel_f}) still allows action $a$ to be performed. If we want specification $S$ to never reach the error state for \emph{all} possible environments, we have to disable output action $a!$ from location 1. Yet, in the example we are composing $S$ with the specific environment $T$, which can help $S$ in avoiding the error state.
Therefore, as long as we are composing components of a system together, we should not apply adversarial pruning on intermediate specifications.}

\new{We still would like to simplify intermediate specifications as much as possible before and after performing parallel composition without any loss of possible implementations. This is captured in the following definition of cooperative pruning. 
}
\begin{definition}\label{def:cooperativepruning}
	\new{Given a specification $S = (Q,s_0,\mathit{Act},\rightarrow)$, the result of \emph{cooperative pruning} of $S$, denoted by $S^{\forall}$, is a subspecification $S^{\forall} = (Q^{\forall},s_0,\mathit{Act},\rightarrow^{\forall})$ with $S^{\forall}\subseteq S$ and $\rightarrow^{\forall} \subseteq \rightarrow$ such that for all specifications $T$ composable with $S$ it holds that $\mod{S\parallel T} = \mod{ S^{\forall} \parallel T}$}
\end{definition}

\new{Unfortunately, the best we can do, in the sense of removing states, transitions, or both, is to remove nothing, i.e., cooperative pruning is the identity transformation. We prove this with the following lemma.
}

\begin{lemma}\label{lemma:cooperativepruning}
	\new{Given a specification $S = (Q^S,s_0,\mathit{Act}^S,\rightarrow^S)$ and its cooperatively pruned subspecification $S^{\forall}$. It holds that $S = S^{\forall}$.}
\end{lemma}
\noindent The proof of Lemma~\ref{lemma:cooperativepruning} can be found in Appendix~\ref{app:proofs-comp}.


We now switch to the symbolic representation. Parallel composition of two TIOA is defined in the following way. 
\begin{definition}\label{def:parallelcompositionTIOA}
	Given two specification automata $A^i = (\mathit{Loc}^i, q_0^i, \allowbreak\mathit{Act}^i, \mathit{Clk}^i,\allowbreak E^i, \allowbreak\mathit{Inv}^i), i=1,2$ where $\mathit{Act}_o^1 \cap \mathit{Act}_o^2 =\emptyset$, the \emph{parallel composition} of $A^1$ and $A^2$, denoted by $A^1\parallel A^2$, is TIOA $(\mathit{Loc}^1\times\mathit{Loc}^2, (q_0^1,q_0^2), \mathit{Act}, \mathit{Clk}^1\uplus\mathit{Clk}^2, E, \mathit{Inv})$ where $\mathit{Act} = \mathit{Act}_i \uplus \mathit{Act}_o$ with $\mathit{Act}_i = (\mathit{Act}_i^1\setminus\mathit{Act}_o^2) \cup (\mathit{Act}_i^2\setminus\mathit{Act}_o^1)$ and $\mathit{Act}_o= \mathit{Act}_o^1 \cup \mathit{Act}_o^2$, $\mathit{Inv}((q^1,q^2)) = \mathit{Inv}(q^1) \wedge \mathit{Inv}(q^2)$, and $E$ is defined as
	\begin{itemize}
		\item $((q_1^1,q_1^2),a, \varphi^1\wedge\varphi^2, c^1\cup c^2,(q_2^1,q_2^2)) \in E$ if $a \in \mathit{Act}^1\cap\mathit{Act}^2$, $(q_1^1, a, \varphi^1, c^1,q_2^1) \in E^1$, and $(q_1^1,a, \varphi^2, c^2,q_2^1) \in E^1$
		\item $((q_1^1,q^2),a, \varphi^1, c^1,(q_2^1,q^2)) \in E$ if $a \in \mathit{Act}^1\setminus\mathit{Act}^2$, $(q_1^1, a, \varphi^1, c^1,q_2^1) \in E^1$, and $q^2\in\mathit{Loc}^2$
		\item $((q^1,q_1^2),a, \varphi^2, c^2,(q^1,q_2^2)) \in E$ if $a \in \mathit{Act}^2\setminus\mathit{Act}^1$, $(q_1^2, a, \varphi^2, c^2,q_2^2) \in E^2$, and $q^1\in\mathit{Loc}^1$
	\end{itemize}
\end{definition}

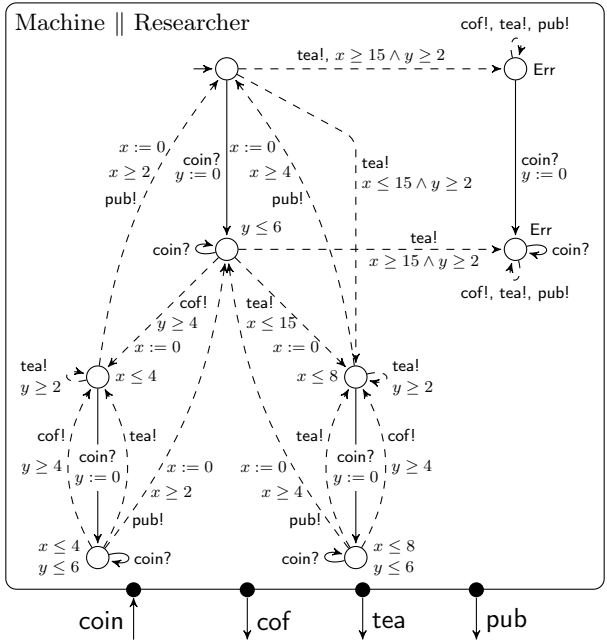
\begin{figure}
	\centering
	
	\begin{tikzpicture}[->,>=stealth',shorten >=1pt, font=\small, scale=1, transform shape]
		\node[draw, shape=rectangle, rounded corners] (machres) {\begin{tikzpicture}[->,>=stealth',shorten >=1pt, align=left,node distance=3cm, scale=.7]
				\node[main node, initial, initial text={}, initial where=left] (1) {};
				\node[main node, label={above right:$y \leq 6$}] (2) [below = of 1] {};
				\node[main node, label={right:$x \leq 4$}] (3) [below left = of 2] {};
				\node[main node, label={left:$x \leq 8$}] (4) [below right = of 2] {};
				\node[main node, label={left:$x \leq 4$\\$y \leq 6$}] (5) [below = of 3] {};
				\node[main node, label={right:$x \leq 8$\\$y \leq 6$}] (6) [below = of 4] {};
				\node[main node, label={right:Err}] (7) [right = 5cm of 1] {};
				\node[main node, label={above right:Err}] (8) [right = 5cm of 2] {};
				
				\path[every node/.style={action}]
				(1) edge node[left] {coin?} node[left, pos=0.6] {$y := 0$}  (2)
				(2) edge[dashed] node[right,pos=.4] {cof!} node[right,pos=.6] {$y \geq 4$} node[right,pos=.8] {$x := 0$} (3)
				(2) edge[dashed] node[left,pos=.4] {tea!} node[left,pos=.6] {$x \leq 15$} node[left,pos=.8] {$x := 0$} (4)
				(3) edge node[pos=.43,fill=white] {coin?} node[pos=0.57,fill=white] {$y := 0$} (5)
				(4) edge node[pos=.43,fill=white] {coin?} node[ pos=0.57,fill=white] {$y := 0$} (6)
				(1) edge[dashed] node[above] {tea!, $x \geq 15\wedge y\geq 2$} (7)
				(2) edge[dashed] node[above,pos=.7] {tea!} node[below,pos=.7] {$x \geq 15\wedge y\geq 2$} (8)
				(7) edge node[right] {coin?} node[right, pos=0.6] {$y := 0$}  (8)
				;
				\path[dashed,every node/.style={action}]
				(3) edge[bend left = 15] node[left,pos=.55] {pub!} node[left,pos=.65] {$x \geq 2$} node[left,pos=.75] {$x := 0$} (1)
				(4) edge[bend right=15] node[left,pos=.55] {pub!} node[left,pos=.65] {$x \geq 4$} node[left,pos=.75] {$x := 0$} (1)
				(5) edge[bend right = 15] node[right,pos=.1] {pub!} node[right,pos=.2] {$x \geq 2$} node[right,pos=.3] {$x := 0$} (2)
				(6) edge[bend left = 15] node[left,pos=.1] {pub!} node[left,pos=.2] {$x \geq 4$} node[left,pos=.3] {$x := 0$} (2)
				(5) edge[bend right = 30] node[right,pos=.7] {tea!} (3)
				(6) edge[bend left = 30] node[left,pos=.7] {tea!} (4)
				(5) edge[bend left = 30] node[left,pos=.7] {cof!} node[left,pos=.5] {$y \geq 4$} (3)
				(6) edge[bend right = 30] node[right,pos=.7] {cof!} node[right,pos=.5] {$y \geq 4$} (4)
				;
				\node (help) [below = 1cm of 1] {};
				\draw[dashed, rounded corners, every node/.style={action}]
				(1) -- (help -| 4) -- node[right, pos=.1] {tea!} node[right, pos=0.2] {$x\leq 15 \wedge y \geq 2$} (4)
				;
				\path[every node/.style={action}]
				(2) edge[loop left] node[left] {coin?} (2)
				(3) edge[loop left,dashed] node[left] {tea!\\$y\geq2$} (3)
				(4) edge[loop right,dashed] node[right] {tea!\\$y\geq2$} (4)
				(5) edge[loop right] node[right] {coin?} (5)
				(6) edge[loop left] node[left] {coin?} (6)
				(7) edge[loop above,dashed] node[above] {cof!, tea!, pub!} (7)
				(8) edge[loop below,dashed] node[below] {cof!, tea!, pub!} (8)
				(8) edge[loop right] node[right] {coin?} (8)
				;
		\end{tikzpicture}};
		\node[anchor=north west] () at (machres.north west) {Machine $\parallel$ Researcher};
		
		\node[connector] (coin) at (machres.-120) {};
		\node[connector] (cof) at (machres.-101) {};
		\node[connector] (tea) at (machres.-79) {};
		\node[connector] (pub) at (machres.-60) {};
		
		\draw[->,>=stealth',shorten >=1pt, every node/.style={action}]
		(cof) edge node[right] {cof} ++(-90:0.7)
		(tea) edge node[right] {tea} ++(-90:0.7)
		(pub) edge node[right] {pub} ++(-90:0.7)
		;
		\draw[<-,>=stealth',shorten >=1pt, every node/.style={action}]
		(coin) edge node[left] {coin} ++(-90:0.7)
		;
	\end{tikzpicture}
	\caption{\new{The parallel composition of the \textsf{Machine} and \textsf{Researcher} from Figure~\ref{fig:university}.}}
	\label{fig:machineResearcher}
\end{figure}

\new{Figure~\ref{fig:machineResearcher} shows the parallel composition \textsf{Machine}$\parallel$\textsf{Researcher} where \textsf{Machine} and \textsf{Researcher} are from Figure~\ref{fig:university}. As typical for composing automata, the parallel composition of \textsf{Machine} and \textsf{Researcher} looks much more complicated that the two individual specifications. Furthermore, the actions \textsf{cof} and \textsf{tea}, which were outputs in \textsf{Machine} and inputs in \textsf{Researcher}, have become outputs in the combined specification.}

Finally, the following theorem lifts all the results from timed input/output transition systems to the symbolic representation level. \new{Similarly to Theorem~\ref{thrm:conjunctionTSandA}, we need to take the special case from Figure~\ref{fig:example_condition_conjunction} into account (but now consider action $c$ to be an input for $A^{2}$). The transition in Figure~\ref{fig:example_condition_conjunction} (e) from $(1,4)\xlongrightarrow{a!}(2,4)$ can be `removed' with adversarial pruning by realizing that the target state $(2,4)$ is an inconsistent state (you can see this by noticing that no time delay, including a zero time delay, is possible).}

\begin{theorem}\label{thrm:parallelcompositionTSandA}
  Given two specification automata $A^i = (\mathit{Loc}^i, l_0^i, \mathit{Act}^i, \mathit{Clk}^i, E^i,\allowbreak \mathit{Inv}^i), i=1,2$ where $\mathit{Act}_o^1 \cap \mathit{Act}_o^2 =\emptyset$. Then \new{$(\sem{A^1 \parallel A^2})^{\Delta} = (\sem{A^1} \parallel\sem{A^2})^{\Delta}$}.
\end{theorem}

\new{Before we can prove this theorem, we have to introduce several lemmas. These lemmas are almost identical to the ones in Section~\ref{sec:consandconj} for the conjunction. Therefore, we have omitted the proof.}


\begin{lemma}\label{lemma:parallelcompositionTSandAsamestateset}
	\new{Given two TIOAs $A^i = (\mathit{Loc}^i, l_0^i, \mathit{Act}^i, \mathit{Clk}^i, E^i, \mathit{Inv}^i), i=1,2$ where $\mathit{Act}_o^1 \cap \mathit{Act}_o^2 =\emptyset$. Then $Q^{\sem{A^1 \parallel A^2}} = Q^{\sem{A^1} \parallel \sem{A^2}}$ and $q_0^{\sem{A^1 \parallel A^2}} = q_0^{\sem{A^1} \parallel \sem{A^2}}$.}
\end{lemma}

\begin{lemma}\label{lemma:parallelcompositionTSandAdelay}
	\new{Given two TIOAs $A^i = (\mathit{Loc}^i, l_0^i, \mathit{Act}^i, \mathit{Clk}^i, E^i, \mathit{Inv}^i), i=1,2$ where $\mathit{Act}_o^1 \cap \mathit{Act}_o^2 =\emptyset$. Denote $X = \sem{A^1 \parallel A^2}$ and $Y = \sem{A^1} \parallel \sem{A^2}$, and let $d\in\mathbb{R}_{\geq 0}$ and $q_1,q_2\in Q^X\cap Q^Y$. Then $q_1\xlongrightarrow{d}{}^{\!\! X} q_2$ if and only if $q_1\xlongrightarrow{d}{}^{\!\! Y} q_2$.}
\end{lemma}

\begin{lemma}\label{lemma:parallelcompositionTSandAsharedaction}
	\new{Given two TIOAs $A^i = (\mathit{Loc}^i, l_0^i, \mathit{Act}^i, \mathit{Clk}^i, E^i, \mathit{Inv}^i), i=1,2$ where $\mathit{Act}_o^1 \cap \mathit{Act}_o^2 =\emptyset$. Denote $X = \sem{A^1 \parallel A^2}$ and $Y = \sem{A^1} \parallel \sem{A^2}$, and let $a\in\mathit{Act}^1\cap\mathit{Act}^2$ and $q_1,q_2\in Q^X\cap Q^Y$. Then $q_1\xlongrightarrow{a}{}^{\!\! X} q_2$ if and only if $q_1\xlongrightarrow{a}{}^{\!\! Y} q_2$.}
\end{lemma}

\begin{lemma}\label{lemma:parallelcompositionTSandAnonsharedaction}
	\new{Given two TIOAs $A^i = (\mathit{Loc}^i, l_0^i, \mathit{Act}^i, \mathit{Clk}^i, E^i, \mathit{Inv}^i), i=1,2$ where $\mathit{Act}_o^1 \cap \mathit{Act}_o^2 =\emptyset$. Denote $X = \sem{A^1 \parallel A^2}$ and $Y = \sem{A^1} \parallel \sem{A^2}$, and let $a\in\mathit{Act}^1\setminus\mathit{Act}^2$ and $q_1,q_2\in Q^X\cap Q^Y$, where $q_2=(l_2^1,l_2^2,v_2)$. If $v_2\models\mathit{Inv}^2(l_2)$, then $q_1\xlongrightarrow{a}{}^{\!\! X} q_2$ if and only if $q_1\xlongrightarrow{a}{}^{\!\! Y} q_2$.}
\end{lemma}

\begin{corollary}\label{cor:parallelcompositionTSandAnonsharedaction}
	\new{Given two TIOAs $A^i = (\mathit{Loc}^i, l_0^i, \mathit{Act}^i, \mathit{Clk}^i, E^i, \mathit{Inv}^i), i=1,2$ where $\mathit{Act}_o^1 \cap \mathit{Act}_o^2 =\emptyset$. Denote $X = \sem{A^1 \parallel A^2}$ and $Y = \sem{A^1} \parallel \sem{A^2}$, and let $a\in\mathit{Act}^1\setminus\mathit{Act}^2$ and $q_1,q_2\in Q^X\cap Q^Y$. If $q_1\xlongrightarrow{a}{}^{\!\! X} q_2$, then $q_1\xlongrightarrow{a}{}^{\!\! Y} q_2$.}
\end{corollary}
\begin{lemma}\label{lemma:parallelcompositionTSandAsameerror}
	\new{Given two TIOAs $A^i = (\mathit{Loc}^i, l_0^i, \mathit{Act}^i, \mathit{Clk}^i, E^i, \mathit{Inv}^i), i=1,2$ where $\mathit{Act}_o^1 \cap \mathit{Act}_o^2 =\emptyset$. Let $Q \subseteq \mathit{Loc}^1 \times \mathit{Loc}^2 \times [(\mathit{Clk}^1 \cup \mathit{Clk}^2)\mapsto \mathbb{R}_{\geq 0}]$. Then $\mathrm{err}^{\sem{A^1 \parallel A^2}}(Q) = \mathrm{err}^{\sem{A^1} \parallel \sem{A^2}}(Q)$.}
\end{lemma}

\begin{lemma}\label{lemma:parallelcompositionTSandAsamecons}
	\new{Given two TIOAs $A^i = (\mathit{Loc}^i, l_0^i, \mathit{Act}^i, \mathit{Clk}^i, E^i, \mathit{Inv}^i), i=1,2$ where $\mathit{Act}_o^1 \cap \mathit{Act}_o^2 =\emptyset$. Then $\mathrm{cons}^{\sem{A^1 \parallel A^2}} = \mathrm{cons}^{\sem{A^1} \parallel \sem{A^2}}$.}
\end{lemma}

\new{Finally, we are ready to proof Theorem~\ref{thrm:parallelcompositionTSandA}.}
 
\begin{proof}[Proof of Theorem~\ref{thrm:parallelcompositionTSandA}]
	\new{We will prove this theorem by showing that $(\sem{A^1 \parallel A^2})^{\Delta}$ and $(\sem{A^1} \parallel \sem{A^2})^{\Delta}$ have the same set of states, same initial state, same set of actions, and same transition relation.}
	
	\new{It follows from Lemma~\ref{lemma:parallelcompositionTSandAsamestateset} that $\sem{A^1 \parallel A^2}$ and $\sem{A^1} \parallel \sem{A^2}$ have the same state set and initial state. As $\mathrm{cons}^{\sem{A^1 \parallel A^2}} = \mathrm{cons}^{\sem{A^1} \parallel \sem{A^2}} = \mathrm{cons}$ from Lemma~\ref{lemma:parallelcompositionTSandAsamecons}, it follows that $(\sem{A^1 \parallel A^2})^{\Delta}$ and $(\sem{A^1} \parallel \sem{A^2})^{\Delta}$ have the same state set and initial state. Also, observe that the semantic of a TIOA and adversarial pruning do not alter the action set. Therefore, it follows directly that $(\sem{A^1 \parallel A^2})^{\Delta}$ and $(\sem{A^1} \parallel \sem{A^2})^{\Delta}$ have the same action set and partitioning into input and output actions.}
	
	\new{It remains to show that $(\sem{A^1 \parallel A^2})^{\Delta}$ and $(\sem{A^1} \parallel \sem{A^2})^{\Delta}$ have the same transition relation. In the remainder of the proof, we will use $v^1$ and $v^2$ to indicate the part of a valuation $v$ of only the clocks of $A^1$ and $A^2$, respectively. Also, for brevity we write $X=(\sem{A^1 \parallel A^2})^{\Delta}$, $Y = (\sem{A^1} \parallel \sem{A^2})^{\Delta}$, and $\mathit{Clk} = \mathit{Clk}^1 \uplus\mathit{Clk}^2$ in the rest of this proof.}
	
	\new{($\Rightarrow$) Assume a transition $q_1^X\xrightarrow{a} q_2^X$ in $X$. From Definition~\ref{def:adversarialpruning} it follows that $q_1^X\xrightarrow{a} q_2^X$ in $\sem{A^1 \parallel A^2}$ and $q_2^X \in \mathrm{cons}$. Following Definition~\ref{def:semanticTIOA} of the semantic, it follows that there exists an edge $(l_1,a,\varphi, c, l_2)\in E^{A^1\parallel A^2}$ with $q_1^X = (l_1,v_1)$, $q_2^X = (l_2,v_2)$, $l_1,l_2 \in \mathit{Loc}^{A^1\parallel A^2}$, $v_1, v_2 \in [\mathit{Clk} \mapsto \mathbb{R}_{\geq 0}]$, $v_1\models \varphi$, $v_2 = v_1[r\mapsto 0]_{r\in c}$, and $v_2\models \mathit{Inv}(l_2)$. Now we consider the three cases of Definition~\ref{def:parallelcompositionTIOA} of the parallel composition for TIOA.}
	\begin{itemize}
		\item \new{$a\in\mathit{Act}^1\cap\mathit{Act}^2$. It follows directly from Lemma~\ref{lemma:parallelcompositionTSandAsharedaction} that $q_1^X \xrightarrow{a} q_2^X$ is a transition in $\sem{A^1} \parallel \sem{A^2}$. Since $q_2^X\in\mathrm{cons}$, it holds that $q_1^X \xrightarrow{a} q_2^X$ is a transition in $Y$.}
		
		\item \new{$a\in\mathit{Act}^1\setminus\mathit{Act}^2$. It follows directly from Corollary~\ref{cor:parallelcompositionTSandAnonsharedaction} that $q_1^X \xrightarrow{a} q_2^X$ is a transition in $\sem{A^1} \parallel \sem{A^2}$. Since $q_2^X\in\mathrm{cons}$, it holds that $q_1^X \xrightarrow{a} q_2^X$ is a transition in $Y$.}
		
		\item \new{$a\in\mathit{Act}^2\setminus\mathit{Act}^1$. It follows directly from Corollary~\ref{cor:parallelcompositionTSandAnonsharedaction} (where we switched $A^1$ and $A^2$) that $q_1^X \xrightarrow{a} q_2^X$ is a transition in $\sem{A^1} \parallel \sem{A^2}$. Since $q_2^X\in\mathrm{cons}$, it holds that $q_1^X \xrightarrow{a} q_2^X$ is a transition in $Y$.}
	\end{itemize}
	\new{Now consider that $a$ is a delay $d$. It follows directly from Lemma~\ref{lemma:parallelcompositionTSandAdelay} that $q_1^X \xrightarrow{d} q_2^X$ is a transition in $\sem{A^1} \parallel \sem{A^2}$. Since $q_2^X\in\mathrm{cons}$, it holds that $q_1^X \xrightarrow{d} q_2^X$ is a transition in $Y$.
	}

	\new{We have shown that when $q_1^X \xrightarrow{a} q_2^X$ is a transition in $X = (\sem{A^1 \parallel A^2})^{\Delta}$, it holds that $q_1^X \xrightarrow{a} q_2^X$ is a transition in $Y = (\sem{A^1} \parallel \sem{A^2})^{\Delta}$. Since the transition is arbitrarily chosen, it holds for all transitions in $X$.}
	
	\new{($\Leftarrow$) Assume a transition $q_1^Y\xrightarrow{a} q_2^Y$ in $Y$. From Definition~\ref{def:adversarialpruning} it follows that $q_1^Y\xrightarrow{a} q_2^Y$ in $\sem{A^1} \wedge \sem{A^2}$ and $q_2^Y \in \mathrm{cons}$. Now we consider the three cases of Definition~\ref{def:parallelcompositionTIOTS} of the parallel composition for TIOTS.}
	\begin{itemize}
		\item \new{$a\in\mathit{Act}^1\cap\mathit{Act}^2$. It follows directly from Lemma~\ref{lemma:parallelcompositionTSandAsharedaction} that $q_1^Y \xrightarrow{a} q_2^Y$ is a transition in $\sem{A^1 \parallel A^2}$. Since $q_2^Y\in\mathrm{cons}$, it holds that $q_1^Y \xrightarrow{a} q_2^Y$ is a transition in $X$.}
		
		\item \new{$a\in\mathit{Act}^1\setminus\mathit{Act}^2$. From time reflexivity of Definition~\ref{def:tiots} we have that $q_2^Y\xlongrightarrow{d}$ with $d = 0$. From Definitions~\ref{def:adversarialpruning} and~\ref{def:parallelcompositionTIOTS} it follows that $q_2^{\sem{A^1}} \xlongrightarrow{d}$ and $q^{\sem{A^2}} \xlongrightarrow{d}$. Now, from Definition~\ref{def:semanticTIOA} it follows that $v^2 + d \models \mathit{Inv}^2(l^2)$, i.e., $v^2 \models \mathit{Inv}^2(l^2)$.}
		
		\new{It now follows directly from Lemma~\ref{lemma:parallelcompositionTSandAnonsharedaction} that $q_1^Y \xrightarrow{a} q_2^Y$ is a transition in $\sem{A^1 \parallel A^2}$. Since $q_2^Y\in\mathrm{cons}$, it holds that $q_1^Y \xrightarrow{a} q_2^Y$ is a transition in $X$.}
		
		\item \new{$a\in\mathit{Act}^2\setminus\mathit{Act}^1$. From time reflexivity of Definition~\ref{def:tiots} we have that $q_2^Y\xlongrightarrow{d}$ with $d = 0$. From Definitions~\ref{def:adversarialpruning} and~\ref{def:parallelcompositionTIOTS} it follows that $q^{\sem{A^1}} \xlongrightarrow{d}$ and $q_2^{\sem{A^2}} \xlongrightarrow{d}$. Now, from Definition~\ref{def:semanticTIOA} it follows that $v^1 + d \models \mathit{Inv}^1(l^1)$, i.e., $v^1 \models \mathit{Inv}^1(l^1)$.}
		
		\new{It now follows directly from Lemma~\ref{lemma:parallelcompositionTSandAnonsharedaction} (where we switched $A^1$ and $A^2$) that $q_1^Y \xrightarrow{a} q_2^Y$ is a transition in $\sem{A^1 \parallel A^2}$. Since $q_2^Y\in\mathrm{cons}$, it holds that $q_1^Y \xrightarrow{a} q_2^Y$ is a transition in $X$.}
	\end{itemize}
	\new{Now consider that $a$ is a delay $d$. It follows directly from Lemma~\ref{lemma:parallelcompositionTSandAdelay} that $q_1^Y \xrightarrow{d} q_2^Y$ is a transition in $\sem{A^1 \parallel A^2}$. Since $q_2^Y\in\mathrm{cons}$, it holds that $q_1^Y \xrightarrow{d} q_2^Y$ is a transition in $X$.
	}

	\new{We have shown that when $q_1^Y \xrightarrow{a} q_2^Y$ is a transition in $Y = \sem{A^1} \parallel \sem{A^2}$, it holds that $q_1^Y \xrightarrow{a} q_2^Y$ is a transition in $X = \sem{A^1 \parallel A^2}$. Since the transition is arbitrarily chosen, it holds for all transitions in $Y$.}
\end{proof}

\new{Finally, the following corollary describes a special case of Theorem~\ref{thrm:parallelcompositionTSandA}, which happens to be one of the unproven main theorems in our previous work~\cite{david_real-time_2015}.}
\begin{corollary}
  \new{Given two specification automata $A^i = (\mathit{Loc}^i, l_0^i, \mathit{Act}^i, \mathit{Clk}^i, E^i, \allowbreak\mathit{Inv}^i), i=1,2$ where $\mathit{Act}_o^1 \cap \mathit{Act}_o^2 =\emptyset$ and $\mathit{Act}^1 = \mathit{Act}^2$. Then \new{$\sem{A^1 \parallel A^2} = \sem{A^1} \parallel\sem{A^2}$}.}
\end{corollary}
\begin{proof}
	\new{This corollary follows directly as a special case from the proof of Theorem~\ref{thrm:parallelcompositionTSandA}. The special case only depends on Lemmas~\ref{lemma:parallelcompositionTSandAsamestateset} and~\ref{lemma:parallelcompositionTSandAsharedaction}, which do not require adversarial pruning to be applied.}
\end{proof}

\section{Quotient}
\label{sec:quotient}


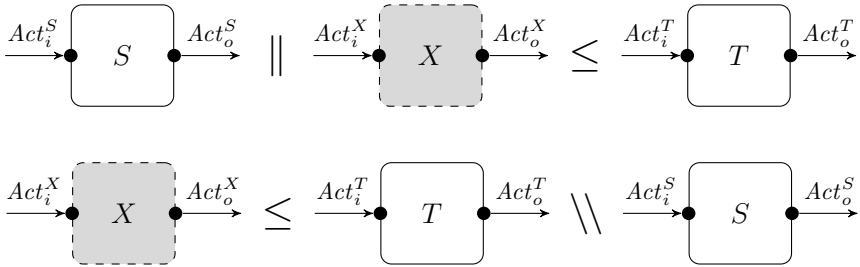
\begin{figure}
	\centering
	\begin{tikzpicture}[->,>=stealth',shorten >=1pt, font=\small, node distance=3cm, scale=.9, transform shape]
		\node[draw, shape=rectangle, minimum size=1.5cm, rounded corners] (S) {\large $S$};
		\node[draw, shape=rectangle, minimum size=1.5cm, rounded corners, dashed, fill=gray!30!white] (X) [right = of S] {\large $X$};
		\node[draw, shape=rectangle, minimum size=1.5cm, rounded corners] (T) [right = of X] {\large $T$};
		
		\path (S) -- node {\Large $\parallel$} (X);
		\path (X) -- node {\Large $\leq$} (T);
		
		\node[connector] (Sin) at (S.180) {};
		\node[connector] (Sout) at (S.0) {};
		\node[connector] (Xin) at (X.180) {};
		\node[connector] (Xout) at (X.0) {};
		\node[connector] (Tin) at (T.180) {};
		\node[connector] (Tout) at (T.0) {};
		
		\draw[->,>=stealth',shorten >=1pt, every node/.style={action}] 
		(Sout) edge node[above] {$\mathit{Act}_o^S$} ++(0:1)
		(Xout) edge node[above] {$\mathit{Act}_o^X$} ++(0:1)
		(Tout) edge node[above] {$\mathit{Act}_o^T$} ++(0:1)
		;
		\draw[<-,>=stealth',shorten >=1pt, every node/.style={action}] 
		(Sin) edge node[above] {$\mathit{Act}_i^S$} ++(180:1)
		(Xin) edge node[above] {$\mathit{Act}_i^X$} ++(180:1)
		(Tin) edge node[above] {$\mathit{Act}_i^T$} ++(180:1)
		;
	\end{tikzpicture}
	
	\vspace{2em}
	\begin{tikzpicture}[->,>=stealth',shorten >=1pt, font=\small, node distance=3cm, scale=.9, transform shape]
		\node[draw, shape=rectangle, minimum size=1.5cm, rounded corners, dashed, fill=gray!30!white] (X) {\large $X$};
		\node[draw, shape=rectangle, minimum size=1.5cm, rounded corners] (T) [right = of X] {\large $T$};
		\node[draw, shape=rectangle, minimum size=1.5cm, rounded corners] (S) [right = of T] {\large $S$};
		
		\path (X) -- node {\Large $\leq$} (T);
		\path (T) -- node {\Large $\quotient$} (S);
		
		\node[connector] (Sin) at (S.180) {};
		\node[connector] (Sout) at (S.0) {};
		\node[connector] (Xin) at (X.180) {};
		\node[connector] (Xout) at (X.0) {};
		\node[connector] (Tin) at (T.180) {};
		\node[connector] (Tout) at (T.0) {};
		
		\draw[->,>=stealth',shorten >=1pt, every node/.style={action}] 
		(Sout) edge node[above] {$\mathit{Act}_o^S$} ++(0:1)
		(Xout) edge node[above] {$\mathit{Act}_o^X$} ++(0:1)
		(Tout) edge node[above] {$\mathit{Act}_o^T$} ++(0:1)
		;
		\draw[<-,>=stealth',shorten >=1pt, every node/.style={action}] 
		(Sin) edge node[above] {$\mathit{Act}_i^S$} ++(180:1)
		(Xin) edge node[above] {$\mathit{Act}_i^X$} ++(180:1)
		(Tin) edge node[above] {$\mathit{Act}_i^T$} ++(180:1)
		;
	\end{tikzpicture}
	\caption{\new{Concept of quotient for given specifications $S$ and $T$ and the unknown implementation $X$.}}
	\label{fig:quotient}
\end{figure}

An essential operator in a complete specification theory is the one of {\em quotienting}. It allows for factoring out behavior from a larger component. If one has a large component specification $T$ and a small one $S$ then $T \quotient S$ is the specification of exactly those components that when composed with $S$ refine $T$. (\new{In this context, larger does not necessarily means bigger in terms of the size of the state set, action set or transition relation, but more like higher in the refinement hierarchy.}) In other words, $T\quotient S$ specifies the work that still needs to be done, given availability of an implementation of $S$, in order to provide an implementation of $T$.

\new{
Figure~\ref{fig:quotient} shows the conceptual idea behind the quotient operator, which is formalized in Theorem~\ref{thm:quotient_correct} later in this section. Given two specification $S$ and $T$, each with its own action sets, the purpose of the quotient operator is to calculate the specification for the missing implementation $X$. Once $X$ is put parallel to $S$, the combined system should refine $T$. From this figure, we can derive the true minimal requirement on the quotient operator (our previous work had stricter requirements): $\mathit{Act}_o^S \cap \mathit{Act}_i^T = \emptyset$, i.e., an action cannot be classified as an output action in $S$ and an input action in $T$. 
}

We proceed like for structural and logical compositions and start with a \new{quotient} that may introduce error states. Those errors \new{can} then pruned \new{if desired}. 
\begin{definition}\label{def:quotientTIOTS}
	Given specifications $S = (Q^S,q_0^S,\mathit{Act}^S,\rightarrow^S)$ and $T = (Q^T,q_0^T,\allowbreak\mathit{Act}^T,\allowbreak\rightarrow^T)$ where $\mathit{Act}_o^S\cap\mathit{Act}_i^T=\emptyset$. The \emph{quotient of $T$ and $S$}, denoted by $T\quotient S$, is a specification $(Q^T\times Q^S\cup\{u, e\}, (q_0^T,q_0^S), \mathit{Act},\allowbreak \rightarrow)$ where $u$ is the universal state, $e$ the inconsistent state, $\mathit{Act}=\mathit{Act}_i\uplus\mathit{Act}_o$ with $\mathit{Act}_i=\mathit{Act}_i^T\cup\mathit{Act}_o^S$ and $\mathit{Act}_o=\mathit{Act}_o^T\setminus\mathit{Act}_o^S \cup \mathit{Act}_i^S\setminus\mathit{Act}_i^T$, and $\rightarrow$ is defined as
	\begin{enumerate}
		\item $(q_1^T,q_1^S)\xlongrightarrow{a} (q_2^T,q_2^S)$ if $a\in\mathit{Act}^S\cap\mathit{Act}^T$, $q_1^T\xlongrightarrow{a}{}^{\!\! T} q_2^T$, and $q_1^S\xlongrightarrow{a}{}^{\!\! S} q_2^S$
		\item $(q^T,q_1^S)\xlongrightarrow{a} (q^T,q_2^S)$ if $a\in\mathit{Act}^S\setminus\mathit{Act}^T$, $q^T\in Q^T$, and $q_1^S\xlongrightarrow{a}{}^{\!\! S} q_2^S$
		\item $(q_1^T,q^S)\xlongrightarrow{a} (q_2^T,q^S)$ if $a\in\mathit{Act}^T\setminus\mathit{Act}^S$, $q^S\in Q^S$, and $q_1^T\xlongrightarrow{a}{}^{\!\! T} q_2^T$
		\item $(q_1^T,q_1^S)\xlongrightarrow{d} (q_2^T,q_2^S)$ if $d \in \mathbb{R}_{\geq 0}$, $q_1^T\xlongrightarrow{d}{}^{\!\! T} q_2^T$, and $q_1^S\xlongrightarrow{d}{}^{\!\! S} q_2^S$
		\item \label{enum:unreachable_action} $(q^T,q^S)\xlongrightarrow{a} u$ if $a\in\mathit{Act}_o^S$, $q^T\in Q^T$, and $q^S\arrownot\xlongrightarrow{a}{}^{\!\! S}$
		\item \label{enum:unreachable_delay} $(q^T,q^S)\xlongrightarrow{d} u$ if $d \in \mathbb{R}_{\geq 0}$, $q^T\in Q^T$, and $q^S\arrownot\xlongrightarrow{d}{}^{\!\! S}$
		\item $(q^T,q^S)\xlongrightarrow{a} e$ if $a\in\mathit{Act}_o^S\cap\mathit{Act}_o^T$, $q^T\arrownot\xlongrightarrow{a}{}^{\!\! T}$, and $q^S\xlongrightarrow{a}{}^{\!\! S}$
		\item $u\xlongrightarrow{a}u$ if $a\in\mathit{Act}\cup\mathbb{R}_{\geq 0}$
		\item $e\xlongrightarrow{a}e$ if $a\in\mathit{Act}_i$
	\end{enumerate}
\end{definition}
In this definition, $u$ and $e$ are fresh states such that $u$ is universal (allows arbitrary behavior) and $e$ is inconsistent (no output-controllable behavior can satisfy it). State $e$ disallows progress of time and has no output transitions. The universal state guarantees nothing about the behavior of its implementations (thus any refinement with a suitable alphabet is possible), and dually the inconsistent state  allows no implementations.

\new{The first four rules are part of the standard rules of parallel composition, see Definition~\ref{def:parallelcompositionTIOTS}. Rules 5 and 6 capture the situation where $S$ does not allow a particular output action or delay, respectively, so the parallel composition of $S$ and the quotient also does not allow this to happen. Therefore, it technically does not matter what the quotient does after performing these transitions, hence they go to the universal state $u$. Rule 7 captures the situation that an output shared between $S$ and $T$ as causes a problem in the refinement $S\leq T$ as $T$ is blocking the output. Thus the quotient, representing the missing component put into parallel composition with $S$, needs to block $S$ from performing this output action. But the output action has become an input action in the quotient, so we redirect this output to the error state to `memorize' this problem. Finally, rules 8 and 9 simply express what we mean by universal and error state, respectively.}

Theorem~\ref{thm:quotient_correct} states that the proposed quotient operator has exactly the property that it is dual of structural composition with regards to refinement.

\begin{theorem}\label{thm:quotient_correct} 
    For any two specifications $S$ and $T$ such that the quotient $T\quotient S$ is defined, and for any implementation $X$ over the same alphabet as $T\quotient S$, we have that $S \parallel X$ is defined and $S \parallel X \leq T$ iff $X \leq T \quotient S$.     
\end{theorem}
\noindent The proof of Theorem~\ref{thm:quotient_correct} can be found in Appendix~\ref{app:proofs-quotient}.


Quotienting for TIOA is defined in the following way.
\begin{definition}\label{def:quotientTIOA}
	Given specification automata $S = (\mathit{Loc}^S,l_0^S,\mathit{Act}^S,\mathit{Clk}^S, E^S,\mathit{Inv}^S)$ and $T = (\mathit{Loc}^T,l_0^T,\allowbreak\mathit{Act}^T,\mathit{Clk}^T, E^T,\allowbreak\mathit{Inv}^T)$ where $\mathit{Act}_o^S\cap\mathit{Act}_i^T=\emptyset$. The \emph{quotient of $T$ and $S$}, denoted by $T\quotient S$, is a specification automaton $(\mathit{Loc}^T\times \mathit{Loc}^S\cup\{l_u, l_e\}, (l_0^T,l_0^S), \mathit{Act}, \mathit{Clk}^T \uplus \mathit{Clk}^S \uplus \{x_{\mathit{new}}\}, E, \mathit{Inv})$ where $l_u$ is the universal state, $l_e$ the inconsistent state, $\mathit{Act}=\mathit{Act}_i\uplus\mathit{Act}_o$ with $\mathit{Act}_i=\mathit{Act}_i^T\cup\mathit{Act}_o^S\cup \{i_{\mathit{new}}\}$ and $\mathit{Act}_o=\mathit{Act}_o^T\setminus\mathit{Act}_o^S \cup \mathit{Act}_i^S\setminus\mathit{Act}_i^T$, $\mathit{Inv}((l^T,l^S)) = \mathit{Inv}(l_u) = \mathbf{T}$, $\mathit{Inv}(l_e) = x_{\mathit{new}} \leq 0$ and $E$ is defined as
	\begin{enumerate}
		\item $((l_1^T,l_1^S), a, \varphi^T\wedge \new{\mathit{Inv}(l_2^T)[r\mapsto 0]_{r\in c^T}} \wedge \varphi^S \wedge \new{\mathit{Inv}(l_1^S)} \wedge \new{\mathit{Inv}(l_2^S)[r\mapsto 0]_{r\in c^S}}, c^T\cup c^S, (l_2^T,l_2^S)) \in E$ if $a\in\mathit{Act}^S\cap\mathit{Act}^T$, $(l_1^T, a, \varphi^T,c^T, l_2^T)\in E^T$, and $(l_1^S, a, \varphi^S, c^S, l_2^S)\in E^S$\footnote{\new{Only the target invariant of $T$ matters. $\mathit{Inv}(l_1^S)$ is used to force the complementary edge to the universal state (which depends on $S$, see rules~\ref{enum:unreachable_action} and~\ref{enum:unreachable_delay} in Definition~\ref{def:quotientTIOTS} of quotient for TIOTS), $\mathit{Inv}(l_2^S)[r\mapsto 0]_{r\in c^S}$ is used to ensure the transition only appears in feasible states in the semantic representation as the location invariants are removed.}}
		\item $((l^T,l_1^S), a, \varphi^S \wedge \new{\mathit{Inv}(l_1^S)} \wedge \new{\mathit{Inv}(l_2^S)[r\mapsto 0]_{r\in c^S}}, c^S, (l^T,l_2^S)) \in E$ if $a\in\mathit{Act}^S\setminus\mathit{Act}^T$, $l^T\in \mathit{Loc}^T$, and $(l_1^S, a, \varphi^S, c^S, l_2^S)\in E^S$
		\item $((l^T,l_1^S), a, \neg G_S, \emptyset, l_u) \in E$ if $a\in\mathit{Act}_o^S$, $l^T\in \mathit{Loc}^T$ and $G_S = \bigvee \{\varphi^S \wedge \new{\mathit{Inv}(l_2^S)[r\mapsto 0]_{r\in c^S}} \mid (l_1^S,a,\varphi^S,c^S,l_2^S) \in E^S\}$
		\item $((l^T,l^S), a, \neg \mathit{Inv}(l^S), \new{\emptyset}, l_u) \in E$ if $a \in \mathit{Act}$, $l^T\in \mathit{Loc}^T$, and $l^S\in\mathit{Loc}^S$
		\item $((l_1^T,l_1^S), a, \varphi^S \wedge \new{\mathit{Inv}(l_1^S)} \wedge \new{\mathit{Inv}(l_2^S)[r\mapsto 0]_{r\in c^S}} \wedge \neg G_T, \{x_{\mathit{new}}\}, l_e) \in E$ if $a\in\mathit{Act}_o^S\cap\mathit{Act}_o^T$, $(l^S, a, \varphi^S, c^S, l_2^S)\in E^S$, and $G_T = \bigvee \{\varphi^T \wedge \new{\mathit{Inv}(l_2^T)[r\mapsto 0]_{r\in c^T}} \mid (l_1^T,a,\varphi^T,c^T,l_2^T) \in E^T\}$
		\item $((l^T,l^S), i_{\mathit{new}}, \neg \mathit{Inv}(l^T) \wedge \mathit{Inv}(l^S), \{x_{\mathit{new}}\}, l_e) \in E$ if $l^T\in \mathit{Loc}^T$ and $l^S\in\mathit{Loc}^S$
		\item \new{$((l^T,l^S), i_{\mathit{new}}, \mathit{Inv}(l^T) \vee \neg \mathit{Inv}(l^S), \emptyset, (l^T,l^S)) \in E$ if $l^T\in \mathit{Loc}^T$ and $l^S\in\mathit{Loc}^S$}
		\item $((l_1^T,l^S), a, \varphi^T \wedge \new{\mathit{Inv}(l_2^T)[r\mapsto 0]_{r\in c^T} \wedge \mathit{Inv}(l^S)}, c^T, (l_2^T,l^S)) \in E$ if $a\in\mathit{Act}^T\setminus\mathit{Act}^S$, $l^S\in \mathit{Loc}^S$, and $(l_1^T, a, \varphi^T,c^T, l_2^T)\in E^T$\footnote{\new{Location invariant $\neg \mathit{Inv}(l^S)$ is added to this transition to avoid nondeterminism caused by rule 4. This problem is not present in Definition~\ref{def:quotientTIOTS} of the quotient for TIOTS, as there we can directly refer to the delay action $d$ in rule 5.}}
		\item $(l_u, a, \mathbf{T}, \emptyset, l_u) \in E$ if $a\in\mathit{Act}$
		\item $(l_e, a, x_{\mathit{new}}=0, \emptyset, l_e) \in E$ if $a\in\mathit{Act}_i$
	\end{enumerate}
	and the conjunction of an empty set equals false ($\bigvee\emptyset = \mathbf{F}$). 
\end{definition}

\new{Compared to definitions of the quotient for TIOAs in our previous works, we made several changes to correct minor mistakes. 1) Location invariants in the quotient are simply $\mathbf{T}$, hence the location invariants of the specifications $S$ and $T$ are now included in the transitions of the quotient. For example, rule~1 captures rule~1 of the quotient for TIOTS (Definition~\ref{def:quotientTIOTS}) where transitions are both possible in $S$ and $T$. A transition is possible when the guard is satisfied, captured by the quotient definitions in previous works, \emph{and} the updated valuation satisfies the target location's invariant, see Definition~\ref{def:semanticTIOA} of the semantics. 2) We resolved a potential nondeterminism caused by the combination of rules~4 and~8, i.e., in Definition~\ref{def:quotientTIOA} for any given state either the edge emanating from rule~4 is enabled or the one from rule~8, or none, but never both. 3) Similarly, we resolved a potential nondeterminism caused by the combination of rules~4 and~5. 4) Rule~7 is added to ensure that the quotient is actually input enabled by construction for the new input action $i_{\mathit{new}}$.}

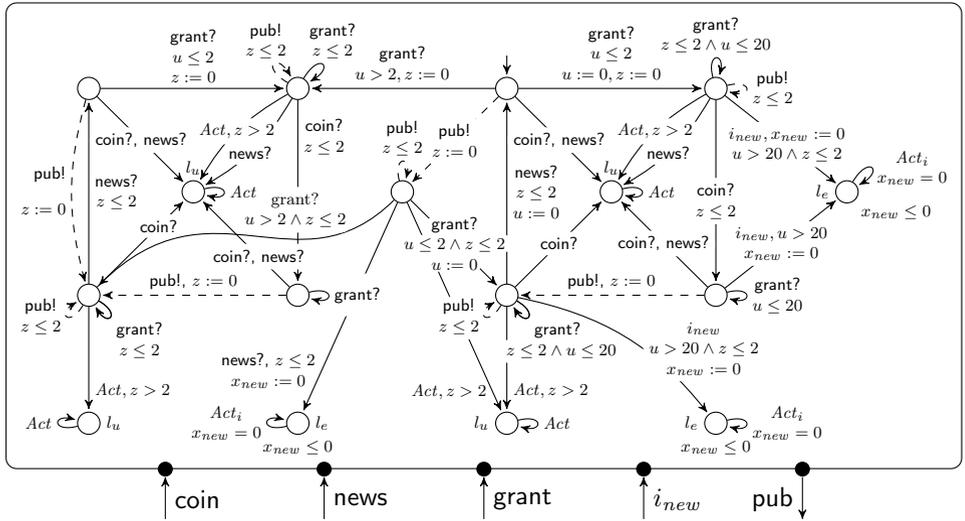
\begin{figure}
	\centering
	\begin{tikzpicture}[->,>=stealth',shorten >=1pt, font=\small, scale=1, transform shape]
		\node[draw, shape=rectangle, rounded corners] (adm) {\begin{tikzpicture}[->,>=stealth',shorten >=1pt, align=left,node distance=3.5cm, align=center, scale=.7]
				\node[main node, initial, initial text={}, initial where=above] (1) {};
				\node[main node] (2) [right = of 1] {};
				\node[main node] (3) [below = of 2] {};
				\node[main node] (4) [below = of 1] {};
				
				\path[every node/.style={action}]
				(1) edge node[above] {grant?\\ $u \leq 2$ \\$u:= 0, z := 0$} (2)
				(2) edge node[pos=.55,fill=white] {coin?\\$z\leq 2$} (3)
				(3) edge[dashed] node[above] {pub!, $z := 0$} (4)
				(4) edge node[right] {news?\\$z\leq 2$\\$u:=0$} (1)
				;
				\path[every node/.style={action}]
				(2) edge[loop above] node[above] {grant?\\$z\leq 2 \wedge u \leq 20$} (2)
				(2) edge[loop right,dashed] node[right] {pub!\\$z\leq 2$} (2)
				(3) edge[loop right] node[right] {grant?\\$u \leq 20$} (3)
				;
				\path[every node/.style={action}] 
				(4) edge [out=330,in=300,looseness=8, min distance=5mm] node[below right] {grant?} node[below right=12pt and -15pt,fill=white] {$z\leq 2 \wedge u\leq 20$} (4)
				(4) edge [out=330,in=300,looseness=8, min distance=5mm] (4) 
				(4) edge [out=240,in=210,looseness=8, min distance=5mm,dashed] node[left] {pub!\\$z\leq 2$} (4)
				;
				
				\node[main node] (2s) [left = of 1] {};
				\node[main node] (1s) [left = of 2s] {};
				\node[main node] (3s) [below = of 2s] {};
				\node[main node] (4s) [below = of 1s] {};
				
				\path[every node/.style={action}]
				(1s) edge node[above] {grant?\\ $u \leq 2$ \\$z := 0$} (2s)
				(2s) edge node[right, pos=.2] {coin?\\$z\leq 2$} (3s)
				(3s) edge[dashed] node[above] {pub!, $z := 0$} (4s)
				(4s) edge node[right] {news?\\$z\leq 2$} (1s)
				;
				\path[bend right = 15, every node/.style={action}]
				(1s) edge[dashed] node[left,pos=.4] {pub!} node[left,pos=.6] {$z:= 0$} (4s)
				;
				\path[every node/.style={action}]
				(3s) edge[loop right] node[right] {grant?} (3s)
				;
				\path[every node/.style={action}] 
				(4s) edge [out=330,in=300,looseness=8, min distance=5mm] node[below right,fill=white] {grant?\\$z\leq 2$} (4s)
				(4s) edge [out=330,in=300,looseness=8, min distance=5mm] (4s) 
				(4s) edge [out=240,in=210,looseness=8, min distance=5mm,dashed] node[left] {pub!\\$z\leq 2$} (4s)
				(2s) edge[out=150,in=120,looseness=8, min distance=6mm,dashed] node[above left = 0pt and -10pt] {pub!\\$z\leq 2$} (2s)
				(2s) edge[out=60,in=30,looseness=8, min distance=6mm] node[above right = 0pt and -10pt] {grant?\\$z\leq 2$} (2s)
				;
				
				\path (1) -- node[main node] (c) {} (3s);
				\path[every node/.style={action}]
				(1) edge node[above] {grant?\\$u > 2, z:= 0$} (2s)
				(1) edge[dashed] node[fill=white] {pub!\\$z:= 0$} (c)
				(c) edge node[fill=white,rectangle,minimum width=7em,minimum height=1em] {} node[fill=white,rectangle,minimum width=3.3em,minimum height=3.3em] {} node[] {grant?\\$u \leq 2 \wedge z\leq 2$\\$u:= 0$} (4)
				(c) edge[loop above,dashed] node[above] {pub!\\$z\leq 2$} (c)
				;
				
				\path (2s) -- coordinate[pos=0.4] (x) (4s);
				\draw 
				(c) .. controls (3s) and (x) .. node[above=3.5pt,pos=0.35,fill=white] {grant?\\$u > 2 \wedge z \leq 2$} (4s);
				
				\path (1) -- node[main node, label={above:$l_u$}] (u1) {} (3);
				\path[every node/.style={action, fill=white}]
				(1) edge node {coin?, news?} (u1)
				(2) edge[bend left=15,pos=.55] node {news?} (u1)
				(2) edge[bend right=15] node {$\mathit{Act}, z > 2$} (u1)
				(3) edge node {coin?, news?} (u1)
				(4) edge node {coin?} (u1)
				;
				
				\node[main node, label={left:$l_u$}] (u3) [below = 2cm of 4] {};
				\begin{scope}[on background layer]
					\path[every node/.style={action}]
					(4) edge node[right, pos=.8] {$\mathit{Act}, z > 2$} (u3)
					;
				\end{scope}
				
				\path (1s) -- node[main node, label={above:$l_u$}] (u1s) {} (3s);
				\path[every node/.style={action, fill=white}]
				(1s) edge node {coin?, news?} (u1s)
				(2s) edge[bend left=15,pos=.55] node {news?} (u1s)
				(2s) edge[bend right=15] node {$\mathit{Act}, z > 2$} (u1s)
				(3s) edge node[pos=.35] {coin?, news?} (u1s)
				(4s) edge node[pos=.7] {coin?} (u1s)
				;
			
				\node[main node, label={right:$l_u$}] (u3s) [below = 2cm of 4s] {};
				\begin{scope}[on background layer]
					\path[every node/.style={action}]
					(4s) edge node[right, pos=.8] {$\mathit{Act}, z > 2$} (u3s)
					;
				\end{scope}
			
				\path[every node/.style={action}]
				(u1) edge[loop right] node[right] {$\mathit{Act}$} (u1)
				(u1s) edge[loop right] node[right] {$\mathit{Act}$} (u1s)
				(u3) edge[loop right] node[right] {$\mathit{Act}$} (u3)
				(u3s) edge[loop left] node[left] {$\mathit{Act}$} (u3s)
				;
			
				\node[main node, label={right:$l_e$}, label={below:$x_{\mathit{new}} \leq 0$}] (e1) [left = of u3] {};
				
				\begin{scope}[on background layer]
					\path[every node/.style={action}]
					(c) edge node[fill=white,rectangle,minimum width=3em,minimum height=2.5em, pos=.44] {} node[left,pos=.8] {news?, $z \leq 2$\\$x_{\mathit{new}} := 0$} (e1)
					(c) edge node[fill=white,rectangle,minimum width=3em,minimum height=2.5em, pos=.55] {} node[left,pos=.9] {$\mathit{Act}, z>2$} (u3)
					;
				\end{scope}
				\path[every node/.style={action}]
				(e1) edge[loop left] node[left] {$\mathit{Act}_i$\\$x_{\mathit{new}} = 0$} (e1)
				;
				
				\node[main node, label={left:$l_e$}, label={below right:$x_{\mathit{new}} \leq 0$}] (e2) [right = 4cm of u1] {};
				\begin{scope}[on background layer]
					\path[every node/.style={action}]
					(2) edge node[pos=.55, fill=white] {$i_{\mathit{new}}, x_{\mathit{new}} := 0$\\$u > 20 \wedge z \leq 2$} (e2)
					(3) edge node[pos=.5, fill=white] {$i_{\mathit{new}}, u > 20$\\$x_{\mathit{new}} := 0$} (e2)
					(e2) edge[out=60,in=30,looseness=8, min distance=6mm] node[right] {$\mathit{Act}_i$\\$x_{\mathit{new}} = 0$} (e2)
					;
				\end{scope}
			
				\node[main node, label={left:$l_e$}, label={below:$x_{\mathit{new}} \leq 0$}] (e3) [right = of u3] {};
				\begin{scope}[on background layer]
					\path[every node/.style={action}]
					(4) edge[bend left=20] node[right, pos=.6, fill=white] {$i_{\mathit{new}}$\\$u > 20 \wedge z \leq 2$\\$x_{\mathit{new}} := 0$} (e3)
					(e3) edge[loop right] node[right] {$\mathit{Act}_i$\\$x_{\mathit{new}} = 0$} (e3)
					;
				\end{scope}
		\end{tikzpicture}};
		
		\path (adm.south west) -- node[connector, pos=0.167] (coin) {} (adm.south east);
		\path (adm.south west) -- node[connector, pos=0.333] (news) {} (adm.south east);
		\path (adm.south west) -- node[connector, pos=0.5] (grant) {} (adm.south east);
		\path (adm.south west) -- node[connector, pos=0.667] (inew) {} (adm.south east);
		\path (adm.south west) -- node[connector, pos=0.833] (pub) {} (adm.south east);
		
		\draw[<-,>=stealth',shorten >=1pt, every node/.style={action}]
		(coin) edge node[right] {coin} ++(-90:0.7)
		(news) edge node[right] {news} ++(-90:0.7)
		(grant) edge node[right] {grant} ++(-90:0.7)
		(inew) edge node[right] {$i_{\mathit{new}}$} ++(-90:0.7)
		;
		\draw[->,>=stealth',shorten >=1pt, every node/.style={action}]
		(pub) edge node[left] {pub} ++(-90:0.7)
		;
	\end{tikzpicture}
	\caption{\new{The quotient \textsf{University} $\quotient$ \textsf{Administration}. To increase readability, we included multiple universal locations and error locations, and omitted the sefloops labeled with $i_{\mathit{new}}$ (see rule 7 of Definition~\ref{def:quotientTIOA}).}}
	\label{fig:quotientexample}
\end{figure}

\new{Figure~\ref{fig:quotientexample} shows the quotient \textsf{University} $\quotient$ \textsf{Administration} as an example. Note that, to increase the readability of the figure, we included multiple universal locations and error locations, while in theory there is only a single universal location and a single error location. Furthermore, we omitted the selfloops labeled with $i_{\mathit{new}}$ that are generated by rule~7 from Definition~\ref{def:quotientTIOA}. As can be seen from this example, there are two potential problems that can result in errors for the university specification: it can take too long ($u > 20$) to create a \textsf{news} item on the research after a \textsf{grant} has been received, or a \textsf{pub} is produces followed by a \textsf{news} item before an actual \textsf{grant} has been received. Any further refinements, including implementations, of \textsf{University} $\quotient$ \textsf{Administration} should resolve these two problems.
}

\new{Observe that $\sem{T\quotient S}$ and $\sem{T}  \quotient \sem{S}$ have different state and action sets. For example, $\sem{T\quotient S}$ has a set of error states $\{(l_e, v) \mid v \in [\mathit{Clk}^{T\quotient S}\mapsto \mathbb{R}_{\geq 0}]\}$, while $\sem{T} \quotient \sem{S}$ only has a single error state $e$. Or $\sem{T\quotient S}$ contains the input action $i_{\mathit{new}}$ which $\sem{T} \quotient \sem{S}$ lacks. This makes relating the quotient for TIOTS and TIOA much more tedious that those theorems in previous sections. Therefore, we have to use bisimulation, Definition~\ref{def:bisimulation}, in the following main theorem, Theorem~\ref{thrm:quotientTSandA}, that lifts all the results from timed input/output transition systems to the symbolic representation level\footnote{\new{In previous works, the bisimulation was replaced by an equation sign, which is technically incorrect.}}.
}

\new{
	\begin{definition}\label{def:bisimulation}
		Given specifications $S = (Q^S,q_0^S,\mathit{Act}^S,\rightarrow^S)$ and $T = (Q^T,q_0^T,\allowbreak\mathit{Act}^T,\allowbreak\rightarrow^T)$. $S$ and $T$ are \emph{bisimilar}, denoted by $S\simeq T$, iff there exists a bisimulation relation $R\subseteq Q^S\times Q^T$ containing $(q_0^S,q_0^T)$ such that for each pair of states $(s,t)\in R$ it holds that
		\begin{enumerate}[itemsep=0pt]
			\item whenever $s\xlongrightarrow{a}{}^{\!\!S} s'$ for some $s'\in Q^S$ and $a\in\mathit{Act}^S\cap\mathit{Act}^T$, then $t\xlongrightarrow{a}{}^{\!\!T}t'$ and $(s',t')\in R$ for some $t'\in Q^T$
			\item whenever $s\xlongrightarrow{a}{}^{\!\!S} s'$ for some $s'\in Q^S$ and $a\in\mathit{Act}^S\setminus\mathit{Act}^T$, then $(s',t)\in R$
			\item whenever $t\xlongrightarrow{a}{}^{\!\!T} t'$ for some $t'\in Q^T$ and $a\in\mathit{Act}^T\cap\mathit{Act}^S$, then $s\xlongrightarrow{a}{}^{\!\!S}s'$ and $(s',t')\in R$ for some $s'\in Q^S$
			\item whenever $t\xlongrightarrow{a}{}^{\!\!T} t'$ for some $t'\in Q^T$ and $a\in\mathit{Act}^T\setminus\mathit{Act}^S$, then $(s,t')\in R$
			\item whenever $s\xlongrightarrow{d}{}^{\!\!S} s'$ for some $s'\in Q^S$ and $d\in \mathbb{R}_{\geq 0}$, then $t\xlongrightarrow{d}{}^{\!\!T}t'$ and $(s',t')\in R$ for some $t'\in Q^T$
			\item whenever $t\xlongrightarrow{d}{}^{\!\!T} t'$ for some $t'\in Q^T$ and $d\in \mathbb{R}_{\geq 0}$, then $s\xlongrightarrow{d}{}^{\!\!S}s'$ and $(s',t')\in R$ for some $s'\in Q^S$
		\end{enumerate}
		Two specification automata $A$ and $B$ are bisimilar, denoted by $A\simeq B$, iff $\sem{A} \simeq \sem{B}$.
	\end{definition}
}

\begin{theorem}\label{thrm:quotientTSandA}
	Given specification automata $S = (\mathit{Loc}^S,l_0^S,\mathit{Act}^S,\mathit{Clk}^S, E^S,\mathit{Inv}^S)$ and $T = (\mathit{Loc}^T,l_0^T,\allowbreak\mathit{Act}^T,\mathit{Clk}^T, E^T,\allowbreak\mathit{Inv}^T)$ where $\mathit{Act}_o^S\cap\mathit{Act}_i^T=\emptyset$. Then \new{$(\sem{T\quotient S})^{\Delta} \simeq (\sem{T}  \quotient \sem{S})^{\Delta}$}.
\end{theorem}
\noindent The proof of Theorem~\ref{thrm:quotientTSandA} can be found in Appendix~\ref{app:proofs-quotient}.

\section{Concluding Remarks}
\label{sec:conclusion}

We have proposed a \new{complete and fully proven} game-based specification theory for timed systems, in which we distinguish between a component and the environment in which it is used. Our contribution is a game-based approach to support both refinement, consistency checking, logical and structural composition, and quotient. 

\new{In the future one could extend our model with (discrete) variables to ease the modeling of cyber-physical systems. This was already suggested by Berendsen and Vaandrager in~\cite{berendsen_compositional_2008}, but only for structural composition.}
One could also investigate whether our approach can be used to perform scheduling of timed systems (see~\cite{alfaro_accelerated_2007,henzinger_embedded_2006,deng_scheduling_1997} for examples). For example, the quotient operation could perhaps be used to synthesize a scheduler for such problem.

\new{In this paper, we ignored the notion of time divergence as defined in~\cite{alfaro_element_2003}. In Section~\ref{sec:specandref} we observed that to verify whether an implementation has time divergence, we need to analyze it in the context of an environment to form a closed-system, as an environment could both ensure or prevent the diverging of time. It would be interesting to investigate whether one could investigate time divergence in a compositional manner.}

\begin{figure}
	\centering
	\includegraphics[width=\textwidth]{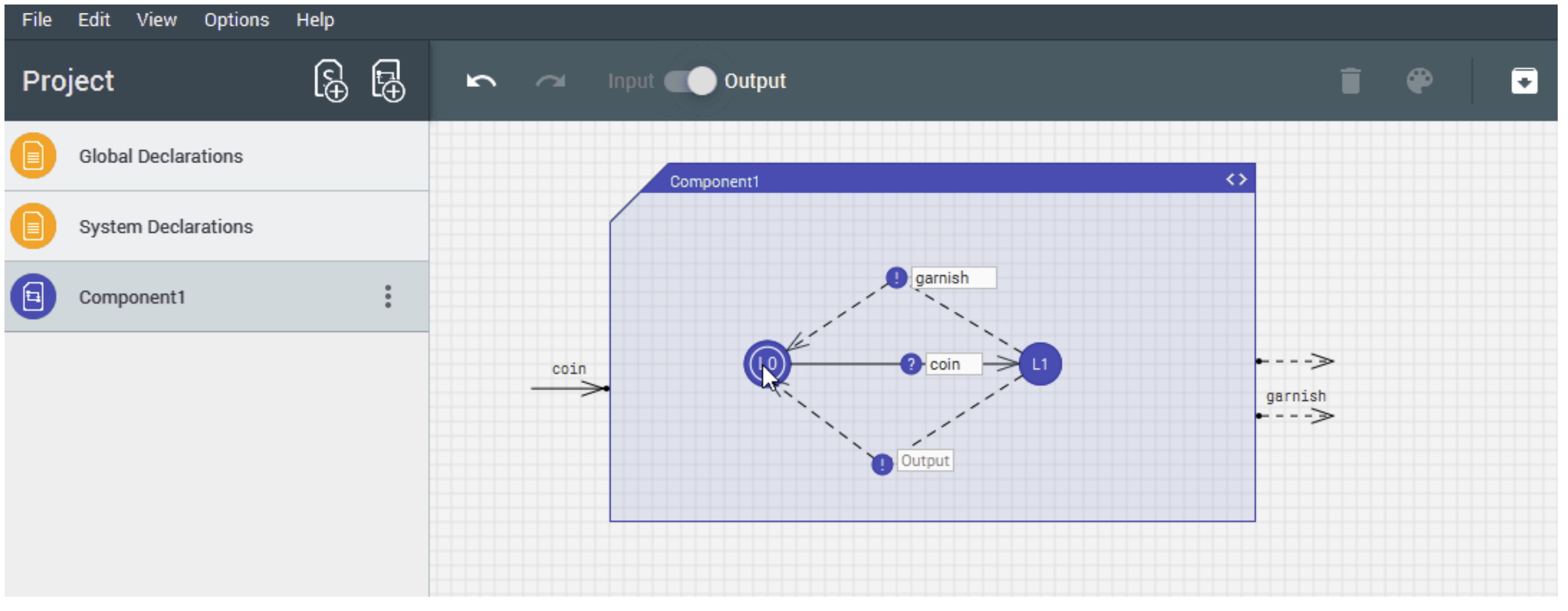}
	\caption{Screenshot of the GUI of Ecdar 2.4}
	\label{fig:EcdarGUI}
\end{figure}

\new{Finally our methodology is being implemented in the open-source tool ECDAR\footnote{\url{http://ecdar.net}}}. Conjunction, composition, and quotienting are simple product constructions allowing for consistency checking to be solved using the zone-based algorithms for synthesizing winning strategies in timed games~\cite{maler_synthesis_1995,cassez_efficident_2005}. Refinement between specifications can be checked using a variant of the pre-existing efficient game-based algorithm~\cite{Bulychev_efficient_2009}. \new{A previous version of the tool was closed-source, contained a few issues and did not implement some of the features, like quotient. The version currently in development contains all features, is thoroughly tested and will support Boolean variables. Besides the implementation of the algorithms, we are also working an on graphical user interface (GUI) to model systems. Figure~\ref{fig:EcdarGUI} shows a screenshot of the current state of the GUI.}


\backmatter

%
%
%

\bmhead{Acknowledgments}

This study was funded by the European Research Council (ERC) Advance grant LASSO, the Villum Investigator grant S4OS, and the Digital Research Center Denmark (DIREC) Bridge grant `Verifiable and Safe AI for Autonomous Systems'.

\section*{Declarations}

The authors have no competing interests to declare that are relevant to the content of this article.

%
%
%

%
%
%
%

\begin{appendices}

\section{Proofs}\label{app:proofs}

This section contains all the proofs not included in the main text of the paper. We will repeat the theorem for clarity before actually providing the proof.

\subsection{Omitted proofs of Section~\ref{sec:specandref}}\label{app:proofs-spec}
\newtheorem*{T1}{Theorem~\ref{thm:local-consistency}}
\begin{T1}
	Every locally consistent specification is consistent in the sense of Definition~\ref{def:consistency}.  
\end{T1}
\begin{proof}
	Let us begin with defining an auxiliary function $\delta$ which chooses a delay for every state $s$ in a \new{locally consistent} specification $S$: 
	\begin{equation*}
		\delta\new{(s)} = 
		\begin{cases}
			d & \text{\new{the infimum $d$} such that } s\xlongrightarrow{d}{}^{\!\!S} s' \text{ and } \exists
			o!: s'\xlongrightarrow{o!}{}^{\!\!S}\\
			+\infty & \text{\new{otherwise}}
		\end{cases}
	\end{equation*}
	Note that \new{since $s$ allows independent progress, it always hold that $s \xlongrightarrow{\delta(s)}{}^{\!\!S}$.} $\delta$ is time additive in the following sense: if $s\xlongrightarrow{d}{}^{\!\!S} s'$ and $d\leq \delta(s)$ then $\delta(s') + d = \delta(s)$, \new{which is} due to time additivity of $\rightarrow^S$, and local consistency of $S$.
	
	We want to show for an arbitrary locally consistent specifications $S$ that it has an implementation. This can be shown by synthesizing an implementation $P = (\new{Q^S, s_0, \mathit{Act}^S}, \rightarrow^P)$, where $\rightarrow^P$ is the largest transition relation generated by the following rules:
	\begin{align*}
		\new{s}\xlongrightarrow{i?}{}^{\!\!P} \new{s'} & \mbox{ if } s\xlongrightarrow{i?}{}^{\!\!S} s' \wedge i?\in\mathit{Act}_i^S \\
		\new{s}\xlongrightarrow{o!}{}^{\!\!P} \new{s'} & \mbox{ if } s\xlongrightarrow{o!}{}^{\!\!S} s' \wedge o!\in\mathit{Act}_o^S \wedge \delta(s) = 0 \\
		\new{s}\xlongrightarrow{d}{}^{\!\!P} \new{s'} & \mbox{ if } s\xlongrightarrow{d}{}^{\!\!S} s' \wedge d\in\mathbb{R}_{\geq 0} \wedge d \leq \delta(s)
	\end{align*}
	
	Since $P$ only takes a subset of transitions of $S$, the determinism of $S$ implies determinism of $P$. The transition relation of $P$ is time-additive due to time additivity of $\rightarrow^S$ and of $\delta$. It is also time-reflexive due to the last rule ($0 \leq \delta(s)$ for every state $s$ and $\rightarrow^S$ was time reflexive). So $P$ is a TIOTS.
	
	The new transition relation is also input enabled as it inherits of input transitions from $S$, which was input enabled. So $P$ is a specification. The second rule guarantees that outputs are urgent (\new{by construction} $P$ only outputs when no further delays are possible). Moreover $P$ observes independent progress. Consider a state \new{$s$ in $P$}. Then if $\delta(s) = +\infty$ clearly $s$ can delay indefinitely. If $\delta(s)$ is finite, then by definition of $\delta$ and of $P$, the state $s$ can delay and hence produce an output. Thus $P$ is an implementation in the sense of Definition~\ref{def:implementation}. 
	
	Now an unsurprising coinductive argument shows that the following relation $R\subseteq Q^S\times Q^S$ witnesses $P\sat S$:
	\begin{equation*}
		R = \left\{ (s,s) \mid s\in Q^S \right\}.
	\end{equation*}
\end{proof}

\newtheorem*{T2}{Theorem~\ref{thm:implementations-minimum}}
\begin{T2}
	Any locally consistent specification $S$ refining an implementation $P$ is an implementation as per Definition~\ref{def:implementation}.
\end{T2}
\begin{proof}
	Observe first that $S$ is already locally consistent, so all its states warrant independent progress. We only need to argue that it satisfies output urgency. Without loss of generality, assume that $S$ only contains states which are reachable by (sequences of) discrete or timed transitions.  
	
	If $S$ only contains reachable states, every state of $S$ has to be related to some state of $P$ in a relation $R$ witnessing $S \leq P$ (output and delay transitions need to be matched in the refinement; input transitions also need to be matched as $P$ is input enabled and $S$ is deterministic). This can be argued for using a standard, though slightly lengthy argument, by formalizing reachable states as a fixpoint of a monotonic operator.
	
	Now that we know that every state of $S$ is related to some state of $P$ consider an arbitrary $s\in Q^S$ and let $p\in Q^P$ be such that $(s,p)\in R$. Then if $s\xlongrightarrow{o!}{}^{\!\!S} s'$ for some state $s'\in Q^S$ and an output $o!\in\mathit{Act}_o^S$, it must be that also $p\xlongrightarrow{o!}{}^{\!\!P} p'$ for some state $p'\in Q^P$ (and $(s',p')\in R$). 
	But since $P$ is an implementation, its outputs must be urgent, so $p\arrownot\xlongrightarrow{d}{}^{\!\!P}$ for all $d>0$, and consequently $s\arrownot\xlongrightarrow{d}{}^{\!\!S}$ for all $d>0$. We have shown that all states of $S$ have urgent outputs (if any) and thus $S$ is an implementation.
\end{proof}

\subsection{Omitted proofs of Section~\ref{sec:consandconj}}\label{app:proofs-conj}
\newtheorem*{T3}{Theorem~\ref{thm:consistency}}
\begin{T3}
	A specification \(S = (Q,s_0,\mathit{Act},\rightarrow)\) is consistent iff $s_0\in\mathrm{cons}$.
\end{T3}
\begin{proof}
	First, assume that $s_0\in\mathrm{cons}^S$. Show that $S$ is consistent in the sense of Definition~\ref{def:consistency}. In a similar fashion to the proof of Theorem~\ref{thm:local-consistency} we first postulate existence of a function $\delta$, which chooses a delay and an output for every consistent state $s$:
	\begin{equation*}
		\delta\new{(s)} = 
		\begin{cases}
			d & \text{if } \exists s', s'' \in \mathrm{cons}^S: \text{ \new{the infimum $d$} such that } s\xlongrightarrow{d}{}^{\!\!S} s' \\
			& \text{and } \exists
			o!: s'\xlongrightarrow{o!}{}^{\!\!S} s''\\
			+\infty & \text{\new{otherwise}}
		\end{cases}
	\end{equation*}
	
	\new{Note that $\delta$ is time additive in the following sense: if $s\xlongrightarrow{d}{}^{\!\!S} s'$ and $d\leq \delta(s)$ then $\delta(s') + d = \delta(s)$, which is due to} time additivity of $\rightarrow^S$ and the fact that $\mathrm{cons}^S$ is a fixpoint of $\Theta^S$.
	
	We show this by constructing an implementation $P=(\new{Q^S,s_0,\mathit{Act}^S},\rightarrow^P)$ \new{where} the transition relation is the largest relation generated by the following rules:
	\begin{enumerate}
		\item $s\xlongrightarrow{o!}{}^{\!\!P} s'$ iff $s\xlongrightarrow{o!}{}^{\!\!S} s'$ and $s'\in\mathrm{cons}^S$ and $\delta_s = 0$,
		\item $s\xlongrightarrow{i?}{}^{\!\!P} s'$ iff $s\xlongrightarrow{i?}{}^{\!\!S} s'$,
		\item $s\xlongrightarrow{d}{}^{\!\!P} s'$ iff $s\xlongrightarrow{d}{}^{\!\!S} s'$ and $d \leq \delta_s$.
	\end{enumerate}
	
	Observe that the construction of $P$ is essentially identical to the one in the proof of Theorem~\ref{thm:local-consistency} above. It can be argued in almost the same way as in the above proof, that $P$ satisfies the axioms of TIO\new{TS}s and is an implementation. Here one has to use the definition of $\Theta^S$ in order to see that the side condition in the first rule, that is $s'\in\mathrm{cons}^S$, does not introduce a violation of independent progress.
	
	It remains to argue that $P\sat S$.  This is done by arguing that the following relation $R$
	\begin{equation*} 
		R = \left\{ (p,s) \in Q^S \times Q^S \mid p = s \right\} 
	\end{equation*}
	witnesses the refinement of $S$ by $P$.
	
	Consider now the other direction. Assume that $S$ is consistent and show that $s_0\in\mathrm{cons}^S$.  In the following we write that a state $s$ is consistent meaning that a specification would be consistent if $s$ was the initial state. Let $X = \{ s\in Q^S \mid s \text{ is consistent} \}$. It suffices to show that $X$ is a post-fixed point of $\Theta^S$, thus $X \subseteq \Theta^S(X)$ (then $s_0 \in X=\mathrm{cons}^S$). 
	
	Since $s$ is consistent, let us consider an implementation $P$ and a state $p$ such that $p\sat s$. We will show that $s\in \Theta^S(X)$. Consider an arbitrary $d\geq 0$ and the first disjunct in the definition of $\Theta^S$. If $p\xlongrightarrow{d}{}^{\!\!P} p^d$ then also $s\xlongrightarrow{d}{}^{\!\!S} s^d$ and $p^d\sat s^d$, so \(s^d\in X\). Consider an arbitrary input $i?$ such that $s^d\xlongrightarrow{i?}{}^{\!\!S} s'$. Then also $p^d\xlongrightarrow{i?}{}^{\!\!P} p'$ and $p'\sat s'$ (by satisfaction). But then $s'\in X$. So by the first disjunct of definition of $\Theta^S$ we have that $s\in \Theta^S(X)$.
	
	If $p\arrownot\xlongrightarrow{d}{}^{\!\!P}$ for our fixed value of $d$, then by independent progress of $p$ there exists a $d_\mathrm{max} < d$ such that $p\xlongrightarrow{d_\mathrm{max}}{}^{\!\!P} p'$ for some $p'$ and $p'\xlongrightarrow{o!}{}^{\!\!P} p''$ for some $p''$ and some output $o!$. By $p\sat s$ there also exist $s'$ and $s''$ such that $s\xlongrightarrow{d_\mathrm{max}}{}^{\!\!S} s'$ and $s'\xlongrightarrow{o!}{}^{\!\!S} s''$. Moreover $p''\sat s''$, so $s''\in X$, which by the second disjunct in the definition of $\Theta^S$ implies that $s\in\Theta^S(X)$.
	
	So we conclude that $X$ is a fixpoint of $\Theta^S$. Since $s_0$ is consistent by assumption, then $s_0\in X \subseteq \mathrm{cons}^S$.
\end{proof}

\newtheorem*{T4}{Theorem~\ref{thm:prune}}
\begin{T4}
	For a consistent specification $S$, $S^{\Delta}$ is locally consistent and $\mod{S} = \mod{S^{\Delta}}$.
\end{T4}
\begin{proof}
	\new{We first proof that $S^{\Delta}$ is locally consistent. From Definitions~\ref{def:local-consistency} and~\ref{def:implementation} of local consistency and implementation, respectively, it follows that we have to show that $\forall q\in Q^{S^{\Delta}}$: either $\forall d \in \mathbb{R}_{\geq 0} : q\xlongrightarrow{d}{}^{\!\!P}$ or $\exists d\in\mathbb{R}_{\geq 0}, \exists o!\in\mathit{Act}_o$ s.t. $q\xlongrightarrow{d}{}^{\!\!P}q'$ and $q'\xlongrightarrow{o!}{}^{\!\!P}$. From Definition~\ref{def:adversarialpruning} of adversarial pruning it follows that $Q^{S^{\Delta}} = \mathrm{cons}$.}
	
	\new{Consider a state $q\in \mathrm{cons}$. From the definition of $\Theta$, it follows that $q \in \overline{\mathrm{err}(\overline{\mathrm{cons}})}$ and $ q \in \{q_1\in Q\mid \forall d \geq 0: [\forall q_2\in Q: q_1\xlongrightarrow{d} q_2 \Rightarrow q_2\in \mathrm{cons} \wedge \forall i?\in \mathit{Act}_i: \exists q_3\in \mathrm{cons}: q_2\xlongrightarrow{i?} q_3]\ \vee [\exists d'\leq d \wedge \exists q_2,q_3\in \mathrm{cons} \wedge \exists o!\in\mathit{Act}_o: q_1\xlongrightarrow{d'} q_2 \wedge q_2\xlongrightarrow{o!} q_3 \wedge \forall i?\in\mathit{Act}_i: \exists q_4\in \mathrm{cons}: q_2\xlongrightarrow{i?} q_4]\}$. In case that the condition $[\exists d'\leq d \wedge \exists q_2,q_3\in \mathrm{cons} \wedge \exists o!\in\mathit{Act}_o: q_1\xlongrightarrow{d'} q_2 \wedge q_2\xlongrightarrow{o!} q_3 \wedge \forall i?\in\mathit{Act}_i: \exists q_4\in \mathrm{cons}: q_2\xlongrightarrow{i?} q_4]$ holds for some $d$, then it follows immediately that $q$ allows independent progress. In the other case, i.e., there does not exists a $d$ such that $[\exists d'\leq d \wedge \exists q_2,q_3\in \mathrm{cons} \wedge \exists o!\in\mathit{Act}_o: q_1\xlongrightarrow{d'} q_2 \wedge q_2\xlongrightarrow{o!} q_3 \wedge \forall i?\in\mathit{Act}_i: \exists q_4\in \mathrm{cons}: q_2\xlongrightarrow{i?} q_4]$ holds, it follows from the fact that $q \in \overline{\mathrm{err}(\overline{\mathrm{cons}})}$ and Definition~\ref{def:error} that $\forall d \in \mathbb{R}_{\geq 0} : q\xlongrightarrow{d}{}^{\!\! P}$, thus allowing independent progress.}
	
	\new{We now show that $\mod{S} = \mod{S^{\Delta}}$. From Definition~\ref{def:satisfaction} it follows that $\mod{S} = \mod{S^{\Delta}}$ iff for all implementations $P$ it holds that $P\leq S \Leftrightarrow P\leq S^{\Delta}$.}
	
	\new{($P\leq S \Rightarrow P\leq S^{\Delta}$)
		Consider an implementation $P$ such that $P\leq S$. This implies from Definition~\ref{def:refinement} of refinement that there exists a relation $R \subseteq Q^P \times Q^S$ witnessing the refinement. We will arguing that for any pair $(p,s) \in R$ it holds that $s\in\mathrm{cons}$.}
	
	\new{For this, consider the controllable predecessor operator $\pi$ and $\pi(\mathrm{imerr})$ to understand what it exactly calculates with respect to the definition of a consistent specification. A state $q\in\pi(\mathrm{imerr})$ is either directly an error state or it can first delay followed by an input action to reach an error state \emph{without} encountering an output action preventing it reaching an error state. With other words, no implementation can prevent state $q$ from reaching an error state.}
	
	\new{Now, denote $\pi^n(\mathrm{err})$ the $n$-th iteration of the fixed-point calculation, i.e., $\pi^1(\mathrm{imerr}) = \pi(\mathrm{imerr})$, $\pi^2(\mathrm{imerr}) = \pi(\pi(\mathrm{imerr}))$, etc. Following the above reasoning about the effect of $\pi$ on the reachability of error states, we can formulate the following fixed-point invariant: for each $n$ and $q\in\pi^n(\mathrm{err})$, there does not exists an implementation preventing $q$ from reaching an error state. Once the fixed-point $\mathrm{incons} = \pi(\mathrm{incons}) = \pi^N(\mathrm{imerr})$ for some $N$ is reached, we know for all $q\in \overline{\mathrm{incons}}$ that it cannot reach the fixed-point $\mathrm{incons}$ because either $\mathrm{incons}$ is just simply unreachable by any means or an implementation can prevent it from reaching it.}
	
	\new{Consider a pair $(p,s) \in R$ where $s \in \mathrm{incons}$. This means that specification $S$ cannot be prevented from reaching an error state $s'$. If we follow this path, we end up with pair $(p', s') \in R$. Now, $s'$ is an error state, which either cannot progress time indefinitely and do an output. But since $p'$ is a state from an implementation $P$, it has the independent progress property. Therefore, once the specification wants to do an output or (indefinite) delay, the second or third property from Definition~\ref{def:refinement} is violated. Therefore, we can conclude that for pair $(p,s) \in R$, $s \notin \mathrm{incons}$, i.e., $s\in\mathrm{cons}$. As the argument above does not rely on a specific state $s$ in $S$, it holds for all states $s \in Q^S$.}
	
	\new{Now, we effectively have that $R \subseteq Q^P \times \mathrm{cons}$, thus it follows from Definition~\ref{def:adversarialpruning} of adversarial pruning that $R$ is also a relation witnessing the refinement $P \leq S^{\Delta}$. As we considered an arbitrarily implementation $P$ refining $S$, it holds for all implementations $P$ refining $S$. Therefore, we conclude that $P\leq S \Leftarrow P\leq S^{\Delta}$.}
	
	\new{($P\leq S \Leftarrow P\leq S^{\Delta}$)
		This case follows directly from the construction of $S^{\Delta}$ and the fact that $\mathrm{cons}\subseteq Q^S$, i.e., for all implementations $P$ that refine $S^{\Delta}$, the binary relation $R\subseteq Q^P \times \mathrm{cons}$ also witnesses the refinement of $P$ and $S$.}
\end{proof}

\newtheorem*{T5}{Lemma~\ref{lem:common-implementation}}
\begin{T5}
	For two specifications $S$, $T$, and their states $s$ and $t$, respectively, if there exists an implementation $P$ and its state $p$ such that simultaneously $p\sat s$ and $p\sat t$ then $(s,t)\in\mathrm{cons}^{S\wedge T}$.  
\end{T5}
\begin{proof}
	This is shown by arguing that the following set $X$ of states of $S\wedge T$ is a postfixed point of $\Theta$ (then $(s,t)\in X \subseteq \Theta(X) \subseteq \mathrm{cons}^{S\wedge T}$):
	\begin{equation*}
		X = \{ (s,t) \mid \exists P: \exists p\in Q^P : p\sat s \land p\sat t\}.
	\end{equation*}
	
	This is done by checking that $X\subseteq\Theta(X)$. Take $(s,t)\in X$, show that $(s,t)\in\Theta(X)$. So consider an arbitrary $d_0 \geq 0$. We know that there exists state $p$ such that $p\sat s$ and $p\sat t$. Since $p$ is a state of an implementation it guarantees independent progress, so there exists a delay $d^p$ such that $p\xlongrightarrow{d^p}{}^{\!\!P} p'$ for some state $p'$. Now the proof is split in two cases, proceeding by coinduction.
	
	\begin{itemize}
		\item $d^p \leq d_0$ is used to show that $(s,t)\in \Theta(X)$ using a standard argument with the second disjunct in definition of $\Theta$ (namely that $p$ can delay and output leading to a refinement of successors of $s$ and $t$, which again will be in $X$).
		
		\item $d^p > d_0$ is used to show that $(s,t)\in \Theta(X)$ using the same kind of argument with the first disjunct in the definition of $\Theta$ (namely that then $p$ can delay $d_0$ time and by refinement for any input transition it can advanced to a state refining successors of $s$ and $t$, which are in $X$).
	\end{itemize}
\end{proof}

\newtheorem*{T6}{Lemma~\ref{lemma:conjunctionTSandAsamestateset}}
\begin{T6}
	\new{Given two TIOAs $A^i = (\mathit{Loc}^i, l_0^i, \mathit{Act}^i, \mathit{Clk}^i, E^i, \mathit{Inv}^i), i=1,2$ where $\mathit{Act}_i^1 \cap \mathit{Act}_o^2 =\emptyset \wedge \mathit{Act}_o^1 \cap \mathit{Act}_i^2 =\emptyset$. Then $Q^{\sem{A^1 \wedge A^2}} = Q^{\sem{A^1} \wedge \sem{A^2}}$ and $q_0^{\sem{A^1 \wedge A^2}} = q_0^{\sem{A^1} \wedge \sem{A^2}}$.}
\end{T6}
\begin{proof}
	\new{For brevity, we write $X = \sem{A^1 \wedge A^2}$, $Y = \sem{A^1} \wedge \sem{A^2}$, and $\mathit{Clk} = \mathit{Clk}^1 \uplus\mathit{Clk}^2$ in the rest of this proof. Following Definition~\ref{def:semanticTIOA} of semantic of a TIOA, Definition~\ref{def:adversarialpruning} of adversarial pruning, Definition~\ref{def:conjunctionTIOTS} of the conjunction for TIOTS, and Definition~\ref{def:conjunctionTIOA} of the conjunction for TIOA, the set of states of $X$ is $Q^X = (\mathit{Loc}^1\times \mathit{Loc}^2) \times [\mathit{Clk}\mapsto \mathbb{R}_{\geq 0}] = \mathit{Loc}^1\times \mathit{Loc}^2 \times [\mathit{Clk}\mapsto \mathbb{R}_{\geq 0}]$ and the set of states of $Y$ is $Q^Y = (\mathit{Loc}^1\times[\mathit{Clk}^1\mapsto \mathbb{R}_{\geq 0}]) \times (\mathit{Loc}^2\times[\mathit{Clk}^2\mapsto \mathbb{R}_{\geq 0}])= \mathit{Loc}^1\times \mathit{Loc}^2\times [\mathit{Clk}\mapsto \mathbb{R}_{\geq 0}]$. Therefore, $Q^X = Q^Y$. Furthermore, it now also follows immediately from the same definitions that $q_0^X = q_0^Y$, as none of these definitions alter the initial location of a TIOA or initial state of a TIOTS.}
\end{proof}

\newtheorem*{T7}{Lemma~\ref{lemma:conjunctionTSandAdelay}}
\begin{T7}
	\new{Given two TIOAs $A^i = (\mathit{Loc}^i, l_0^i, \mathit{Act}^i, \mathit{Clk}^i, E^i, \mathit{Inv}^i), i=1,2$ where $\mathit{Act}_i^1 \cap \mathit{Act}_o^2 =\emptyset \wedge \mathit{Act}_o^1 \cap \mathit{Act}_i^2 =\emptyset$. Denote $X = \sem{A^1 \wedge A^2}$ and $Y = \sem{A^1} \wedge \sem{A^2}$, and let $d\in\mathbb{R}_{\geq 0}$ and $q_1,q_2\in Q^X\cap Q^Y$. Then $q_1\xlongrightarrow{d}{}^{\!\! X} q_2$ if and only if $q_1\xlongrightarrow{d}{}^{\!\! Y} q_2$.}
\end{T7}
\begin{proof}
	\new{First, from Lemma~\ref{lemma:conjunctionTSandAsamestateset} it follows that $Q^X = Q^Y$. Consider a delay $d\in\mathbb{R}_{\geq 0}$. For brevity, in the rest of this proof we write  $\mathit{Clk} = \mathit{Clk}^1 \uplus\mathit{Clk}^2$, and $u^1$ and $u^2$ to indicate the part of a valuation $u$ of only the clocks of $A^1$ and $A^2$, respectively. 
	}
	
	\new{($\Rightarrow$) 
		Assume that $\exists q_1,q_2\in Q^X$ such that $q_1\xlongrightarrow{d}{}^{\!\! X} q_2$. From Definition~\ref{def:semanticTIOA} of the semantic of a TIOA it follows that $q_1 = (l,v)$, $q_2 = (l, v + d)$, $l \in \mathit{Loc}^{A^1 \wedge A^2}$, $v\in [\mathit{Clk}\mapsto \mathbb{R}_{\geq 0}]$, $v + d\models \mathit{Inv}^{A^1 \wedge A^2}(l)$, and $\forall d'\in\mathbb{R}_{\geq 0}, d' < d: v + d'\models \mathit{Inv}^{A^1\wedge A^2}(l)$. From Definition~\ref{def:conjunctionTIOA} of the conjunction for TIOA it follows that $l = (l^1, l^2)$, $l^1 \in \mathit{Loc}^1$, $l^2\in\mathit{Loc}^2$, and $\mathit{Inv}^{A^1\wedge A^2}(l) = \mathit{Inv}^1(l^1) \wedge \mathit{Inv}^2(l^2)$. Therefore, $v + d\models \mathit{Inv}^1(l^1) \wedge \mathit{Inv}^2(l^2)$, and thus $v + d\models\mathit{Inv}^1(l^1)$ and $v + d\models\mathit{Inv}^2(l^2)$. Similarly, $v + d'\models \mathit{Inv}^1(l^1) \wedge \mathit{Inv}^2(l^2)$, and thus $v + d'\models\mathit{Inv}^1(l^1)$ and $v + d'\models\mathit{Inv}^2(l^2)$. Because $\mathit{Clk}^1\cap\mathit{Clk}^2=\emptyset$, it follows that $v^1 + d\models\mathit{Inv}^1(l^1)$, $v^2 + d\models\mathit{Inv}^2(l^2)$, $v^1 + d'\models\mathit{Inv}^1(l^1)$, and $v^2 + d'\models\mathit{Inv}^2(l^2)$. Now, from Definition~\ref{def:semanticTIOA} of the semantic of a TIOA, it follows that $(l^1, v^1)\xlongrightarrow{d}{}^{\!\! \sem{A^1}} (l^1, v^1 + d)$ and $(l^2, v^2)\xlongrightarrow{d}{}^{\!\! \sem{A^2}} (l^2, v^2 + d)$. Finally, from Definition~\ref{def:conjunctionTIOTS} of the conjunction for TIOTS, if follows that $(l^1, v^1, l^2, v^2)\xlongrightarrow{d}{}^{\!\! Y} (l^1, v^1 + d, l^2, v^2 + d)$. Again by using that $\mathit{Clk}^1\cap\mathit{Clk}^2=\emptyset$, we can rewrite the states: $(l^1, v^1, l^2, v^2) = (l^1, l^2, v) = q_1$ and $(l^1, v^1 + d, l^2, v^2 + d) = (l^1, l^2, v + d) = q_2$. Thus $q_1 \xlongrightarrow{d}{}^{\!\! Y} q_2$.
	}
	
	\new{($\Leftarrow$)
		Assume that $\exists q_1,q_2\in Q^Y$ such that $q_1\xlongrightarrow{d}{}^{\!\! Y} q_2$. From Definition~\ref{def:conjunctionTIOTS} of the conjunction for TIOTS it follows that $q_1 = (q_1^1, q_1^2)$, $q_2 = (q_2^1, q_2^2)$, $q_1^1,q_2^1 \in Q^{\sem{A^1}}$, $q_1^2,q_2^2 \in Q^{\sem{A^2}}$, $q_1^1\xlongrightarrow{d}{}^{\!\! \sem{A^1}} q_2^1$, and $q_1^2\xlongrightarrow{d}{}^{\!\! \sem{A^2}} q_2^2$. From Definition~\ref{def:semanticTIOA} of the semantic of a TIOA it follows that for $i = 1,2$: $q_1^i = (l^i, v^i)$, $q_2^i = (l^i, v^i + d)$, $l^i \in \mathit{Loc}^i$, $v^i \in [\mathit{Clk}^i \mapsto \mathbb{R}_{\geq 0}]$, $v^i + d\models \mathit{Inv}^i(l^i)$, and $\forall d'\in\mathbb{R}_{\geq 0}, d' < d: v + d'\models \mathit{Inv}^i(l^i)$. Because $\mathit{Clk}^1\cap\mathit{Clk}^2 = \emptyset$, it follows that for $i=1,2$: $v + d\models \mathit{Inv}^i(l^i)$ and $v + d'\models \mathit{Inv}^i(l^i)$. Now, from Definition~\ref{def:conjunctionTIOA} it follows that $\mathit{Inv}^{A^1 \wedge A^2}(l^1,l^2) = \mathit{Inv}^1(l^1) \wedge \mathit{Inv}^2(l^2)$. Thus we know that $v + d\models \mathit{Inv}^{A^1 \wedge A^2}((l^1, l^2))$ and $v + d'\models \mathit{Inv}^{A^1 \wedge A^2}((l^1, l^2))$. Therefore, using Definition~\ref{def:semanticTIOA} of the semantic of a TIOA, it follows that $(l^1, l^2, v)\xlongrightarrow{d}{}^{\!\! X} (l^1, l^2, v + d)$. Again by using that $\mathit{Clk}^1\cap\mathit{Clk}^2=\emptyset$, we can rewrite the states: $(l^1, l^2, v) = (l^1, v^1, l^2, v^2) = q_1$ and $(l^1, l^2, v + d) = (l^1, v^1 + d, l^2, v^2 + d) = q_2$. Thus $q_1 \xlongrightarrow{d}{}^{\!\! X} q_2$.
	}
	
	\new{As the analysis above holds for any chosen $d \in \mathbb{R}_{\geq 0}$, it holds for all $d$. This concludes the proof.}
\end{proof}

\newtheorem*{T8}{Lemma~\ref{lemma:conjunctionTSandAsharedaction}}
\begin{T8}
	\new{Given two TIOAs $A^i = (\mathit{Loc}^i, l_0^i, \mathit{Act}^i, \mathit{Clk}^i, E^i, \mathit{Inv}^i), i=1,2$ where $\mathit{Act}_i^1 \cap \mathit{Act}_o^2 =\emptyset \wedge \mathit{Act}_o^1 \cap \mathit{Act}_i^2 =\emptyset$. Denote $X = \sem{A^1 \wedge A^2}$ and $Y = \sem{A^1} \wedge \sem{A^2}$, and let $a\in\mathit{Act}^1\cap\mathit{Act}^2$ and $q_1,q_2\in Q^X\cap Q^Y$. Then $q_1\xlongrightarrow{a}{}^{\!\! X} q_2$ if and only if $q_1\xlongrightarrow{a}{}^{\!\! Y} q_2$.}
\end{T8}
\begin{proof}
	\new{First, from Lemma~\ref{lemma:conjunctionTSandAsamestateset} it follows that $Q^X = Q^Y$. For brevity, in the rest of this proof we write  $\mathit{Clk} = \mathit{Clk}^1 \uplus\mathit{Clk}^2$, and $v^1$ and $v^2$ to indicate the part of a valuation $v$ of only the clocks of $A^1$ and $A^2$, respectively. 
	}
	
	\new{($\Rightarrow$) 
		Assume a transition $q_1^X\xrightarrow{a} q_2^X$ in $X$. Following Definition~\ref{def:semanticTIOA} of the semantic, it follows that there exists an edge $(l_1,a,\varphi, c, l_2)\in E^{A^1\wedge A^2}$ with $q_1^X = (l_1,v_1)$, $q_2^X = (l_2,v_2)$, $l_1,l_2 \in \mathit{Loc}^{A^1\wedge A^2}$, $v_1, v_2 \in [\mathit{Clk} \mapsto \mathbb{R}_{\geq 0}]$, $v_1\models \varphi$, $v_2 = v_1[r\mapsto 0]_{r\in c}$, and $v_2\models \mathit{Inv}(l_2)$. 
	}
	
	\new{From Definition~\ref{def:conjunctionTIOA} of the conjunction for TIOA it follows that $(l_1^1, a,\varphi^1, c^1, l_2^1)$ is an edge in $A^1$ and $(l_1^2, a, \varphi^2, c^2, l_2^2)$ in $A^2$, $l_1 = (l_1^1,l_1^2)$, $l_2=(l_2^1,l_2^2)$, $\varphi = \varphi^1 \wedge \varphi^2$, $c  =c^1\cup c^2$. Since $v_1\models \varphi$, it holds that $v_1\models \varphi^1$ and $v_1\models \varphi^2$. Because $\mathit{Clk}^1\cap\mathit{Clk}^2 = \emptyset$, it holds that $v_1^1\models \varphi^1$ and $v_1^2\models\varphi^2$. Also, since $v_2 = v_1[r\mapsto 0]_{r\in c}$, it holds that $v_2^1 = v_1^1[r\mapsto 0]_{r\in c^1}$ and $v_2^2 = v_1^2[r\mapsto 0]_{r\in c^2}$. Finally, because $\mathit{Inv}^{A^1\wedge A^2}(l_2) = \mathit{Inv}^1(l_2^1) \wedge \mathit{Inv}^2(l_2^2)$ (see Definition~\ref{def:conjunctionTIOA}) and $v_2\models \mathit{Inv}^{A^1\wedge A^2}(l_2)$, it follows that $v_2\models\mathit{Inv}^1(l_2^1)$ and $v_2\models\mathit{Inv}^2(l_2^2)$. Since $\mathit{Clk}^1\cap\mathit{Clk}^2 = \emptyset$, it follows that $v_2^1\models\mathit{Inv}^1(l_2^1)$ and $v_2^2\models\mathit{Inv}^2(l_2^2)$.
	}
	
	\new{Combining all the information about $A^1$, we have that $(l_1^1, a,\varphi^1, c^1, l_2^1)$ is an edge in $A^1$, $v_1^1\models \varphi^1$, $v_2^1 = v_1^1[r\mapsto 0]_{r\in c^1}$, and $v_2^1\models\mathit{Inv}^1(l_2^1)$. Therefore, from Definition~\ref{def:semanticTIOA} it follows that $(l_1^1,v_1^1) \xrightarrow{a} (l_2^1,v_2^1)$ is a transition in $\sem{A^1}$. Combining all the information about $A^2$, we have that $(l_1^2, a,\varphi^2, c^2, l_2^2)$ is an edge in $A^2$, $v_1^2\models \varphi^2$, $v_2^2 = v_1^2[r\mapsto 0]_{r\in c^2}$, and $v_2^2\models\mathit{Inv}^2(l_2^2)$. Therefore, from Definition~\ref{def:semanticTIOA} it follows that $(l_1^2,v_1^2) \xrightarrow{a} (l_2^2,v_2^2)$ is a transition in $\sem{A^2}$.
	}
	
	\new{Now, from Definition~\ref{def:conjunctionTIOTS} of the conjunction for TIOTS it follows that $((l_1^1,v_1^1) , (l_1^2,v_1^2)) \xrightarrow{a} ((l_2^1,v_2^1), \allowbreak (l_2^2, v_2^2))$ is a transition in $\sem{A^1} \wedge \sem{A^2}$. Because $\mathit{Clk}^1\cap\mathit{Clk}^2 = \emptyset$, we can rearrange the states into $((l_1^1,v_1^1) , (l_1^2,v_1^2)) = ((l_1^1,l_1^2),v_1) = q_1^X$ and $((l_2^1,v_2^1) , (l_2^2,v_2^2)) = ((l_2^1,l_2^2),v_2) = q_2^X$. Thus, $q_1^X \xrightarrow{a} q_2^X$ is a transition in $\sem{A^1} \wedge \sem{A^2} = Y$.
	}
	
	\new{($\Leftarrow$)
		Assume a transition $q_1^Y\xrightarrow{a} q_2^Y$ in $Y$. From Definition~\ref{def:conjunctionTIOTS} of the conjunction for TIOTS it follows that $q_1^{\sem{A^1}}\xrightarrow{a} q_2^{\sem{A^1}}$ is a transition in $\sem{A^1}$ and $q_1^{\sem{A^2}}\xrightarrow{a} q_2^{\sem{A^2}}$ in $\sem{A^2}$, $q_1^Y = (q_1^{\sem{A^1}},q_1^{\sem{A^2}})$, and $q_2^Y=(q_2^{\sem{A^1}},q_2^{\sem{A^2}})$. From Definition~\ref{def:semanticTIOA} of semantic it follows that there exists an edge $(l_1^1,a,\varphi^1, c^1, l_2^1)\in E^1$ with $q_1^{\sem{A^1}} = (l_1^1,v_1^1)$, $q_2^{\sem{A^1}} = (l_2^1,v_2^1)$, $l_1^1,l_2^1 \in \mathit{Loc}^1$, $v_1^1, v_2^1 \in [\mathit{Clk}^1 \mapsto \mathbb{R}_{\geq 0}]$, $v_1^1\models \varphi^1$, $v_2^1 = v_1^1[r\mapsto 0]_{r\in c^1}$, and $v_2^1\models \mathit{Inv}^1(l_2^1)$. Similarly, it follows from the same definition that there exists an edge $(l_1^2,a,\varphi^2, c^2, l_2^2)\in E^2$ with $q_1^{\sem{A^2}} = (l_1^2,v_1^2)$, $q_2^{\sem{A^2}} = (l_2^2,v_2^2)$, $l_1^2,l_2^2 \in \mathit{Loc}^2$, $v_1^2, v_2^2 \in [\mathit{Clk}^2 \mapsto \mathbb{R}_{\geq 0}]$, $v_1^2\models \varphi^2$, $v_2^2 = v_1^2[r\mapsto 0]_{r\in c^2}$, and $v_2^2\models \mathit{Inv}^2(l_2^2)$.
	}
	
	\new{Now, from Definition~\ref{def:conjunctionTIOA} of the conjunction for TIOA, it follows that there exists an edge $((l_1^1, l_1^2),a,\varphi^1 \wedge \varphi^2, c^1 \cup c^2, (l_2^1,l_2^2))$ in $A^1 \wedge A^2$. Let $v_i, i=1,2$ be the valuations that combines the one from $A^1$ with the one from $A^2$, i.e. $\forall r \in \mathit{Clk}^1: v_i(r) = v_i^1(r)$ and $\forall r \in \mathit{Clk}^2: v_i(r) = v_i^2(r)$. Because $\mathit{Clk}^1\cap\mathit{Clk}^2 = \emptyset$, it holds that $v_1\models \varphi^1$ and $v_1\models \varphi^2$, thus $v_1\models \varphi^1\wedge\varphi^2$; $v_2 = v_1[r\mapsto 0]_{r\in c^1\cup c^2}$; and $v_2\models\mathit{Inv}^1(l_2^1)$ and $v_2\models\mathit{Inv}^2(l_2^2)$, thus $v_2\models\mathit{Inv}^1(l_2^1)\wedge\mathit{Inv}^2(l_2^2)$. 
	}
	
	\new{From Definition~\ref{def:semanticTIOA} it now follows that $((l_1^1,l_1^2), v_1) \xrightarrow{a} ((l_2^1,l_2^2), v_2)$ is a transition in $\sem{A^1 \wedge A^2}$. Because $\mathit{Clk}^1\cap\mathit{Clk}^2 = \emptyset$, we can rearrange the states into $((l_1^1,l_1^2),v_1) = ((l_1^1,v_1^1) , (l_1^2,v_1^2)) = q_1^Y$ and $((l_2^1,l_2^2),v_2) = ((l_2^1,v_2^1) , (l_2^2,v_2^2)) = q_2^Y$. Thus, $q_1^Y \xrightarrow{a} q_2^Y$ is a transition in $\sem{A^1 \wedge A^2} = Y$.}
\end{proof}

\newtheorem*{T9}{Lemma~\ref{lemma:conjunctionTSandAnonsharedaction}}
\begin{T9}
	\new{Given two TIOAs $A^i = (\mathit{Loc}^i, l_0^i, \mathit{Act}^i, \mathit{Clk}^i, E^i, \mathit{Inv}^i), i=1,2$ where $\mathit{Act}_i^1 \cap \mathit{Act}_o^2 =\emptyset \wedge \mathit{Act}_o^1 \cap \mathit{Act}_i^2 =\emptyset$. Denote $X = \sem{A^1 \wedge A^2}$ and $Y = \sem{A^1} \wedge \sem{A^2}$, and let $a\in\mathit{Act}^1\setminus\mathit{Act}^2$ and $q_1,q_2\in Q^X\cap Q^Y$, where $q_2=(l_2^1,l_2^2,v_2)$. If $v_2\models\mathit{Inv}^2(l_2)$, then $q_1\xlongrightarrow{a}{}^{\!\! X} q_2$ if and only if $q_1\xlongrightarrow{a}{}^{\!\! Y} q_2$.}
\end{T9}
\begin{proof}
	\new{First, from Lemma~\ref{lemma:conjunctionTSandAsamestateset} it follows that $Q^X = Q^Y$. For brevity, in the rest of this proof we write  $\mathit{Clk} = \mathit{Clk}^1 \uplus\mathit{Clk}^2$, and $v^1$ and $v^2$ to indicate the part of a valuation $v$ of only the clocks of $A^1$ and $A^2$, respectively. 
	}
	
	\new{($\Rightarrow$) 
		Assume a transition $q_1^X\xrightarrow{a} q_2^X$ in $X$. Following Definition~\ref{def:semanticTIOA} of the semantic, it follows that there exists an edge $(l_1,a,\varphi, c, l_2)\in E^{A^1\wedge A^2}$ with $q_1^X = (l_1,v_1)$, $q_2^X = (l_2,v_2)$, $l_1,l_2 \in \mathit{Loc}^{A^1\wedge A^2}$, $v_1, v_2 \in [\mathit{Clk} \mapsto \mathbb{R}_{\geq 0}]$, $v_1\models \varphi$, $v_2 = v_1[r\mapsto 0]_{r\in c}$, and $v_2\models \mathit{Inv}(l_2)$. 
	}
	
	\new{From Definition~\ref{def:conjunctionTIOA} of the conjunction for TIOA it follows that $(l_1^1, a,\varphi^1, c^1, l_2^1)$ is an edge in $A^1$, $l_1 = (l_1^1,l_1^2)$, $l_2=(l_2^1,l_2^2)$, $l_1^2 = l_2^2 = l^2$, $\varphi = \varphi^1$, $c = c^1$. Since $v_1\models \varphi$ and $\mathit{Clk}^1\cap\mathit{Clk}^2 = \emptyset$, it holds that $v_1^1\models \varphi^1$. Also, since $v_2 = v_1[r\mapsto 0]_{r\in c}$ and $c = c^1$, it holds that $v_2^1 = v_1^1[r\mapsto 0]_{r\in c^1}$ and $v_2^2 = v_1^2$. Finally, because $\mathit{Inv}^{A^1\wedge A^2}(l_2) = \mathit{Inv}^1(l_2^1) \wedge \mathit{Inv}^2(l^2)$ (see Definition~\ref{def:conjunctionTIOA}) and $v_2\models \mathit{Inv}^{A^1\wedge A^2}(l_2)$, it follows that $v_2\models\mathit{Inv}^1(l_2^1)$ and $v_2\models\mathit{Inv}^2(l^2)$\footnote{So the if condition in the lemma is always satisfied once we know that $q_1\xlongrightarrow{a}{}^{\!\! X} q_2$ is a transition in $X$. We formalize this in Corollary~\ref{cor:conjunctionTSandAnonsharedaction}.}. Since $\mathit{Clk}^1\cap\mathit{Clk}^2 = \emptyset$, it follows that $v_2^1\models\mathit{Inv}^1(l_2^1)$ and $v_2^2\models\mathit{Inv}^2(l^2)$.
	}
	
	\new{Combining all the information about $A^1$, we have that $(l_1^1, a,\varphi^1, c^1, l_2^1)$ is an edge in $A^1$, $v_1^1\models \varphi^1$, $v_2^1 = v_1^1[r\mapsto 0]_{r\in c^1}$, and $v_2^1\models\mathit{Inv}^1(l_2^1)$. Therefore, from Definition~\ref{def:semanticTIOA} it follows that $(l_1^1,v_1^1) \xrightarrow{a} (l_2^1,v_2^1)$ is a transition in $\sem{A^1}$. Combining all the information about $A^2$, we have that $v_1^2 = v_2^2$ and $v_2^2\models\mathit{Inv}^2(l^2)$.
	}
	
	\new{Now, from Definition~\ref{def:conjunctionTIOTS} of the conjunction for TIOTS it follows that $((l_1^1,v_1^1) , (l^2,v_1^2)) \xrightarrow{a} ((l_2^1,v_2^1), \allowbreak (l^2, v_1^2))$ is a transition in $\sem{A^1} \wedge \sem{A^2}$. Because $\mathit{Clk}^1\cap\mathit{Clk}^2 = \emptyset$, we can rearrange the states into $((l_1^1,v_1^1) , (l^2,v_1^2)) = ((l_1^1,l^2),v_1) = q_1^X$ and $((l_2^1,v_2^1) , (l^2,v_2^2)) = ((l_2^1,l^2),v_2) = q_2^X$. Thus, $q_1^X \xrightarrow{a} q_2^X$ is a transition in $\sem{A^1} \wedge \sem{A^2} = Y$. 
	}
	
	\new{($\Leftarrow$)
		Assume a transition $q_1^Y\xrightarrow{a} q_2^Y$ in $Y$. From Definition~\ref{def:conjunctionTIOTS} of the conjunction for TIOTS it follows that $q_1^{\sem{A^1}}\xrightarrow{a} q_2^{\sem{A^1}}$ is a transition in $\sem{A^1}$, $q^{\sem{A^2}}\in Q^{\sem{A^2}}$, $q_1^Y = (q_1^{\sem{A^1}},q^{\sem{A^2}})$, and $q_2^Y=(q_2^{\sem{A^1}},q^{\sem{A^2}})$. From Definition~\ref{def:semanticTIOA} of semantic it follows that there exists an edge $(l_1^1,a,\varphi^1, c^1, l_2^1)\in E^1$ with $q_1^{\sem{A^1}} = (l_1^1,v_1^1)$, $q_2^{\sem{A^1}} = (l_2^1,v_2^1)$, $l_1^1,l_2^1 \in \mathit{Loc}^1$, $v_1^1, v_2^1 \in [\mathit{Clk}^1 \mapsto \mathbb{R}_{\geq 0}]$, $v_1^1\models \varphi^1$, $v_2^1 = v_1^1[r\mapsto 0]_{r\in c^1}$, and $v_2^1\models \mathit{Inv}^1(l_2^1)$. Similarly, it follows from the same definition that $q^{\sem{A^2}} = (l^2,v^2)$, $l^2\in \mathit{Loc}^2$, and $v^2 \in [\mathit{Clk}^2 \mapsto \mathbb{R}_{\geq 0}]$.
	}
	
	\new{Now, from Definition~\ref{def:conjunctionTIOA} of the conjunction for TIOA, it follows that there exists an edge $((l_1^1, l^2),a,\varphi^1, c^1, (l_2^1,l^2))$ in $A^1 \wedge A^2$. Let $v_i, i=1,2$, be a valuation that combines the one from $A^1$ with the one from $A^2$, i.e. $\forall r \in \mathit{Clk}^1: v_i(r) = v_i^1(r)$ and $\forall r \in \mathit{Clk}^2: v_i(r) = v_i^2(r)$. Because $\mathit{Clk}^1\cap\mathit{Clk}^2 = \emptyset$, it holds that $v_1\models \varphi^1$; $v_2 = v_1[r\mapsto 0]_{r\in c^1}$ with $v_1^2 = v_2^2$; and $v_2\models\mathit{Inv}^1(l_2^1)$. As the antecedent states that $v_2\models\mathit{Inv}^2(l^2)$, it follows that $v^2\models\mathit{Inv}(l_2^1)\wedge\mathit{Inv}(l^2)$.
	}
	
	\new{From Definition~\ref{def:semanticTIOA} it now follows that $((l_1^1,l^2), v_1) \xrightarrow{a} ((l_2^1,l^2), v_2)$ is a transition in $\sem{A^1 \wedge A^2}$. Because $\mathit{Clk}^1\cap\mathit{Clk}^2 = \emptyset$, we can rearrange the states into $((l_1^1,l^2),v_1) = ((l_1^1,v_1^1) , (l^2,v_1^2)) = q_1^Y$ and $((l_2^1,l^2),v_2) = ((l_2^1,v_2^1) , (l^2,v_2^2)) = q_2^Y$. Thus, $q_1^Y \xrightarrow{a} q_2^Y$ is a transition in $\sem{A^1 \wedge A^2} = Y$.}
\end{proof}

\newtheorem*{T10}{Corollary~\ref{cor:conjunctionTSandAnonsharedaction}}
\begin{T10}
	\new{Given two TIOAs $A^i = (\mathit{Loc}^i, l_0^i, \mathit{Act}^i, \mathit{Clk}^i, E^i, \mathit{Inv}^i), i=1,2$ where $\mathit{Act}_i^1 \cap \mathit{Act}_o^2 =\emptyset \wedge \mathit{Act}_o^1 \cap \mathit{Act}_i^2 =\emptyset$. Denote $X = \sem{A^1 \wedge A^2}$ and $Y = \sem{A^1} \wedge \sem{A^2}$, and let $a\in\mathit{Act}^1\setminus\mathit{Act}^2$ and $q_1,q_2\in Q^X\cap Q^Y$. If $q_1\xlongrightarrow{a}{}^{\!\! X} q_2$, then $q_1\xlongrightarrow{a}{}^{\!\! Y} q_2$.}
\end{T10}
\begin{proof}
	\new{First, from Lemma~\ref{lemma:conjunctionTSandAsamestateset} it follows that $Q^X = Q^Y$. For brevity, in the rest of this proof we write  $\mathit{Clk} = \mathit{Clk}^1 \uplus\mathit{Clk}^2$, and $v^1$ and $v^2$ to indicate the part of a valuation $v$ of only the clocks of $A^1$ and $A^2$, respectively. 
	}
	
	\new{Following Definition~\ref{def:semanticTIOA} of the semantic, it follows that there exists an edge $(l_1,a,\varphi, c, l_2)\in E^{A^1\wedge A^2}$ with $q_1^X = (l_1,v_1)$, $q_2^X = (l_2,v_2)$, $l_1,l_2 \in \mathit{Loc}^{A^1\wedge A^2}$, $v_1, v_2 \in [\mathit{Clk} \mapsto \mathbb{R}_{\geq 0}]$, $v_1\models \varphi$, $v_2 = v_1[r\mapsto 0]_{r\in c}$, and $v_2\models \mathit{Inv}(l_2)$. From Definition~\ref{def:conjunctionTIOA} of the conjunction for TIOA it follows that $l_1 = (l_1^1,l_1^2)$, $l_2=(l_2^1,l_2^2)$, $l_1^2 = l_2^2 = l^2$, and $\mathit{Inv}^{A^1\wedge A^2}(l_2) = \mathit{Inv}^1(l_2^1) \wedge \mathit{Inv}^2(l^2)$. Since $v_2\models \mathit{Inv}^{A^1\wedge A^2}(l_2)$, it follows that $v_2\models\mathit{Inv}^1(l_2^1)$ and $v_2\models\mathit{Inv}^2(l^2)$.
	}
	
	\new{It now follows directly from Lemma~\ref{lemma:conjunctionTSandAnonsharedaction} that $q_1\xlongrightarrow{a}{}^{\!\! Y} q_2$.}
\end{proof}

\newtheorem*{T11}{Lemma~\ref{lemma:conjunctionTSandAsameerror}}
\begin{T11}
	\new{Given two TIOAs $A^i = (\mathit{Loc}^i, l_0^i, \mathit{Act}^i, \mathit{Clk}^i, E^i, \mathit{Inv}^i), i=1,2$ where $\mathit{Act}_i^1 \cap \mathit{Act}_o^2 =\emptyset \wedge \mathit{Act}_o^1 \cap \mathit{Act}_i^2 =\emptyset$. Let $Q \subseteq \mathit{Loc}^1 \times \mathit{Loc}^2 \times [(\mathit{Clk}^1 \cup \mathit{Clk}^2)\mapsto \mathbb{R}_{\geq 0}]$. Then $\mathrm{err}^{\sem{A^1 \wedge A^2}}(Q) = \mathrm{err}^{\sem{A^1} \wedge \sem{A^2}}(Q)$.}
\end{T11}
\begin{proof}
	\new{It follows from Lemma~\ref{lemma:conjunctionTSandAsamestateset} that $\sem{A^1 \wedge A^2}$ and $\sem{A^1} \wedge \sem{A^2}$ have the same state set. We will show that $\mathrm{err}^{\sem{A^1 \wedge A^2}}(Q) \subseteq \mathrm{err}^{\sem{A^1} \wedge \sem{A^2}}(Q)$ and $\mathrm{err}^{\sem{A^1} \wedge \sem{A^2}}(Q) \subseteq \mathrm{err}^{\sem{A^1 \wedge A^2}}(Q)$. For brevity, we write $X = \sem{A^1 \wedge A^2}$, $Y = \sem{A^1} \wedge \sem{A^2}$, and $\mathit{Clk} = \mathit{Clk}^1 \uplus\mathit{Clk}^2$ in the rest of this proof. Also, we will use $v^1$ and $v^2$ to indicate the part of a valuation $v$ of only the clocks of $A^1$ and $A^2$, respectively. 
	}
	
	\new{($\mathrm{err}^X(Q) \subseteq\mathrm{err}^Y(Q)$) 
		Consider a state $q^X\in\mathrm{err}^X(Q)$. From Definition~\ref{def:error} of error states we know that $\exists d \in\mathbb{R}_{\geq 0}$ s.t. $q^X\arrownot\xlongrightarrow{d}{}^{\!\! X}$ and $\forall d'\in\mathbb{R}_{\geq 0}\forall o!\in\mathit{Act}_o\forall q_2\in Q^X: q^X\xlongrightarrow{d'} q_2 \Rightarrow (q_2\arrownot\xlongrightarrow{o!}{}^{\!\! X} \vee \forall q_3\in Q^X: q_2\xlongrightarrow{o!}{}^{\!\! X} q_3 \Rightarrow q_3\in Q)$. From Definition~\ref{def:semanticTIOA} of the semantic of a TIOA it follows that $q^X = (l_1,v)$ for some $l_1 \in \mathit{Loc}^{A^1\wedge A^2}$ and $v \in [\mathit{Clk}\mapsto \mathbb{R}_{\geq 0}]$, $v + d \not\models \mathit{Inv}^{A^1\wedge A^2}(l_1)$, and $v+d'\models\mathit{Inv}^{A^1\wedge A^2}(l_1) \implies [\nexists (l_1,o!,\varphi,c,l_3)\in E^{A^1\wedge A^2} \vee \forall (l_1,o!,\varphi,c,l_3)\in E^{A^1\wedge A^2}: v+d'\not\models\varphi \vee v+d'[r\mapsto 0]_{r\in c} \not\models \mathit{Inv}^{A^1\wedge A^2}(l_3) \vee (l_3, v + d'[r\mapsto 0]_{r\in c}) \in Q]$.
	}
	
	\new{From Lemma~\ref{lemma:conjunctionTSandAdelay} it follows immediately that $q^X\arrownot\xlongrightarrow{d}{}^{\!\! X}$ implies that $q^X\arrownot\xlongrightarrow{d}{}^{\!\! Y}$. So the first condition in the definition of error states holds for $q^X$ in $Y$.
	}
	
	\new{Now, pick any $d'$, $q_2$, and $o!$ such that $v+d'\models\mathit{Inv}^{A^1\wedge A^2}(l_1) \implies [\nexists (l_1,o!,\varphi,c,l_3)\in E^{A^1\wedge A^2} \vee \forall (l_1,o!,\varphi,c,l_3)\in E^{A^1\wedge A^2}: v+d'\not\models\varphi \vee v+d'[r\mapsto 0]_{r\in c} \not\models \mathit{Inv}^{A^1\wedge A^2}(l_3) \vee (l_3, v + d'[r\mapsto 0]_{r\in c}) \in Q]$. The implication holds if $v+d'\not\models\mathit{Inv}^{A^1\wedge A^2}(l_1)$ or $v+d'\models\mathit{Inv}^{A^1\wedge A^2}(l_1) \wedge [\nexists (l_1,o!,\varphi,c,l_3)\in E^{A^1\wedge A^2} \vee \forall (l_1,o!,\varphi,c,l_3)\in E^{A^1\wedge A^2}: v+d'\not\models\varphi \vee v+d'[r\mapsto 0]_{r\in c} \not\models \mathit{Inv}^{A^1\wedge A^2}(l_3) \vee (l_3, v + d'[r\mapsto 0]_{r\in c}) \in Q]$. The first case follows directly from Lemma~\ref{lemma:conjunctionTSandAdelay} that shows that $q^X\arrownot\xlongrightarrow{d'}{}^{\!\! Y}$, which ensures that the second condition in the definition of error states holds for $q^X$ in $Y$. For the second case we again use Lemma~\ref{lemma:conjunctionTSandAdelay}, thus $q^X\xlongrightarrow{d'}{}^{\!\! Y} q_2$, where $q_2 = (l_1, v + d')$. Now consider the two cases in the right-hand side of the implication.
	}
	
	\begin{itemize}
		\item \new{$\nexists (l_1,o!,\varphi,c,l_3)\in E^{A^1\wedge A^2}$. We have to consider the three cases from Definition~\ref{def:conjunctionTIOA} of the conjunction for TIOA.}
		\begin{itemize}
			\item \new{$o!\in\mathit{Act}^1\cap\mathit{Act}^2$. In this case, we know that $\nexists (l_1^1,o!,\varphi^1,c^1,l_3^1)\in E^1$ or $\nexists (l_1^2,o!,\varphi^2,c^2,l_3^2)\in E^2$ (or both). Therefore, it follows from Definition~\ref{def:semanticTIOA} of the semantic of a TIOA that $(l_1^1,v^1 + d')\arrownot\xlongrightarrow{o!}{}^{\!\! \sem{A^1}}$ or $(l_1^2,v^2 + d')\arrownot\xlongrightarrow{o!}{}^{\!\! \sem{A^2}}$ (or both). Now, from Definition~\ref{def:conjunctionTIOTS} of the conjunction for TIOTS it follows that $((l_1^1,v^1 + d'),(l_1^2, v^2 + d'))\arrownot\xlongrightarrow{o!}{}^{\!\! Y}$\footnote{Alternatively, we could use Lemma~\ref{lemma:conjunctionTSandAsharedaction} to come to the same conclusion. This also holds for the other two cases, where we have to use Corollary~\ref{cor:conjunctionTSandAnonsharedaction} instead.}.
			}
			
			\item \new{$o!\in\mathit{Act}^1\setminus\mathit{Act}^2$. In this case, we know that $\nexists (l_1^1,o!,\varphi^1,c^1,l_3^1)\in E^1$. Therefore, it follows from Definition~\ref{def:semanticTIOA} of the semantic of a TIOA that $(l_1^1,v^1 + d')\arrownot\xlongrightarrow{o!}{}^{\!\! \sem{A^1}}$. Now, from Definition~\ref{def:conjunctionTIOTS} of the conjunction for TIOTS it follows that $((l_1^1,v^1 + d'),(l_1^2, v^2 + d'))\arrownot\xlongrightarrow{o!}{}^{\!\! Y}$.
			}
			
			\item \new{$o!\in\mathit{Act}^2\setminus\mathit{Act}^1$. In this case, we know that $\nexists (l_1^2,o!,\varphi^2,c^2,l_3^2)\in E^2$. Therefore, it follows from Definition~\ref{def:semanticTIOA} of the semantic of a TIOA that $(l_1^2,v^2 + d')\arrownot\xlongrightarrow{o!}{}^{\!\! \sem{A^2}}$. Now, from Definition~\ref{def:conjunctionTIOTS} of the conjunction for TIOTS it follows that $((l_1^1,v^1 + d'),(l_1^2, v^2 + d'))\arrownot\xlongrightarrow{o!}{}^{\!\! Y}$.}
		\end{itemize}
		\new{So, in all three cases we can show that $((l_1^1,v^1 + d'),(l_1^2, v^2 + d'))\arrownot\xlongrightarrow{o!}{}^{\!\! Y}$. And note that $((l_1^1,v^1 + d'),(l_1^2, v^2 + d')) = q_2$.}
		
		\item \new{$\forall (l_1,o!,\varphi,c,l_3)\in E^{A^1\wedge A^2}: v+d'\not\models\varphi \vee v+d'[r\mapsto 0]_{r\in c} \not\models \mathit{Inv}^{A^1\wedge A^2}(l_3) \vee (l_3, v + d'[r\mapsto 0]_{r\in c}) \in Q$. For each edge $(l_1,o!,\varphi,c,l_3)\in E^{A^1\wedge A^2}$, we have to consider the three cases from Definition~\ref{def:conjunctionTIOA} of the conjunction for TIOA.
		}
		\begin{itemize}
			\item \new{$o!\in\mathit{Act}^1\cap\mathit{Act}^2$. In this case, we know that $(l_1^1,o!,\varphi^1,c^1,l_3^1) \in E^1$, $(l_1^2,o!,\varphi^2,c^2,l_3^2) \in E^2$, $\varphi = \varphi^1 \wedge \varphi^2$, and $c = c^1 \cup c^2$. Now consider the three cases that should hold for each edge $(l_1,o!,\varphi,c,l_3)\in E^{A^1\wedge A^2}$.
			}
			\begin{itemize}
				\item \new{$v+d'\not\models\varphi$. In this case, we know that $v+d'\not\models\varphi$ implies that $v+d'\not\models\varphi^1$ or $v+d'\not\models\varphi^2$ (or both). Because $\mathit{Clk}^1\cap\mathit{Clk}^2 = \emptyset$, it holds that $v^1+d'\not\models\varphi^1$ or $v^2+d'\not\models\varphi^2$ (or both). Therefore, it follows from Definition~\ref{def:semanticTIOA} of the semantic of a TIOA that $(l_1^1,v^1 + d')\arrownot\xlongrightarrow{o!}{}^{\!\! \sem{A^1}}$ or $(l_1^2,v^2 + d')\arrownot\xlongrightarrow{o!}{}^{\!\! \sem{A^2}}$ (or both). Now, from Definition~\ref{def:conjunctionTIOTS} of the conjunction for TIOTS it follows that $((l_1^1,v^1 + d'),(l_1^2, v^2 + d'))\arrownot\xlongrightarrow{o!}{}^{\!\! Y}$.
				}
				
				\item \new{$v+d'[r\mapsto 0]_{r\in c} \not\models \mathit{Inv}^{A^1\wedge A^2}(l_3)$. In this case, we know that $v+d'[r\mapsto 0]_{r\in c} \not\models \mathit{Inv}^{A^1\wedge A^2}(l_3)$ implies that $v+d'[r\mapsto 0]_{r\in c} \not\models \mathit{Inv}^1(l_3^1)$ or $v+d'[r\mapsto 0]_{r\in c} \not\models \mathit{Inv}^2(l_3^2)$ (or both). Because $\mathit{Clk}^1\cap\mathit{Clk}^2 = \emptyset$, it holds that $v^1+d'[r\mapsto 0]_{r\in c^1} \not\models \mathit{Inv}^1(l_3^1)$ or $v^2+d'[r\mapsto 0]_{r\in c^2} \not\models \mathit{Inv}^2(l_3^2)$ (or both). Therefore, it follows from Definition~\ref{def:semanticTIOA} of the semantic of a TIOA that $(l_1^1,v^1 + d')\arrownot\xlongrightarrow{o!}{}^{\!\! \sem{A^1}}$ or $(l_1^2,v^2 + d')\arrownot\xlongrightarrow{o!}{}^{\!\! \sem{A^2}}$ (or both). Now, from Definition~\ref{def:conjunctionTIOTS} of the conjunction for TIOTS it follows that $((l_1^1,v^1 + d'),(l_1^2, v^2 + d'))\arrownot\xlongrightarrow{o!}{}^{\!\! Y}$.}
				
				\item \new{$(l_3, v + d'[r\mapsto 0]_{r\in c})\in Q$. In this case, assume that $v+d'\models\varphi$ and $v+d[r\mapsto 0]_{r\in c} \models \mathit{Inv}^{A^1\wedge A^2}(l_3)$ (otherwise, one of the above cases can be used instead). Because $\mathit{Clk}^1\cap\mathit{Clk}^2 = \emptyset$, it follows that $v^1+d'\models\varphi^1$, $v^2+d'\models\varphi^2$,  $v^1+d'[r\mapsto 0]_{r\in c^1} \models \mathit{Inv}^1(l_3^1)$, and $v^2+d'[r\mapsto 0]_{r\in c^2} \models \mathit{Inv}^2(l_3^2)$. Therefore, it follows from Definition~\ref{def:semanticTIOA} of the semantic of a TIOA that $(l_1^1,v^1 + d')\xlongrightarrow{o!}{}^{\!\! \sem{A^1}} (l_3^1, v^1 + d'[r\mapsto 0]_{r\in c^1})$ and $(l_1^2,v^2 + d')\xlongrightarrow{o!}{}^{\!\! \sem{A^2}} (l_3^2, v^2 + d'[r\mapsto 0]_{r\in c^2})$. Now, from Definition~\ref{def:conjunctionTIOTS} of the conjunction for TIOTS it follows that $((l_1^1,v^1 + d'), (l_1^2,v^2 + d'))\xlongrightarrow{o!}{}^{\!\! Y} ((l_3^1, v^1 + d'[r\mapsto 0]_{r\in c^1}), (l_3^2, v^2 + d'[r\mapsto 0]_{r\in c^2}))$. And note that $((l_3^1, v^1 + d'[r\mapsto 0]_{r\in c^1}), (l_3^2, v^2 + d'[r\mapsto 0]_{r\in c^2})) = (l_3^1, l_3^2, v + d'[r\mapsto 0]_{r\in c}) = (l_3, v + d'[r\mapsto 0]_{r\in c})$.}
			\end{itemize}
			\new{So, in the first two cases we have shown that $((l_1^1,v^1 + d'),(l_1^2, v^2 + d'))\arrownot\xlongrightarrow{o!}{}^{\!\! Y}$\footnote{Alternatively, we could use Lemma~\ref{lemma:conjunctionTSandAsharedaction} to come to the same conclusion. This also holds for the other two cases, where we have to use Corollary~\ref{cor:conjunctionTSandAnonsharedaction} instead.} and in the third case that $((l_1^1,v^1 + d'), (l_1^2,v^2 + d'))\xlongrightarrow{o!}{}^{\!\! Y} (l_3, v + d'[r\mapsto 0]_{r\in c})$.
			}
			
			\item \new{$o!\in\mathit{Act}^1\setminus\mathit{Act}^2$. In this case, we know that $(l_1^1,o!,\varphi^1,c^1,l_3^1) \in E^1$, $\varphi = \varphi^1$, and $c = c^1$. Now consider the three cases that should hold for each edge $(l_1,o!,\varphi,c,l_3)\in E^{A^1\wedge A^2}$.
			}
			\begin{itemize}
				\item \new{$v+d'\not\models\varphi$. In this case, we know that $v+d'\not\models\varphi$ implies that $v+d'\not\models\varphi^1$. Because $\mathit{Clk}^1\cap\mathit{Clk}^2 = \emptyset$, it holds that $v^1+d'\not\models\varphi^1$. Therefore, it follows from Definition~\ref{def:semanticTIOA} of the semantic of a TIOA that $(l_1^1,v^1 + d')\arrownot\xlongrightarrow{o!}{}^{\!\! \sem{A^1}}$. Now, from Definition~\ref{def:conjunctionTIOTS} of the conjunction for TIOTS it follows that $((l_1^1,v^1 + d'),(l_1^2, v^2 + d'))\arrownot\xlongrightarrow{o!}{}^{\!\! Y}$.
				}
				
				\item \new{$v+d[r\mapsto 0]_{r\in c} \not\models \mathit{Inv}^{A^1\wedge A^2}(l_3)$. In this case, we know that $v+d[r\mapsto 0]_{r\in c} \not\models \mathit{Inv}^{A^1\wedge A^2}(l_3)$ implies that $v+d[r\mapsto 0]_{r\in c} \not\models \mathit{Inv}^1(l_3^1)$. Because $\mathit{Clk}^1\cap\mathit{Clk}^2 = \emptyset$, it holds that $v^1+d[r\mapsto 0]_{r\in c^1} \not\models \mathit{Inv}^1(l_3^1)$. Therefore, it follows from Definition~\ref{def:semanticTIOA} of the semantic of a TIOA that $(l_1^1,v^1 + d')\arrownot\xlongrightarrow{o!}{}^{\!\! \sem{A^1}}$. Now, from Definition~\ref{def:conjunctionTIOTS} of the conjunction for TIOTS it follows that $((l_1^1,v^1 + d'),(l_1^2, v^2 + d'))\arrownot\xlongrightarrow{o!}{}^{\!\! Y}$.}
				
				\item \new{$(l_3, v + d'[r\mapsto 0]_{r\in c})\in Q$. In this case, assume that $v+d'\models\varphi$ and $v+d[r\mapsto 0]_{r\in c} \models \mathit{Inv}^{A^1\wedge A^2}(l_3)$ (otherwise, one of the above cases can be used instead). Because $\mathit{Clk}^1\cap\mathit{Clk}^2 = \emptyset$, it follows that $v^1+d'\models\varphi^1$ and $v^1+d'[r\mapsto 0]_{r\in c^1} \models \mathit{Inv}^1(l_3^1)$. Therefore, it follows from Definition~\ref{def:semanticTIOA} of the semantic of a TIOA that $(l_1^1,v^1 + d')\xlongrightarrow{o!}{}^{\!\! \sem{A^1}} (l_3^1, v^1 + d'[r\mapsto 0]_{r\in c^1})$. Now, from Definition~\ref{def:conjunctionTIOTS} of the conjunction for TIOTS it follows that $((l_1^1,v^1 + d'), (l_1^2,v^2 + d'))\xlongrightarrow{o!}{}^{\!\! Y} ((l_3^1, v^1 + d'[r\mapsto 0]_{r\in c^1}), (l_1^2, v^2 + d'))$. And note that $((l_3^1, v^1 + d'[r\mapsto 0]_{r\in c^1}), (l_1^2, v^2 + d')) = (l_3^1, l_1^2, v + d'[r\mapsto 0]_{r\in c}) = (l_3, v + d'[r\mapsto 0]_{r\in c})$. Now notice that $(l_3, v + d'[r\mapsto 0]_{r\in c}) \in Q$.}
			\end{itemize}
			\new{So, in the first two cases we have shown that $((l_1^1,v^1 + d'),(l_1^2, v^2 + d'))\arrownot\xlongrightarrow{o!}{}^{\!\! Y}$ and in the third case that $((l_1^1,v^1 + d'), (l_1^2,v^2 + d'))\xlongrightarrow{o!}{}^{\!\! Y} (l_3, v + d'[r\mapsto 0]_{r\in c})$.}
			
			\item \new{$o!\in\mathit{Act}^2\setminus\mathit{Act}^1$. In this case, we know that $(l_1^2,o!,\varphi^2,c^2,l_3^2) \in E^2$, $\varphi = \varphi^2$, and $c = c^2$. Now consider the three cases that should hold for each edge $(l_1,o!,\varphi,c,l_3)\in E^{A^1\wedge A^2}$.
			}
			\begin{itemize}
				\item \new{$v+d'\not\models\varphi$. In this case, we know that $v+d\not\models\varphi$ implies that $v+d'\not\models\varphi^2$. Because $\mathit{Clk}^1\cap\mathit{Clk}^2 = \emptyset$, it holds that $v^2+d'\not\models\varphi^2$. Therefore, it follows from Definition~\ref{def:semanticTIOA} of the semantic of a TIOA that $(l_1^2,v^2 + d')\arrownot\xlongrightarrow{o!}{}^{\!\! \sem{A^2}}$. Now, from Definition~\ref{def:conjunctionTIOTS} of the conjunction for TIOTS it follows that $((l_1^1,v^1 + d'),(l_1^2, v^2 + d'))\arrownot\xlongrightarrow{o!}{}^{\!\! Y}$.
				}
				\item \new{$v+d[r\mapsto 0]_{r\in c} \not\models \mathit{Inv}^{A^1\wedge A^2}(l_3)$. In this case, we know that $v+d[r\mapsto 0]_{r\in c} \not\models \mathit{Inv}^{A^1\wedge A^2}(l_3)$ implies that $v+d[r\mapsto 0]_{r\in c} \not\models \mathit{Inv}^2(l_3^2)$. Because $\mathit{Clk}^1\cap\mathit{Clk}^2 = \emptyset$, it holds that $v^2+d[r\mapsto 0]_{r\in c^2} \not\models \mathit{Inv}^2(l_3^2)$. Therefore, it follows from Definition~\ref{def:semanticTIOA} of the semantic of a TIOA that $(l_1^2,v^2 + d')\arrownot\xlongrightarrow{o!}{}^{\!\! \sem{A^2}}$. Now, from Definition~\ref{def:conjunctionTIOTS} of the conjunction for TIOTS it follows that $((l_1^1,v^1 + d'),(l_1^2, v^2 + d'))\arrownot\xlongrightarrow{o!}{}^{\!\! Y}$.}
				
				\item \new{$(l_3, v + d'[r\mapsto 0]_{r\in c})\in Q$. In this case, assume that $v+d'\models\varphi$ and $v+d[r\mapsto 0]_{r\in c} \models \mathit{Inv}^{A^1\wedge A^2}(l_3)$ (otherwise, one of the above cases can be used instead). Because $\mathit{Clk}^1\cap\mathit{Clk}^2 = \emptyset$, it follows that $v^2+d'\models\varphi^2$ and $v^2+d'[r\mapsto 0]_{r\in c^2} \models \mathit{Inv}^2(l_3^2)$. Therefore, it follows from Definition~\ref{def:semanticTIOA} of the semantic of a TIOA that $(l_1^2,v^2 + d')\xlongrightarrow{o!}{}^{\!\! \sem{A^2}} (l_3^2, v^2 + d'[r\mapsto 0]_{r\in c^2})$. Now, from Definition~\ref{def:conjunctionTIOTS} of the conjunction for TIOTS it follows that $((l_1^1,v^1 + d'), (l_1^2,v^2 + d'))\xlongrightarrow{o!}{}^{\!\! Y} ((l_1^1, v^1 + d'[r\mapsto 0]_{r\in c^1}), (l_3^2, v^2 + d'[r\mapsto 0]_{r\in c^2}))$. And note that $((l_1^1, v^1 + d'), (l_3^2, v^2 + d'[r\mapsto 0]_{r\in c^2})) = (l_1^1, l_3^2, v + d'[r\mapsto 0]_{r\in c}) = (l_3, v + d'[r\mapsto 0]_{r\in c})$.}
			\end{itemize}
			\new{So, in the first two cases we have shown that $((l_1^1,v^1 + d'),(l_1^2, v^2 + d'))\arrownot\xlongrightarrow{o!}{}^{\!\! Y}$ and in the third case that $((l_1^1,v^1 + d'), (l_1^2,v^2 + d'))\xlongrightarrow{o!}{}^{\!\! Y} (l_3, v + d'[r\mapsto 0]_{r\in c})$.}
		\end{itemize}
		\new{So, in all three cases we have shown that $((l_1^1,v^1 + d'),(l_1^2, v^2 + d'))\arrownot\xlongrightarrow{o!}{}^{\!\! Y}$ or $((l_1^1,v^1 + d'), (l_1^2,v^2 + d'))\xlongrightarrow{o!}{}^{\!\! Y} (l_3, v + d'[r\mapsto 0]_{r\in c})$. And note that $((l_1^1,v^1 + d'),(l_1^2, v^2 + d')) = q_2$ and $(l_3, v + d'[r\mapsto 0]_{r\in c}) = q_3$.}
	\end{itemize}
	\new{So we have shown that $((l_1^1,v^1 + d'),(l_1^2, v^2 + d'))\arrownot\xlongrightarrow{o!}{}^{\!\! Y}$ or $((l_1^1,v^1 + d'), (l_1^2,v^2 + d'))\xlongrightarrow{o!}{}^{\!\! Y} (l_3, v + d'[r\mapsto 0]_{r\in c})$ with $(l_3, v + d'[r\mapsto 0]_{r\in c}) \in Q$. We can rewrite this into $q^X \xlongrightarrow{d'}{}^{\!\! Y} q_2 \arrownot\xlongrightarrow{o!}{}^{\!\! Y}$ or $q^X \xlongrightarrow{d'}{}^{\!\! Y} q_2 \xlongrightarrow{o!}{}^{\!\! Y} q_3$. Since we have chosen $d'$, $q_2$, $q_3$, and $o!$ arbitrarily, the conclusion holds for all $d'$, $q_2$, $q_3$, and $o!$. Therefore, the second condition in the definition of error states hold for $q^X$.
	}
	
	\new{Now, since both conditions in the definition of the error states hold for $q^X$, we know that $q^X \in \mathrm{err}^Y(Q)$. Since we have chosen $q^X$ arbitrarily from $\mathrm{err}^X(Q)$, it holds for all $q^X \in \mathrm{err}^X(Q)$. Therefore, it holds that $\mathrm{err}^X(Q) \subseteq \mathrm{err}^Y(Q)$.
	}
	
	\new{($\mathrm{err}^Y \subseteq\mathrm{err}^X$)
		Consider a state $q^Y\in\mathrm{err}^Y$. From Definition~\ref{def:error} of error states we know that $\exists d \in\mathbb{R}_{\geq 0}$ s.t. $q^Y\arrownot\xlongrightarrow{d'}{}^{\!\! Y}$ and $\forall d'\in\mathbb{R}_{\geq 0}\forall o!\in\mathit{Act}_o\forall q_2\in Q^Y: q^Y\xlongrightarrow{d} q_2 \Rightarrow (q_2\arrownot\xlongrightarrow{o!}{}^{\!\! Y} \vee \forall q_3\in Q^Y: q_2\xlongrightarrow{o!}{}^{\!\! Y} q_3 \Rightarrow q_3\in Q)$. From Definition~\ref{def:conjunctionTIOTS} of the conjunction for TIOTS it follows that $q^Y=(q^{\sem{A^1}}, q^{\sem{A^2}})$ and $q_2=(q_2^{\sem{A^1}}, q_2^{\sem{A^2}})$. 
	}
	
	\new{First, consider the first condition in the definition of error states.	From Lemma~\ref{lemma:conjunctionTSandAdelay} it follows immediately that $q^Y\arrownot\xlongrightarrow{d}{}^{\!\! Y}$ implies that $q^Y\arrownot\xlongrightarrow{d}{}^{\!\! X}$. So the first condition in the definition of error states holds for $q^Y$ in $X$.
	}
	
	\new{Now, consider the second condition in the definition of error states. Pick any $d'$, $q_2$, and $o!$ such that $q^Y\xlongrightarrow{d'} q_2 \Rightarrow (q_2\arrownot\xlongrightarrow{o!}{}^{\!\! Y} \vee \forall q_3\in Q^Y: q_2\xlongrightarrow{o!}{}^{\!\! Y} q_3 \Rightarrow q_3\in Q)$. The implication holds if $q^Y\arrownot\xlongrightarrow{d'}{}^{\!\! Y}$ or $q^Y\xlongrightarrow{d'} q_2 \wedge (q_2\arrownot\xlongrightarrow{o!}{}^{\!\! Y} \vee \forall q_3\in Q^Y: q_2\xlongrightarrow{o!}{}^{\!\! Y} q_3 \Rightarrow q_3\in Q)$. The first case follows directly from Lemma~\ref{lemma:conjunctionTSandAdelay} that shows that
		$q^Y\arrownot\xlongrightarrow{d}{}^{\!\! Y}$ implies that $q^Y\arrownot\xlongrightarrow{d}{}^{\!\! X}$, which ensures that the second condition in the definition of error states holds for $q^Y$ in $X$. For the second case we again use Lemma~\ref{lemma:conjunctionTSandAdelay}, thus $q^Y\xlongrightarrow{d'}{}^{\!\! X} q_2$, where $q^Y = (l_1^1,l_1^2, v)$ and $q_2 = (l_1^1,l_1^2, v + d)$. 
	}
	
	\new{It remains to be shown that $q_2\arrownot\xlongrightarrow{o!}{}^{\!\! Y} \vee \forall q_3\in Q^Y: q_2\xlongrightarrow{o!}{}^{\!\! Y} q_3 \Rightarrow q_3\in Q$ in $Y$ implies that $q_2\arrownot\xlongrightarrow{o!}{}^{\!\! X} \vee \forall q_3\in Q^X: q_2\xlongrightarrow{o!}{}^{\!\! X} q_3 \Rightarrow q_3\in Q$ in $X$. We have to consider the three cases from Definition~\ref{def:conjunctionTIOTS} of the conjunction for TIOTS.
	}
	\begin{itemize}
		\item \new{$o!\in\mathit{Act}^1\cap\mathit{Act}^2$. It follows directly from Lemma~\ref{lemma:conjunctionTSandAsharedaction} that $q_2\arrownot\xlongrightarrow{o!}{}^{\!\! X} \vee \forall q_3\in Q^X: q_2\xlongrightarrow{o!}{}^{\!\! X} q_3 \Rightarrow q_3\in Q$.
		}
		
		\item \new{$o!\in\mathit{Act}^1\setminus\mathit{Act}^2$. Using Definition~\ref{def:semanticTIOA} of the semantic of a TIOA, we now know that $\nexists (l_1^1, o!, \varphi^1, c^1, l_3^1)\in E^1$ or $\forall (l_1^1, o!, \varphi^1, c^1, l_3^1)\in E^1 : v^1 + d' \not\models\varphi^1 \vee v^1 + d'[r\mapsto 0]_{r\in c^1}\not\models \mathit{Inv}^1(l_3^1) \vee (l_3^1, l_1^2, v + d'[r\mapsto 0]_{r\in c^1}) \in Q$. 
		}
		
		\new{In case that $\nexists (l_1^1, o!, \varphi^1, c^1, l_3^1)\in E^1$, it follows directly from Definition~\ref{def:conjunctionTIOA} of the conjunction for TIOA that $\nexists ((l_1^1, l_1^2), o!, \varphi^1, c^1, (l_3^1,l_3^2))\in E^{A^1 \wedge A^2}$. Then, with Definition~\ref{def:semanticTIOA} of the semantic of a TIOA, it follows that $(l_1^1,l_1^2, v + d') \arrownot\xlongrightarrow{o!}{}^{\!\! X}$.
		}
		
		\new{In case that $\forall (l_1^1, o!, \varphi^1, c^1, l_3^1)\in E^1 : v^1 + d' \not\models\varphi^1 \vee v^1 + d'[r\mapsto 0]_{r\in c^1}\not\models \mathit{Inv}^1(l_3^1)$, it follows from Definition~\ref{def:conjunctionTIOA} that for each edge $(l_1^1, o!, \varphi^1, c^1, l_3^1)\in E^1$, $\exists ((l_1^1,l_1^2),o!,\varphi^1, c^1, (l_3^1, l_3^2)) \in E^{A^1\wedge A^2}$. Because $\mathit{Clk}^1\cap\mathit{Clk}^2 = \emptyset$, it holds that $v + d' \not\models\varphi^1 \vee v + d'[r\mapsto 0]_{r\in c}\not\models \mathit{Inv}^1(l_3^1)$. Therefore, it also holds that $v + d' \not\models\varphi^1 \vee v + d'[r\mapsto 0]_{r\in c}\not\models \mathit{Inv}^1(l_3^1) \wedge\mathit{Inv}^2(l_1^2)$. Note that from Definition~\ref{def:conjunctionTIOA} we know that $\mathit{Inv}^{A^1\wedge A^2}((l_3^1,l_1^2)) = \mathit{Inv}^1(l_3^1) \wedge \mathit{Inv}^2(l_1^2)$. As we have shown that $v + d' \not\models\varphi^1 \vee v + d'[r\mapsto 0]_{r\in c}\not\models \mathit{Inv}^1(l_3^1) \wedge\mathit{Inv}^2(l_1^2)$ for all edges labeled with $o!$ from $(l_1^1, l_1^2)$, it follows from Definition~\ref{def:semanticTIOA} of the semantic of a TIOA that $(l_1^1,l_1^2, v + d') \arrownot\xlongrightarrow{o!}{}^{\!\! X}$.
		}
		
		\new{In case that $\forall (l_1^1, o!, \varphi^1, c^1, l_3^1)\in E^1 : (l_3^1, l_1^2, v + d'[r\mapsto 0]_{r\in c^1}) \in Q$, it follows from Definition~\ref{def:conjunctionTIOA} that for each edge $(l_1^1, o!, \varphi^1, c^1, l_3^1)\in E^1$, $\exists ((l_1^1,l_1^2),o!,\varphi^1, c^1, (l_3^1, l_3^2)) \in E^{A^1\wedge A^2}$. Because $\mathit{Clk}^1\cap\mathit{Clk}^2 = \emptyset$, it holds that $v + d' \models\varphi^1 \wedge v + d'[r\mapsto 0]_{r\in c}\models \mathit{Inv}^1(l_3^1)$ (in case one of them does not hold, we can use the argument above). Therefore, it also holds that $v + d' \models\varphi^1 \wedge v + d'[r\mapsto 0]_{r\in c}\models \mathit{Inv}^1(l_3^1) \wedge\mathit{Inv}^2(l_1^2)$. Note that from Definition~\ref{def:conjunctionTIOA} we know that $\mathit{Inv}^{A^1\wedge A^2}((l_3^1,l_1^2)) = \mathit{Inv}^1(l_3^1) \wedge \mathit{Inv}^2(l_1^2)$. As we have shown that $v + d' \models\varphi^1 \wedge v + d'[r\mapsto 0]_{r\in c}\models \mathit{Inv}^1(l_3^1) \wedge\mathit{Inv}^2(l_1^2)$ for all edges labeled with $o!$ from $(l_1^1, l_1^2)$, it follows from Definition~\ref{def:semanticTIOA} of the semantic of a TIOA that $(l_1^1,l_1^2, v + d') \xlongrightarrow{o!}{}^{\!\! X} (l_3^1,l_1^2, v + d'[r\mapsto 0]_{r\in c})$. Now notice that $(l_3^1,l_1^2, v + d'[r\mapsto 0]_{r\in c}) \in Q$.}
		
		\item \new{$o!\in\mathit{Act}^2\setminus\mathit{Act}^1$. Using Definition~\ref{def:semanticTIOA} of the semantic of a TIOA, we now know that $\nexists (l_1^2, o!, \varphi^2, c^2, l_3^2)\in E^2$ or $\forall (l_1^2, o!, \varphi^2, c^2, l_3^2)\in E^2 : v^2 + d' \not\models\varphi^2 \vee v^2 + d'[r\mapsto 0]_{r\in c^2}\not\models \mathit{Inv}^2(l_3^2)  \vee (l_1^1, l_3^2, v + d'[r\mapsto 0]_{r\in c^1}) \in Q$. 
		}
		
		\new{In case that $\nexists (l_1^2, o!, \varphi^2, c^2, l_3^2)\in E^2$, it follows directly from Definition~\ref{def:conjunctionTIOA} of the conjunction for TIOA that $\nexists ((l_1^1, l_1^2), o!, \varphi^2, c^2, (l_3^1,l_3^2))\in E^{A^1 \wedge A^2}$. Then, with Definition~\ref{def:semanticTIOA} of the semantic of a TIOA, it follows that $(l_1^1,l_1^2, v + d') \arrownot\xlongrightarrow{o!}{}^{\!\! X}$.
		}
		
		\new{In case that $\forall (l_1^2, o!, \varphi^2, c^2, l_3^2)\in E^2 : v^2 + d' \not\models\varphi^2 \vee v^2 + d'[r\mapsto 0]_{r\in c^2}\not\models \mathit{Inv}^2(l_3^2)$, it follows from Definition~\ref{def:conjunctionTIOA} that for each edge $(l_1^2, o!, \varphi^2, c^2, l_3^2)\in E^2$, $\exists ((l_1^1,l_1^2),o!,\varphi^2, c^2, (l_3^1, l_3^2)) \in E^{A^1\wedge A^2}$. Because $\mathit{Clk}^1\cap\mathit{Clk}^2 = \emptyset$, it holds that $v + d' \not\models\varphi^2 \vee v + d'[r\mapsto 0]_{r\in c}\not\models \mathit{Inv}^2(l_3^2)$. Therefore, it also holds that $v + d' \not\models\varphi^2 \vee v + d'[r\mapsto 0]_{r\in c}\not\models \mathit{Inv}^2(l_3^2) \wedge\mathit{Inv}^2(l_3^2)$. Note that from Definition~\ref{def:conjunctionTIOA} we know that $\mathit{Inv}^{A^1\wedge A^2}((l_3^1,l_1^2)) = \mathit{Inv}^1(l_3^1) \wedge \mathit{Inv}^2(l_1^2)$. As we have shown that $v + d' \not\models\varphi^2 \vee v + d'[r\mapsto 0]_{r\in c}\not\models \mathit{Inv}^1(l_1^1) \wedge\mathit{Inv}^2(l_3^2)$ for all edges labeled with $o!$ from $(l_1^1, l_1^2)$, it follows from Definition~\ref{def:semanticTIOA} of the semantic of a TIOA that $(l_1^1,l_1^2, v + d') \arrownot\xlongrightarrow{o!}{}^{\!\! X}$.}
		
		\new{In case that $\forall (l_1^2, o!, \varphi^2, c^2, l_3^2)\in E^2 : (l_1^1, l_3^2, v + d'[r\mapsto 0]_{r\in c^2}) \in Q$, it follows from Definition~\ref{def:conjunctionTIOA} that for each edge $(l_1^2, o!, \varphi^2, c^2, l_3^2)\in E^2$, $\exists ((l_1^1,l_1^2),o!,\varphi^2, c^2, (l_3^1, l_2^2)) \in E^{A^1\wedge A^2}$. Because $\mathit{Clk}^1\cap\mathit{Clk}^2 = \emptyset$, it holds that $v + d' \models\varphi^2 \wedge v + d'[r\mapsto 0]_{r\in c}\models \mathit{Inv}^2(l_3^2)$ (in case one of them does not hold, we can use the argument above). Therefore, it also holds that $v + d' \models\varphi^2 \wedge v + d'[r\mapsto 0]_{r\in c}\models \mathit{Inv}^1(l_1^1) \wedge\mathit{Inv}^2(l_3^2)$. Note that from Definition~\ref{def:conjunctionTIOA} we know that $\mathit{Inv}^{A^1\wedge A^2}((l_1^1,l_3^2)) = \mathit{Inv}^1(l_1^1) \wedge \mathit{Inv}^2(l_3^2)$. As we have shown that $v + d' \models\varphi^2 \wedge v + d'[r\mapsto 0]_{r\in c}\models \mathit{Inv}^1(l_1^1) \wedge\mathit{Inv}^2(l_3^2)$ for all edges labeled with $o!$ from $(l_1^1, l_1^2)$, it follows from Definition~\ref{def:semanticTIOA} of the semantic of a TIOA that $(l_1^1,l_1^2, v + d') \xlongrightarrow{o!}{}^{\!\! X} (l_1^1,l_3^2, v + d'[r\mapsto 0]_{r\in c})$. Now notice that $(l_1^1,l_3^2, v + d'[r\mapsto 0]_{r\in c}) \in Q$.}
	\end{itemize}
	\new{So, in all three cases, we have shown that $(l_1^1, l_1^2, v + d')\arrownot\xlongrightarrow{o!}{}^{\!\! X}$ or $((l_1^1,v^1 + d'), (l_1^2,v^2 + d'))\xlongrightarrow{o!}{}^{\!\! X} (l_3, v + d'[r\mapsto 0]_{r\in c})$ with $(l_3, v + d'[r\mapsto 0]_{r\in c}) \in Q$.  We can rewrite this into $q^Y \xlongrightarrow{d'}{}^{\!\! X} q_2 \arrownot\xlongrightarrow{o!}{}^{\!\! X}$ or $q^X \xlongrightarrow{d'}{}^{\!\! Y} q_2 \xlongrightarrow{o!}{}^{\!\! Y} q_3$. Since we have chosen $d'$, $q_2$, $q_3$, and $o!$ arbitrarily, the conclusion holds for all $d'$, $q_2$, $q_3$, and $o!$. Therefore, the second condition in the definition of error states hold for $q^Y$.
	}
	
	\new{Now, since both conditions in the definition of the error states hold for $q^Y$, we know that $q^Y \in \mathrm{err}^X$. Since we have chosen $q^Y$ arbitrarily, it holds for all $q^Y \in \mathrm{err}^Y$. Therefore, it holds that $\mathrm{err}^Y \subseteq \mathrm{err}^X$.}
\end{proof}

\newtheorem*{T12}{Lemma~\ref{lemma:conjunctionTSandAsamecons}}
\begin{T12}
	\new{Given two TIOAs $A^i = (\mathit{Loc}^i, l_0^i, \mathit{Act}^i, \mathit{Clk}^i, E^i, \mathit{Inv}^i), i=1,2$ where $\mathit{Act}_i^1 \cap \mathit{Act}_o^2 =\emptyset \wedge \mathit{Act}_o^1 \cap \mathit{Act}_i^2 =\emptyset$. Then $\mathrm{cons}^{\sem{A^1 \wedge A^2}} = \mathrm{cons}^{\sem{A^1} \wedge \sem{A^2}}$.}
\end{T12}
\begin{proof}
	\new{We will proof this by using the $\Theta$ operator. It follows from Lemma~\ref{lemma:conjunctionTSandAsamestateset} that $\sem{A^1 \wedge A^2}$ and $\sem{A^1} \wedge \sem{A^2}$ have the same state set. Also, observe that the semantic of a TIOA, conjunction, and adversarial pruning do not alter the action set. Therefore, it follows that $\sem{A^1 \wedge A^2}$ and $\sem{A^1} \wedge \sem{A^2}$ have the same action set and partitioning into input and output actions. We will show for any postfixed point $P$ of $\Theta$ that $\Theta^{\sem{A^1 \wedge A^2}}(P) \subseteq \Theta^{\sem{A^1} \wedge \sem{A^2}}(P)$ and $\Theta^{\sem{A^1} \wedge \sem{A^2}}(P) \subseteq \Theta^{\sem{A^1 \wedge A^2}}(P)$. For brevity, we write $X = \sem{A^1 \wedge A^2}$, $Y = \sem{A^1} \wedge \sem{A^2}$, and $\mathit{Clk} = \mathit{Clk}^1 \uplus\mathit{Clk}^2$ in the rest of this proof. Also, we will use $v^1$ and $v^2$ to indicate the part of a valuation $v$ of only the clocks of $A^1$ and $A^2$, respectively. 
	}
	
	\new{($\Theta^X(P) \subseteq \Theta^Y(P)$)
		Consider a state $q^X \in P$. Because $P$ is a postfixed point of $\Theta^X$, it follows that $q^X\in\Theta^X(P)$. From the definition of $\Theta$, it follows that $q^X \in \overline{\mathrm{err}^X(\overline{P})}$ and $ q^X \in \{q_1\in Q^X\mid \forall d \geq 0: [\forall q_2\in Q^X: q_1\xlongrightarrow{d}{}^{\!\! X} q_2 \Rightarrow q_2\in P \wedge \forall i?\in \mathit{Act}_i^X: \exists q_3\in P: q_2\xlongrightarrow{i?}{}^{\!\! X} q_3]\ \vee [\exists d'\leq d \wedge \exists q_2,q_3\in P \wedge \exists o!\in\mathit{Act}_o^X: q_1\xlongrightarrow{d'}{}^{\!\! X} q_2 \wedge q_2\xlongrightarrow{o!}{}^{\!\! X} q_3 \wedge \forall i?\in\mathit{Act}_i^X: \exists q_4\in P: q_2\xlongrightarrow{i?}{}^{\!\! X} q_4]\}$. From Lemma~\ref{lemma:conjunctionTSandAsameerror} it follows directly that $q^X \in \overline{\mathrm{err}^Y(\overline{P})}$. Now we only focus on the second part of the definition of $\Theta$.
	}
	
	\new{Consider a $d\in\mathbb{R}_{\geq 0}$. Then the left-hand side or the right-hand side of the disjunction is true (or both). }
	\begin{itemize}
		\item \new{Assume the left-hand side is true, i.e., $\forall q_2\in Q^X: q^X\xlongrightarrow{d}{}^{\!\! X} q_2 \Rightarrow q_2\in P \wedge \forall i?\in \mathit{Act}_i^X: \exists q_3\in P: q_2\xlongrightarrow{i?}{}^{\!\! X} q_3$. Pick a $q_2 \in Q^X$. The implication is true when $q^X\arrownot\xlongrightarrow{d}{}^{\!\! X} q_2$ or $q^X\xlongrightarrow{d}{}^{\!\! X} q_2 \wedge q_2\in P \wedge \forall i?\in \mathit{Act}_i^X: \exists q_3\in P: q_2\xlongrightarrow{i?}{}^{\!\! X} q_3$. 
		}
		
		\begin{itemize}
			\item \new{Consider the first case. From Lemma~\ref{lemma:conjunctionTSandAdelay} it follows that $q^X \arrownot\xlongrightarrow{d}{}^{\!\! Y}$. Note that $q^X = (l_1^1, v_1^1, l_1^2, v_1^2)$. Thus the implication also holds for $q_2$ in $Y$.
			}
			
			\item \new{Consider the second case. From Lemma~\ref{lemma:conjunctionTSandAdelay}, we have that $q^X\xlongrightarrow{d}{}^{\!\! X} q_2$ implies that $q^X\xlongrightarrow{d}{}^{\!\! Y} q_2$, and from Definition~\ref{def:semanticTIOA} of the semantic of a TIOA it follows that $v_1 + d\models \mathit{Inv}^{A^1\wedge A^2}(l_1)$ for $q^X = (l_1, v_1)$, $q_2 = (l_1, v_1 + d)$, $l_1\in \mathit{Loc}^{A^1\wedge A^2}$, and $v_1\in[\mathit{Clk}\mapsto \mathbb{R}_{\geq 0}]$. Now, pick $i?\in\mathit{Act}_i^X$ and $q_3\in Q^X$ such that $q_2 \xlongrightarrow{i?}{}^{\!\! X} q_3$ and $q_3\in P$. From Definition~\ref{def:semanticTIOA} of the semantic of a TIOA it follows that $(l_1, i?, \varphi, c, l_3) \in E^{A^1\wedge A^2}$, $q_3 = (l_3, v_3)$, $v_1 + d \models \varphi$, $v_3 = v_1+d[r\mapsto 0]_{r\in c}$, and $v_3\models \mathit{Inv}^{A^1\wedge A^2}(l_3)$. From Definition~\ref{def:conjunctionTIOA} of the conjunction of TIOA it follows that $l_1 = (l_1^1, l_1^2)$, $l_3=(l_3^1, l_3^2)$, $\mathit{Inv}^{A^1\wedge A^2}(l_1) = \mathit{Inv}^1(l_1^1)\wedge \mathit{Inv}^2(l_1^2)$, and $\mathit{Inv}^{A^1\wedge A^2}(l_3) = \mathit{Inv}^1(l_3^1)\wedge \mathit{Inv}^2(l_3^2)$. We have to consider the three cases of Definition~\ref{def:conjunctionTIOA} in relation to $i?$.
			}
			\begin{itemize}
				\item \new{$i? \in \mathit{Act}_i^1 \cap \mathit{Act}_i^2$. It follows directly from Lemma~\ref{lemma:conjunctionTSandAsharedaction} that $q_2 \xrightarrow{i?} q_3$ is a transition in $Y$.}
				
				\item \new{$i? \in \mathit{Act}_i^1 \setminus \mathit{Act}_i^2$. It follows directly from Corollary~\ref{cor:conjunctionTSandAnonsharedaction} that $q_2 \xrightarrow{i?} q_3$ is a transition in $Y$.}
				
				\item \new{$i? \in \mathit{Act}_i^2 \setminus \mathit{Act}_i^1$. It follows directly from Corollary~\ref{cor:conjunctionTSandAnonsharedaction} (where we switched $A^1$ and $A^2$) that  $q_2 \xrightarrow{i?} q_3$ is a transition in $Y$.}
			\end{itemize}
			
			\new{So, in all three cases we have that $q_2 \xrightarrow{i?} q_3$ is a transition in $Y$. As the analysis above is independent of the particular $i?$, $q_2 \xrightarrow{i?} q_3$ is a transition in $Y$ for all $i?$. Because both $q_2, q_3 \in P$ and $q^X\xlongrightarrow{d}{}^{\!\! Y} q_2$, we have that the implication also holds for $q_2 \in Y$.}
		\end{itemize}
		\new{So, in both cases we have that for $q^X\xlongrightarrow{d}{}^{\!\! Y} q_2 \Rightarrow q_2\in P \wedge \forall i?\in \mathit{Act}_i^Y: \exists q_3\in P: q_2\xlongrightarrow{i?}{}^{\!\! Y} q_3$. As $q_2$ is chosen arbitrarily, it holds for all $q_2\in Q^X = Q^Y$. Therefore, the left-hand side is true.	
		}
		
		\item \new{Assume the right-hand side is true, i.e., $\exists d'\leq d \wedge \exists q_2,q_3\in P \wedge \exists o!\in\mathit{Act}_o^X: q^X\xlongrightarrow{d'}{}^{\!\! X} q_2 \wedge q_2\xlongrightarrow{o!}{}^{\!\! X} q_3 \wedge \forall i?\in\mathit{Act}_i^X: \exists q_4\in P: q_2\xlongrightarrow{i?}{}^{\!\! X} q_4$. 
		}
		
		\new{First, following Definition~\ref{def:semanticTIOA} of the semantic of a TIOA, we have that $q^X = (l_1,v_1)$, $q_2 = (l_1, v_1 + d')$, $q_3 = (l_3, v_3)$, $q_4 = (l_4,v_4)$, $l_1, l_3, l_4\in \mathit{Loc}^{A^1 \wedge A^2}$, $v_1, v_3, v_4\in [\mathit{Clk}\mapsto \mathbb{R}_{\geq 0}]$, $v_1 + d'\models \mathit{Inv}^{A^1\wedge A^2}(l_1)$, $\exists (l_1,o!,\varphi, c, l_3)\in E^{A^1\wedge A^2}$, $v_1 + d'\models\varphi$, $v_3 = v_1 + d'[r\mapsto 0]_{r\in c}$, and $v_3\models \mathit{Inv}^{A^1\wedge A^2}(l_3)$. First, focus on the delay transition. From Lemma~\ref{lemma:conjunctionTSandAdelay} it follows that $q^X \xlongrightarrow{d'}{}^{\!\! Y} q_2$ in $Y$, with $q^X = (l_1^1, v_1^1, l_1^2, v_1^2) = (l_1^1, l_1^2, v_1)$ and $q_2 = (l_1^1, v_1^1 + d', l_1^2, v_1^2 + d') = (l_1^1, l_1^2, v_1 + d')$.
		}
		
		\new{Now consider the output transition labeled with $o!$. We have to consider the three cases from Definition~\ref{def:conjunctionTIOA}.}
		\begin{itemize}
			\item \new{$o! \in \mathit{Act}_o^1\cap\mathit{Act}_o^2$. It follows directly from Lemma~\ref{lemma:conjunctionTSandAsharedaction} that $q_2 \xrightarrow{o!} q_3$ is a transition in $Y$.}
			
			\item \new{$o! \in \mathit{Act}_o^1 \setminus \mathit{Act}_o^2$. It follows directly from Corollary~\ref{cor:conjunctionTSandAnonsharedaction} that $q_2 \xrightarrow{o!} q_3$ is a transition in $Y$.}
			
			\item \new{$o! \in \mathit{Act}_o^2 \setminus \mathit{Act}_o^1$. It follows directly from Corollary~\ref{cor:conjunctionTSandAnonsharedaction} (where we switched $A^1$ and $A^2$) that $q_2 \xrightarrow{o!} q_3$ is a transition in $Y$.}
		\end{itemize}
		\new{Thus, in all three cases we have that $q_2 \xrightarrow{o!} q_3$ is a transition in $Y$. Therefore, we can conclude that ${q^X\xlongrightarrow{d'}{}^{\!\! Y} q_2} \wedge q_2\xlongrightarrow{o!}{}^{\!\! Y} q_3$ with $q_2, q_3 \in P$.}
		
		\new{Finally, consider the input transitions labeled with $i?$. Using the same argument as before, we can show that $q_2\xrightarrow{i?} q_4$ in $X$ is also a transition in $Y$, and $q_4\in P$. Therefore, we can conclude that ${q^X\xlongrightarrow{d'}{}^{\!\! Y} q_2} \wedge q_2\xlongrightarrow{o!}{}^{\!\! Y} q_3 \wedge \forall i?\in\mathit{Act}_i^Y: \exists q_4\in P: q_2\xlongrightarrow{i?}{}^{\!\! Y} q_4$ with $q_2, q_3, q_4 \in P$. Thus, the right-hand side is true.}
	\end{itemize}
	\new{Thus, we have shown that when the left-hand side is true for $q^X$ in $X$, it is also true for $q^X$ in $Y$; and that when the right-hand side is true for $q^X$ in $X$, it is also true for $q^X$ in $Y$. Thus, $q^X\in\Theta^Y(P)$. Since $q^X\in P$ was chosen arbitrarily, it holds for all states in $P$. Once we choose $P$ to be the fixed-point of $\Theta^X$, we have that $\Theta^X(P)\subseteq\Theta^Y(P)$.
	}
	
	\new{($\Theta^Y(P) \subseteq \Theta^X(P)$)
		Consider a state $q^Y \in P$. Because $P$ is a postfixed point of $\Theta^Y$, it follows that $p\in\Theta^X(Y)$. From the definition of $\Theta$, it follows that $q^Y \in \overline{\mathrm{err}^Y(\overline{P})}$ and $ q^Y \in \{q\in Q^Y\mid \forall d \geq 0: [\forall q_2\in Q^Y: q\xlongrightarrow{d}{}^{\!\! Y} q_2 \Rightarrow q_2\in P \wedge \forall i?\in \mathit{Act}_i^Y: \exists q_3\in P: q_2\xlongrightarrow{i?}{}^{\!\! Y} q_3]\ \vee [\exists d'\leq d \wedge \exists q_2,q_3\in P \wedge \exists o!\in\mathit{Act}_o^Y: q\xlongrightarrow{d'}{}^{\!\! Y} q_2 \wedge q_2\xlongrightarrow{o!}{}^{\!\! Y} q_3 \wedge \forall i?\in\mathit{Act}_i: \exists q_4\in P: q_2\xlongrightarrow{i?}{}^{\!\! Y} q_4]\}$. From Lemma~\ref{lemma:conjunctionTSandAsameerror} it follows directly that $q^X \in \overline{\mathrm{err}^X(\overline{P})}$. Now we only focus on the second part of the definition of $\Theta$.
	}
	
	\new{Consider a $d\in\mathbb{R}_{\geq 0}$. Then the left-hand side or the right-hand side of the disjunction is true (or both). }
	\begin{itemize}
		\item \new{Assume the left-hand side is true, i.e., $\forall q_2\in Q^Y: q^Y\xlongrightarrow{d}{}^{\!\! Y} q_2 \Rightarrow q_2\in P \wedge \forall i?\in \mathit{Act}_i^Y: \exists q_3\in P: q_2\xlongrightarrow{i?}{}^{\!\! Y} q_3$. Pick a $q_2 \in Q^Y$. The implication is true when $q^Y\arrownot\xlongrightarrow{d}{}^{\!\! Y} q_2$ or $q^Y\xlongrightarrow{d}{}^{\!\! Y} q_2 \wedge q_2\in P \wedge \forall i?\in \mathit{Act}_i^Y: \exists q_3\in P: q_2\xlongrightarrow{i?}{}^{\!\! Y} q_3$. 
		}
		
		\begin{itemize}
			\item \new{Consider the first case. From Lemma~\ref{lemma:conjunctionTSandAdelay} it follows that $q^Y \arrownot\xlongrightarrow{d}{}^{\!\! X}$. Note that $q^Y = (l^1, v^1, l^2, v^2)$. Thus the implication also holds for $q_2$ in $X$.
			}
			
			\item \new{Consider the second case. From Lemma~\ref{lemma:conjunctionTSandAdelay} we have that $q^Y \xlongrightarrow{d}{}^{\!\! Y} q_2$ implies that $q^Y \xlongrightarrow{d}{}^{\!\! X} q_2$, and from Definition~\ref{def:conjunctionTIOTS} of the conjunction for TIOTS that $q^Y = (q_1^1,q_1^2)$ and $q_2 = (q_2^1, q_2^2)$. Also, using Definition~\ref{def:semanticTIOA} of the semantic of a TIOA it follows for $i=1,2$ that $q_1^i = (l_1^i, v_1^i)$, $q_2^i=(l_1^i, v_1^i + d)$, $l_1^i \in \mathit{Loc}^i$, and $v_1^i \in [\mathit{Clk}^i \mapsto \mathbb{R}_{\geq 0}]$.
				Now, pick an $i?\in\mathit{Act}_i^Y$ with its corresponding $q_3$ according to the implication. We have to consider the three cases from Definition~\ref{def:conjunctionTIOTS}.
			}
			\begin{itemize}
				\item \new{$i?\in\mathit{Act}_i^1 \cap\mathit{Act}_i^2$. It follows directly from Lemma~\ref{lemma:conjunctionTSandAsharedaction} that $q_2 \xlongrightarrow{i?}{}^{\!\! X} q_3$.}
				
				\item \new{$i?\in\mathit{Act}_i^1\setminus\mathit{Act}_i^2$. From the fact that $q^Y \xlongrightarrow{d}{}^{\!\! X} q_2$\footnote{This fact is key for finalizing the proof of Theorem~\ref{thrm:conjunctionTSandA}: without adversarial pruning in that theorem, you cannot assume this, and you get stuck in proving that $v_3 \models \mathit{Inv}^2(l_1^2)$ and thus $v_3\models\mathit{Inv}^{A^1\wedge A^2}((l_3^1,l_1^2))$, i.e., you cannot prove that.}, it follows from Definitions~\ref{def:semanticTIOA} and~\ref{def:conjunctionTIOTS} that $v_1^2 + d\models\mathit{Inv}^2(l_1^2)$ (see also proof of Lemma~\ref{lemma:conjunctionTSandAdelay}). Observe that $v_1^2 + d[r\mapsto 0]_{r\in c^1} = v_1^2 + d$, so $v_3 \models \mathit{Inv}^2(l_1^2)$. Now it follows directly from Lemma~\ref{lemma:conjunctionTSandAnonsharedaction} that $q_2 \xlongrightarrow{i?}{}^{\!\! X} q_3$.
				}
				
				\item \new{$i?\in\mathit{Act}_i^2\setminus\mathit{Act}_i^1$.  From the fact that $q^Y \xlongrightarrow{d}{}^{\!\! X} q_2$, it follows from Definitions~\ref{def:semanticTIOA} and~\ref{def:conjunctionTIOTS} that $v_1^1 + d\models\mathit{Inv}^1(l_1^1)$ (see also proof of Lemma~\ref{lemma:conjunctionTSandAdelay}). Observe that $v_1^1 + d[r\mapsto 0]_{r\in c^2} = v_1^1 + d$, so $v_3 \models \mathit{Inv}^1(l_1^1)$. Now it follows directly from Lemma~\ref{lemma:conjunctionTSandAnonsharedaction} (where we switched $A^2$ and $A^2$) that $q_2 \xlongrightarrow{i?}{}^{\!\! X} q_3$.}
			\end{itemize}
			\new{Thus, in all three cases we can show that $q_2\xlongrightarrow{i?}{}^{\!\! Y} q_3$ implies $q_2 \xlongrightarrow{i?}{}^{\!\! X} q_3$. Since we have chosen an arbitrarily $i?\in\mathit{Act}_i^Y$, it holds for all $i?\in \mathit{Act}_i^Y$. Thus the implication also holds for $q_2$ in $X$.}
		\end{itemize}
		\new{Thus, in both cases the implication holds. Therefore, we can conclude that $q^Y\xlongrightarrow{d}{}^{\!\! X} q_2 \Rightarrow q_2\in P \wedge \forall i?\in \mathit{Act}_i^X: \exists q_3\in P: q_2\xlongrightarrow{i?}{}^{\!\! X} q_3$. As $q_2$ is chosen arbitrarily, it holds for all $q_2\in Q^X = Q^Y$. Therefore, the left-hand side is true.
		}
		
		\item \new{Assume the right-hand side is true, i.e., $\exists d'\leq d \wedge \exists q_2,q_3\in P \wedge \exists o!\in\mathit{Act}_o^Y: q\xlongrightarrow{d'}{}^{\!\! Y} q_2 \wedge q_2\xlongrightarrow{o!}{}^{\!\! Y} q_3 \wedge \forall i?\in\mathit{Act}_i: \exists q_4\in P: q_2\xlongrightarrow{i?}{}^{\!\! Y} q_4$. First, focus on the delay. From Lemma~\ref{lemma:conjunctionTSandAdelay} it follows that $q\xlongrightarrow{d'}{}^{\!\! Y} q_2$ implies $q\xlongrightarrow{d'}{}^{\!\! X} q_2$, and from Definition~\ref{def:conjunctionTIOTS} of the conjunction for TIOTS that $q^Y = (q_1^1,q_1^2)$ and $q_2 = (q_2^1, q_2^2)$. Also, using Definition~\ref{def:semanticTIOA} of the semantic of a TIOA it follows for $i=1,2$ that $q_1^i = (l_1^i, v_1^i)$, $q_2^i=(l_1^i, v_1^i + d')$, $l_1^i \in \mathit{Loc}^i$, and $v_1^i \in [\mathit{Clk}^i \mapsto \mathbb{R}_{\geq 0}]$. Now, consider the output transition labeled with $o!$. We have to consider the three cases from Definition~\ref{def:conjunctionTIOTS} of the conjunction for TIOTS.
		}
		\begin{itemize}
			\item \new{$o!\in\mathit{Act}_o^1\cap\mathit{Act}_o^2$. It follows directly from Lemma~\ref{lemma:conjunctionTSandAsharedaction} that $q_2 \xlongrightarrow{o!}{}^{\!\! X} q_3$.
			}
			
			\item \new{$o!\in\mathit{Act}_o^1\subset\mathit{Act}_o^2$. From the fact that $q^Y \xlongrightarrow{d'}{}^{\!\! X} q_2$, it follows from Definitions~\ref{def:semanticTIOA} and~\ref{def:conjunctionTIOTS} that $v_1^2 + d'\models\mathit{Inv}^2(l_1^2)$ (see also proof of Lemma~\ref{lemma:conjunctionTSandAdelay}). Observe that $v_1^2 + d'[r\mapsto 0]_{r\in c^1} = v_1^2 + d'$, so $v_3 \models \mathit{Inv}^2(l_1^2)$. Now it follows directly from Lemma~\ref{lemma:conjunctionTSandAnonsharedaction} that $q_2 \xlongrightarrow{o!}{}^{\!\! X} q_3$.
			}
			
			\item \new{$o!\in\mathit{Act}_o^2\subset\mathit{Act}_o^1$. From the fact that $q^Y \xlongrightarrow{d'}{}^{\!\! X} q_2$, it follows from Definitions~\ref{def:semanticTIOA} and~\ref{def:conjunctionTIOTS} that $v_1^1 + d'\models\mathit{Inv}^1(l_1^1)$ (see also proof of Lemma~\ref{lemma:conjunctionTSandAdelay}). Observe that $v_1^1 + d'[r\mapsto 0]_{r\in c^2} = v_1^1 + d'$, so $v_3 \models \mathit{Inv}^1(l_1^1)$. Now it follows directly from Lemma~\ref{lemma:conjunctionTSandAnonsharedaction} (where we switched $A^2$ and $A^2$) that $q_2 \xlongrightarrow{o!}{}^{\!\! X} q_3$.}
		\end{itemize} 
		\new{Thus, in all three cases we have that $q_2 \xrightarrow{o!}{}^{\!\! X} q_3$ is a transition in $X$. Therefore, we can conclude that ${q^Y\xlongrightarrow{d'}{}^{\!\! X} q_2} \wedge q_2\xlongrightarrow{o!}{}^{\!\! X} q_3$ with $q_2, q_3 \in P$. Thus, the right-hand side is true.}
		
		\new{Finally, consider the input transitions labeled with $i?$. Using the same argument as before, we can show that $q_2\xrightarrow{i?} q_4$ in $Y$ is also a transition in $X$, and $q_4\in P$. Therefore, we can conclude that ${q^Y\xlongrightarrow{d'}{}^{\!\! X} q_2} \wedge q_2\xlongrightarrow{o!}{}^{\!\! X} q_3 \wedge \forall i?\in\mathit{Act}_i^X: \exists q_4\in P: q_2\xlongrightarrow{i?}{}^{\!\! X} q_4$ with $q_2, q_3, q_4 \in P$. Thus, the right-hand side is true.}
	\end{itemize}
	\new{Thus, we have shown that when the left-hand side is true for $q^Y$ in $Y$, it is also true for $q^Y$ in $X$; and that when the right-hand side is true for $q^Y$ in $Y$, it is also true for $q^Y$ in $X$. Thus, $q^Y\in\Theta^X(P)$. Since $q^Y\in P$ was chosen arbitrarily, it holds for all states in $P$. Once we choose $P$ to be the fixed-point of $\Theta^Y$, we have that $\Theta^Y(P)\subseteq\Theta^X(P)$.}
\end{proof}

\subsection{Omitted proofs of Section~\ref{sec:compandcomp}}\label{app:proofs-comp}
\newtheorem*{T13}{Lemma~\ref{lemma:parallelcompositionTSconsistent}}
\begin{T13}
	\new{Given two locally consistent specifications $S^i = (Q^i,q_0^i,\mathit{Act}^i,\rightarrow^i), i=1,2$ where $\mathit{Act}_o^1 \cap \mathit{Act}_o^2 =\emptyset$. Then $S^1 \parallel S^2$ is locally consistent.}
\end{T13}
\begin{proof}
	\new{Since, $S^1$ and $S^2$ are locally consistent, the only reason why $S^1 \parallel S^2$ could be inconsistent is when a new error state is created by the parallel composition. We show by contradiction that this is not possible.}
	
	\new{Assume that state $q_1\in S^1 \parallel S^2$ is an error state. From Definition~\ref{def:error} of the error state it follows that $\exists d_1\in\mathbb{R}_{\geq 0}: q_1\arrownot\xlongrightarrow{d_1} \wedge \forall d_2\in\mathbb{R}_{\geq 0}\forall o!\in\mathit{Act}_o\forall q_2\in Q: q_1\xlongrightarrow{d_2} q_2 \Rightarrow q_2\arrownot\xlongrightarrow{o!}$. From Definition~\ref{def:parallelcompositionTIOTS} of the parallel composition for TIOTS it follows that (1) $q_1 = (q_1^1, q_1^2)$ with $q_1^1\in Q^1$ and $q_1^2\in Q^2$, and that either $q^1\arrownot\xlongrightarrow{d_1}{}^{\!\! 1}$ or $q^2\arrownot\xlongrightarrow{d_2}{}^{\!\! 2}$ (or both); (2) that $q_2 = (q_2^1, q_2^2)$ with $q_2^1\in Q^1$ and $q_2^2\in Q^2$, and that $q_1^1\xlongrightarrow{d_2}{}^{\!\! 1} q_2^1$ and $q_1^2\xlongrightarrow{d_2}{}^{\!\! 2} q_2^2$; and (3) that $o!\in\mathit{Act}_o^1$ and possibly $o?\in\mathit{Act}_i^2$, or $o!\in\mathit{Act}_o^2$ and possibly $o?\in\mathit{Act}_i^1$. In the next step we assume that $o!\in\mathit{Act}_o^1$ and possibly $o?\in\mathit{Act}_i^2$, as the other case is symmetrical. Consider two cases and Definition~\ref{def:parallelcompositionTIOTS}:
		\begin{itemize}
			\item $o?\in\mathit{Act}_i^2$. As $S^2$ is a specification, it is input-enabled. Therefore, $q_1\xlongrightarrow{d_2} q_2 \Rightarrow q_2\arrownot\xlongrightarrow{o!}$ implies that $q_1^1\xlongrightarrow{d_2} q_2^1 \Rightarrow q_2^1\arrownot\xlongrightarrow{o!}$.
			\item $o?\notin\mathit{Act}_i^2$. This directly results in that $q_1\xlongrightarrow{d_2} q_2 \Rightarrow q_2\arrownot\xlongrightarrow{o!}$ implies that $q_1^1\xlongrightarrow{d_2} q_2^1 \Rightarrow q_2^1\arrownot\xlongrightarrow{o!}$.
		\end{itemize}
		Applying the above reasoning for all output actions and knowing that $\mathit{Act}_o = \mathit{Act}_o^1\cup\mathit{Act}_o^2$ from Definition~\ref{def:parallelcompositionTIOTS}, it follows that $\forall o!\in\mathit{Act}_o^1: q_1^1\xlongrightarrow{d_2}{}^{\!\! 1} q_2^1 \implies q_2^1\arrownot\xlongrightarrow{o!}$ and $\forall o!\in\mathit{Act}_o^2: q_1^2\xlongrightarrow{d_2}{}^{\!\! 2} q_2^2 \implies q_2^2\arrownot\xlongrightarrow{o!}$. As this is independent of the actual value of $d_2$, it holds for all $d_2$.}
	
	\new{Finally, since either $q^1\arrownot\xlongrightarrow{d_1}{}^{\!\! 1}$ or $q^2\arrownot\xlongrightarrow{d_2}{}^{\!\! 2}$ (or both), it follows that either $\exists d_1\in\mathbb{R}_{\geq 0}: q_1^1\arrownot\xlongrightarrow{d_1}{}^{\!\! 1} \wedge \forall d_2\in\mathbb{R}_{\geq 0}\forall o!\in\mathit{Act}_o^1\forall q_2^1\in Q^1: q_1^1\xlongrightarrow{d_2}{}^{\!\! 1} q_2^1 \Rightarrow q_2^1\arrownot\xlongrightarrow{o!}{}^{\!\! 1}$ or $\exists d_1\in\mathbb{R}_{\geq 0}: q_1^2\arrownot\xlongrightarrow{d_1}{}^{\!\! 2} \wedge \forall d_2\in\mathbb{R}_{\geq 0}\forall o!\in\mathit{Act}_o^2\forall q_2^2\in Q^1: q_1^2\xlongrightarrow{d_2}{}^{\!\! 2} q_2^2 \Rightarrow q_2^2\arrownot\xlongrightarrow{o!}{}^{\!\! 2}$ (or both). Therefore, either $q_1^1$ or $q_1^2$ (or both) is an error state, which contradicts with the antecedent stating that $S^1$ and $S^2$ are consistent.}
\end{proof}

\newtheorem*{T14}{Lemma~\ref{thm:precongruence}}
\begin{T14}
	Refinement is a pre-congruence with respect to parallel composition: for any specifications $S^1$, $S^2$, and $T$ such that $S^1 \leq S^2$ and $S^1$ is composable with $T$, we have that $S^2$ is composable with $T$ and $S^1\parallel T \leq S^2\parallel T$. 
\end{T14}
\begin{proof}
	\new{
		$S^1 \leq S^2$ implies that $\mathit{Act}_o^{S^2} \subseteq \mathit{Act}_o^{S^2}$ (see Definition~\ref{def:refinement}), and $S^1$ is composable with $T$ implies that $\mathit{Act}_o^{S^1} \cap \mathit{Act}_o^{T} = \emptyset$. Combining this results immediately in that $\mathit{Act}_o^{S^2} \cap \mathit{Act}_o^{T} = \emptyset$, thus $S^2$ is composable with $T$. Furthermore, since $S^1 \leq S^2$, there exists a relation $R\in Q^1\times Q^2$ with the properties given in Definition~\ref{def:refinement} of the refinement. Construct relation $R' = \{ ((q^1, q^T),(q^2,q^T)) \in Q^{S^1 \parallel T}\times Q^{S^2 \parallel T} \mid (q^1,q^2) \in R\}$. We show that $R'$ witnesses $S^1\parallel T \leq S^2\parallel T$. Consider the five cases of refinement for a state pair $((q_1^1, q_1^T),(q_1^2,q_1^T)) \in R'$.
	}
	
	\begin{enumerate}
		
		\item \new{$(q_1^2,q_1^T)\xlongrightarrow{i?}{}^{\!\! S^2 \parallel T} (q_2^2,q_2^T)$ for some $(q_2^2,q_2^T)\in Q^{S^2 \parallel T}$ and $i?\in\mathit{Act}_i^{S^2 \parallel T} \cap \mathit{Act}_i^{S^1 \parallel T}$. Consider the five feasible combinations for input action $i?$ using Definition~\ref{def:parallelcompositionTIOTS} such that $i?\in\mathit{Act}_i^{S^2 \parallel T} \cap \mathit{Act}_i^{S^1 \parallel T}$.
		}
		\begin{itemize}
			\item \new{$i?\in\mathit{Act}_i^{S^1}$, $i?\in\mathit{Act}_i^{S^2}$, and $i?\in\mathit{Act}_i^{T}$. In this case, it follows from Definition~\ref{def:parallelcompositionTIOTS} that $q_1^2\xlongrightarrow{i?}{}^{\!\! S^2} q_2^2$ and $q_1^T\xlongrightarrow{i?}{}^{\!\! T} q_2^T$. Now, using $R$ and Definition~\ref{def:refinement}, it follows that $q_1^1\xlongrightarrow{i?}{}^{\!\! S^1} q_2^1$ and $(q_2^1,q_2^2) \in R$. Thus, following Definition~\ref{def:parallelcompositionTIOTS} again, we have that $(q_1^1,q_1^T)\xlongrightarrow{i?}{}^{\!\! S^1 \parallel T} (q_2^1,q_2^T)$. From the construction of $R'$ we confirm that $((q_2^1, q_2^T),(q_2^2,q_2^T)) \in R'$.
			}
			
			\item \new{$i?\in\mathit{Act}_i^{S^1}$, $i?\in\mathit{Act}_i^{S^2}$, and $i?\notin\mathit{Act}^{T}$. In this case, it follows from Definition~\ref{def:parallelcompositionTIOTS} that $q_1^2\xlongrightarrow{i?}{}^{\!\! S^2} q_2^2$ and $q_1^T = q_2^T$. Now, using $R$ and Definition~\ref{def:refinement}, it follows that $q_1^1\xlongrightarrow{i?}{}^{\!\! S^1} q_2^1$ and $(q_2^1,q_2^2) \in R$. Thus, following Definition~\ref{def:parallelcompositionTIOTS} again, we have that $(q_1^1,q_1^T)\xlongrightarrow{i?}{}^{\!\! S^1 \parallel T} (q_2^1,q_2^T)$ with $q_1^T = q_2^T$. From the construction of $R'$ we confirm that $((q_2^1, q_2^T),(q_2^2,q_2^T)) \in R'$. 
			}
			
			\item \new{$i?\in\mathit{Act}_i^{S^1}$, $i\notin\mathit{Act}^{S^2}$, and $i?\in\mathit{Act}_i^{T}$. This case is infeasible, as Definition~\ref{def:refinement} of refinement requires that $\mathit{Act}_i^{S^1} \subseteq \mathit{Act}_i^{S^2}$.
			}
			
			\item \new{$i?\notin\mathit{Act}^{S^1}$, $i?\in\mathit{Act}_i^{S^2}$, and $i?\in\mathit{Act}_i^{T}$. In this case, it follows from Definition~\ref{def:parallelcompositionTIOTS} that $q_1^2\xlongrightarrow{i?}{}^{\!\! S^2} q_2^2$ and $q_1^T\xlongrightarrow{i?}{}^{\!\! T} q_2^T$. Now, using $R$ and Definition~\ref{def:refinement}, it follows that $(q_2^1,q_2^2) \in R$ and $q_1^1 = q_2^1$. Thus, following Definition~\ref{def:parallelcompositionTIOTS} again, we have that $(q_1^1,q_1^T)\xlongrightarrow{i?}{}^{\!\! S^1 \parallel T} (q_2^1,q_2^T)$ and $q_1^1 = q_2^1$. From the construction of $R'$ we confirm that $((q_2^1, q_2^T),(q_2^2,q_2^T)) \in R'$.
			}
			
			\item \new{$i?\notin\mathit{Act}^{S^1}$, $i?\notin\mathit{Act}^{S^2}$, and $i?\in\mathit{Act}_i^{T}$. In this case, it follows from Definition~\ref{def:parallelcompositionTIOTS} that $q_1^T\xlongrightarrow{i?}{}^{\!\! T} q_2^T$ and $q_1^2 = q_2^2$. Following Definition~\ref{def:parallelcompositionTIOTS} again, we have that $(q_1^1,q_1^T)\xlongrightarrow{i?}{}^{\!\! S^1 \parallel T} (q_2^1,q_2^T)$ and $q_1^1 = q_2^1$. From the construction of $R'$ we confirm that $((q_2^1, q_2^T),(q_2^2,q_2^T)) \in R'$.}
		\end{itemize}
		\new{So, in all feasible cases we can show that $(q_1^1,q_1^T)\xlongrightarrow{i?}{}^{\!\! S^1 \parallel T} (q_2^1,q_2^T)$ and $((q_2^1, q_2^T),(q_2^2,q_2^T)) \in R'$.
		}
		
		\item \new{$(q_1^2,q_1^T)\xlongrightarrow{i?}{}^{\!\! S^2 \parallel T} (q_2^2,q_2^T)$ for some $(q_2^2,q_2^T)\in Q^{S^2 \parallel T}$ and $i?\in\mathit{Act}_i^{S^2 \parallel T} \setminus \mathit{Act}_i^{S^1 \parallel T}$. In this case it follows from Definition~\ref{def:refinement} and~\ref{def:parallelcompositionTIOTS} that $i?\in\mathit{Act}_i^{S^2}$, $i?\notin\mathit{Act}_i^{S^1}$, and $i?\notin\mathit{Act}_i^{T}$. Therefore, from the same definitions, we have that $q_1^2\xlongrightarrow{i?}{}^{\!\! S^2} q_2^2$ and $q_1^T = q_2^T$. Now, using $R$ and Definition~\ref{def:refinement}, it follows that $(q_2^1,q_2^2) \in R$ and $q_1^1 = q_2^1$. From the construction of $R'$ we confirm that $((q_2^1, q_2^T),(q_2^2,q_2^T)) \in R'$.
		}
		
		\item \new{$(q_1^1, q_1^T)\xlongrightarrow{o!}{}^{\!\! S^1 \parallel T} (q_2^1, q_2^T)$ for some $(q_2^1, q_2^T)\in Q^{S^1 \parallel T}$ and $o!\in\mathit{Act}_o^{S^1 \parallel T} \cap \mathit{Act}_o^{S^2 \parallel T}$. Consider the eight feasible combinations for output action $o!$ using Definition~\ref{def:parallelcompositionTIOTS} such that $o!\in\mathit{Act}_o^{S^2 \parallel T} \cap \mathit{Act}_o^{S^1 \parallel T}$, already taking into account that if $o\in\mathit{Act}^{S^1}$ and $o\in\mathit{Act}^{S^2}$ then $o!\in\mathit{Act}_o^{S^1}$ and $o!\in\mathit{Act}_o^{S^2}$ or $o?\in\mathit{Act}_i^{S^1}$ and $o?\in\mathit{Act}_i^{S^2}$ (see Definition~\ref{def:refinement}).}
		\begin{itemize}
			\item \new{$o!\in\mathit{Act}_o^{S^1}$, $o!\in\mathit{Act}_o^{S^2}$, and $o\in\mathit{Act}^{T}$\footnote{With this notation, we indicate that it does not matter whether $o!\in\mathit{Act}_o^T$ or $o?\in\mathit{Act}_i^T$.}. In this case, it follows from Definition~\ref{def:parallelcompositionTIOTS} that $q_1^1\xlongrightarrow{o!}{}^{\!\! S^1} q_2^1$ and $q_1^T\xlongrightarrow{o}{}^{\!\! T} q_2^T$. Now, using $R$ and Definition~\ref{def:refinement}, it follows that $q_1^2\xlongrightarrow{o!}{}^{\!\! S^2} q_2^2$ and $(q_2^1,q_2^2) \in R$. Thus, following Definition~\ref{def:parallelcompositionTIOTS} again, we have that $(q_1^2,q_1^T)\xlongrightarrow{o!}{}^{\!\! S^2 \parallel T} (q_2^2,q_2^T)$. From the construction of $R'$ we confirm that $((q_2^1, q_2^T),(q_2^2,q_2^T)) \in R'$.
			}
			
			\item \new{$o?\in\mathit{Act}_i^{S^1}$, $o?\in\mathit{Act}_i^{S^2}$, and $o!\in\mathit{Act}_o^{T}$. In this case, it follows from Definition~\ref{def:parallelcompositionTIOTS} that $q_1^1\xlongrightarrow{o?}{}^{\!\! S^1} q_2^1$ and $q_1^T\xlongrightarrow{o!}{}^{\!\! T} q_2^T$. As $S^2$ is input-enabled, it follows that $q_1^2\xlongrightarrow{o?}{}^{\!\! S^2} q_2^2$ for some $q_2^2 \in Q^2$. Now, using $R$ and Definition~\ref{def:refinement}, it follows that $(q_2^1,q_2^2) \in R$. Thus, following Definition~\ref{def:parallelcompositionTIOTS} again, we have that $(q_1^2,q_1^T)\xlongrightarrow{o!}{}^{\!\! S^2 \parallel T} (q_2^2,q_2^T)$. From the construction of $R'$ we confirm that $((q_2^1, q_2^T),(q_2^2,q_2^T)) \in R'$.
			}
			
			\item \new{$o!\in\mathit{Act}_o^{S^1}$, $o!\in\mathit{Act}_o^{S^2}$, and $o!\notin\mathit{Act}^{T}$. In this case, it follows from Definition~\ref{def:parallelcompositionTIOTS} that $q_1^1\xlongrightarrow{o!}{}^{\!\! S^1} q_2^1$ and $q_1^T = q_2^T$. Now, using $R$ and Definition~\ref{def:refinement}, it follows that $q_1^2\xlongrightarrow{o!}{}^{\!\! S^2} q_2^2$ and $(q_2^1,q_2^2) \in R$. Thus, following Definition~\ref{def:parallelcompositionTIOTS} again, we have that $(q_1^2,q_1^T)\xlongrightarrow{o!}{}^{\!\! S^2 \parallel T} (q_2^2,q_2^T)$ with $q_1^T = q_2^T$. From the construction of $R'$ we confirm that $((q_2^1, q_2^T),(q_2^2,q_2^T)) \in R'$. 
			}
			
			\item \new{$o!\in\mathit{Act}_o^{S^1}$, $o!\notin\mathit{Act}^{S^2}$, and $o!\in\mathit{Act}_o^{T}$. In this case, it follows from Definition~\ref{def:parallelcompositionTIOTS} that $q_1^1\xlongrightarrow{o!}{}^{\!\! S^1} q_2^1$ and $q_1^T\xlongrightarrow{o!}{}^{\!\! T} q_2^T$. Now, using $R$ and Definition~\ref{def:refinement}, it follows that $(q_2^1,q_2^2) \in R$ and $q_1^2 = q_2^2$. Thus, following Definition~\ref{def:parallelcompositionTIOTS} again, we have that $(q_1^2,q_1^T)\xlongrightarrow{o!}{}^{\!\! S^2 \parallel T} (q_2^2,q_2^T)$ and $q_1^2 = q_2^2$. From the construction of $R'$ we confirm that $((q_2^1, q_2^T),(q_2^2,q_2^T)) \in R'$.
			}
			
			\item \new{$o?\in\mathit{Act}_i^{S^1}$, $o!\notin\mathit{Act}^{S^2}$, and $o!\in\mathit{Act}_o^{T}$. This case is infeasible, as Definition~\ref{def:refinement} of refinement requires that $\mathit{Act}_i^{S^1} \subseteq \mathit{Act}_i^{S^2}$.
			}
			
			\item \new{$o!\notin\mathit{Act}^{S^1}$, $o!\in\mathit{Act}_o^{S^2}$, and $o!\in\mathit{Act}_o^{T}$. This case is infeasible, as Definition~\ref{def:refinement} of refinement requires that $\mathit{Act}_o^{S^2} \subseteq \mathit{Act}_o^{S^1}$.
			}
			
			\item \new{$o!\notin\mathit{Act}^{S^1}$, $o?\in\mathit{Act}_i^{S^2}$, and $o!\in\mathit{Act}_o^{T}$. In this case, it follows from Definition~\ref{def:parallelcompositionTIOTS} that $q_1^T\xlongrightarrow{o!}{}^{\!\! T} q_2^T$ and $q_1^1 = q_2^1$. As $S^2$ is input-enabled, it follows that $q_1^2\xlongrightarrow{o?}{}^{\!\! S^2} q_2^2$ for some $q_2^2 \in Q^2$. Now, using $R$ and Definition~\ref{def:refinement}, it follows that $(q_2^1,q_2^2) \in R$. Thus, following Definition~\ref{def:parallelcompositionTIOTS} again, we have that $(q_1^2,q_1^T)\xlongrightarrow{o!}{}^{\!\! S^2 \parallel T} (q_2^2,q_2^T)$. From the construction of $R'$ we confirm that $((q_2^1, q_2^T),(q_2^2,q_2^T)) \in R'$.
			}
			
			\item \new{$o!\notin\mathit{Act}^{S^1}$, $o!\notin\mathit{Act}^{S^2}$, and $o!\in\mathit{Act}_i^{T}$. In this case, it follows from Definition~\ref{def:parallelcompositionTIOTS} that $q_1^T\xlongrightarrow{o!}{}^{\!\! T} q_2^T$ and $q_1^1 = q_2^1$. Following Definition~\ref{def:parallelcompositionTIOTS} again, we have that $(q_1^2,q_1^T)\xlongrightarrow{o!}{}^{\!\! S^2 \parallel T} (q_2^2,q_2^T)$ and $q_1^2 = q_2^2$. From the construction of $R'$ we confirm that $((q_2^1, q_2^T),(q_2^2,q_2^T)) \in R'$.}
		\end{itemize}
		\new{So, in all feasible cases we can show that $(q_1^2,q_1^T)\xlongrightarrow{o!}{}^{\!\! S^2 \parallel T} (q_2^2,q_2^T)$ and $((q_2^1,q_2^T), (q_2^2,q_2^T)) \in R'$.}
		
		\item \new{$(q_1^1, q_1^T)\xlongrightarrow{o!}{}^{\!\! S^1 \parallel T} (q_2^1, q_2^T)$ for some $(q_2^1, q_2^T)\in Q^{S^1 \parallel T}$ and $o!\in\mathit{Act}_o^{S^1 \parallel T} \setminus \mathit{Act}_o^{S^2 \parallel T}$. In this case it follows from Definitions~\ref{def:refinement} and~\ref{def:parallelcompositionTIOTS} that $o!\in\mathit{Act}_o^{S^1}$, $o\notin\mathit{Act}^{S^2}$, and $o\notin\mathit{Act}^{T}$. Therefore, from the same definitions, we have that $q_1^1\xlongrightarrow{o!}{}^{\!\! S^1} q_2^1$ and $q_1^T = q_2^T$. Now, using $R$ and Definition~\ref{def:refinement}, it follows that $(q_2^1,q_2^2) \in R$ and $q_1^2 = q_2^2$. From the construction of $R'$ we confirm that $((q_2^1,q_2^T), (q_2^2,q_2^T)) \in R'$.}
		
		\item \new{$(q_1^1, q_1^T)\xlongrightarrow{d}{}^{\!\! S^1 \parallel T} (q_2^1, q_2^T)$ for some $(q_2^1, q_2^T)\in Q^{S^1 \parallel T}$ and $d\in \mathbb{R}_{\geq 0}$. In this case, it follows from Definition~\ref{def:parallelcompositionTIOTS} that $q_1^2\xlongrightarrow{d}{}^{\!\! S^1} q_2^2$ and $q_1^T\xlongrightarrow{d}{}^{\!\! T} q_2^T$. Now, using $R$ and Definition~\ref{def:refinement}, it follows that $q_1^2\xlongrightarrow{d}{}^{\!\! S^2} q_2^2$ and $(q_2^1,q_2^2) \in R$. Thus, following Definition~\ref{def:parallelcompositionTIOTS} again, we have that $(q_1^2,q_1^T)\xlongrightarrow{d}{}^{\!\! S^2 \parallel T} (q_2^2,q_2^T)$. From the construction of $R'$ we confirm that $((q_2^1,q_2^T), (q_2^2,q_2^T)) \in R'$.}
	\end{enumerate}
\end{proof}

\newtheorem*{T15}{Lemma~\ref{lemma:cooperativepruning}}
\begin{T15}
	\new{Given a specification $S = (Q^S,s_0,\mathit{Act}^S,\rightarrow^S)$ and its cooperatively pruned subspecification $S^{\forall}$. It holds that $S = S^{\forall}$.}
\end{T15}
\begin{proof}
	\new{The idea of pruning is to remove error states and related transitions from a specification that violate the independent progress property, as all states of any implementation of that specification need to have independent progress, see Definition~\ref{def:implementation}. So, for a state $q_{\mathrm{imerr}} \in \mathrm{imerr}^S$ of $S$ (see Definition~\ref{def:immediateerror}), it holds that $(\exists d\in\mathbb{R}_{\geq 0}: q_{\mathrm{imerr}}\arrownot\xlongrightarrow{d}) \wedge \forall d\in\mathbb{R}_{\geq 0}\forall o!\in\mathit{Act}_o^S\forall q'\in Q^S: q_{\mathrm{imerr}}\xlongrightarrow{d} q' \Rightarrow q'\arrownot\xlongrightarrow{o!}$.}
	
	\new{Now, consider a specification $T = (t, t, \mathit{Act}^T, \rightarrow^T)$ with a single state $t$, $\mathit{Act}^T = \mathit{Act}_o^T = \mathit{Act}_i^S \cup \{\tau\}$ with $\tau\notin\mathit{Act}^S$, and $\rightarrow^T = \{(t,a,t) \mid a \in \mathit{Act}^T\} \cup \{(t, d, t) \mid d \in \mathbb{R}_{\geq 0}\}$. The unique event $\tau$ is present to ensure that the following argument holds in case $S$ does not have any input actions. In the composition $S\parallel T$, it still holds that $(q_{\mathrm{imerr}}, t) \xlongrightarrow{d} (q', t)$ (see Definition~\ref{def:parallelcompositionTIOTS}). Since a specification is input enabled, Definition~\ref{def:specification}, we know that in the composition $S\parallel T$ there exist an output action $o!\in\mathit{Act}^T$ such that $(q',t) \xlongrightarrow{o!}$. Thus, in the composition $S\parallel T$, the state $(q_{\mathrm{imerr}}, t)$ is no longer an immediate error state. As this holds for all $q_{\mathrm{imerr}} \in \mathrm{imerr}^S$, we have that $\mathrm{imerr}^{S\parallel T} = \emptyset$. And once $\mathrm{imerr}^{S\parallel T} = \emptyset$, we have that $\mathrm{err}^{S\parallel T}(\emptyset) = \emptyset$ and therefore $\mathrm{incons}^{S\parallel T} = \emptyset$ (using the fixed-point operator $\pi$). Thus for this $T$ we need to keep all states of $S$ in $S^{\forall}$ to ensure that $\mod{S\parallel T} = \mod{S^{\forall} \parallel T}$. }
\end{proof}

\subsection{Omitted proofs of Section~\ref{sec:quotient}}\label{app:proofs-quotient}

\begin{lemma}\label{lemma:quotientparallelcompositionalphabet}
	For any two specifications $S$ and $T$ such that the quotient $T\quotient S$ is defined, and for any implementation $X$ over the same alphabet as $T\quotient S$, we have that $S \parallel X$ is defined, $\mathit{Act}_i^{S \parallel X} = \mathit{Act}_i^T$ and $\mathit{Act}_0^{S \parallel X} = \mathit{Act}_o^S \cup \mathit{Act}_o^T \cup \mathit{Act}_i^S\setminus\mathit{Act}_i^T$.
\end{lemma}
\begin{proof}
	\new{We will first show that $S \parallel X$ is defined. This boils down to show that $S$ and $X$ are composable., i.e., $\mathit{Act}_o^S \cap \mathit{Act}_o^X = \emptyset$. From Definition~\ref{def:quotientTIOTS} and the assumption that $X$ has the same alphabet as $T\quotient S$, it follows that $\mathit{Act}_o^X = \mathit{Act}_o^T\setminus\mathit{Act}_o^S \cup \mathit{Act}_i^S\setminus\mathit{Act}_i^T$. Thus it holds that $\mathit{Act}_o^S \cap \mathit{Act}_o^X = \emptyset$.
	}
	
	\new{To show that $\mathit{Act}_i^{S \parallel X} = \mathit{Act}_i^T$, we follow Definition~\ref{def:parallelcompositionTIOTS} of the parallel composition and Definition~\ref{def:quotientTIOTS} of the quotient and use careful rewriting to get to this conclusion.
		\begin{align*}
			\mathit{Act}_i^{S \parallel X} &= \mathit{Act}_i^S\setminus \mathit{Act}_o^X \cup \mathit{Act}_i^X \setminus\mathit{Act}_o^S \\
			&= \mathit{Act}_i^S\setminus (\mathit{Act}_o^T\setminus\mathit{Act}_o^S \cup \mathit{Act}_i^S\setminus\mathit{Act}_i^T) \cup (\mathit{Act}_i^T\cup\mathit{Act}_o^S)\setminus\mathit{Act}_o^S \\
			&= \mathit{Act}_i^S\setminus (\mathit{Act}_o^T\setminus\mathit{Act}_o^S \cup \mathit{Act}_i^S\setminus\mathit{Act}_i^T) \cup \mathit{Act}_i^T\setminus\mathit{Act}_o^S \cup \mathit{Act}_o^S\setminus\mathit{Act}_o^S \\
			&= \mathit{Act}_i^S\setminus (\mathit{Act}_o^T\setminus\mathit{Act}_o^S \cup \mathit{Act}_i^S\setminus\mathit{Act}_i^T) \cup \mathit{Act}_i^T \\
			&= \left(\mathit{Act}_i^S\setminus (\mathit{Act}_o^T\setminus\mathit{Act}_o^S) \cap \mathit{Act}_i^S\setminus (\mathit{Act}_i^S\setminus\mathit{Act}_i^T)\right) \cup \mathit{Act}_i^T \\
			&= \left(\mathit{Act}_i^S\setminus (\mathit{Act}_o^T\setminus\mathit{Act}_o^S) \cap \mathit{Act}_i^S\cap\mathit{Act}_i^T\right) \cup \mathit{Act}_i^T \\
			&= \left(\left((\mathit{Act}_i^S \cap \mathit{Act}_o^S) \cup (\mathit{Act}_i^S \setminus \mathit{Act}_o^T)\right) \cap \mathit{Act}_i^S\cap\mathit{Act}_i^T\right) \cup \mathit{Act}_i^T \\
			&= \left( \mathit{Act}_i^S \setminus \mathit{Act}_o^T\cap \mathit{Act}_i^S\cap\mathit{Act}_i^T\right) \cup \mathit{Act}_i^T \\
			&= \left( \mathit{Act}_i^S \cap \mathit{Act}_i^S\cap\mathit{Act}_i^T \right) \setminus \mathit{Act}_o^T \cup \mathit{Act}_i^T  \\
			&= \left( \mathit{Act}_i^S\cap\mathit{Act}_i^T \right) \setminus \mathit{Act}_o^T \cup \mathit{Act}_i^T \\
			&= \left( \mathit{Act}_i^S\cap (\mathit{Act}_i^T \setminus \mathit{Act}_o^T) \right) \cup \mathit{Act}_i^T \\
			&= \left( \mathit{Act}_i^S\cap \mathit{Act}_i^T \right) \cup \mathit{Act}_i^T \\
			&= \mathit{Act}_i^T 
		\end{align*}
	}
	
	\new{To show that $\mathit{Act}_0^{S \parallel X} = \mathit{Act}_o^S \cup \mathit{Act}_o^T \cup \mathit{Act}_i^S\setminus\mathit{Act}_i^T$, we follow again Definition~\ref{def:parallelcompositionTIOTS} of the parallel composition and Definition~\ref{def:quotientTIOTS} of the quotient and use careful rewriting to get to this conclusion.
		\begin{align*}
			\mathit{Act}_0^{S \parallel X} &= \mathit{Act}_o^S \cup \mathit{Act}_o^X \\
			&= \mathit{Act}_o^S \cup (\mathit{Act}_o^T\setminus\mathit{Act}_o^S \cup \mathit{Act}_i^S\setminus\mathit{Act}^T) \\
			&= \mathit{Act}_o^S \cup \mathit{Act}_o^T \cup \mathit{Act}_i^S\setminus\mathit{Act}_i^T
		\end{align*}	
	}
\end{proof}

\newtheorem*{T16}{Lemma~\ref{thm:quotient_correct}}
\begin{T16}
	For any two specifications $S$ and $T$ such that the quotient $T\quotient S$ is defined, and for any implementation $X$ over the same alphabet as $T\quotient S$, we have that $S \parallel X$ is defined and $S \parallel X \leq T$ iff $X \leq T \quotient S$.     
\end{T16}
\begin{proof}
	\new{It is shown in Lemma~\ref{lemma:quotientparallelcompositionalphabet} that $S \parallel X$ is defined. The alphabet pre-condition of Definition~\ref{def:refinement} is satisfied for $X \leq T \quotient S$ by definition of $X$; using Lemma~\ref{lemma:quotientparallelcompositionalphabet} we can see that this is also the case for $S \parallel X \leq T$. So we only have to show that $S \parallel X \leq T$ iff $X \leq T \quotient S$.
	}
	
	\new{($S \parallel X \leq T \Rightarrow X \leq T \quotient S$) Since $S\parallel X \leq T$, it follows from Definition~\ref{def:refinement} of refinement that there exists a relation $R \in Q^{S\parallel X} \times Q^T$ that witness the refinement. Note that $Q^{S \parallel X} = Q^S \times Q^X$ according to Definition~\ref{def:parallelcompositionTIOTS}. Construct relation $R' = \{(q_1^X,(q_1^T,q_1^S))\in Q^X \times Q^{T \quotient S} \mid ((q_1^S,q_1^X), q_1^T) \in R\} \cup \{ (q_1^X,u)\in Q^X \times Q^{T \quotient S} \mid q_1^X \in Q^X \}$. We will show that $R'$ witnesses $X\leq T\quotient S$. First consider the five cases of Definition~\ref{def:refinement} for a state pair $(q_1^X, (q_1^T,q_1^S)) \in R'$.
	}
	\begin{enumerate}
		\item \new{$(q_1^T,q_1^S)\xlongrightarrow{i?}{}^{\!\! T\quotient S} (q_2^T,q_2^S)$ for some $ (q_2^T,q_2^S)\in Q^{T\quotient S}$ and $i?\in\mathit{Act}_i^{T\quotient S} \cap \mathit{Act}_i^X$. By definition of $X$ it follows that $\mathit{Act}_i^{T\quotient S} = \mathit{Act}_i^X$. Consider the following five possible cases from Definition~\ref{def:quotientTIOTS} of the quotient that might result in $i?\in\mathit{Act}_i^{T\quotient S} (= \mathit{Act}_i^T \cup \mathit{Act}_o^S)$.}
		\begin{itemize}
			\item \new{$i?\in\mathit{Act}_i^T$ and $i!\in\mathit{Act}_o^S$. This case is actually not feasible, since Definition~\ref{def:quotientTIOTS} also requires that $\mathit{Act}_o^S\cap\mathit{Act}_i^T=\emptyset$.
			}
			\item \new{$i?\in\mathit{Act}_i^T$ and $i?\in\mathit{Act}_i^S$. In this case, it follows from Definition~\ref{def:quotientTIOTS} that $q_1^T\xlongrightarrow{i?}{}^{\!\! T} q_2^T$ and $q_1^S\ {\xlongrightarrow{i?}{}^{\!\! S}} q_2^S$. Now, using $R$, the first case of Definition~\ref{def:refinement} of refinement, and the fact that $\mathit{Act}_i^{S\parallel X} = \mathit{Act}_i^T$ (Lemma~\ref{lemma:quotientparallelcompositionalphabet}) it follows that $(q_1^S,q_1^X)\xlongrightarrow{i?}{}^{\!\! S\parallel X} (q_2^S,q_2^X)$ and $((q_2^S,q_2^X), q_2^T) \in R$. From Definition~\ref{def:parallelcompositionTIOTS} of parallel composition it follows that $q_1^X\xlongrightarrow{i?}{}^{\!\! X} q_2^X$. From the construction of $R'$ we confirm that $(q_2^X,(q_2^T,q_2^S)) \in R'$.
			}
			\item \new{$i?\in\mathit{Act}_i^T$ and $i?\notin\mathit{Act}^S$. In this case, it follows from Definition~\ref{def:quotientTIOTS} that $q_1^T\xlongrightarrow{i?}{}^{\!\! T} q_2^T$ and $q_1^S = q_2^S$. Now, using $R$, the first case of Definition~\ref{def:refinement} of refinement, and the fact that $\mathit{Act}_i^{S\parallel X} = \mathit{Act}_i^T$ (Lemma~\ref{lemma:quotientparallelcompositionalphabet}) it follows that $(q_1^S,q_1^X)\xlongrightarrow{i?}{}^{\!\! S\parallel X} (q_2^S,q_2^X)$ and $((q_2^S,q_2^X), q_2^T) \in R$. From Definition~\ref{def:parallelcompositionTIOTS} of parallel composition it follows that $q_1^X\xlongrightarrow{i?}{}^{\!\! X} q_2^X$. From the construction of $R'$ we confirm that $(q_2^X,(q_2^T,q_2^S)) \in R'$.
			}
			\item \new{$i!\in\mathit{Act}_o^T$ and $i!\in\mathit{Act}_o^S$. In this case, there are three possible options from Definition~\ref{def:quotientTIOTS}.} 
				
				\begin{itemize} 
					\item \new{$q_1^T\xlongrightarrow{i!}{}^{\!\! T} q_2^T$ and $q_1^S{\xlongrightarrow{i!}{}^{\!\! S}} q_2^S$. 
						Since $X$ is an implementation and  $i? \in \mathit{Act}_i^X$, it follows that $q_1^X\xlongrightarrow{i?}{}^{\!\! X} q_2^X$ for some $q_2^X \in Q^X$ (any implementation is a specification, see Definition~\ref{def:implementation}, which is input-enabled, see Definition~\ref{def:specification}). Now, using Definition~\ref{def:parallelcompositionTIOTS} of parallel composition it follows that $(q_1^S,q_1^X)\ {\xlongrightarrow{i!}{}^{\!\! S\parallel X}} (q_2^S,q_2^X)$. Using $R$ and the third case of Definition~\ref{def:refinement} of refinement, it follows that $((q_2^S,q_2^X), q_2^T) \in R$. Thus from the construction of $R'$ we confirm that $(q_2^X,(q_2^T,q_2^S)) \in R'$.}
					\item \new{$q_1^S{\arrownot\xlongrightarrow{i!}{}^{\!\! S}}$.
						In this case, $(q_2^T,q_2^S) = u$. Again, since $X$ is an implementation and $i? \in \mathit{Act}_i^X$, it follows that $q_1^X\xlongrightarrow{i?}{}^{\!\! X} q_2^X$ for some $q_2^X \in Q^X$. By construction of $R'$ it follows that $(q_2^X,(q_2^T,q_2^S)) =  (q_2^X,u)\in R'$.}
					\item \new{$q_1^T\arrownot\xlongrightarrow{i!}{}^{\!\! T}$ and $q_1^S{\xlongrightarrow{i!}{}^{\!\! S}} q_2^S$. 
						Since $S\parallel X \leq T$ holds and $(q_1^X, (q_1^T,q_1^S)) \in R'$ implies, by construction $R'$, that $((q_1^S,q_1^X), q_1^T) \in R$, we can conclude that $q_1^T\arrownot\xlongrightarrow{i!}{}^{\!\! T}$ implies $(q_1^S,q_1^X)\arrownot\xlongrightarrow{i!}{}^{\!\! S\parallel X}$ from Definition~\ref{def:refinement} of refinement. Since $X$ is an implementation and $i? \in \mathit{Act}_i^X$, it follows that $q_1^X\xlongrightarrow{i?}{}^{\!\! X}$. Therefore, from Definition~\ref{def:parallelcompositionTIOTS} of parallel composition it follows that $q_1^S\arrownot\xlongrightarrow{i!}{}^{\!\! S}$. This contradicts with $q_1^S{\xlongrightarrow{i!}{}^{\!\! S}} q_2^S$, so this option is infeasible.}
		 	\end{itemize}
			
			\item \new{$i!\notin\mathit{Act}^T$ and $i!\in\mathit{Act}_o^S$. In this case, there are two possible options from Definition~\ref{def:quotientTIOTS}. } 
				\begin{itemize}
					\item \new{$q_1^S\xlongrightarrow{i!}{}^{\!\! S} q_2^S$ and $q_1^T = q_2^T$. 
						Since $X$ is an implementation and  $i? \in \mathit{Act}_i^X$, it follows that $q_1^X\xlongrightarrow{i?}{}^{\!\! X} q_2^X$ for some $q_2^X \in Q^X$ (any implementation is a specification, see Definition~\ref{def:implementation}, which is input-enabled, see Definition~\ref{def:specification}). Now, using Definition~\ref{def:parallelcompositionTIOTS} of parallel composition it follows that $(q_1^S,q_1^X)\ {\xlongrightarrow{i!}{}^{\!\! S\parallel X}} (q_2^S,q_2^X)$. Using $R$ and the forth case of Definition~\ref{def:refinement} of refinement, it follows that $((q_2^S,q_2^X), q_2^T) \in R$. Thus from the construction of $R'$ we confirm that $(q_2^X,(q_2^T,q_2^S)) \in R'$.}
					\item \new{$q_1^S{\arrownot\xlongrightarrow{i!}{}^{\!\! S}}$.
						In this case, $(q_2^T,q_2^S) = u$. Again, since $X$ is an implementation and $i? \in \mathit{Act}_i^X$, it follows that $q_1^X\xlongrightarrow{i?}{}^{\!\! X} q_2^X$ for some $q_2^X \in Q^X$. By construction of $R'$ it follows that $(q_2^X,(q_2^T,q_2^S)) =  (q_2^X,u)\in R'$.}
				\end{itemize}
		\end{itemize}
		\new{So, in all feasible cases we can show that $q_1^X\xlongrightarrow{i?}{}^{\!\! X} q_2^X$  and $(q_2^X,(q_2^T,q_2^S)) \in R'$.}
		
		\item \new{$ (q_1^T,q_1^S)\xlongrightarrow{i?}{}^{\!\! T\quotient S}  (q_2^T,q_2^S)$ for some $ (q_2^T,q_2^S) \in Q^{T\quotient S}$ and $i?\in\mathit{Act}_i^{T\quotient S} \setminus \mathit{Act}_i^X$. By definition of $X$ it follows that $\mathit{Act}_i^{T\quotient S} \setminus \mathit{Act}_i^X = \emptyset$, so this case can be ignored.
		}
		
		\item \new{$q_1^X\xlongrightarrow{o!}{}^{\!\! X} q_2^X$ for some $q_2^X\in Q^X$ and $o!\in\mathit{Act}_o^X \cap \mathit{Act}_o^{T\quotient S}$. By definition of $X$ it follows that $\mathit{Act}_o^X = \mathit{Act}_o^{T\quotient S}$. Consider the following five possible cases from Definition~\ref{def:quotientTIOTS} of the quotient that might result in $o!\in\mathit{Act}_o^{T\quotient S} (= \mathit{Act}_o^T\setminus\mathit{Act}_o^S \cup \mathit{Act}_i^S\setminus\mathit{Act}_i^T)$.
		}
		\begin{itemize}
			\item \new{$o!\in\mathit{Act}_o^T\setminus\mathit{Act}_o^S$ and $o?\in\mathit{Act}_i^S\setminus\mathit{Act}_i^T$. It follows from Definition~\ref{def:specification} of a specification that $S$ is input-enabled. Therefore, there is a transition $q_1^S\xlongrightarrow{o?}{}^{\!\! S} q_2^S$ for some $q_2^S\in Q^S$. Now, from Definition~\ref{def:parallelcompositionTIOTS} of parallel composition it follows that there is a transition $(q_1^S,q_1^X)\xlongrightarrow{o!}{}^{\!\! S\parallel X}  (q_2^S, q_2^X)$. Using $R$ and the third case of Definition~\ref{def:refinement} of refinement, it follows that $q_1^T\xlongrightarrow{o!}{}^{\!\! T} q_2^T$ and $((q_2^S,q_2^X), q_2^T) \in R$. Now, using Definition~\ref{def:quotientTIOTS} of the quotient, it follows that $(q_1^T,q_1^S)\xlongrightarrow{o!}{}^{\!\! T\quotient S}  (q_2^T,q_2^S)$. And from the construction of $R'$ we confirm that $(q_2^X,(q_2^T,q_2^S)) \in R'$.
			}
			\item \new{$o!\in\mathit{Act}_o^T\setminus\mathit{Act}_o^S$ and $o?\in\mathit{Act}_i^S\cap\mathit{Act}_i^T$. This case is not feasible, as an action cannot be both an output and input in $T$.
			}
			\item \new{$o!\in\mathit{Act}_o^T\setminus\mathit{Act}_o^S$ and $o?\notin\mathit{Act}_i^S$. In this case, it follows that $o \notin\mathit{Act}^S$ at all. Then from Definition~\ref{def:parallelcompositionTIOTS} it follows that there is a transition $(q_1^S,q_1^X)\xlongrightarrow{i?}{}^{\!\! S\parallel X}  (q_2^S, q_2^X)$ and $q_1^S = q_2^S$. Using $R$ and the third case of Definition~\ref{def:refinement} of refinement, it follows that $q_1^T\xlongrightarrow{o!}{}^{\!\! T} q_2^T$ and $((q_2^S,q_2^X), q_2^T) \in R$. Now, using Definition~\ref{def:quotientTIOTS} of the quotient, it follows that $(q_1^T,q_1^S)\xlongrightarrow{o!}{}^{\!\! T\quotient S}  (q_2^T,q_2^S)$. And from the construction of $R'$ we confirm that $(q_2^X,(q_2^T,q_2^S)) \in R'$.
			}
			\item \new{$o!\in\mathit{Act}_o^T\cap\mathit{Act}_o^S$ and $o?\in\mathit{Act}_i^S\setminus\mathit{Act}_i^T$. This case is not feasible, as an action cannot be both an output and input in $S$.
			}
			\item \new{$o!\notin\mathit{Act}_o^T$ and $o?\in\mathit{Act}_i^S\setminus\mathit{Act}_i^T$.  It follows from Definition~\ref{def:specification} of a specification that $S$ is input-enabled. Therefore, there is a transition $q_1^S\xlongrightarrow{o?}{}^{\!\! S} q_2^S$ for some $q_2^S\in Q^S$. Now, from Definition~\ref{def:parallelcompositionTIOTS} of parallel composition it follows that there is a transition $(q_1^S,q_1^X)\xlongrightarrow{i?}{}^{\!\! S\parallel X}  (q_2^S, q_2^X)$. Using $R$ and the forth case of Definition~\ref{def:refinement} of refinement, it follows that $q_1^T = q_2^T$ and $((q_2^S,q_2^X), q_2^T) \in R$. Now, using Definition~\ref{def:quotientTIOTS} of the quotient, it follows that $(q_1^T,q_1^S)\xlongrightarrow{o!}{}^{\!\! T\quotient S}  (q_2^T,q_2^S)$. And from the construction of $R'$ we confirm that $(q_2^X,(q_2^T,q_2^S)) \in R'$.}
		\end{itemize}
		\new{So, in all feasible cases we can show that $(q_1^T,q_1^S)\xlongrightarrow{o!}{}^{\!\! T\quotient S}  (q_2^T,q_2^S)$ and $(q_2^X,(q_2^T,q_2^S)) \in R'$.}
		
		\item \new{$q_1^X\xlongrightarrow{o!}{}^{\!\! X} q_2^X$ for some $q_2^X\in Q^X$ and $o!\in\mathit{Act}_o^X \setminus \mathit{Act}_o^{T\quotient S}$. By definition of $X$ it follows that $\mathit{Act}_o^X \setminus \mathit{Act}_o^{T\quotient S} = \emptyset$, so this case can be ignored. 
		}
		
		\item \new{$q_1^X\xlongrightarrow{d}{}^{\!\! X} q_2^X$ for some $q_2^X\in Q^X$ and $d\in \mathbb{R}_{\geq 0}$. Consider two cases in $S$.
		}
		\begin{itemize}
			\item \new{$q_1^S\xlongrightarrow{d}{}^{\!\! S}$. In this case, there exists some $q_2^S\in Q^S$ such that $q_1^S\xlongrightarrow{d}{}^{\!\! S} q_2^S$. Now, from Definition~\ref{def:parallelcompositionTIOTS} of parallel composition it follows that there is a transition $(q_1^S,q_1^X)\xlongrightarrow{d}{}^{\!\! S\parallel X}  (q_2^S, q_2^X)$. Using $R$ and the fifth case of Definition~\ref{def:refinement} of refinement, it follows that $q_1^T\xlongrightarrow{d}{}^{\!\! T} q_2^T$ and $((q_2^S,q_2^X), q_2^T) \in R$. Now, using Definition~\ref{def:quotientTIOTS} of the quotient, it follows that $(q_1^T,q_1^S)\xlongrightarrow{d}{}^{\!\! T\quotient S}  (q_2^T,q_2^S)$. And from the construction of $R'$ we confirm that $(q_2^X,(q_2^T,q_2^S)) \in R'$.} 
			\item \new{$q_1^S\arrownot\xlongrightarrow{d}{}^{\!\! S}$. In this case, it follows from Definition~\ref{def:parallelcompositionTIOTS} of parallel composition that there is no transition in $S\parallel X$, i.e., $(q_1^S,q_1^X)\arrownot\xlongrightarrow{d}{}^{\!\! S\parallel X} $. Furthermore, from Definition~\ref{def:quotientTIOTS} it follows that $(q_1^T,q_1^S)\xlongrightarrow{d}{}^{\!\! T\quotient S} u$. And from the construction of $R'$ we confirm that $(q_2^X, u) \in R'$.}
		\end{itemize}
		\new{So, in both cases we can show that $(q_1^T,q_1^S)\xlongrightarrow{d}{}^{\!\! T\quotient S} (q_2^X,q^{T\quotient S})$ and $(q_2^X,q^{T\quotient S}) \in R'$ with $q^{T\quotient S} = (q_2^T,q_2^S)$ or $q^{T\quotient S} = u$.}
	\end{enumerate}
	\new{So for all state pairs $(q_1^X, (q_1^T,q_1^S)) \in R'$ we have shown that $R'$ witnesses the refinement $X \leq T\quotient S$. Now consider the five cases of Definition~\ref{def:refinement} for a state pair $(q_1^X, u) \in R'$.
	}
	
	\begin{enumerate}
		\item \new{$u\xlongrightarrow{i?}{}^{\!\! T\quotient S} u$ for some $i?\in\mathit{Act}_i^{T\quotient S} \cap \mathit{Act}_i^X$. By definition of $X$ it follows that $\mathit{Act}_i^{T\quotient S} = \mathit{Act}_i^X$. Since $X$ is an implementation and  $i? \in \mathit{Act}_i^X$, it follows that $q_1^X\xlongrightarrow{i?}{}^{\!\! X} q_2^X$ for some $q_2^X \in Q^X$ (any implementation is a specification, see Definition~\ref{def:implementation}, which is input-enabled, see Definition~\ref{def:specification}). By construction of $R'$ it follows that $(q_2^X,u)\in R'$.
		}
		
		\item \new{$ u\xlongrightarrow{i?}{}^{\!\! T\quotient S}  u$ for some $i?\in\mathit{Act}_i^{T\quotient S} \setminus \mathit{Act}_i^X$. By definition of $X$ it follows that $\mathit{Act}_i^{T\quotient S} \setminus \mathit{Act}_i^X = \emptyset$, so this case can be ignored.
		}
		
		\item \new{$q_1^X\xlongrightarrow{o!}{}^{\!\! X} q_2^X$ for some $q_2^X\in Q^X$ and $o!\in\mathit{Act}_o^X \cap \mathit{Act}_o^{T\quotient S}$. By definition of $X$ it follows that $\mathit{Act}_o^X = \mathit{Act}_o^{T\quotient S}$. From Definition~\ref{def:quotientTIOTS} of the quotient it follows that $u \xlongrightarrow{o!}{}^{\!\! T\quotient S}  u$. By construction of $R'$ it also follows that $(q_2^X,u)\in R'$.
		}
		
		\item \new{$q_1^X\xlongrightarrow{o!}{}^{\!\! X} q_2^X$ for some $q_2^X\in Q^X$ and $o!\in\mathit{Act}_o^X \setminus \mathit{Act}_o^{T\quotient S}$. By definition of $X$ it follows that $\mathit{Act}_o^X \setminus \mathit{Act}_o^{T\quotient S} = \emptyset$, so this case can be ignored. 
		}
		
		\item \new{$q_1^X\xlongrightarrow{d}{}^{\!\! X} q_2^X$ for some $q_2^X\in Q^X$ and $d\in \mathbb{R}_{\geq 0}$. From Definition~\ref{def:quotientTIOTS} of the quotient it follows that $u \xlongrightarrow{d}{}^{\!\! T\quotient S}  u$. By construction of $R'$ it also follows that $(q_2^X,u)\in R'$. }
	\end{enumerate}
	\new{So for all state pairs $(q_1^X,u) \in R'$ we have shown that $R'$ witnesses the refinement $X \leq T\quotient S$. Finally, since $R$ witnesses $S\parallel X \leq T$ it holds that $((q_0^S, q_o^X), q_0^T) \in R$ (see Definition~\ref{def:refinement}). Thus by construction of $R'$ it holds that $(q_0^X, (q_0^T,q_0^S)) \in R'$. Therefore, we can now conclude that $R'$ witnesses $X \leq T\quotient X$.
	}
	
	\new{($S \parallel X \leq T \Leftarrow X \leq T \quotient S$) Since $X \leq T \quotient S$, it follows from Definition~\ref{def:refinement} of refinement that there exists a relation $R \in Q^X \times Q^{T \quotient S}$ that witness the refinement. Note that $Q^{S \parallel X} = Q^S \times Q^X$ according to Definition~\ref{def:parallelcompositionTIOTS}. Construct relation $R' = \{((q_1^S,q_1^X), q_1^T)\in Q^X \times Q^{T \quotient S} \mid (q_1^X,(q_1^T,q_1^S)) \in R\}$. We will show that $R'$ witnesses $S\parallel X\leq T$. First consider the five cases of Definition~\ref{def:refinement} for a state pair $((q_1^S,q_1^X), q_1^T) \in R'$.
	}
	\begin{enumerate}
		\item \new{$q_1^T\xlongrightarrow{i?}{}^{\!\! T} q_2^T$ for some $q_2^T\in Q^T$ and $i?\in\mathit{Act}_i^T \cap \mathit{Act}_i^{S \parallel X}$. From Lemma~\ref{lemma:quotientparallelcompositionalphabet} it follows that $\mathit{Act}_i^T = \mathit{Act}_i^{S \parallel X}$. Consider the following five possible cases from Definition~\ref{def:parallelcompositionTIOTS} of the parallel composition that might result in $i?\in\mathit{Act}_i^{S \parallel X} (= \mathit{Act}_i^S\setminus \mathit{Act}_o^X \cup \mathit{Act}_i^X\setminus\mathit{Act}_o^S)$.
		}
		
		\begin{itemize}
			\item \new{$i?\in \mathit{Act}_i^S\setminus \mathit{Act}_o^X$ and $i?\in \mathit{Act}_i^X\setminus \mathit{Act}_o^S$. Since $S$ and $X$ are specifications and $i? \in \mathit{Act}_i^S \cap\mathit{Act}_i^X$, it follows that $q_1^S\xlongrightarrow{i?}{}^{\!\! S} q_2^S$ for some $q_2^S \in Q^S$  and $q_1^X\xlongrightarrow{i?}{}^{\!\! X} q_2^X$ for some $q_2^X \in Q^X$ (any specification is input-enabled, see Definition~\ref{def:specification}). Therefore, using Definition~\ref{def:parallelcompositionTIOTS} of parallel composition, it follows that $(q_1^S,q_1^X)\xlongrightarrow{i?}{}^{\!\! S \parallel X} (q_2^S,q_2^X)$. Also, using Definition~\ref{def:quotientTIOTS} of the quotient it follows that $(q_1^T,q_1^S)\xlongrightarrow{i?}{}^{\!\! T\quotient S} (q_2^T,q_2^S)$. Now, using $R$, the first case of Definition~\ref{def:refinement} of refinement, and $\mathit{Act}^X = \mathit{Act}^{T\quotient S}$ by construction, it follows that $(q_2^X, (q_2^T, q_2^S)) \in R$. And from the construction of $R'$ we confirm that $((q_2^S,q_2^X),q_2^T) \in R'$.
			}
			
			\item \new{$i?\in \mathit{Act}_i^S\setminus \mathit{Act}_o^X$ and $i?\in \mathit{Act}_i^X\cap \mathit{Act}_o^S$. This case is infeasible, as an action cannot be both an output and input in $S$.
			}
			
			\item \new{$i?\in \mathit{Act}_i^S\setminus \mathit{Act}_o^X$ and $i?\notin \mathit{Act}_i^X$. This case is infeasible, as $i?\in \mathit{Act}_i^S\setminus \mathit{Act}_o^X$ and $i?\notin \mathit{Act}_i^X$ implies that $i\notin\mathit{Act}^X$, but from Definition~\ref{def:quotientTIOTS} of the quotient it follows that $i?\in\mathit{Act}_i^S$ implies that $i\in\mathit{Act}^{T\quotient S} (= \mathit{Act}^X)$.
			}
			
			\item \new{$i?\in \mathit{Act}_i^S\cap \mathit{Act}_o^X$ and $i?\in \mathit{Act}_i^X\setminus \mathit{Act}_o^S$. This case is infeasible, as an action cannot be both an output and input in $X$.
			}
			\item \new{$i?\notin \mathit{Act}_i^S$ and $i?\in \mathit{Act}_i^X\setminus \mathit{Act}_o^S$. Since $i?\in \mathit{Act}_i^X\setminus \mathit{Act}_o^S$ implies that $i!\notin \mathit{Act}_o^S$, it follows that $i\notin\mathit{Act}^S$. From Definition~\ref{def:quotientTIOTS} of quotient if follows that $(q_1^T,q_1^S)\xlongrightarrow{i?}{}^{\!\! T\quotient S} (q_2^T,q_2^S)$ and $q_1^S = q_2^S$. Now, using $R$, the first case of Definition~\ref{def:refinement} of refinement, and $\mathit{Act}^X = \mathit{Act}^{T\quotient S}$ by construction, it follows that $q_1^X\xlongrightarrow{i?}{}^{\!\! X} q_2^X$ and $(q_2^X, (q_2^T, q_2^S)) \in R$. Using Definition~\ref{def:parallelcompositionTIOTS} of the parallel composition, it follows that $(q_1^S,q_1^X)\xlongrightarrow{i?}{}^{\!\! S\parallel X} (q_2^S,q_2^X)$. And from the construction of $R'$ we confirm that $((q_2^S,q_2^X),q_2^T) \in R'$.}
		\end{itemize}
		\new{So, in all feasible cases we can show that $(q_1^S,q_1^X)\xlongrightarrow{i?}{}^{\!\! S\parallel X} (q_2^S,q_2^X)$  and $((q_2^S,q_2^X),q_2^T) \in R'$.}
		
		\item \new{$q_1^T\xlongrightarrow{i?}{}^{\!\! T} q_2^T$ for some $q_2^T\in Q^T$ and $i?\in\mathit{Act}_i^T \setminus \mathit{Act}_i^{S \parallel X}$. From Lemma~\ref{lemma:quotientparallelcompositionalphabet} it follows that $\mathit{Act}_i^T \setminus \mathit{Act}_i^{S \parallel X} = \emptyset$, so this case can be ignored.
		}
		
		\item \new{$(q_1^S, q_1^X)\xlongrightarrow{o!}{}^{\!\! S \parallel X} (q_2^S, q_2^X)$ for some $(q_2^S, q_2^X)\in Q^{S \parallel X}$ and $o!\in\mathit{Act}_o^{S \parallel X} \cap \mathit{Act}_o^T$. From Lemma~\ref{lemma:quotientparallelcompositionalphabet} we have that $\mathit{Act}_o^{S\parallel X} = \mathit{Act}_o^S \cup \mathit{Act}_o^T \cup \mathit{Act}_i^S\setminus \mathit{Act}_i^T$. Consider the following three cases that might result in $o!\in\mathit{Act}_o^{S \parallel X}$ and $o!\in \mathit{Act}_o^T$.
		}
		\begin{itemize}
			\item \new{$o!\in\mathit{Act}_o^S$ and $o!\in\mathit{Act}_o^T$. In this case we have that $o?\in\mathit{Act}_i^{T\quotient S}$ by Definition~\ref{def:quotientTIOTS}, and thus by construction of $X$ that $o?\in\mathit{Act}_i^X$. Now, using Definition~\ref{def:parallelcompositionTIOTS} of the parallel composition, it follows that $q_1^S\xlongrightarrow{o!}{}^{\!\! S} q_2^S$ and $q_1^X\xlongrightarrow{o?}{}^{\!\! X} q_2^X$. Consider the following two cases for $T$.
			}
			\begin{itemize}
				\item \new{$q_1^T\xlongrightarrow{o!}{}^{\!\! T} q_2^T$. In this case it follows that from Definition~\ref{def:quotientTIOTS} of the quotient that $(q_1^T,q_1^S)\ {\xlongrightarrow{o?}{}^{\!\! T\quotient S}}\allowbreak (q_2^T,q_2^S)$. Using $R$, the first case of Definition~\ref{def:refinement} of refinement, and $\mathit{Act}^X = \mathit{Act}^{T\quotient S}$ by construction, it follows that $(q_2^X, (q_2^T, q_2^S)) \in R$. And from the construction of $R'$ we confirm that $((q_2^S,q_2^X),q_2^T) \in R'$.}
				\item \new{$q_1^T\arrownot\xlongrightarrow{o!}{}^{\!\! T}$. In this case it follows from Definition~\ref{def:quotientTIOTS} of the quotient that $(q_1^T,q_1^S)\xlongrightarrow{o?}{}^{\!\! T\quotient S} e$. By construction of $e$, it does not allow independent progress. But, since $X$ is an implementation, all states in $X$ allow independent progress, see Definition~\ref{def:implementation}\footnote{\new{This is the reason why $X$ is assumed to be an implementation and not just a specification.}}. Therefore, either $X$ can delay indefinitely from state $q_2^X$ or there exists a delay after which $X$ can perform an output action. Neither of these options can be simulated by $T\quotient S$ when in state $e$. Thus $(q_2^X, e) \notin R$, i.e., $X \nleq T\quotient S$. This contradicts with the assumption, thus this is not a feasible case.}
			\end{itemize}
			
			\item \new{$o?\in\mathit{Act}_i^S$ and $o!\in\mathit{Act}_o^T$. In this case we have that $o!\in\mathit{Act}_o^{T\quotient S}$ by Definition~\ref{def:quotientTIOTS}, and thus by construction of $X$ that $o!\in\mathit{Act}_o^X$. Now, using Definition~\ref{def:parallelcompositionTIOTS} of the parallel composition, it follows that $q_1^S\xlongrightarrow{o?}{}^{\!\! S} q_2^S$ and $q_1^X\xlongrightarrow{o!}{}^{\!\! X} q_2^X$. Using $R$, the third case of Definition~\ref{def:refinement} of refinement, and $\mathit{Act}^X = \mathit{Act}^{T\quotient S}$ by construction, it follows that $(q_1^T,q_1^S)\xlongrightarrow{o!}{}^{\!\! T\quotient S} (q_2^T,q_2^S)$ and $(q_2^X, (q_2^T, q_2^S)) \in R$. Now, using Definition~\ref{def:quotientTIOTS} of quotient again, it follows that $q_1^T\xlongrightarrow{o!}{}^{\!\! T} q_2^T$. And from the construction of $R'$ we confirm that $((q_2^S,q_2^X),q_2^T) \in R'$.
			}
			
			\item \new{$o\notin\mathit{Act}^S$ and $o!\in\mathit{Act}_o^T$. In this case we have that $o!\in\mathit{Act}_o^{T\quotient S}$ by Definition~\ref{def:quotientTIOTS}, and thus by construction of $X$ that $o!\in\mathit{Act}_o^X$. Now, using Definition~\ref{def:parallelcompositionTIOTS} of the parallel composition, it follows that $q_1^X\xlongrightarrow{o!}{}^{\!\! X} q_2^X$ and $q_1^S = q_2^S$. Using $R$, the third case of Definition~\ref{def:refinement} of refinement, and $\mathit{Act}^X = \mathit{Act}^{T\quotient S}$ by construction, it follows that $(q_1^T,q_1^S)\xlongrightarrow{o!}{}^{\!\! T\quotient S} (q_2^T,q_2^S)$ and $(q_2^X, (q_2^T, q_2^S)) \in R$. Now, using Definition~\ref{def:quotientTIOTS} of quotient again, it follows that $q_1^T\xlongrightarrow{o!}{}^{\!\! T} q_2^T$. And from the construction of $R'$ we confirm that $((q_2^S,q_2^X),q_2^T) \in R'$.}
		\end{itemize}
		\new{So, in all feasible cases we can show that $q_1^T\xlongrightarrow{o!}{}^{\!\! T} q_2^T$  and $((q_2^S,q_2^X),q_2^T) \in R'$.}
		
		\item \new{$(q_1^S, q_1^X)\xlongrightarrow{o!}{}^{\!\! S \parallel X} (q_2^S, q_2^X)$ for some $(q_2^S, q_2^X)\in Q^{S \parallel X}$ and $o!\in\mathit{Act}_o^{S \parallel X} \setminus \mathit{Act}_o^T$. From Lemma~\ref{lemma:quotientparallelcompositionalphabet} we have that $\mathit{Act}_o^{S\parallel X} = \mathit{Act}_o^S \cup \mathit{Act}_o^T \cup \mathit{Act}_i^S\setminus \mathit{Act}_i^T$. So $\mathit{Act}_o^{S \parallel X} \setminus \mathit{Act}_o^T = (\mathit{Act}_o^S \cup \mathit{Act}_i^S\setminus \mathit{Act}_i^T) \setminus \mathit{Act}_o^T = \mathit{Act}_o^S\setminus\mathit{Act}_o^T \cup (\mathit{Act}_i^S\setminus \mathit{Act}_i^T)\setminus\mathit{Act}_o^T = \mathit{Act}_o^S\setminus\mathit{Act}_o^T \cup \mathit{Act}_i^S \setminus\mathit{Act}^T$. Consider the following five cases that might result in $o!\in \mathit{Act}_o^{S \parallel X} \setminus \mathit{Act}_o^T$. 
		}
		\begin{itemize}
			\item \new{$o!\in\mathit{Act}_o^S\setminus\mathit{Act}_o^T$ and $o?\in\mathit{Act}_i^S\setminus\mathit{Act}^T$. This case is infeasible, as an action cannot be both an output and input in $S$.
			}
			
			\item \new{$o!\in\mathit{Act}_o^S\setminus\mathit{Act}_o^T$ and $o?\in\mathit{Act}_i^S\cap\mathit{Act}^T$. This case is infeasible, as an action cannot be both an output and input in $S$.
			}
			
			\item \new{$o!\in\mathit{Act}_o^S\setminus\mathit{Act}_o^T$ and $o?\notin\mathit{Act}_i^S$. In this case, we have that $o?\in \mathit{Act}_i^{T\quotient S}$ from Definition~\ref{def:quotientTIOTS} of the quotient. Therefore,  $o?\in \mathit{Act}_i^X$ by construction of $X$. Now, using Definition~\ref{def:parallelcompositionTIOTS} of the parallel composition, it follows that $q_1^S\xlongrightarrow{o!}{}^{\!\! S} q_2^S$ and $q_1^X\xlongrightarrow{o?}{}^{\!\! X} q_2^X$. Since Definition~\ref{def:quotientTIOTS} also requires that $\mathit{Act}_o^S\cap\mathit{Act}_i^T=\emptyset$, it follows that in this case $o\notin\mathit{Act}^T$. Thus, from Definition~\ref{def:quotientTIOTS} it follows that $(q_1^T,q_1^S)\xlongrightarrow{o?}{}^{\!\! T\quotient S} (q_2^T,q_2^S)$ and $q_1^T = q_2^T$. Using $R$, the first case of Definition~\ref{def:refinement} of refinement, and $\mathit{Act}^X = \mathit{Act}^{T\quotient S}$ by construction, it follows that $(q_2^X, (q_2^T, q_2^S)) \in R$. And from the construction of $R'$ we confirm that $((q_2^S,q_2^X),q_2^T) \in R'$.
			}
			
			\item \new{$o!\in\mathit{Act}_o^S\cap\mathit{Act}_o^T$ and $o?\in\mathit{Act}_i^S\setminus\mathit{Act}^T$. This case is infeasible, as an action cannot be both an output and input in $S$.
			}
			
			\item \new{$o!\notin\mathit{Act}_o^S$ and $o?\in\mathit{Act}_i^S\setminus\mathit{Act}^T$. In this case, we have that $o!\in \mathit{Act}_o^{T\quotient S}$ from Definition~\ref{def:quotientTIOTS} of the quotient. Therefore, $o!\in \mathit{Act}_o^X$ by construction of $X$. Now, using Definition~\ref{def:parallelcompositionTIOTS} of the parallel composition, it follows that $q_1^S\xlongrightarrow{o!}{}^{\!\! S} q_2^S$ and $q_1^X\xlongrightarrow{o?}{}^{\!\! X} q_2^X$. Using $R$, the fourth case of Definition~\ref{def:refinement} of refinement, $\mathit{Act}^X = \mathit{Act}^{T\quotient S}$ by construction, and $o\notin\mathit{Act}^T$, it follows that $(q_2^X, (q_2^T, q_2^S)) \in R$ and $q_1^T = q_2^T$. And from the construction of $R'$ we confirm that $((q_2^S,q_2^X),q_2^T) \in R'$.}
		\end{itemize}
		\new{So, in all feasible cases we can show that $o\notin\mathit{Act}^T$, $q_1^T = q_2^T$, and $((q_2^S,q_2^X),q_2^T) \in R'$.}
		
		\item \new{$(q_1^S, q_1^X)\xlongrightarrow{d}{}^{\!\! S \parallel X} (q_2^S, q_2^X)$ for some $(q_2^S, q_2^X)\in Q^{S \parallel X}$ and $d\in \mathbb{R}_{\geq 0}$. It follows from Definition~\ref{def:parallelcompositionTIOTS} of the parallel composition that $q_1^S\xlongrightarrow{d}{}^{\!\! S} q_2^S$ and $q_1^X\xlongrightarrow{d}{}^{\!\! X} q_2^X$. Using $R$ and the fifth case of Definition~\ref{def:refinement} of refinement it follows that $(q_1^T,q_1^S)\xlongrightarrow{d}{}^{\!\! T\quotient S} q_2$ for some $q_2^{T\quotient S}\in Q^{T \quotient S}$ and $(q_2^X, q_2^{T\quotient S})) \in R$. Now, by Definition~\ref{def:quotientTIOTS} of the quotient it follows that $q_1^T\xlongrightarrow{d}{}^{\!\! T} q_2^T$ (and $q_1^S\xlongrightarrow{d}{}^{\!\! S} q_2^S$). And from the construction of $R'$ we confirm that $((q_2^S,q_2^X),q_2^T) \in R'$.}
	\end{enumerate}
	\new{So for all state pairs $((q_1^S,q_1^X),q_1^T) \in R'$ we have shown that $R'$ witnesses the refinement $S\parallel X \leq T$. Finally, since $R$ witnesses $X \leq T\quotient S$ it holds that $(q_0^X, (q_0^T,q_0^S)) \in R$ (see Definition~\ref{def:refinement}). Thus by construction of $R'$ it holds that $((q_0^S, q_o^X), q_0^T) \in R'$. Therefore, we can now conclude that $R'$ witnesses $S\parallel X \leq T$.}
\end{proof}

\begin{definition}\label{def:equivalencequotientTIOTS}
	\new{Given a TIOTS $S = (Q,q_0,\mathit{Act},\rightarrow)$ and equivalence relation $\sim$ on the set of states $Q$. The \emph{$\sim$-quotient $S$}, denoted by $S/{\sim}$, is a specification $([Q]_{\sim}, [q_0]_{\sim}, \mathit{Act}, {\rightarrow}/{\sim})$ where $[Q]_{\sim}$ is the set of all equivalence classes of $Q$\footnote{\new{Recall that an equivalent class is defined as $[q]_{\sim} = \{r \in Q \mid q \sim r\}$.}} and ${\rightarrow}/{\sim}$ being defined as $([q_1], a, [q_2]) \in {\rightarrow}/{\sim}$ if $(q_1, a, q_2) \in\rightarrow$ for some $q_1\in[q_1]$ and $q_2\in[q_2]$.}
\end{definition}

\begin{lemma}\label{lemma:quotientequivalencequotient}
	\new{Given specification automata $S = (\mathit{Loc}^S,l_0^S,\mathit{Act}^S,\mathit{Clk}^S, E^S,\mathit{Inv}^S)$ and $T = (\mathit{Loc}^T,l_0^T,\allowbreak\mathit{Act}^T,\mathit{Clk}^T, E^T,\allowbreak\mathit{Inv}^T)$ where $\mathit{Act}_o^S\cap\mathit{Act}_i^T=\emptyset$. Let $V_0 = \{v \in [\mathit{Clk}^{T\quotient S} \mapsto \mathbb{R}_{\geq 0}] \mid v(x_{\mathit{new}}) = 0 \}$, $V_{>0} = [\mathit{Clk}^{T\quotient S} \mapsto \mathbb{R}_{\geq 0}] \setminus V_0$, and ${\sim} = \{ (q_1, q_2) \mid q_1, q_2 \in \{l_e\} \times V_0\} \cup \{ (q, q) \mid q \in \{l_e\} \times V_{>0}\} \cup \{ (q_1, q_2) \mid q_1, q_2 \in \{l_u\} \times [\mathit{Clk}^{T \quotient S} \mapsto \mathbb{R}_{\geq 0}] \} \cup \{ ((l, v_1), (l, v_2)) \mid l \in \mathit{Loc}^{T\quotient S}\setminus \{l_e, l_u\}, v_1, v_2 \in [\mathit{Clk}^{T\quotient S}\mapsto \mathbb{R}_{\geq 0}], \forall c\in\mathit{Clk}^{T\quotient S}\setminus \{x_{\mathit{new}}\}, v_1(c) = v_2(c) \}$. Then $\sem{T\quotient S} \simeq \sem{T\quotient S}/{\sim}$.}
\end{lemma}
\begin{proof}
	\new{It follows directly from the definition of $\sim$ that it is reflexive, symmetric, and transitive, thus it is an equivalence relation. Now, observe from Definition~\ref{def:equivalencequotientTIOTS} that an equivalence quotient of a TIOTS does not alter the action set, i.e., $\mathit{Act}^{\sem{T\quotient S}} = \mathit{Act}^{\sem{T\quotient S}/{\sim}}$. Let $R = \{(q, [q]_{\sim}) \mid q \in Q^{\sem{T\quotient S}} \}$. We will show that $R$ is a bisimulation relation. First, observe that $(q_0, [q_0]_{\sim}) \in R$. Consider a state pair $(q_1, [r_1]_{\sim}) \in R$. We have to check whether the six cases from Definition~\ref{def:bisimulation} of bisimulation hold.
	}
	
	\begin{enumerate}
		\item \new{$q_1 \xlongrightarrow{a}{}^{\!\! \sem{T\quotient S}} q_2$, $q_2\in Q^{\sem{T\quotient S}}$, and $a \in \mathit{Act}^{\sem{T\quotient S}}\cap\mathit{Act}^{\sem{T\quotient S}/{\sim}}$. By the definition of an equivalence class and Definition~\ref{def:equivalencequotientTIOTS} it follows immediately that $[q_1]_{\sim} \xlongrightarrow{a}{}^{\!\! \sem{T\quotient S}/{\sim}} [q_2]_{\sim}$. By construction of $R$ it follows that $(q_2, [q_2]_{\sim}) \in R$.
		}
		
		\item \new{$q_1 \xlongrightarrow{a}{}^{\!\! \sem{T\quotient S}} q_2$, $q_2\in Q^{\sem{T\quotient S}}$, and $a \in \mathit{Act}^{\sem{T\quotient S}}\setminus\mathit{Act}^{\sem{T\quotient S}/{\sim}}$. This case is infeasible, since $\mathit{Act}^{\sem{T\quotient S}} = \mathit{Act}^{\sem{T\quotient S}/{\sim}}$. 
		}
		
		\item \new{$[r_1]_{\sim} \xlongrightarrow{a}{}^{\!\! \sem{T\quotient S}/{\sim}} [r_2]_{\sim}$, $[r_2]_{\sim}\in Q^{\sem{T\quotient S}/{\sim}}$, and $a \in \mathit{Act}^{\sem{T\quotient S}/{\sim}}\cap \mathit{Act}^{\sem{T\quotient S}}$. By construction of $R$, we have to show that $\forall q_1\in[r_1]_{\sim} \exists q_2\in Q^{\sem{T\quotient S}}: q_1 \xlongrightarrow{a}{}^{\!\! \sem{T\quotient S}} q_2$, $q_2\in [r_2]_{\sim}$, and $(q_2, [r_2]_{\sim})\in R$. Consider the following four cases based on the construction of $\sim$:
		}
		\begin{itemize}
			\item \new{$[r_1]_{\sim} = \{ q \mid q \in \{l_e\} \times V_0\}$. In this case, let $q_1 = (l_e, v_1) \in [r_1]_{\sim}$ for some $v_1 \in  V_0$. From Definition~\ref{def:semanticTIOA} of the semantic of a TIOA it follows that $\sem{T\quotient S}$ is in location $l_e$. From Definition~\ref{def:quotientTIOA} of the quotient it follows that the only possible transition in $T\quotient S$ is $(l_e, a, x_{\mathit{new}} = 0, \emptyset, l_e)$. Furthermore, since $[r_1]_{\sim} \xlongrightarrow{a}{}^{\!\! \sem{T\quotient S}/{\sim}} [r_2]_{\sim}$, it holds that $\exists r_1, r_2 \in Q^{\sem{T\quotient S}} : r_1 \xlongrightarrow{a}{}^{\!\! \sem{T\quotient S}} r_2$. Following Definition~\ref{def:semanticTIOA} and the above observation, it holds that $r_1 = (l_e, v_1')$ and $r_2 = (l_e, v_2')$ for some $v_1', v_2'\in [\mathit{Clk}^{T\quotient S}\mapsto \mathbb{R}_{\geq 0}]$, $v_1'\models x_{\mathit{new}} = 0$, and $v_1' = v_2'$. From  $v_1'\models x_{\mathit{new}} = 0$ it follows that $v_1'(x_{\mathit{new}}) = 0$ and $v_1', v_2'\in V_0$, and from $v_1' = v_2'$ that $[r_2]_{\sim} = [r_1]_{\sim}$. Thus we can conclude that $q_1\xlongrightarrow{a}{}^{\!\! \sem{T\quotient S}} q_2$ with $q_2 \in [r_2]_{\sim}$. By construction of $R$ it follows that $(q_2, [r_2]_{\sim}) \in R$.
			}
			
			\item \new{$[r_1]_{\sim} = \{ q \mid q \in \{l_e\} \times V_{>0}\}$. This case is trivial, since $[r_1]_{\sim} = \{r_1\} = \{q_1\}$. Therefore, if $[r_1]_{\sim} \allowbreak {\xlongrightarrow{a}{}^{\!\! \sem{T\quotient S}/{\sim}}} [r_2]_{\sim}$, $\exists q_2 \in [r_2]_{\sim}$ such that $q_1 \xlongrightarrow{a}{}^{\!\! \sem{T\quotient S}} q_2$.
			}
			
			\item \new{$[r_1]_{\sim} = \{ q \mid q \in \{l_u\} \times [\mathit{Clk}^{T\quotient S} \mapsto \mathbb{R}_{\geq 0}]\}$. In this case, let $q_1 = (l_u, v_1) \in [r_1]_{\sim}$ for some $v_1 \in  [\mathit{Clk}^{T\quotient S} \mapsto \mathbb{R}_{\geq 0}]$. From Definition~\ref{def:semanticTIOA} of the semantic of a TIOA it follows that $\sem{T\quotient S}$ is in location $l_u$. From Definition~\ref{def:quotientTIOA} of the quotient it follows that the only possible transition in $T\quotient S$ is $(l_u, a, \mathbf{T}, \emptyset, l_u)$. Furthermore, since $[r_1]_{\sim} \xlongrightarrow{a}{}^{\!\! \sem{T\quotient S}/{\sim}} [r_2]_{\sim}$, it holds that $\exists r_1, r_2 \in Q^{\sem{T\quotient S}} : r_1 \xlongrightarrow{a}{}^{\!\! \sem{T\quotient S}} r_2$. Following Definition~\ref{def:semanticTIOA} and the above observation, it holds that $r_1 = (l_u, v_1')$ and $r_2 = (l_u, v_2')$ for some $v_1', v_2'\in [\mathit{Clk}^{T\quotient S}\mapsto \mathbb{R}_{\geq 0}]$, $v_1'\models \mathbf{T}$, and $v_1' = v_2'$. From $v_1' = v_2'$ it follows that $[r_2]_{\sim} = [r_1]_{\sim}$. Thus we can conclude that $q_1\xlongrightarrow{a}{}^{\!\! \sem{T\quotient S}} q_2$ with $q_2 \in [r_2]_{\sim}$. By construction of $R$ it follows that $(q_2, [r_2]_{\sim}) \in R$.
			}
			
			\item \new{In this case, since $[r_1]_{\sim} \xlongrightarrow{a}{}^{\!\! \sem{T\quotient S}/{\sim}} [r_2]_{\sim}$, it holds that $\exists r_1, r_2 \in Q^{\sem{T\quotient S}} : r_1 \xlongrightarrow{a}{}^{\!\! \sem{T\quotient S}} r_2$. Following Definition~\ref{def:semanticTIOA} of the semantic of a TIOA, it holds that $(l_1, a, \varphi, c, l_2) \in E^{T\quotient S}$, $r_1 = (l_1, v_1)$, $r_2 = (l_2, v_2)$, $l_1, l_2 \in \mathit{Loc}^{T\quotient S}$, $v_1, v_2\in [\mathit{Clk}^{T\quotient S}\mapsto \mathbb{R}_{\geq 0}]$, $v_1\models \mathbf{\varphi}$, $v_2 = v_1[r\mapsto 0]_{r\in c}$, and $v_2\models \mathit{Inv}(l_2)$. From the construction of $\sim$, it follows that for any  state $(l_1', v_1') \in [r_1]_{\sim}$ it holds that $l_1' = l_1$, $l_1 \neq l_e$, and $\forall c \in \mathit{Clk}^{T \quotient S} \setminus \{x_{\mathit{new}}\} : v_1'(c) = v_1(c)$. Since $x_{\mathit{new}} \notin \mathit{Clk}^T\cup \mathit{Clk}^S$ and none of the possible rules for this location from Definition~\ref{def:quotientTIOA} of the quotient for TIOA use $x_{\mathit{new}}$ in its guard, it follows that $v_1'\models \varphi$. Furthermore, no matter whether $x_{\mathit{new}}\in c$ or not, we have for $v_2' = v_1'[r\mapsto 0]_{r\in c}$ that $\forall c \in \mathit{Clk}^{T \quotient S} \setminus \{x_{\mathit{new}}\} : v_2'(c) = v_2(c)$. Now consider the following three options for the target location $l_2$.
			}
			\begin{itemize}
				\item \new{If $l_2 = (l^T, l^S)$ with $l^T\in\mathit{Loc}^T$ and $l^S\in\mathit{Loc}^S$, then $\mathit{Inv}(l_2) = \mathbf{T}$. Thus $v_2'\models \mathit{Inv}(l_2)$.}
				\item \new{If $l_2 = l_u$, then $\mathit{Inv}(l_2) = \mathbf{T}$. Thus $v_2'\models \mathit{Inv}(l_2)$.}
				\item \new{If $l_2 = l_e$, then $\mathit{Inv}(l_2) = x_{\mathit{new}} = 0$. Also, $c = \{x_{\mathit{new}}\}$, thus $v_2(x_{\mathit{new}}) = v_2'(x_{\mathit{new}}) = 0$. Thus $v_2'\models \mathit{Inv}(l_2)$.}
			\end{itemize}
			\new{Therefore, we can conclude that $(l_1', v_1') \xlongrightarrow{a}{}^{\!\! \sem{T\quotient S}} (l_2, v_2')$, $(l_2, v_2')\in [r_2]_{\sim}$, and by construction of $R$  that $((l_2, v_2'), [r_2]_{\sim}) \in R$. Since we picked any state $(l_1', v_1') \in [r_1]_{\sim}$, the conclusion holds for all states $q_1 \in [r_1]_{\sim}$.}
		\end{itemize}
		
		\item \new{$[r_1]_{\sim} \xlongrightarrow{a}{}^{\!\! \sem{T\quotient S}/{\sim}} [r_2]_{\sim}$, $[r_2]_{\sim}\in Q^{\sem{T\quotient S}/{\sim}}$, and $a \in \mathit{Act}^{\sem{T\quotient S}/{\sim}}\setminus \mathit{Act}^{\sem{T\quotient S}}$. This case is infeasible, since $\mathit{Act}^{\sem{T\quotient S}} = \mathit{Act}^{\sem{T\quotient S}/{\sim}}$. 
		}
		
		\item \new{$q_1 \xlongrightarrow{d}{}^{\!\! \sem{T\quotient S}} q_2$, $q_2\in Q^{\sem{T\quotient S}}$, and $d \in \mathbb{R}_{\geq 0}$.  By the definition of an equivalence class and Definition~\ref{def:equivalencequotientTIOTS} it follows immediately that $[q_1]_{\sim} \xlongrightarrow{a}{}^{\!\! \sem{T\quotient S}/{\sim}} [q_2]_{\sim}$. By construction of $R$ it follows that $(q_2, [q_2]_{\sim}) \in R$.
		}
		
		\item \new{$[r_1]_{\sim} \xlongrightarrow{d}{}^{\!\! \sem{T\quotient S}/{\sim}} [r_2]_{\sim}$, $[r_2]_{\sim}\in Q^{\sem{T\quotient S}/{\sim}}$, and $d \in \mathbb{R}_{\geq 0}$. By construction of $R$, we have to show that $\forall q_1\in[r_1]_{\sim} \exists q_2\in Q^{\sem{T\quotient S}}: q_1 \xlongrightarrow{a}{}^{\!\! \sem{T\quotient S}} q_2$, $q_2\in [r_2]_{\sim}$, and $(q_2, [r_2]_{\sim})\in R$. Consider the following three cases based on the construction of $\sim$:
		}
		
		\begin{itemize}
			\item \new{$[r_1]_{\sim} = \{ q \mid q \in \{l_e\} \times V_0\}$. In this case, let $q_1 = (l_e, v_1) \in [r_1]_{\sim}$ for some $v_1 \in  V_0$. From Definition~\ref{def:semanticTIOA} of the semantic of a TIOA it follows that $\sem{T\quotient S}$ is in location $l_e$. From Definition~\ref{def:quotientTIOA} of the quotient it follows that $\mathit{Inv}(l_e) = x_{\mathit{new}} = 0$. Furthermore, since $[r_1]_{\sim} \xlongrightarrow{d}{}^{\!\! \sem{T\quotient S}/{\sim}} [r_2]_{\sim}$, it holds that $\exists r_1, r_2 \in Q^{\sem{T\quotient S}} : r_1 \xlongrightarrow{d}{}^{\!\! \sem{T\quotient S}} r_2$. Following Definition~\ref{def:semanticTIOA} and the above observation, it holds that $r_1 = (l_e, v_1')$ and $r_2 = (l_e, v_2')$ for some $v_1', v_2'\in [\mathit{Clk}^{T\quotient S}\mapsto \mathbb{R}_{\geq 0}]$, $v_2' = v_1' + d$ and $v_2'\models\mathit{Inv}(l_e)$. From  $v_2'\models\mathit{Inv}(l_e)$ it follows that $v_2'(x_{\mathit{new}}) = 0$, thus $d = 0$, $v_1' = v_2'$, $v_1', v_2'\in V_0$, and $[r_2]_{\sim} = [r_1]_{\sim}$. Thus we can conclude that $q_1\xlongrightarrow{d}{}^{\!\! \sem{T\quotient S}} q_2$ with $q_2 \in [r_2]_{\sim}$. By construction of $R$ it follows that $(q_2, [r_2]_{\sim}) \in R$.
			}
			
			\item \new{$[r_1]_{\sim} = \{ q \mid q \in \{l_e\} \times V_{>0}\}$. This case is trivial, since $[r_1]_{\sim} = \{r_1\} = \{q_1\}$. Therefore, if $[r_1]_{\sim} \xlongrightarrow{d}{}^{\!\! \sem{T\quotient S}/{\sim}} [r_2]_{\sim}$, $\exists q_2 \in [r_2]_{\sim}$ such that $q_1 \xlongrightarrow{d}{}^{\!\! \sem{T\quotient S}} q_2$.
			}
			
			\item \new{$[r_1]_{\sim} = \{ q \mid q \in \{l_u\} \times [\mathit{Clk}^{T\quotient S} \mapsto \mathbb{R}_{\geq 0}]\}$. In this case, let $q_1 = (l_u, v_1) \in [r_1]_{\sim}$ for some $v_1 \in  V_0$. From Definition~\ref{def:semanticTIOA} of the semantic of a TIOA it follows that $\sem{T\quotient S}$ is in location $l_u$. From Definition~\ref{def:quotientTIOA} of the quotient it follows that $\mathit{Inv}(l_u) = \mathbf{T}$. Furthermore, since $[r_1]_{\sim} \xlongrightarrow{d}{}^{\!\! \sem{T\quotient S}/{\sim}} [r_2]_{\sim}$, it holds that $\exists r_1, r_2 \in Q^{\sem{T\quotient S}} : r_1 \xlongrightarrow{d}{}^{\!\! \sem{T\quotient S}} r_2$. Following Definition~\ref{def:semanticTIOA} and the above observation, it holds that $r_1 = (l_u, v_1')$ and $r_2 = (l_u, v_2')$ for some $v_1', v_2'\in [\mathit{Clk}^{T\quotient S}\mapsto \mathbb{R}_{\geq 0}]$, $v_2' = v_1' + d$ and $v_2'\models\mathit{Inv}(l_u)$. Now it follows that $(l_u, v_2')\in [r_1]_{\sim}$, thus $[r_2]_{\sim} = [r_1]_{\sim}$. Therefore, we can conclude that $q_1\xlongrightarrow{d}{}^{\!\! \sem{T\quotient S}} q_2$ with $q_2 \in [r_2]_{\sim}$ and by construction of $R$ it follows that $(q_2, [r_2]_{\sim}) \in R$.
			}
			
			\item \new{In this case, since $[r_1]_{\sim} \xlongrightarrow{d}{}^{\!\! \sem{T\quotient S}/{\sim}} [r_2]_{\sim}$, it holds that $\exists r_1, r_2 \in Q^{\sem{T\quotient S}} : r_1 \xlongrightarrow{d}{}^{\!\! \sem{T\quotient S}} r_2$. Following Definition~\ref{def:semanticTIOA} of the semantic of a TIOA, it holds that $r_1 = (l, v_1)$, $r_2 = (l, v_2)$, $l \in \mathit{Loc}^{T\quotient S}$, $v_1, v_2\in [\mathit{Clk}^{T\quotient S}\mapsto \mathbb{R}_{\geq 0}]$, $v_2 = v_1 + d$, $v_2\models \mathit{Inv}(l)$, and $\forall d'\in\mathbb{R}_{\geq 0}, d' < d: v_1 + d'\models \mathit{Inv}(l)$. 
				From the construction of $\sim$, it follows that for any  state $(l_1', v_1') \in [r_1]_{\sim}$ it holds that $l_1' = l_1$, $l_1 \neq l_e$, and $\forall c \in \mathit{Clk}^{T \quotient S} \setminus \{x_{\mathit{new}}\} : v_1'(c) = v_1(c)$. Therefore, we have for $v_2' = v_1' + d$ that $\forall c \in \mathit{Clk}^{T \quotient S} \setminus \{x_{\mathit{new}}\} : v_2'(c) = v_2(c)$; similarly, for $v_1' + d'$ we have that $\forall c \in \mathit{Clk}^{T \quotient S} \setminus \{x_{\mathit{new}}\} : v_1'+d'(c) = v_1 + d'(c)$. From Definition~\ref{def:quotientTIOA} of the quotient for TIOA it follows that $\mathit{Inv}(l) = \mathit{Inv}(l') = \mathbf{T}$. Thus $v_2'\models \mathit{Inv}(l')$ and $v_1' + d'\models \mathit{Inv}(l')$. Therefore, from Definition~\ref{def:semanticTIOA} again we have that $(l_1', v_1') \xlongrightarrow{d}{}^{\!\! \sem{T\quotient S}} (l_1', v_2')$, $(l_2, v_2')\in [r_2]_{\sim}$, and by construction of $R$  that $((l_2, v_2'), [r_2]_{\sim}) \in R$. Since we picked any state $(l_1', v_1') \in [r_1]_{\sim}$, the conclusion holds for all states $q_1 \in [r_1]_{\sim}$.}
		\end{itemize}
	\end{enumerate}
\end{proof}

\new{The following definition defines the TIOTS of the $\sim$-quotient of $\sem{T\quotient S}$ where all states consisting of the error location and a valuation where $u(x_{\mathit{new}}) > 0$ are removed, as these states are never reachable.}
\begin{definition}\label{def:quotientreducedquotient}
	\new{Given specification automata $S = (\mathit{Loc}^S,l_0^S,\mathit{Act}^S,\mathit{Clk}^S, E^S,\mathit{Inv}^S)$ and $T = (\mathit{Loc}^T,l_0^T,\allowbreak\mathit{Act}^T,\mathit{Clk}^T, E^T,\allowbreak\mathit{Inv}^T)$ where $\mathit{Act}_o^S\cap\mathit{Act}_i^T=\emptyset$. Let $V_0 = \{u \in [\mathit{Clk}^{T\quotient S} \mapsto \mathbb{R}_{\geq 0}] \mid u(x_{\mathit{new}}) = 0 \}$, $V_{>0} = [\mathit{Clk}^{T\quotient S} \mapsto \mathbb{R}_{\geq 0}] \setminus V_0$, and ${\sim} = \{ (q_1, q_2) \mid q_1, q_2 \in \{l_e\} \times V_0\} \cup \{ (q, q) \mid q \in \{l_e\} \times V_{>0}\} \cup \{ (q_1, q_2) \mid q_1, q_2 \in \{l_u\} \times [\mathit{Clk}^{T \quotient S} \mapsto \mathbb{R}_{\geq 0}] \} \cup \{ ((l, v_1), (l, v_2)) \mid l \in \mathit{Loc}^{T\quotient S}\setminus \{l_e, l_u\}, v_1, v_2 \in [\mathit{Clk}^{T\quotient S}\mapsto \mathbb{R}_{\geq 0}], \forall c\in\mathit{Clk}^{T\quotient S}\setminus \{x_{\mathit{new}}\}, v_1(c) = v_2(c) \}$. The \emph{reduced $\sim$-quotient of $\sem{T\quotient S}$}, denoted by $\sem{T\quotient S}^{\rho}$, is defined as TIOTS $(Q^{\rho}, q_0^{\rho}, \mathit{Act}^{T\quotient S}, \rightarrow^{\rho})$ where $Q^{\rho} = Q^{\sem{T\quotient S}/{\sim}} \setminus \{[q] \mid q \in \{l_e\} \times V_{>0}\}$, $q_0^{\rho} = q_0^{\sem{T\quotient S}/{\sim}}$, and $\rightarrow^{\rho} = \rightarrow^{\sem{T\quotient S}/{\sim}} \cap \{(q_1, a, q_2) \mid q_1, q_2 \in Q, a \in\mathit{Act}^{T\quotient S}\}$.}
\end{definition}

\begin{lemma}\label{lemma:quotientreducedquotient}
	\new{Given specification automata $S = (\mathit{Loc}^S,l_0^S,\mathit{Act}^S,\mathit{Clk}^S, E^S,\mathit{Inv}^S)$ and $T = (\mathit{Loc}^T,l_0^T,\allowbreak\mathit{Act}^T,\mathit{Clk}^T, E^T,\allowbreak\mathit{Inv}^T)$ where $\mathit{Act}_o^S\cap\mathit{Act}_i^T=\emptyset$. Then $ \sem{T\quotient S} \simeq \sem{T\quotient S}^{\rho}$.}
\end{lemma}
\begin{proof}
	\new{Since bisimulation relation is an equivalence relation, it follows from Lemma~\ref{lemma:quotientequivalencequotient} that it suffice to show that $\sem{T\quotient S}/{\sim} \simeq \sem{T\quotient S}^{\rho}$. Let $R = \{(q, q) \mid q \in Q^{\sem{T\quotient S}^{\rho}} \}$. We will show that $R$ is a bisimulation relation. First, observe that $(q_0, q_0) \in R$ by definition of $\sem{T\quotient S}^{\rho}$. Instead of checking all six cases of bisimulation (Definition~\ref{def:bisimulation}), we will show that $q_1 \arrownot\xlongrightarrow{a}{}^{\!\! \sem{T\quotient S}/{\sim}} q_2$ for any $a\in \mathit{Act}^{T\quotient S} \cup \mathbb{R}_{\geq 0}$ where $q_1 \in Q^{\sem{T\quotient S}^{\rho}}$ and $q_2 \in \{l_e\} \times V_{>0}$ (i.e., $q_2 \notin Q^{\sem{T\quotient S}^{\rho}}$). Only rules 5, 7, and 11 of Definition~\ref{def:quotientTIOA} of the quotient for TIOA have target location $l_e$, and thus could become $q_2$ in the semantic of it. But notice that all three cases have clock reset $c = \{x_{\mathit{new}}\}$. Therefore, any state $(l_e, u)$ reached after taking a transition matching one of these three rules has a valuation $u(x_{\mathit{new}}) = 0$. Thus $(l_e, u) \notin \{l_e\} \times V_{>0}$ and $q_1 \arrownot\xlongrightarrow{a}{}^{\!\! \sem{T\quotient S}/{\sim}} q_2$. Therefore, all reachable state pairs by bisimulation remains within $R$.}
\end{proof}

\begin{lemma}~\label{lemma:quotientTSandAsamestateset} 
	\new{Given specification automata $S = (\mathit{Loc}^S,l_0^S,\mathit{Act}^S,\mathit{Clk}^S, E^S,\mathit{Inv}^S)$ and $T = (\mathit{Loc}^T,l_0^T,\mathit{Act}^T,\mathit{Clk}^T, E^T,\allowbreak\mathit{Inv}^T)$ where $\mathit{Act}_o^S\cap\mathit{Act}_i^T=\emptyset$.  Let $f: Q^{\sem{T\quotient S}^{\rho}} \rightarrow Q^{\sem{T} \quotient \sem{S}}$ be defined as }
	\begin{itemize}
		\item \new{$f([((l^T,l^S),v)]_{\sim}) = ((l^T,v^T), (l^S,v^S))$ for any $v\in (\mathit{Clk}^{\sem{T\quotient S}}\times \mathbb{R}_{\geq 0})$, $l^T\in \mathit{Loc}^T$, $v^T\in (\mathit{Clk}^T\times \mathbb{R}_{\geq 0})$, $l^S\in \mathit{Loc}^S$, and $v^S\in (\mathit{Clk}^S\times \mathbb{R}_{\geq 0})$ such that $\forall x \in\mathit{Clk}^T: v(x) = v^T(x)$ and $\forall x \in\mathit{Clk}^S: v(x) = v^S(x)$.
		}
		\item \new{$f([(l_u,v)]_{\sim}) = u$ for any $v\in (\mathit{Clk}^{\sem{T\quotient S}}\times \mathbb{R}_{\geq 0})$.}
		\item \new{$f([(l_e,v)]_{\sim}) = e$ for any $v\in V_0$.}
	\end{itemize}
	\new{Then $f$ is a bijective function.}
\end{lemma} 
\begin{proof}
	\new{It follows directly from the definition that $f$ is injective. We only have to show that $f$ is surjective, where the last two cases are again trivial by definition of $f$. Thus we only have to show that any state $((l^T,v^T), (l^S,v^S))$ maps to only a single state $[((l^T,l^S),v)]_{\sim}$ in $\sem{T\quotient S}^{\rho}$. For this, note that $\sim$ in Definition~\ref{def:quotientreducedquotient} contains $ \{ ((l, v_1), (l, v_2)) \mid l \in \mathit{Loc}^{T\quotient S}\setminus \{l_e, l_u\}, v_1, v_2 \in [\mathit{Clk}^{T\quotient S}\mapsto \mathbb{R}_{\geq 0}], \forall c\in\mathit{Clk}^{T\quotient S}\setminus \{x_{\mathit{new}}\}, v_1(c) = v_2(c) \}$. Now we will show that state $((l^T,v^T), (l^S,v^S))$ maps to only a single state $[((l^T,l^S),v)]_{\sim}$ using contradiction. Assume that state $((l^T,v^T), (l^S,v^S))$ maps to two (or more) states $[((l_1^T,l_1^S),v_1)]_{\sim}$ and $[((l_2^T,l_2^S),v_1)]_{\sim}$. From $\sim$ it follows that either $l_1^T \neq l_2^T$,  $l_1^S \neq l_2^S$, or $\exists c \in\mathit{Clk}^{T\quotient S}\setminus \{x_{\mathit{new}}\}: v_1(c) \neq v_2(c)$. But since we only consider a single state $((l^T,v^T), (l^S,v^S))$, none of these options can hold. Thus our assumption does not hold, which concludes the proof.}
\end{proof}

\new{Since we now have a bijective function $f$ relating states in $\sem{T\quotient S}^{\rho}$ and $\sem{T} \quotient \sem{S}$ together, we can effectively relabel the states in $\sem{T\quotient S}^{\rho}$ from $[((l^T, l^S), v)]_{\sim}$ to $((l^T, l^S), v^{T,S})$ in all proofs below, where $v^{T,S}\in [\mathit{Clk}^T\cup\mathit{Clk}^S \mapsto \mathbb{R}_{\geq 0}]$ with $\forall c\in \mathit{Clk}^T \cup \mathit{Clk}^S: v^{T,S}(c) = v(c)$. Notice that we remove the clock $x_{\mathit{new}}$ from the state label, as this clock is not present in the state labels in $\sem{T} \quotient \sem{S}$. Thus $Q^{\sem{T\quotient S}^{\rho}} = \{ ((l^T,l^S),v) \mid l^T\in\mathit{Loc}^T, l^S\in \mathit{Loc}^S, v\in [\mathit{Clk}^T \cup\mathit{Clk}^S \mapsto \mathbb{R}_{\geq 0}]\} \cup \{u, e\}$.}

\begin{lemma}\label{lemma:quotientreducedsamecons}
	\new{Given specification automata $S = (\mathit{Loc}^S,l_0^S,\mathit{Act}^S,\mathit{Clk}^S, E^S,\mathit{Inv}^S)$ and $T = (\mathit{Loc}^T,l_0^T,\allowbreak\mathit{Act}^T,\mathit{Clk}^T, E^T,\allowbreak\mathit{Inv}^T)$ where $\mathit{Act}_o^S\cap\mathit{Act}_i^T=\emptyset$. Then $\forall [q]\in Q^{\sem{T\quotient S}^{\rho}}, \forall q \in [q]_{\sim}$: $q\in \mathrm{cons}^{\sem{T\quotient S}}$ iff $[q]\in\mathrm{cons}^{\sem{T\quotient S}^{\rho}}$.}
\end{lemma}
\begin{proof}
	\new{From Lemmas~\ref{lemma:quotientequivalencequotient} and~\ref{lemma:quotientreducedquotient} it follows that $\sem{T\quotient S} \simeq \sem{T\quotient S}^{\rho}$. With $R_1 = \{(q, [q]_{\sim}) \mid q \in Q^{\sem{T\quotient S}} \}$ being the bisimulation relation for $\sem{T\quotient S} \simeq \sem{T\quotient S}/{\sim}$ and $R_2 = \{(q, q) \mid q \in Q^{\sem{T\quotient S}^{\rho}} \}$ the bisimulation relation for $\sem{T\quotient S}/{\sim} \simeq \sem{T\quotient S}^{\rho}$, we have that $R = \{(q, [q]_{\sim}) \mid [q]_{\sim} \in Q^{\sem{T\quotient S}^{\rho}} \}$ is a bisimulation relation for $\sem{T\quotient S} \simeq \sem{T\quotient S}^{\rho}$. Using this bisimulation relation, we can easily see that $q$ is an error state iff $[q]_{\sim}$ is an error state.
	}
	
	\new{We will now proof $q\in \mathrm{cons}^{\sem{T\quotient S}}$ iff $[q]\in\mathrm{cons}^{\sem{T\quotient S}^{\rho}}$ by contradiction. First, assume that $[q]\in\mathrm{cons}^{\sem{T\quotient S}^{\rho}}$, but $\exists q' \in [q]_{\sim}$ such that $q'\notin \mathrm{cons}^{\sem{T\quotient S}}$. That means that there exists a path from $q'$ to an error state $q''$. But since $\sem{T\quotient S} \simeq \sem{T\quotient S}^{\rho}$, it follows that $\sem{T\quotient S}^{\rho}$ can simulate the same path from $[q]_{\sim}$, and using $R$ we have that $\sem{T\quotient S}^{\rho}$ reaches state $[q'']_{\sim}$. But since we assume that $[q]\in\mathrm{cons}^{\sem{T\quotient S}^{\rho}}$, it must hold that $[q'']_{\sim}$ is not an error state. But this contradicts with the previous observation on error states. Showing the contradiction the other way around follows the same argument. Therefore, we can conclude that $q\in \mathrm{cons}^{\sem{T\quotient S}}$ iff $[q]\in\mathrm{cons}^{\sem{T\quotient S}^{\rho}}$.}
\end{proof}

\begin{lemma}\label{lemma:quotientadversarialreduced}
	\new{Given specification automata $S = (\mathit{Loc}^S,l_0^S,\mathit{Act}^S,\mathit{Clk}^S, E^S,\mathit{Inv}^S)$ and $T = (\mathit{Loc}^T,l_0^T,\allowbreak\mathit{Act}^T,\mathit{Clk}^T, E^T,\allowbreak\mathit{Inv}^T)$ where $\mathit{Act}_o^S\cap\mathit{Act}_i^T=\emptyset$. Then $(\sem{T\quotient S})^{\Delta} \simeq (\sem{T\quotient S}^{\rho})^{\Delta}$.}
\end{lemma}
\begin{proof}
	\new{First, observe from Definition~\ref{def:adversarialpruning} that adversarial pruning does not alter the action set. Therefore, together with Definition~\ref{def:quotientreducedquotient} of the reduced quotient it follows that $(\sem{T\quotient S})^{\Delta}$ and $(\sem{T\quotient S}^{\rho})^{\Delta}$ have the same action set. From the proof of Lemma~\ref{lemma:quotientreducedsamecons} it follows that $R = \{(q, [q]_{\sim}) \mid q \in Q^{(\sem{T\quotient S})^{\Delta}}\}$ is a bisimulation relation showing $\sem{T\quotient S} \simeq \sem{T\quotient S}^{\rho}$. Finally, using the result of Lemma~\ref{lemma:quotientreducedsamecons} that $\forall [q]\in Q^{\sem{T\quotient S}^{\rho}}, \forall q \in [q]_{\sim}$: $q\in \mathrm{cons}^{\sem{T\quotient S}}$ iff $[q]\in\mathrm{cons}^{\sem{T\quotient S}^{\rho}}$ together with Definition~\ref{def:adversarialpruning}, we can immediately conclude that $R = \{(q, [q]_{\sim}) \mid q \in Q^{(\sem{T\quotient S})^{\Delta}}\}$ is also a bisimulation relation showing $(\sem{T\quotient S})^{\Delta} \simeq (\sem{T\quotient S}^{\rho})^{\Delta}$.}
\end{proof}

\begin{lemma}\label{lemma:quotientTSandAerror}
	\new{Given specification automata $S = (\mathit{Loc}^S,l_0^S,\mathit{Act}^S,\mathit{Clk}^S, E^S,\mathit{Inv}^S)$ and $T = (\mathit{Loc}^T,l_0^T,\mathit{Act}^T,\mathit{Clk}^T, E^T,\allowbreak\mathit{Inv}^T)$ where $\mathit{Act}_o^S\cap\mathit{Act}_i^T=\emptyset$. Then $\mathrm{imerr}^{\sem{T\quotient S}^{\rho}} \subseteq \mathrm{imerr}^{\sem{T} \quotient \sem{S}}$ and $\mathrm{imerr}^{\sem{T} \quotient \sem{S}} \subseteq \mathrm{incons}^{\sem{T\quotient S}^{\rho}}$.}
\end{lemma}
\begin{proof}
	\new{First, observe that the semantic of a TIOA and the reduced quotient do not alter the action set. Therefore, it follows directly that $\sem{T\quotient S}^{\rho}$ and $\sem{T} \quotient \sem{S}$ have the same action set and partitioning into input and output actions, except that $\sem{T\quotient S}^{\rho}$ has an additional input event $i_{\mathit{new}}$, i.e., $\mathit{Act}^{\sem{T\quotient S}^{\rho}} \cup \{i_{\mathit{new}}\} = \mathit{Act}^{\sem{T} \quotient \sem{S}}$. 
	}
	
	\new{It follows from Lemma~\ref{lemma:quotientTSandAsamestateset} that there is a bijective function $f$ relating states from $\sem{T\quotient S}^{\rho}$ and $\sem{T} \quotient \sem{S}$ together. Therefore, we can effectively say that they have the same state set (up to relabeling), i.e., $Q^{\sem{T\quotient S}^{\rho}} = Q^{\sem{T\quotient S}}$. For brevity, in the rest of this proof we write we write $X=\sem{T\quotient S}^{\rho}$, $Y = \sem{T} \quotient \sem{S}$, $\mathit{Clk} = \mathit{Clk}^T \uplus\mathit{Clk}^S$, and $v^S$ and $v^T$ to indicate the part of a valuation $v$ of only the clocks of $S$ and $T$, respectively. Note that $x_{\mathit{new}}\notin\mathit{Clk}$, but $x_{\mathit{new}} \in \mathit{Clk}^X$.
	}
	
	\new{$\mathrm{imerr}^{\sem{T\quotient S}^{\rho}} \subseteq \mathrm{imerr}^{\sem{T} \quotient \sem{S}}$. From Definition~\ref{def:quotientTIOA} of the quotient for TIOA and Definition~\ref{def:quotientreducedquotient} of the reduced $\sim$-quotient of $\sem{T\quotient S}$, it follows that states in $\{(l_e, v)\in Q^{\sem{T\quotient S}^{\rho}} \mid v(x_{\mathit{new}}) = 0 \} = \mathrm{imerr}^{\sem{T\quotient S}^{\rho}} $ are immediate error states, as only states with location $l_e$ have an invariant other than $\mathbf{T}$. From Lemma~\ref{lemma:quotientTSandAsamestateset}, we have that $\forall q \in f(q) = e$ with $e \in Q^{\sem{T} \quotient \sem{S}}$. From Definition~\ref{def:quotientTIOTS} of the quotient for TIOTS, it follows immediately that $e$ is an error state, since only $d = 0$ time delay is possible without any transition labeled with output actions. Thus $e \in \mathrm{imerr}^{\sem{T} \quotient \sem{S}}$. This shows that $\mathrm{imerr}^{\sem{T\quotient S}^{\rho}} \subseteq \mathrm{imerr}^{\sem{T} \quotient \sem{S}}$.
	}
	
	\new{$\mathrm{imerr}^{\sem{T} \quotient \sem{S}} \subseteq \mathrm{incons}^{\sem{T\quotient S}^{\rho}}$. From Definition~\ref{def:quotientTIOTS} of the quotient for TIOTS, it follows that state $e$ is an immediate error state and that states in $\{(q^T, q^S) \in Q^{\sem{T}\quotient \sem{S}} \mid q^T\arrownot\xlongrightarrow{d}{}^{\!\! \sem{T}} \wedge q^S\xlongrightarrow{d}{}^{\!\! \sem{S}}\}$ are potentially error states, as these states have no outgoing delay transition, i.e., $(q^T,q^S)\arrownot\xlongrightarrow{d}{}^{\!\! \sem{T} \quotient \sem{S}}$. Some states of this set are actual immediate error states if $\nexists o!\in\mathit{Act}_o^{\sem{T}\quotient\sem{S}}$ s.t. $(q^T,q^S)\arrownot\xlongrightarrow{o!}{}^{\!\! \sem{T} \quotient \sem{S}}$. By Definition~\ref{def:quotientTIOTS} we have that $\mathit{Act}_o^{\sem{T}\quotient\sem{S}} = \mathit{Act}_o^T\setminus\mathit{Act}_o^S \cup \mathit{Act}_i^S\setminus\mathit{Act}_i^T$. Consider the following two cases.
	}
	
	\begin{itemize}
		\item \new{$o! \in \mathit{Act}_o^T\setminus\mathit{Act}_o^S$. Assume that $(q^T,q^S)\arrownot\xlongrightarrow{o!}{}^{\!\! \sem{T} \quotient \sem{S}}$, such that $(q^T,q^S)$ is actually an error state. It follows from Definition~\ref{def:semanticTIOA} of the semantic that $q^{\sem{T}} = (l^T, v^T)$ and $v^T + d \not\models \mathit{Inv}(l^T)$; similarly we have that $q^{\sem{S}} = (l^S, v^S)$ and $v^S + d \models \mathit{Inv}(l^S)$. Since TIOTSs are time additive, see Definition~\ref{def:tiots}, we can assume that for $\forall d' < d: v^T + d' \not\models \mathit{Inv}(l^T)$\footnote{\new{In case there would be a $d' < d$ such that $v^T + d' \models \mathit{Inv}(l^T)$, we can first delay $d'$ in $\sem{T}\quotient \sem{S}$ such that the reached state can no longer delay.}}. Thus $v^T + 0 \not\models \mathit{Inv}(l^T)$, which simplifies to $v^T \not\models \mathit{Inv}(l^T)$. Again, using time additivity of TIOTS and $v^S + d \models \mathit{Inv}(l^S)$, we have that $v^S + 0 \models \mathit{Inv}(l^S)$. Combining this information, we have that $v\models \neg\mathit{Inv}(l^T) \wedge \mathit{Inv}(l^S)$, where we used the fact that $\mathit{Clk}^T\cap\mathit{Clk}^S=\emptyset$. Now, using Definition~\ref{def:quotientTIOA} of the quotient for TIOA and Definition~\ref{def:semanticTIOA} of the semantics, we have that $(l^T, l^S, v)\xlongrightarrow{i_{\mathit{new}}}{}^{\!\! \sem{T\quotient S}} (l_e, v)$. Since the target state $(l_e, v)$ is an immediate error state and $i_{\mathit{new}}$ is an input action, it follows the controllable predecessor operator that $(l^T, l^S, v) \in \mathrm{incons}^{\sem{T\quotient S}^{\rho}}$.
		}
		
		\item \new{$o?\in \mathit{Act}_i^S\setminus\mathit{Act}_i^T$. Since $S$ is a specification, it is input-enabled, see Definition~\ref{def:specification}. Therefore, $q^S\xlongrightarrow{o?}{}^{\!\!  \sem{S}}$. From the second rule of Definition~\ref{def:quotientTIOTS} of the quotient for TIOTS, it follows that $(q^T,q^S)\xlongrightarrow{o!}{}^{\!\! \sem{T} \quotient \sem{S}}$. Therefore, in this case state $(q^T,q^S)$ is not an error state in $\sem{T}\quotient\sem{S}$.}
	\end{itemize}
\end{proof}

\begin{lemma}\label{lemma:quotientTSandAsamedelay}
	\new{Given specification automata $S = (\mathit{Loc}^S,l_0^S,\mathit{Act}^S,\mathit{Clk}^S, E^S,\mathit{Inv}^S)$ and $T = (\mathit{Loc}^T,l_0^T,\mathit{Act}^T,\mathit{Clk}^T, E^T,\allowbreak\mathit{Inv}^T)$ where $\mathit{Act}_o^S\cap\mathit{Act}_i^T=\emptyset$. Denote $X = \sem{T\quotient S}^{\rho}$ and $Y = \sem{T} \quotient \sem{S}$, and let $d\in\mathbb{R}_{\geq 0}$ and $q_1,q_2\in Q^X\cap Q^Y$ with $q_1 = (l^T,l^S,v)$ for some $v\in(\mathit{Clk}^T\uplus\mathit{Clk}^S \rightarrow \mathbb{R}_{\geq 0})$. If $v\not\models \neg\mathit{Inv}(l^T) \wedge \mathit{Inv}(l^S)$, then $q_1\xlongrightarrow{d}{}^{\!\! X} q_2$ if and only if $q_1\xlongrightarrow{d}{}^{\!\! Y} q_2$.}
\end{lemma}
\begin{proof}
	\new{It follows from Lemma~\ref{lemma:quotientTSandAsamestateset} that there is a bijective function $f$ relating states from $\sem{T\quotient S}^{\rho}$ and $\sem{T} \quotient \sem{S}$ together. Therefore, we can effectively say that they have the same state set (up to relabeling), i.e., $Q^{\sem{T\quotient S}^{\rho}} = Q^{\sem{T\quotient S}}$. For brevity, in the rest of this proof we write we write $\mathit{Clk} = \mathit{Clk}^T \uplus\mathit{Clk}^S$, and $v^S$ and $v^T$ to indicate the part of a valuation $v$ of only the clocks of $S$ and $T$, respectively. Note that $x_{\mathit{new}}\notin\mathit{Clk}$, but $x_{\mathit{new}} \in \mathit{Clk}^X$.
	}
	
	\new{From Definition~\ref{def:quotientTIOA} of the quotient for TIOA it follows that $\mathit{Inv}((l^T, l^S)) = \mathbf{T}$. Therefore, with Definition~\ref{def:semanticTIOA} of the semantic and Definition~\ref{def:quotientreducedquotient} of the $\sim$-reduced quotient of $\sem{T\quotient S}$ it follows that $q_1\xlongrightarrow{d}{}^{\!\! X} q_2$ is possible for any $d\in\mathbb{R}_{\geq 0}$ and any valuation $v$. Thus $q_1\xlongrightarrow{d}{}^{\!\! Y} q_2$ implies $q_1\xlongrightarrow{d}{}^{\!\! X} q_2$.
	}
	
	\new{It remains to show the other way around. Observe from Definition~\ref{def:quotientTIOTS} of the quotient for TIOTS that there are two cases involving a delay (actually three, but we do not consider the universal location in this lemma). So a delay is only possible from $q_1$ if either $q_1^{\sem{T}}\xlongrightarrow{d}{}^{\!\! \sem{T}} q_2^{\sem{T}} \wedge q_1^{\sem{S}}\xlongrightarrow{d}{}^{\!\! \sem{S}} q_2^{\sem{S}}$ or $q_1^{\sem{S}}\arrownot\xlongrightarrow{d}{}^{\!\! \sem{S}}$. So a delay is \emph{not} possible if $q_1^{\sem{T}}\arrownot\xlongrightarrow{d}{}^{\!\! \sem{T}} \wedge q_1^{\sem{S}}\xlongrightarrow{d}{}^{\!\! \sem{S}} q_2^{\sem{S}}$. It follows from Definition~\ref{def:semanticTIOA} of the semantic that $q_1^{\sem{T}} = (l^T, v^T)$ and $v^T + d \not\models \mathit{Inv}(l^T)$ or $\exists d'\in\mathbb{R}_{\geq 0}, d' < d: v^T + d' \not\models \mathit{Inv}(l^T)$; similarly we have that $q_1^{\sem{S}} = (l^S, v^S)$, $v^S + d \models \mathit{Inv}(l^S)$, and $\forall d'\in\mathbb{R}_{\geq 0}, d' < d: v^S + d' \models \mathit{Inv}(l^S)$. Without loss of generality, we can state that $d' = 0$\footnote{\new{In case there would be a $d' > 0$ such that $v^T + d' \models \mathit{Inv}(l^T)$, we can first delay $d'$ in $\sem{T}\quotient \sem{S}$ such that the reached state can no longer delay.}}, so $v^T + 0 \not\models \mathit{Inv}(l^T)$, which simplifies to $v^T \not\models \mathit{Inv}(l^T)$. We have also that $v^S + 0 \models \mathit{Inv}(l^S)$. Combining this information, we have that $v\models \neg\mathit{Inv}(l^T) \wedge \mathit{Inv}(l^S)$, where we used the fact that $\mathit{Clk}^T\cap\mathit{Clk}^S=\emptyset$. But this contradicts with the assumption in the lemma. Thus we can conclude that if $v\not\models \neg\mathit{Inv}(l^T) \wedge \mathit{Inv}(l^S)$, then $q_1\xlongrightarrow{d}{}^{\!\! X} q_2$ implies $q_1\xlongrightarrow{d}{}^{\!\! Y} q_2$.}
\end{proof}

\begin{lemma}\label{lemma:quotientTSandAsamecons}
	\new{Given specification automata $S = (\mathit{Loc}^S,l_0^S,\mathit{Act}^S,\mathit{Clk}^S, E^S,\mathit{Inv}^S)$ and $T = (\mathit{Loc}^T,l_0^T,\mathit{Act}^T,\mathit{Clk}^T, E^T,\allowbreak\mathit{Inv}^T)$ where $\mathit{Act}_o^S\cap\mathit{Act}_i^T=\emptyset$. Then $\mathrm{cons}^{\sem{T\quotient S}^{\rho}} = \mathrm{cons}^{\sem{T} \quotient \sem{S}}$.}
\end{lemma}
\begin{proof}
	\new{We will proof this by using the $\Theta$ operator. First, observe that the semantic of a TIOA and the reduced quotient do not alter the action set. Therefore, it follows directly that $\sem{T\quotient S}^{\rho}$ and $\sem{T} \quotient \sem{S}$ have the same action set and partitioning into input and output actions, except that $\sem{T\quotient S}^{\rho}$ has an additional input event $i_{\mathit{new}}$, i.e., $\mathit{Act}^{\sem{T\quotient S}^{\rho}} \cup \{i_{\mathit{new}}\} = \mathit{Act}^{\sem{T} \quotient \sem{S}}$. 
	}
	
	\new{It follows from Lemma~\ref{lemma:quotientTSandAsamestateset} that there is a bijective function $f$ relating states from $\sem{T\quotient S}^{\rho}$ and $\sem{T} \quotient \sem{S}$ together. Therefore, we can effectively say that they have the same state set (up to relabeling), i.e., $Q^{\sem{T\quotient S}^{\rho}} = Q^{\sem{T\quotient S}}$. For brevity, in the rest of this proof we write we write $X=\sem{T\quotient S}^{\rho}$, $Y = \sem{T} \quotient \sem{S}$, $\mathit{Clk} = \mathit{Clk}^T \uplus\mathit{Clk}^S$, and $v^S$ and $v^T$ to indicate the part of a valuation $v$ of only the clocks of $S$ and $T$, respectively. Note that $x_{\mathit{new}}\notin\mathit{Clk}$, but $x_{\mathit{new}} \in \mathit{Clk}^X$.
	}
	
	\new{We will show for any postfixed point $P$ of $\Theta$ that $\Theta^{\sem{T\quotient S}^{\rho}}(P) \subseteq \Theta^{\sem{T} \quotient \sem{S}}(P)$ and $\Theta^{\sem{T} \quotient \sem{S}}(P) \subseteq \Theta^{\sem{T \quotient S}^{\rho}}(P)$.
	}
	
	\new{($\Theta^X(P) \subseteq \Theta^Y(P)$) Consider a state $q^X\in P$. Because $P$ is a postfixed point of $\Theta^X$, it follows that $q^X\in\Theta^X(P)$. From the definition of $\Theta$, it follows that $q^X \in \overline{\mathrm{err}^X(\overline{P})}$ and $ q^X \in \{q_1\in Q^X \mid \forall d \geq 0: [\forall q_2\in Q^X: q_1\xlongrightarrow{d}{}^{\!\! X} q_2 \Rightarrow q_2\in P \wedge \forall i?\in \mathit{Act}_i^X: \exists q_3\in P: q_2\xlongrightarrow{i?}{}^{\!\! X} q_3] \vee [\exists d'\leq d \wedge \exists q_2,q_3\in P \wedge \exists o!\in\mathit{Act}_o^X: q_1\xlongrightarrow{d'}{}^{\!\! X} q_2 \wedge q_2\xlongrightarrow{o!}{}^{\!\! X} q_3 \wedge \forall i?\in\mathit{Act}_i^X: \exists q_4\in P: q_2\xlongrightarrow{i?}{}^{\!\! X} q_4]\}$. We will focus on the second part of the definition of $\Theta$.
	}
	
	\new{Consider a $d\in\mathbb{R}_{\geq 0}$. Then the left-hand side or the right-hand side of the disjunction is true (or both). }
	\begin{itemize}
		\item \new{Assume the left-hand side is true, i.e., $\forall q_2\in Q^X: q^X\xlongrightarrow{d}{}^{\!\! X} q_2 \Rightarrow q_2\in P \wedge \forall i?\in \mathit{Act}_i^X: \exists q_3\in P: q_2\xlongrightarrow{i?}{}^{\!\! X} q_3$. Pick a $q_2 \in Q^X$. The implication is true when $q^X\arrownot\xlongrightarrow{d}{}^{\!\! X} q_2$ or $q^X\xlongrightarrow{d}{}^{\!\! X} q_2 \wedge q_2\in P \wedge \forall i?\in \mathit{Act}_i^X: \exists q_3\in P: q_2\xlongrightarrow{i?}{}^{\!\! X} q_3$. 
		}
		
		\begin{itemize}
			\item \new{Consider the first case. This case is only applicable if $q^X = (l_e, v)$, since in Definition~\ref{def:quotientTIOA} of the quotient for TIOA only location $l_e$ has an invariant other than $\textbf{T}$. But then $q^X \in \mathrm{imerr}^X$. This contradicts with the fact that $q^X\in\Theta(P)$ implies that $q^X\in\overline{\mathrm{err}^X(\overline{P})}$. Thus this case is infeasible.
			}
			
			\item \new{Consider the second case. From Definition~\ref{def:semanticTIOA} of the semantic of a TIOA and Definition~\ref{def:quotientreducedquotient} of the $\sim$-reduced quotient of $\sem{T\quotient S}$ it follows that $v_1 + d\models \mathit{Inv}^{T\quotient S}(l_1)$ for $q^X = (l_1, v_1)$, $q_2 = (l_1, v_1 + d)$, $l_1\in \mathit{Loc}^{T\quotient S}$, and $v_1\in[\mathit{Clk}\mapsto \mathbb{R}_{\geq 0}]$. Since $q^X\in\overline{\mathrm{err}^X(\overline{P})}$, we have that $l_1 \neq l_e$, thus $\mathit{Inv}^{T\quotient S}(l_1)=\mathbf{T}$. Now, pick $i?\in\mathit{Act}_i^X$ and $q_3\in Q^X$ such that $q_2 \xlongrightarrow{i?}{}^{\!\! X} q_3$ and $q_3\in P$. From Definition~\ref{def:semanticTIOA} of the semantic of a TIOA it follows that $(l_1, i?, \varphi, c, l_3) \in E^{T\quotient S}$, $q_3 = (l_3, v_3)$, $v_1 + d \models \varphi$, $v_3 = v_1+d[r\mapsto 0]_{r\in c}$, and $v_3\models \mathit{Inv}^{T\quotient S}(l_3)$. 
			}
			
			\new{From Lemma~\ref{lemma:quotientTSandAsamedelay} it follows that $q^X\xlongrightarrow{d}{}^{\!\! Y} q_2$ if $v\not\models \neg\mathit{Inv}(l^T) \wedge \mathit{Inv}(l^S)$. In case that $v\models \neg\mathit{Inv}(l^T) \wedge \mathit{Inv}(l^S)$, we have from Definitions~\ref{def:quotientTIOA},~\ref{def:semanticTIOA}, and~\ref{def:quotientreducedquotient} that $q^X\xlongrightarrow{i_{\mathit{new}}}{}^{\!\! X} e$. But since $e \in \mathrm{err}^X(\overline{P})$, it follows that $e \notin P$. Therefore, this case is infeasible. Thus we have that $q^X\xlongrightarrow{d}{}^{\!\! Y} q_2$ in $Y$.
			}
			
			\new{Now, consider the ten cases from Definition~\ref{def:quotientTIOA} of quotient of TIOAs. Remember that $\mathit{Act}_i^X = \mathit{Act}_i^{T\quotient S} = \mathit{Act}_i^T \cup \mathit{Act}_o^S\cup \{i_{\mathit{new}}\}$.}
			\begin{enumerate}
				\item \new{$i?\in\mathit{Act}^S\cap\mathit{Act}^T$, $l_1 = (l_1^T,l_1^S)$, $l_3 = (l_3^T,l_3^S)$, $\varphi = \varphi^T\wedge \mathit{Inv}(l_3^T)[r\mapsto 0]_{r\in c^T}\wedge\varphi^S \wedge \mathit{Inv}(l_1^S) \wedge \mathit{Inv}(l_3^S)[r\mapsto 0]_{r\in c^S}$, $c = c^T\cup c^S$, $(l_1^T, i, \varphi^T,c^T, l_3^T)\in E^T$, and $(l_1^S, i, \varphi^S, c^S, l_3^S)\in E^S$. Since $v_1 + d\models \varphi$, it holds that $v_1 + d\models \varphi^T$, $v_1 + d\models \mathit{Inv}(l_3^T)[r\mapsto 0]_{r\in c^T}$,  $v_1 + d\models \varphi^S$, $v_1 + d\models \mathit{Inv}(l_1^S)$, and $v_1 + d \models \mathit{Inv}(l_3^S)[r\mapsto 0]_{r\in c^S}$. Because $\mathit{Clk}^S\cap\mathit{Clk}^T=\emptyset$, it holds that $v_1^T + d\models \varphi^T$, $v_1^T + d\models \mathit{Inv}(l_3^T)[r\mapsto 0]_{r\in c^T}$, $v_1^S + d\models \varphi^S$, $v_1^S + d \models \mathit{Inv}(l_1^S)$, and $v_1^S + d \models \mathit{Inv}(l_3^S)[r\mapsto 0]_{r\in c^S}$. Since $v_3 = v_1 + d[r\mapsto 0]_{r\in c}$, it holds that $v_3^T = v_1^T + d[r\mapsto 0]_{r\in c^T}$ and $v_3^S = v_1^S + d[r\mapsto 0]_{r\in c^S}$. Therefore, $v_3^T + d\models \mathit{Inv}(l_3^T)$ and $v_3^S + d\models\mathit{Inv}(l_3^S)$. 
				}
				
				\new{Combining all information about $T$, we have that $(l_1^T, i, \varphi^T,c^T, l_2^T)\in E^T$, $v_1^T + d\models \varphi^T$, $v_3^T = v_1^T + d[r\mapsto 0]_{r\in c^T}$, and $v_3^T\models \mathit{Inv}(l_3^T)$. Therefore, from Definition~\ref{def:semanticTIOA} it follows that $(l_1^T,v_1^T + d) \xlongrightarrow{i} (l_3^T,v_ 3^T)$ in $\sem{T}$. Combining all information about $S$, we have that $(l_1^S, i, \varphi^S,c^S, l_2^S)\in E^S$, $v_1^S + d\models \varphi^S$, $v_3^S = v_1^S + d[r\mapsto 0]_{r\in c^S}$, and $v_3^S\models \mathit{Inv}(l_3^S)$. Therefore, from Definition~\ref{def:semanticTIOA} it follows that $(l_1^S,v_1^S + d) \xlongrightarrow{i} (l_3^S,v_3^S)$ in $\sem{S}$. 
				}
				
				\new{Now, from Definition~\ref{def:quotientTIOTS} it follows that $((l_1^T,v_1^T + d),(l_1^S,v_1^S + d)) = (l_1^T, l_1^S, v_1 + d) = q_2^Y \xlongrightarrow{i?}{}^{\!\! Y} ((l_3^T,v_3^T),\allowbreak(l_3^S,v_3^S)) = (l_3^T, l_3^S, v_3) = q_3^Y$ in $Y$. Thus, we can simulate a transition in $Y$. Also, observe now that $q_2 = q_2^Y$ and $q_3=q_3^Y$.
				}
				
				\item \new{$i?\in\mathit{Act}^S\setminus\mathit{Act}^T$, $l_1 = (l^T,l_1^S)$, $l_3 = (l^T,l_3^S)$, $\varphi = \varphi^S \wedge \mathit{Inv}(l_1^S) \wedge \mathit{Inv}(l_3^S)[r\mapsto 0]_{r\in c^S}$, $c = c^S$, $l^T\in \mathit{Loc}^T$, and $(l_1^S, i!, \varphi^S, c^S, l_3^S)\in E^S$. Since $v_1 + d\models \varphi$ and $\mathit{Clk}^S\cap\mathit{Clk}^T=\emptyset$, it holds that $v_1^S + d\models \varphi^S$, $v_1^S + d\models \mathit{Inv}(l_1^S)$, and $v_1^S + d \models \mathit{Inv}(l_3^S)[r\mapsto 0]_{r\in c^S}$. Since $v_3 = v_1 + d[r\mapsto 0]_{r\in c}$ and $c = c^S$, it holds that $v_3^S = v_1^S + d[r\mapsto 0]_{r\in c^S}$, $v_3^T = v_1^T + d$, and $v_3^S\models\mathit{Inv}(l_3^S)$. Combining all information above about $S$, it follows from Definition~\ref{def:semanticTIOA} that $(l_1^S,v_1^S + d) \xlongrightarrow{i!} (l_3^S,v_3^S)$ in $\sem{S}$. From Definition~\ref{def:semanticTIOA} it also follows that $(l^T, v_1^T + d) \in Q^{\sem{T}}$. Therefore, following Definition~\ref{def:quotientTIOTS} it follows that $((l^T,v_1^T + d),(l_1^S,v_1^S + d)) = (l^T, l_1^S, v_1 + d) = q_2^Y \xlongrightarrow{i!}{}^{\!\! Y} ((l^T,v_1^T + d),(l_3^S,v_3^S)) = (l^T, l_3^S, v_3) = q_3^Y$ in $Y$. Thus, we can simulate a transition in $Y$. Also, observe now that $q_2 = q_2^Y$ and $q_3=q_3^Y$.
				}
				
				\item \new{$i!\in\mathit{Act}_o^S$, $l_1 = (l^T,l_1^S)$, $l_3 = l_u$, $\varphi = \neg G_S$, $c = \emptyset$, $l^T\in \mathit{Loc}^T$ and $G_S = \bigvee \{\varphi^S \wedge \mathit{Inv}(l_3^S)[r\mapsto 0]_{r\in c^S} \mid (l_1^S,a,\varphi^S,c^S,l_3^S) \in E^S\}$. Since $v_1 + d\models \varphi$ and $\mathit{Clk}^S\cap\mathit{Clk}^T=\emptyset$, it holds that $v_1^S + d\models \neg G_S$. Therefore, $v_1^S + d\not\models G_S$, which indicates that $\forall (l_1^S,a,\varphi^S,c^S,l_3^S) \in E^S$: $v_1^S + d\not\models \varphi^S \wedge \mathit{Inv}(l_3^S)[r\mapsto 0]_{r\in c^S}$. This means that $v_1^S + d\not\models \varphi^S$ or $v_1^S + d\not\models \mathit{Inv}(l_3^S)[r\mapsto 0]_{r\in c^S}$ or both, where the second option is equivalent to $v_1^S + d[r\mapsto 0]_{r\in c^S} \not\models \mathit{Inv}(l_3^S)$. Following Definition~\ref{def:semanticTIOA}, we can conclude that $(l_1^S,v_1^S + d) \arrownot\xlongrightarrow{\ a\ }$ in $\sem{S}$. From Definition~\ref{def:semanticTIOA} it also follows that $(l^T, v_1^T + d) \in Q^{\sem{T}}$. Now, following Definition~\ref{def:quotientTIOTS}, we have transition $((l^T,v_1^T + d),(l_1^S,v_1^S + d)) = (l^T, l_1^S, v_1 + d) = q_2^Y \xlongrightarrow{a}{}^{\!\! Y} u = q_3^Y$ in $Y$. Thus we can simulate a transition in $Y$. Also, observe now that $q_2 = q_2^Y$ and $q_3=q_3^Y$ (where $(l_u, v_3)$ is mapped into $u$ by $f$ from Lemma~\ref{lemma:quotientTSandAsamestateset}). 
				}
				
				\item \new{$i? \in \mathit{Act}^S\cup\mathit{Act}^T$, $l_1 = (l^T,l^S)$, $l_3 = l_u$, $\varphi = \neg \mathit{Inv}(l^S)$, $c = \emptyset$, $l^T\in \mathit{Loc}^T$, and $l^S\in\mathit{Loc}^S$.(If $i? = i_{\mathit{new}}$, this case is trivial, see item 8 and 10 below.) Since $v_1 + d\models \varphi$ and $\mathit{Clk}^S\cap\mathit{Clk}^T=\emptyset$, it holds that $v_1^S + d\models \neg \mathit{Inv}(l^S)$. Therefore, $v_1^S + d\not\models \mathit{Inv}(l^S)$. Since we delayed into state $q_2^Y$, it must hold that the delay was according to rule 6 of Definition~\ref{def:quotientTIOTS} of the quotient for TIOTS. Therefore, $q_2^Y = u \in P$. From Definition~\ref{def:quotientTIOTS} it also follows that $u = q_2^Y \xlongrightarrow{i?}{}^{\!\! Y} u = q_3^Y$ in $Y$. Thus we can simulate a transition in $Y$. Also, observe now that $q_3=q_3^Y$ (where $(l_u, v_3)$ is mapped into $u$ by $f$ from Lemma~\ref{lemma:quotientTSandAsamestateset}).
				}
				
				\item \new{$i!\in\mathit{Act}_o^S\cap\mathit{Act}_o^T$, $l_1 = (l_1^T,l_1^S)$, $l_3 = l_e$, $\varphi = \varphi^S \wedge \mathit{Inv}{l_1^S} \wedge \mathit{Inv}(l_3^S)[r\mapsto 0]_{r\in c^S} \wedge \neg G_T$, $c = \{x_{\mathit{new}}\}$, $(l_1^S, a, \varphi^S, c^S, l_3^S)\in E^S$, and $G_T = \bigvee \{\varphi^T \wedge \mathit{Inv}(l_3^T)[r\mapsto 0]_{r\in c^T} \mid (l_1^T,a,\varphi^T,c^T, l_3^T) \in E^T\}$. Since the target location is the error location, it holds that $q_3\notin P$. Thus this case is not feasible.
				}
				
				\item \new{$a = i_{\mathit{new}}$, $l_1 = (l^T,l^S)$, $l_3 = l_e$, $\varphi = \neg \mathit{Inv}(l^T) \wedge \mathit{Inv}(l^S)$, $c = \{x_{\mathit{new}}\}$, $l^T\in \mathit{Loc}^T$, and $q^S\in\mathit{Loc}^S$. Since the target location is the error location, it holds that $q_3\notin P$. Thus this case is not feasible.
				}
				
				\item \new{$a = i_{\mathit{new}}$, $l_1 = l_3 = (l_1^T,l_1^S)$, $\varphi = \mathit{Inv}(l^T) \vee \neg\mathit{Inv}(l^S)$ and $c = \emptyset$. First note that $i_{\mathit{new}} \notin\mathit{Act}^Y$. Now, since $c = \emptyset$, it follows that $v_3 = v_1 + d$. Therefore, $q_2 = q_3$. Since $q_3\in P$, it follows $q_2\in P$. Since $q_2 = q_2^Y$, it follows that $q_2^Y \in P$.
				}
				
				\item \new{$i?\in\mathit{Act}^T\setminus\mathit{Act}^S$, $l_1 = (l_1^T,l^S)$, $l_3 = (l_3^T,l^S)$, $\varphi = \varphi^T \wedge \mathit{Inv}(l_3^T)[r\mapsto 0]_{r\in c^T} \wedge \mathit{Inv}(l^S)$, $c = c^T$, $l^S\in \mathit{Loc}^S$, and $(l_1^T, i?, \varphi^T,c^T, l_3^T)\in E^T$. Since $v_1 + d\models \varphi$ and $\mathit{Clk}^S\cap\mathit{Clk}^T=\emptyset$, it holds that $v_1^T + d\models \varphi^T$ and $v_1^T + d \models \mathit{Inv}(l_3^T)[r\mapsto 0]_{r\in c^T}$. Since $v_3 = v_1 + d[r\mapsto 0]_{r\in c}$ and $c = c^T$, it holds that $v_3^T = v_1^T + d[r\mapsto 0]_{r\in c^T}$, $v_3^S = v_1^S + d$, and $v_3^T\models\mathit{Inv}(l_3^T)$. Combining all information above about $T$, it follows from Definition~\ref{def:semanticTIOA} that $(l_1^T,v_1^T + d) \xlongrightarrow{i?} (l_3^T,v_3^T)$ in $\sem{T}$. From Definition~\ref{def:semanticTIOA} it also follows that $(l^S, v_1^S + d) \in Q^{\sem{S}}$. Therefore, following Definition~\ref{def:quotientTIOTS} it follows that $((l_1^T,v_1^T + d),(l^S,v_1^S + d)) = (l_1^T, l^S, v_1 + d) = q_2^Y \xlongrightarrow{i?}{}^{\!\! Y} ((l_3^T,v_3^T),(l^S,v_1^S + d)) = (l_3^T, l^S, v_3) = q_3^Y$ in $Y$. Thus, we can simulate a transition in $Y$. Also, observe now that $q_2 = q_2^Y$ and $q_3=q_3^Y$.
				}
				
				\item \new{$i?\in\mathit{Act}^S\cup\mathit{Act}^T$, $l_1 = l_u$, $l_3 = l_u$, $\varphi = \mathbf{T}$, $c = \emptyset$. Since $q^X = q^Y$, it follows from Definition~\ref{def:quotientTIOTS} of the quotient for TIOTS that $Y$ delayed within state $u$ as well, i.e., $q_2^X = q_2^Y$. Therefore, using Definition~\ref{def:quotientTIOTS} again, we have that there exists a transition $q_2^Y = u \xlongrightarrow{i?}{}^{\!\! Y} u = q_3^Y$ in $Y$. Thus, we can simulate a transition in $Y$. Also, observe now that $q_2 = q_2^Y$ and $q_3=q_3^Y$.
				}
				
				\item \new{$a\in\mathit{Act}_i^S\cup\mathit{Act}_i^T$, $l_1 = l_e$, $l_3 = l_e$, $\varphi = x_{\mathit{new}}=0$, $c = \emptyset$. Since the target location is the error location, it holds that $q_3^X\notin P$. Thus this case is not feasible.}
			\end{enumerate}
			
			\new{So, in all feasible cases we have that $q_2^Y \xrightarrow{i?} q_3^Y$ is a transition in $Y$ if $i?\neq i_{\mathit{new}}$. When $i? = i_{\mathit{new}}$, we have shown explicitly that $q_2^Y \in P$. As the analysis above is independent of the particular $i?$, $q_2^Y \xrightarrow{i?} q_3^Y$ is a transition in $Y$ for all $i?\in \mathit{Act}_i^Y$.  Furthermore, all feasible cases show that $q_2^Y, q_3^Y\in P$ directly, or because $q_2^Y = q_2$ or $q_3^Y = q_3$.}
		\end{itemize}
		
		\new{So, in both cases we have that for $q^X\xlongrightarrow{d}{}^{\!\! Y} q_2^Y \Rightarrow q_2^Y\in P \wedge \forall i?\in \mathit{Act}_i^Y: \exists q_3^Y\in P: q_2^Y\xlongrightarrow{i?}{}^{\!\! Y} q_3^Y$. As $q_2$ is chosen arbitrarily, it holds for all $q_2\in Q^X = Q^Y$. Therefore, the left-hand side is true.}
		
		\item \new{Assume the right-hand side is true, i.e., $\exists d'\leq d \wedge \exists q_2,q_3\in P \wedge \exists o!\in\mathit{Act}_o^X: q^X\xlongrightarrow{d'}{}^{\!\! X} q_2 \wedge q_2\xlongrightarrow{o!}{}^{\!\! X} q_3 \wedge \forall i?\in\mathit{Act}_i^X: \exists q_4\in P: q_2\xlongrightarrow{i?}{}^{\!\! X} q_4$. }
		
		\new{Following Definition~\ref{def:semanticTIOA} of the semantic of a TIOA and Definition~\ref{def:quotientreducedquotient} of the $\sim$-reduced quotient of $\sem{T\quotient S}$, we have that $q^X = (l_1,v_1)$, $q_2 = (l_1, v_1 + d')$, $q_3 = (l_3, v_3)$, $q_4 = (l_4, v_4)$, $l_1, l_3, l_4\in \mathit{Loc}^{T\quotient S}$, $v_1, v_3, v_4\in [\mathit{Clk}\mapsto \mathbb{R}_{\geq 0}]$, $v_1 + d'\models \mathit{Inv}^{T\quotient S}(l_1)$, $\exists (l_1,o!,\varphi, c, l_3)\in E^{T\quotient S}$, $v_1 + d'\models\varphi$, $v_3 = v_1 + d'[r\mapsto 0]_{r\in c}$, and $v_3\models \mathit{Inv}^{T\quotient S}(l_3)$. First, focus on the delay transition. 
		}
		
		\new{From Lemma~\ref{lemma:quotientTSandAsamedelay} it follows that $q^X\xlongrightarrow{d}{}^{\!\! Y} q_2$ if $v\not\models \neg\mathit{Inv}(l^T) \wedge \mathit{Inv}(l^S)$. In case that $v\models \neg\mathit{Inv}(l^T) \wedge \mathit{Inv}(l^S)$, we have from Definitions~\ref{def:quotientTIOA},~\ref{def:semanticTIOA}, and~\ref{def:quotientreducedquotient} that $q^X\xlongrightarrow{i_{\mathit{new}}}{}^{\!\! X} e$. But since $e \in \mathrm{err}^X(\overline{P})$, it follows that $e \notin P$. Since $i_{\mathit{new}}$ is an input action, it must hold that $q_2 \notin P$ (see analysis above in the proof). Therefore, this case is infeasible. Thus we have that $q^X\xlongrightarrow{d}{}^{\!\! Y} q_2$ in $Y$.
		}
		
		\new{Now consider the output transition labeled with $o!$. Remember that $\mathit{Act}_o^{T\quotient S} = \mathit{Act}_o^T\setminus\mathit{Act}_o^S \cup \mathit{Act}_i^S\setminus \mathit{Act}_i^T$. We have to consider the ten cases from Definition~\ref{def:quotientTIOA} of the quotient for TIOA. We can use the exact same argument as before (where now rules 3 and 5 have become infeasible) to show that $q_2 \xrightarrow{o!} q_3$ is a transition in $Y$ for all feasible cases. As the analysis is independent of the particular $o!$, we can conclude that ${q^X\xlongrightarrow{d'}{}^{\!\! Y} q_2} \wedge q_2\xlongrightarrow{o!}{}^{\!\! Y} q_3$ with $q_2, q_3 \in P$.}
		
		\new{Finally, consider the input transitions labeled with $i?$. Using the same argument as before, we can show that $q_2\xrightarrow{i?} q_4$ in $X$ is also a transition in $Y$, and $q_4\in P$. Therefore, we can conclude that ${q^X\xlongrightarrow{d'}{}^{\!\! Y} q_2} \wedge q_2\xlongrightarrow{o!}{}^{\!\! Y} q_3 \wedge \forall i?\in\mathit{Act}_i^Y: \exists q_4\in P: q_2\xlongrightarrow{i?}{}^{\!\! Y} q_4$ with $q_2, q_3, q_4 \in P$. Thus, the right-hand side is true.}
	\end{itemize}
	
	\new{Thus, we have shown that when the left-hand side is true for $q^X$ in $X$, it is also true for $q^X$ in $Y$; and that when the right-hand side is true for $q^X$ in $X$, it is also true for $q^X$ in $Y$. Thus, $q^X\in\Theta^Y(P)$. Since $q^X\in P$ was chosen arbitrarily, it holds for all states in $P$. Once we choose $P$ to be the fixed-point of $\Theta^X$, we have that $\Theta^X(P)\subseteq\Theta^Y(P)$.
	}
	
	\new{($\Theta^Y(P) \subseteq \Theta^X(P)$) 
		Consider a state $q^Y \in P$. Because $P$ is a postfixed point of $\Theta^Y$, it follows that $p\in\Theta^X(Y)$. From the definition of $\Theta$, it follows that $q^Y \in \overline{\mathrm{err}^Y(\overline{P})}$ and $ q^Y \in \{q\in Q^Y\mid \forall d \geq 0: [\forall q_2\in Q^Y: q\xlongrightarrow{d}{}^{\!\! Y} q_2 \Rightarrow q_2\in P \wedge \forall i?\in \mathit{Act}_i^Y: \exists q_3\in P: q_2\xlongrightarrow{i?}{}^{\!\! Y} q_3]\ \vee [\exists d'\leq d \wedge \exists q_2,q_3\in P \wedge \exists o!\in\mathit{Act}_o^Y: q\xlongrightarrow{d'}{}^{\!\! Y} q_2 \wedge q_2\xlongrightarrow{o!}{}^{\!\! Y} q_3 \wedge \forall i?\in\mathit{Act}_i^Y: \exists q_4\in P: q_2\xlongrightarrow{i?}{}^{\!\! Y} q_4]\}$. Now we focus on the second part of the definition of $\Theta$.
	}
	
	\new{Consider a $d\in\mathbb{R}_{\geq 0}$. Then the left-hand side or the right-hand side of the disjunction is true (or both). }
	\begin{itemize}
		\item \new{Assume the left-hand side is true, i.e., $\forall q_2\in Q^Y: q^Y\xlongrightarrow{d}{}^{\!\! Y} q_2 \Rightarrow q_2\in P \wedge \forall i?\in \mathit{Act}_i^Y: \exists q_3\in P: q_2\xlongrightarrow{i?}{}^{\!\! Y} q_3$. Pick a $q_2 \in Q^Y$. The implication is true when $q^Y\arrownot\xlongrightarrow{d}{}^{\!\! Y} q_2$ or $q^Y\xlongrightarrow{d}{}^{\!\! Y} q_2 \wedge q_2\in P \wedge \forall i?\in \mathit{Act}_i^Y: \exists q_3\in P: q_2\xlongrightarrow{i?}{}^{\!\! Y} q_3$. 
		}
		
		\begin{itemize}
			\item \new{Consider the first case. From Lemma~\ref{lemma:quotientTSandAsamedelay} it follows that $q^Y\arrownot\xlongrightarrow{d}{}^{\!\! Y}$ if $v\models \neg\mathit{Inv}(l^T) \wedge \mathit{Inv}(l^S)$ with $q^Y = (l_1,v_1)$. Now we have from Definitions~\ref{def:quotientTIOA},~\ref{def:semanticTIOA}, and~\ref{def:quotientreducedquotient} that $q^Y\xlongrightarrow{i_{\mathit{new}}}{}^{\!\! X} e$. But since $e \in \mathrm{err}^Y(\overline{P})$, it follows that $e \notin P$. Since $i_{\mathit{new}}$ is an input action, it must hold that $(l_1, v) \notin P$ for any valuation $v$ (see analysis above in the proof). Therefore, $q^Y\arrownot\xlongrightarrow{d}{}^{\!\! X}$. Thus the implication also holds for $q_2$ in $X$.
			}
			
			\item \new{Consider the second case. From Definition~\ref{def:quotientTIOA} of the quotient for TIOA it follows that $\mathit{Inv}((l^T, l^S)) = \mathbf{T}$. Therefore, with Definition~\ref{def:semanticTIOA} of the semantic and Definition~\ref{def:quotientreducedquotient} of the $\sim$-reduced quotient of $\sem{T\quotient S}$ it follows that $q^Y\xlongrightarrow{d}{}^{\!\! X} q_2$. Now, pick an $i?\in\mathit{Act}_i^Y$ with its corresponding $q_3$ according to the implication. Remember that $\mathit{Act}_i^Y = \mathit{Act}_i^T \cup \mathit{Act}_o^S$. We have to consider the nine cases from Definition~\ref{def:quotientTIOTS}.
			}
			
			\begin{enumerate}
				\item \new{$i?\in \mathit{Act}^S \cap\mathit{Act}^T$, $q_2^Y = (q_2^{\sem{T}},q_2^{\sem{S}})$, $q_3^Y = (q_3^{\sem{T}},q_3^{\sem{S}})$, $q_2^{\sem{T}} \xlongrightarrow{i}{}^{\!\! \sem{T}} q_3^{\sem{T}}$, and $q_2^{\sem{S}}\allowbreak {\xlongrightarrow{i}{}^{\!\! \sem{S}}} q_3^{\sem{S}}$. From Definition~\ref{def:semanticTIOA} of semantic it follows that there exists an edge $(l_2^T,i,\varphi^T, c^T, l_3^T)\in E^T$ with $q_2^{\sem{T}} = (l_2^T,v_2^T)$, $q_3^{\sem{T}} = (l_3^T,v_3^T)$, $l_2^T,l_3^T \in \mathit{Loc}^T$, $v_2^T, v_3^T \in [\mathit{Clk}^T \mapsto \mathbb{R}_{\geq 0}]$, $v_2^T\models \varphi^T$, $v_3^T = v_2^T[r\mapsto 0]_{r\in c^T}$, and $v_3^T\models \mathit{Inv}^T(l_3^T)$. Similarly, it follows from the same definition that there exists an edge $(l_2^S,i,\varphi^S, c^S, l_3^S)\in E^S$ with $q_2^{\sem{S}} = (l_2^S,v_2^S)$, $q_3^{\sem{S}} = (l_3^S,v_3^S)$, $l_2^S,l_3^S \in \mathit{Loc}^S$, $v_2^S, v_3^S \in [\mathit{Clk}^S \mapsto \mathbb{R}_{\geq 0}]$, $v_2^S\models \varphi^S$, $v_3^S = v_2^S[r\mapsto 0]_{r\in c^S}$, and $v_3^S\models \mathit{Inv}^S(l_3^S)$. Based on Definition~\ref{def:quotientTIOA} of the quotient for TIOA, we need to consider the following two cases.
				}
				
				\begin{itemize}
					\item \new{$v_2^S \models \mathit{Inv}(l_2^S)$. In this case, there exists an edge $((l_2^T, l_2^S),i,\varphi^T \wedge \mathit{Inv}(l_3^T)[r\mapsto 0]_{r\in c^T} \wedge \varphi^S \wedge \mathit{Inv}(l_2^S) \wedge \mathit{Inv}(l_3^S)[r\mapsto 0]_{r\in c^S}, c^T \cup c^S, (l_3^T,l_3^S))$ in $T\quotient S$. Let $v_i, i=1,2$ be the valuations that combines the one from $T$ with the one from $S$, i.e. $\forall r \in \mathit{Clk}^T: v_i(r) = v_i^T(r)$ and $\forall r \in \mathit{Clk}^S: v_i(r) = v_i^S(r)$. Because $\mathit{Clk}^T\cap\mathit{Clk}^S = \emptyset$, it holds that $v_2\models \varphi^T$, $v_2\models \varphi^S$, and $v_2^S \models \mathit{Inv}(l_2^S)$, thus $v_2\models \varphi^T\wedge\varphi^S \wedge \mathit{Inv}(l_2^S)$; $v_3 = v_2[r\mapsto 0]_{r\in c^T\cup c^S}$; and $v_3\models\mathit{Inv}^T(l_3^T)$ and $v_3\models\mathit{Inv}^S(l_3^S)$, thus $v_3\models\mathit{Inv}^T(l_3^T)\wedge\mathit{Inv}^S(l_3^S)$. 
					}
					
					\new{From Definition~\ref{def:semanticTIOA} it now follows that $((l_2^T,l_2^S), v_2) \xrightarrow{i} ((l_3^T,l_3^S), v_3)$ is a transition in $\sem{T\quotient S}$. Because $\mathit{Clk}^T\cap\mathit{Clk}^S = \emptyset$, we can rearrange the states into $((l_2^T,l_2^S),v_2) = ((l_2^T,v_2^S) , (l_2^T,v_2^S)) = q_2^Y$ and $((l_3^T,l_3^S),v_3) = ((l_3^T,v_3^T) , (l_3^S,v_3^S)) = q_3^Y$. Thus, $q_2^Y \xrightarrow{a} q_3^Y$ is a transition in $\sem{T\quotient S} = Y$. Also, observe now that $q_2^X = q_2^Y$ and $q_3^X=q_3^Y$.
					}
					
					\item \new{$v_2^S \not\models \mathit{Inv}(l_2^S)$. In this case, state $q_2 = (l_2^T, v_2^T, l_2^S,v_2^S)$ cannot be reached by delaying into it, since $v_2^S \not\models \mathit{Inv}(l_2^S)$ implies with Definition~\ref{def:semanticTIOA} of the semantic that $\forall q^{\sem{S}} \in Q^{\sem{S}}$ we have $q^{\sem{S}} \arrownot\xlongrightarrow{d}{}^{\!\! \sem{S}} q_2^{\sem{S}}$. From Definition~\ref{def:quotientTIOTS} we have that in this case $q^Y \xlongrightarrow{d}{}^{\!\! Y} u$, and $q_2^Y \neq u$. Thus this case is infeasible.}
				\end{itemize}
				
				\item \new{$i!\in \mathit{Act}^S \setminus \mathit{Act}^T$, $q_2^Y = (q^{\sem{T}},q_2^{\sem{S}})$, $q_3^Y = (q^{\sem{T}},q_3^{\sem{S}})$, $q^{\sem{T}}\in Q^{\sem{T}}$, and $q_2^{\sem{S}}\xlongrightarrow{i!}{}^{\!\! \sem{S}} q_3^{\sem{S}}$. From Definition~\ref{def:semanticTIOA} of semantic it follows that there exists an edge $(l_2^S,i!,\varphi^S, c^S, l_3^S)\in E^S$ with $q_2^{\sem{S}} = (l_2^S,v_2^S)$, $q_3^{\sem{S}} = (l_3^S,v_3^S)$, $l_2^S,l_3^S \in \mathit{Loc}^S$, $v_2^S, v_3^S \in [\mathit{Clk}^S \mapsto \mathbb{R}_{\geq 0}]$, $v_2^S\models \varphi^S$, $v_3^S = v_2^S[r\mapsto 0]_{r\in c^S}$, and $v_3^S\models \mathit{Inv}^S(l_3^S)$. From the same definition, it follows that $q^{\sem{T}} = (l^T, v^T)$ for some $l^T\in\mathit{Loc}^T$ and $v^T \in [\mathit{Clk}^T \mapsto \mathbb{R}_{\geq 0}]$. Based on Definition~\ref{def:quotientTIOA} of the quotient for TIOA, we need to consider the following two cases.
				}
				
				\begin{itemize}
					\item \new{$v_2^S \models \mathit{Inv}(l_2^S)$. In this case, there exists an edge $((l^T, l_2^S),a, \varphi^S \wedge \mathit{Inv}(l_2^S) \wedge \mathit{Inv}(l_3^S)[r\mapsto 0]_{r\in c^S}, c^S,\allowbreak (l^T,l_3^S))$ in $T\quotient S$. Let $v_i, i=1,2$ be the valuations that combines the one from $T$ with the one from $S$, i.e. $\forall r \in \mathit{Clk}^T: v_i(r) = v_i^T(r)$ and $\forall r \in \mathit{Clk}^S: v_i(r) = v_i^S(r)$. Because $\mathit{Clk}^T\cap\mathit{Clk}^S = \emptyset$, it holds that $v_2\models \varphi^S$, and $v_2 \models \mathit{Inv}(l_2^S)$, thus $v_2\models \varphi^S \wedge \mathit{Inv}(l_2^S)$; $v_3 = v_2[r\mapsto 0]_{r\in c^S}$; and $v_3\models\mathit{Inv}^S(l_3^S)$.
					}
					
					\new{Since $\mathit{Inv}((l^T, l_3^S)) = \mathbf{T}$ by definition $T\quotient S$, we have that $v_3\models \mathit{Inv}((l^T, l_3^S))$. From Definition~\ref{def:semanticTIOA} it now follows that $((l^T,l_2^S), v_2) \xrightarrow{i?} ((l^T,l_3^S), v_3)$ is a transition in $\sem{T\quotient S}$. Using Definition~\ref{def:quotientreducedquotient} of the reduced $\sim$-quotient of $\sem{T\quotient S}$ and Lemma~\ref{lemma:quotientTSandAsamestateset}, we can rearrange the states into $((l^T,l_2^S),v_2) = ((l^T,v_2^T) , (l_2^S,v_2^S)) = q_2^Y$ and $((l^T,l_3^S),v_3) = ((l^T,v_3^T) , (l_3^S,v_3^S)) = q_3^Y$, and we can show that $q_2^Y \xrightarrow{i?} q_3^Y$ is a transition in $\sem{T\quotient S}^{\rho} = X$. Also, observe now that $q_2^X = q_2^Y$ and $q_3^X=q_3^Y$.
					}
					
					\item \new{$v_2^S \not\models \mathit{Inv}(l_2^S)$. In this case, state $q_2 = (l_2^T, v_2^T, l_2^S,v_2^S)$ cannot be reached by delaying into it, since $v_2^S \not\models \mathit{Inv}(l_2^S)$ implies with Definition~\ref{def:semanticTIOA} of the semantic that $\forall q^{\sem{S}} \in Q^{\sem{S}}$ we have $q^{\sem{S}} \arrownot\xlongrightarrow{d}{}^{\!\! \sem{S}} q_2^{\sem{S}}$. From Definition~\ref{def:quotientTIOTS} we have that in this case $q^Y \xlongrightarrow{d}{}^{\!\! Y} u$, and $q_2^Y \neq u$. Thus this case is infeasible.}
				\end{itemize}
				
				\item \new{$i?\in \mathit{Act}^T \setminus \mathit{Act}^S$, $q_2^Y = (q_2^{\sem{T}},q^{\sem{S}})$, $q_3^Y = (q_3^{\sem{T}},q^{\sem{S}})$, $q^{\sem{S}}\in Q^{\sem{S}}$, and $q_2^{\sem{T}}\xlongrightarrow{i?}{}^{\!\! \sem{T}} q_3^{\sem{T}}$. 
					From Definition~\ref{def:semanticTIOA} of semantic it follows that there exists an edge $(l_2^T,i?,\varphi^T, c^T, l_3^T)\in E^T$ with $q_2^{\sem{T}} = (l_2^T,v_2^T)$, $q_3^{\sem{T}} = (l_3^T,v_3^T)$, $l_2^T,l_3^T \in \mathit{Loc}^T$, $v_2^T, v_3^T \in [\mathit{Clk}^T \mapsto \mathbb{R}_{\geq 0}]$, $v_2^T\models \varphi^T$, $v_3^T = v_2^T[r\mapsto 0]_{r\in c^T}$, and $v_3^T\models \mathit{Inv}^T(l_3^T)$. From the same definition, it follows that $q^{\sem{S}} = (l^S, v^S)$ for some $l^S\in\mathit{Loc}^S$ and $v^S \in [\mathit{Clk}^S \mapsto \mathbb{R}_{\geq 0}]$. Based on Definition~\ref{def:quotientTIOA} of the quotient for TIOA, we need to consider the following two cases.
				}
				
				\begin{itemize}
					\item \new{$v_2^S \models \mathit{Inv}(l_2^S)$. In this case, there exists an edge $((l_2^T, l^S),i?, \varphi^T \wedge \mathit{Inv}(l_3^T)[r\mapsto 0]_{r\in c^T} \wedge \mathit{Inv}(l^S), c^T,\allowbreak (l_3^T,l^S))$ in $T\quotient S$. Let $v_i, i=1,2$ be the valuations that combines the one from $T$ with the one from $S$, i.e. $\forall r \in \mathit{Clk}^T: v_i(r) = v_i^T(r)$ and $\forall r \in \mathit{Clk}^S: v_i(r) = v_i^S(r)$. Because $\mathit{Clk}^T\cap\mathit{Clk}^S = \emptyset$, it holds that $v_2\models \varphi^T$, and $v_2 \models \mathit{Inv}(l^S)$, thus $v_2\models \varphi^T \wedge \mathit{Inv}(l^S)$; $v_3 = v_2[r\mapsto 0]_{r\in c^T}$; and $v_3\models\mathit{Inv}^T(l_3^T)$. 
					}
					
					\new{Since $\mathit{Inv}((l_3^T, l^S)) = \mathbf{T}$ by definition $T\quotient S$, we have that $v_3\models \mathit{Inv}((l_3^T, l^S))$. From Definition~\ref{def:semanticTIOA} it now follows that $((l_2^T,l^S), v_2) \xrightarrow{i?} ((l_3^T,l^S), v_3)$ is a transition in $\sem{T\quotient S}$. Using Definition~\ref{def:quotientreducedquotient} of the reduced $\sim$-quotient of $\sem{T\quotient S}$ and Lemma~\ref{lemma:quotientTSandAsamestateset}, we can rearrange the states into $((l_2^T,l^S),v_2) = ((l_2^T,v_2^T) , (l^S,v_2^S)) = q_2^Y$ and $((l_3^T,l^S),v_3) = ((l_3^T,v_3^T) , (l^S,v_3^S)) = q_3^Y$, and we can show that $q_2^Y \xrightarrow{i?} q_3^Y$ is a transition in $\sem{T\quotient S}^{\rho} = X$. Also, observe now that $q_2^X = q_2^Y$ and $q_3^X=q_3^Y$.
					}
					
					\item \new{$v_2^S \not\models \mathit{Inv}(l_2^S)$. In this case, state $q_2 = (l_2^T, v_2^T, l_2^S,v_2^S)$ cannot be reached by delaying into it, since $v_2^S \not\models \mathit{Inv}(l_2^S)$ implies with Definition~\ref{def:semanticTIOA} of the semantic that $\forall q^{\sem{S}} \in Q^{\sem{S}}$ we have $q^{\sem{S}} \arrownot\xlongrightarrow{d}{}^{\!\! \sem{S}} q_2^{\sem{S}}$. From Definition~\ref{def:quotientTIOTS} we have that in this case $q^Y \xlongrightarrow{d}{}^{\!\! Y} u$, and $q_2^Y \neq u$. Thus this case is infeasible.}
				\end{itemize}
				
				\item \new{$d\in\mathbb{R}_{\geq 0}$, $q_2^Y = (q_2^{\sem{T}},q_2^{\sem{S}})$, $q_3^Y = (q_3^{\sem{T}},q_3^{\sem{S}})$, $q_2^{\sem{T}} \xlongrightarrow{d}{}^{\!\! \sem{T}} q_3^{\sem{T}}$, and $q_2^{\sem{S}}\allowbreak {\xlongrightarrow{d}{}^{\!\! \sem{S}}} q_3^{\sem{S}}$. This case is infeasible, since $i? \neq d$.
				}
				
				\item \new{$i!\in \mathit{Act}_o^S$, $q_2^Y = (q^{\sem{T}},q^{\sem{S}})$, $q_3^Y = u$, $q^{\sem{T}} \in Q^{\sem{T}}$, and $q^{\sem{S}}\arrownot \xlongrightarrow{i!}{}^{\!\! \sem{S}}$. From Definition~\ref{def:semanticTIOA} of semantic it follows that $q^{\sem{T}} = (l^T,v^T)$ and $q^{\sem{S}} = (l^S,v^S)$. There are two reasons why $q^{\sem{S}}\arrownot \xlongrightarrow{i!}{}^{\!\! \sem{S}}$: there might be no edge in $E^S$ labeled with action $i!$ from location $l^S$ or none of the edges labeled with $i!$ from $l^S$ are enabled. An edge $(l^S,i!,\varphi, c, l^{S\prime}) \in E^S$ is not enabled if $v^S\not\models \varphi$ or $v^S[r\mapsto 0]_{r\in c} \not\models \mathit{Inv}(l^{S\prime})$ (or both), which can also be written as $v^S\not\models \varphi \wedge \mathit{Inv}(l^{S\prime})[r\mapsto 0]_{r\in c}$. Looking at the third rule in Definition~\ref{def:quotientTIOA} of the quotient for TIOA, we have that $((l^T,l^S),i?,\neg G_S, \emptyset, l_u) \in E^{T\quotient S}$ and $v^S\not\models G_S$, or $v^S\models \neg G_S$. Because $\mathit{Clk}^T\cap\mathit{Clk}^S = \emptyset$, it holds that $v\models \neg G_S$. 
				}
				
				\new{Now, since $\mathit{Inv}(l_u) = \textbf{T}$ and no clocks are reset, it holds that $v[r\mapsto 0]_{r\in\emptyset} = v \models \mathit{Inv}(l_u)$. From Definition~\ref{def:semanticTIOA} it now follows that $((l^T,l^S), v) \xrightarrow{i?} (l_u, v_3)$ is a transition in $\sem{T\quotient S}$. From the state label renaming function $f$ from Lemma~\ref{lemma:quotientTSandAsamestateset} we have that $q_3^X = f((l_u, v_3)) = u = q_3^Y$ and $q_2^X = q_2^Y$. And from Definition~\ref{def:quotientreducedquotient} of the reduced $\sim$-quotient of $\sem{T\quotient S}$ we have that $q_2^Y \xrightarrow{i?} q_3^Y$ is a transition in $\sem{T\quotient S}^{\rho} = X$.
				}
				
				\item \new{$d\in\mathbb{R}_{\geq 0}$, $q_2^Y = (q^{\sem{T}},q^{\sem{S}})$, $q_3^Y = u$, $q^{\sem{T}} \in Q^{\sem{T}}$, and $q^{\sem{S}}\arrownot \xlongrightarrow{d}{}^{\!\! \sem{S}}$. This case is infeasible, since $i? \neq d$.
				}
				
				\item \new{$i!\in \mathit{Act}_o^S \cap\mathit{Act}_o^T$, $q_2^Y = (q^{\sem{T}},q^{\sem{S}})$, $q_3^Y = e$, $q^{\sem{T}} \arrownot\xlongrightarrow{a}{}^{\!\! \sem{T}}$, and $q^{\sem{S}} \xlongrightarrow{a}{}^{\!\! \sem{S}}$. Since the target location is the error location, it holds that $q_3\notin P$. Thus this case is not feasible.
				}
				
				\item \new{$i\in \mathit{Act}^T \cup \mathit{Act}^S \cup \mathbb{R}_{\geq 0}$, $q_2^Y = u$, $q_3^Y = u$. There are two cases how $q_2^Y = u$ could have been reached by a delay.
				}
				\begin{itemize}
					\item \new{$q^Y = u$. In this case, it follows directly from Definition~\ref{def:quotientTIOA} that $(l_u, i?, \mathbf{T}, \emptyset, l_u)\in E^{T\quotient S}$. Since any valuation satisfies a true guard and by definition of $T\quotient S$ that $\mathit{Inv}(l_u) = \mathbf{T}$, we have with Definition~\ref{def:semanticTIOA} of semantic that $(l_u, v) \xrightarrow{i?} (l_u, v)$ is a transition in $\sem{T\quotient S}$. From the state label renaming function $f$ from Lemma~\ref{lemma:quotientTSandAsamestateset} we have that $q_2^X = q_2^Y$ and $q_3^X = f((l_u, v)) = u = q_3^Y$. And from Definition~\ref{def:quotientreducedquotient} of the reduced $\sim$-quotient of $\sem{T\quotient S}$ we have that $q_2^Y \xrightarrow{i?} q_3^Y$ is a transition in $\sem{T\quotient S}^{\rho} = X$.
					}
					
					\item \new{$q^Y = (l^T, v^T, l^S, v^S) \in Q^Y$ with $v^S + d \not\models \mathit{Inv}(l^S)$. In this case, it follows from Definitions~\ref{def:quotientTIOA},~\ref{def:semanticTIOA}, and~\ref{def:quotientreducedquotient} that $q^Y \xrightarrow{d} (l^T, l^S, v + d)$ in $X$. Furthermore, it follows directly from Definition~\ref{def:quotientTIOA} that $((l^T, l^S), i?, \neg \mathit{Inv}(l^S), \emptyset, l_u)\in E^{T\quotient S}$. Since $v^S + d \not\models \mathit{Inv}(l^S)$, we have $v^S + d \models \neg \mathit{Inv}(l^S)$. By definition of $T\quotient S$ we have that $\mathit{Inv}(l_u) = \mathbf{T}$, thus $v + d[r\mapsto 0]_{r\in \emptyset} = v + d \models \mathit{Inv}(l_u)$.  Now, with Definition~\ref{def:semanticTIOA} of semantic we it follows that $(l_u, v + d) \xrightarrow{i?} (l_u, v + d)$ is a transition in $\sem{T\quotient S}$. From the state label renaming function $f$ from Lemma~\ref{lemma:quotientTSandAsamestateset} we have that $q_3^X = f((l_u, v + d)) = u = q_3^Y$. And from Definition~\ref{def:quotientreducedquotient} of the reduced $\sim$-quotient of $\sem{T\quotient S}$ we have that $q_2^Y \xrightarrow{i?} q_3^Y$ is a transition in $\sem{T\quotient S}^{\rho} = X$. }
				\end{itemize}
				
				\item \new{$a\in \mathit{Act}_i^T \cup \mathit{Act}_o^S$, $q_2^Y = e$, $q_3^Y = e$. Since the target location is the error location, it holds that $q_3\notin P$. Thus this case is not feasible.}
			\end{enumerate}
			
			\new{Thus, in all feasible cases we can show that $q_2\xlongrightarrow{i?}{}^{\!\! Y} q_3$ implies $q_2 \xlongrightarrow{i?}{}^{\!\! X} q_3$. Since we have chosen an arbitrarily $i?\in\mathit{Act}_i^Y$, it holds for all $i?\in \mathit{Act}_i^Y$. }
			
			\new{It remains to be shown that $q_2 \xlongrightarrow{i_{\mathit{new}}}{}^{\!\! X} q_3$ and $q_3\in P$, since $i_{\mathit{new}}\notin\mathit{Act}_i^Y$. We only have to consider five cases from Definition~\ref{def:quotientTIOA} that involve $i_{\mathit{new}}$ (rule 4, 6, 7, 9, and 10). Using the same arguments as in these cases when we were considering $\Theta^X(P) \subseteq \Theta^Y(P)$ we can conclude that $q_3\in P$ in all feasible cases for $i_{\mathit{new}}$. Thus the implication also holds for $q_2$ in $X$.}
		\end{itemize}
		
		\new{Thus, in both cases the implication holds. Therefore, we can conclude that $q^Y\xlongrightarrow{d}{}^{\!\! X} q_2 \Rightarrow q_2\in P \wedge \forall i?\in \mathit{Act}_i^X: \exists q_3\in P: q_2\xlongrightarrow{i?}{}^{\!\! X} q_3$. As $q_2$ is chosen arbitrarily, it holds for all $q_2\in Q^X = Q^Y$. Therefore, the left-hand side is true.}
		
		\item \new{Assume the right-hand side is true, i.e., $\exists d'\leq d \wedge \exists q_2,q_3\in P \wedge \exists o!\in\mathit{Act}_o^Y: q\xlongrightarrow{d'}{}^{\!\! Y} q_2 \wedge q_2\xlongrightarrow{o!}{}^{\!\! Y} q_3 \wedge \forall i?\in\mathit{Act}_i^Y: \exists q_4\in P: q_2\xlongrightarrow{i?}{}^{\!\! Y} q_4$. First, focus on the delay. From Definition~\ref{def:quotientTIOA} of the quotient for TIOA it follows that $\mathit{Inv}((l^T, l^S)) = \mathbf{T}$. Therefore, with Definition~\ref{def:semanticTIOA} of the semantic and Definition~\ref{def:quotientreducedquotient} of the $\sim$-reduced quotient of $\sem{T\quotient S}$ it follows that $q^Y\xlongrightarrow{d}{}^{\!\! X} q_2$. 
		}
		
		\new{Now, consider the output transition labeled with $o!$. Remember that $\mathit{Act}_o^Y = \mathit{Act}_o^X = \mathit{Act}_o^T\setminus\mathit{Act}_o^S \cup \mathit{Act}_i^S\setminus\mathit{Act}_i^T$. We have to consider the nine cases from Definition~\ref{def:quotientTIOTS}. We can use the exact same argument as before (where now rules 5 and 7 have become infeasible) to show that $q_2\xlongrightarrow{o!}{}^{\!\! X} q_3$ is a transition in $X$ for all feasible cases. Since we have chosen an arbitrarily $o!\in\mathit{Act}_o^Y$, it holds for all $o!\in \mathit{Act}_o^Y$. Therefore, we can conclude that ${q^Y\xlongrightarrow{d'}{}^{\!\! X} q_2} \wedge q_2\xlongrightarrow{o!}{}^{\!\! X} q_3$ with $q_2, q_3 \in P$.}
		
		\new{Finally, consider the input transitions labeled with $i?$. Using the same argument as before, we can show that $q_2\xrightarrow{i?} q_4$ in $Y$ is also a transition in $X$, and $q_4\in P$. Therefore, we can conclude that ${q^Y\xlongrightarrow{d'}{}^{\!\! X} q_2} \wedge q_2\xlongrightarrow{o!}{}^{\!\! X} q_3 \wedge \forall i?\in\mathit{Act}_i^X: \exists q_4\in P: q_2\xlongrightarrow{i?}{}^{\!\! X} q_4$ with $q_2, q_3, q_4 \in P$. Thus, the right-hand side is true.}
	\end{itemize}
	\new{Thus, we have shown that when the left-hand side is true for $q^Y$ in $Y$, it is also true for $q^Y$ in $X$; and that when the right-hand side is true for $q^Y$ in $Y$, it is also true for $q^Y$ in $X$. Thus, $q^Y\in\Theta^X(P)$. Since $q^Y\in P$ was chosen arbitrarily, it holds for all states in $P$. Once we choose $P$ to be the fixed-point of $\Theta^Y$, we have that $\Theta^Y(P)\subseteq\Theta^X(P)$.}
\end{proof}

\new{Finally, we are ready to proof Theorem~\ref{thrm:quotientTSandA}.}

\newtheorem*{T17}{Lemma~\ref{thrm:quotientTSandA}}
\begin{T17}
	Given specification automata $S = (\mathit{Loc}^S,l_0^S,\mathit{Act}^S,\mathit{Clk}^S, E^S,\mathit{Inv}^S)$ and $T = (\mathit{Loc}^T,l_0^T,\allowbreak\mathit{Act}^T,\mathit{Clk}^T, E^T,\allowbreak\mathit{Inv}^T)$ where $\mathit{Act}_o^S\cap\mathit{Act}_i^T=\emptyset$. Then \new{$(\sem{T\quotient S})^{\Delta} \simeq (\sem{T}  \quotient \sem{S})^{\Delta}$}.
\end{T17}
\begin{proof}[Proof of Theorem~\ref{thrm:quotientTSandA}]
	\new{First, observe that the semantic of a TIOA and adversarial pruning do not alter the action set. Therefore, it follows directly that $(\sem{T\quotient S})^{\Delta}$ and $(\sem{T} \quotient \sem{S})^{\Delta}$ have the same action set and partitioning into input and output actions, except that $(\sem{T\quotient S})^{\Delta}$ has an additional input event $i_{\mathit{new}}$, i.e., $\mathit{Act}^{\sem{T\quotient S}} \cup \{i_{\mathit{new}}\} = \mathit{Act}^{\sem{T} \quotient \sem{S}}$. 
	}
	
	\new{Now, it follows from Lemma~\ref{lemma:quotientadversarialreduced} that it suffice to show that $(\sem{T\quotient S}^{\rho})^{\Delta} \simeq (\sem{T} \quotient \sem{S})^{\Delta}$. It follows from Lemma~\ref{lemma:quotientTSandAsamestateset} that there is a bijective function $f$ relating states from $\sem{T\quotient S}^{\rho}$ and $\sem{T} \quotient \sem{S}$ together. Therefore, we can effectively say that they have the same state set (up to relabeling), i.e., $Q^{\sem{T\quotient S}^{\rho}} = Q^{\sem{T\quotient S}}$. For brevity, in the rest of this proof we write we write $X=\sem{T\quotient S}^{\rho}$, $Y = \sem{T} \quotient \sem{S}$, $\mathit{Clk} = \mathit{Clk}^T \uplus\mathit{Clk}^S$, and $v^S$ and $v^T$ to indicate the part of a valuation $v$ of only the clocks of $S$ and $T$, respectively. Note that $x_{\mathit{new}}\notin\mathit{Clk}$, but $x_{\mathit{new}} \in \mathit{Clk}^X$.
	}
	
	\new{Let $A = \{ q \in Q^{X^{\Delta}} \mid q = ((l^T, l^S), v), v \not\models \mathit{Inv}(l^S) \}$. Let $R\subseteq Q^{X^{\Delta}} \times Q^{Y^{\Delta}}$ such that $R = \{ (q, u) \mid q \in A\} \cup \{(q^X,q^Y) \in Q^{X^{\Delta}} \setminus A \times Q^{Y^{\Delta}} \mid q^X = q^Y \}$. We will show that $R$ is a bisimulation relation. First, observe that $(q_0, q_0) \in R$. Consider a state pair $(q_1^X, q_1^Y) \in R$. We have to check whether the six cases from Definition~\ref{def:bisimulation} of bisimulation hold.
	}
	
	\begin{itemize}
		\item \new{$q_1^X\xlongrightarrow{a}{}^{\!\! X^{\Delta}} q_2^X$, $q_2^X\in Q^X$, and $a \in \mathit{Act}^X\cap\mathit{Act}^Y$. Combining Definitions~\ref{def:adversarialpruning},~\ref{def:quotientTIOTS} and~\ref{def:quotientTIOA} it follows that $a\in\mathit{Act}^S\cup\mathit{Act}^T$. From Definition~\ref{def:adversarialpruning} of adversarial pruning we have that $q_1^X\xlongrightarrow{a}{}^{\!\! X} q_2^X$ and $q_1^X, q_2^X\in\mathrm{cons}^X$. Following Definition~\ref{def:semanticTIOA} of the semantic and Definition~\ref{def:quotientreducedquotient} of the reduced $\sim$-quotient of $\sem{T\quotient S}$, it follows that there exists an edge $(l_1,a,\varphi, c, l_2)\in E^{T\quotient S}$ with $q_1^X = (l_1,v_1)$, $q_2^X = (l_2,v_2)$, $l_1,l_2 \in \mathit{Loc}^{T\quotient S}$, $v_1, v_2 \in [\mathit{Clk} \mapsto \mathbb{R}_{\geq 0}]$, $v_1\models \varphi$, $v_2 = v_1[r\mapsto 0]_{r\in c}$, and $v_2\models \mathit{Inv}(l_2)$. Now, consider the ten cases from Definition~\ref{def:quotientTIOA} of quotient of TIOAs. We have to show for feasible each case that we can simulate a transition in $Y$, that the involved states in $Y$ are consistent, and that the resulting state pair is again in the bisimulation relation $R$.
		}
		
		\begin{enumerate}
			\item \new{$a\in\mathit{Act}^S\cap\mathit{Act}^T$, $l_1 = (l_1^T,l_1^S)$, $l_2 = (l_2^T,l_2^S)$, $\varphi = \varphi^T\wedge \mathit{Inv}(l_2^T)[r\mapsto 0]_{r\in c^T}\wedge\varphi^S \wedge \mathit{Inv}(l_1^S) \wedge \mathit{Inv}(l_2^S)[r\mapsto 0]_{r\in c^S}$, $c = c^T\cup c^S$, $(l_1^T, a, \varphi^T,c^T, l_2^T)\in E^T$, and $(l_1^S, a, \varphi^S, c^S, l_2^S)\in E^S$. Since $v_1\models \varphi$, it holds that $v_1\models \varphi^T$, $v_1\models \mathit{Inv}(l_2^T)[r\mapsto 0]_{r\in c^T}$,  $v_1\models \varphi^S$, $v_1\models \mathit{Inv}(l_1^S)$, and $v_1 \models \mathit{Inv}(l_2^S)[r\mapsto 0]_{r\in c^S}$. Because $\mathit{Clk}^S\cap\mathit{Clk}^T=\emptyset$, it holds that $v_1^T\models \varphi^T$, $v_1^T\models \mathit{Inv}(l_2^T)[r\mapsto 0]_{r\in c^T}$, $v_1^S\models \varphi^S$, $v_1^S \models \mathit{Inv}(l_1^S)$, and $v_1^S \models \mathit{Inv}(l_2^S)[r\mapsto 0]_{r\in c^S}$. Since $v_2 = v_1[r\mapsto 0]_{r\in c}$, it holds that $v_2^T = v_1^T[r\mapsto 0]_{r\in c^T}$ and $v_2^S = v_1^S[r\mapsto 0]_{r\in c^S}$. Therefore, $v_2^T\models \mathit{Inv}(l_2^T)$ and $v_2^S\models\mathit{Inv}(l_2^S)$. 
			}
			
			\new{Combining all information about $T$, we have that $(l_1^T, a, \varphi^T,c^T, l_2^T)\in E^T$, $v_1^T\models \varphi^T$, $v_2^T = v_1^T[r\mapsto 0]_{r\in c^T}$, and $v_2^T\models \mathit{Inv}(l_2^T)$. Therefore, from Definition~\ref{def:semanticTIOA} it follows that $(l_1^T,v_1^T) \xlongrightarrow{a} (l_2^T,v_2^T)$ in $\sem{T}$. Combining all information about $S$, we have that $(l_1^S, a, \varphi^S,c^S, l_2^S)\in E^S$, $v_1^S\models \varphi^S$, $v_2^S = v_1^S[r\mapsto 0]_{r\in c^S}$, and $v_2^S\models \mathit{Inv}(l_2^S)$. Therefore, from Definition~\ref{def:semanticTIOA} it follows that $(l_1^S,v_1^S) \xlongrightarrow{a} (l_2^S,v_2^S)$ in $\sem{S}$. 
			}
			
			\new{Now, from Definition~\ref{def:quotientTIOTS} it follows that $((l_1^T,v_1^T),(l_1^S,v_1^S)) = (l_1^T, l_1^S, v_1) = q_1^Y \xlongrightarrow{a}{}^{\!\! Y} ((l_2^T,v_2^T),\allowbreak(l_2^S,v_2^S)) = (l_2^T, l_2^S, v_2) = q_2^Y$ in $Y$. Thus, we can simulate a transition in $Y$. Also, observe now that $q_1^X = q_1^Y$ and $q_2^X=q_2^Y$.
			}
			
			\item \new{$a\in\mathit{Act}^S\setminus\mathit{Act}^T$, $l_1 = (l^T,l_1^S)$, $l_2 = (l^T,l_2^S)$, $\varphi = \varphi^S \wedge \mathit{Inv}(l_1^S) \wedge \mathit{Inv}(l_2^S)[r\mapsto 0]_{r\in c^S}$, $c = c^S$, $l^T\in \mathit{Loc}^T$, and $(l_1^S, a, \varphi^S, c^S, l_2^S)\in E^S$. Since $v_1\models \varphi$ and $\mathit{Clk}^S\cap\mathit{Clk}^T=\emptyset$, it holds that $v_1^S\models \varphi^S$, $v_1^S\models \mathit{Inv}(l_1^S)$, and $v_1^S \models \mathit{Inv}(l_2^S)[r\mapsto 0]_{r\in c^S}$. Since $v_2 = v_1[r\mapsto 0]_{r\in c}$ and $c = c^S$, it holds that $v_2^S = v_1^S[r\mapsto 0]_{r\in c^S}$, $v_2^T = v_1^T$, and $v_2^S\models\mathit{Inv}(l_2^S)$. Combining all information above about $S$, it follows from Definition~\ref{def:semanticTIOA} that $(l_1^S,v_1^S) \xlongrightarrow{a} (l_2^S,v_2^S)$ in $\sem{S}$. From Definition~\ref{def:semanticTIOA} it also follows that $(l^T, v_1^T) \in Q^{\sem{T}}$. Therefore, following Definition~\ref{def:quotientTIOTS} it follows that $((l^T,v_1^T),(l_1^S,v_1^S)) = (l^T, l_1^S, v_1) = q_1^Y \xlongrightarrow{a}{}^{\!\! Y} ((l^T,v_1^T),(l_2^S,v_2^S)) = (l^T, l_2^S, v_2) = q_2^Y$ in $Y$. Thus, we can simulate a transition in $Y$. Also, observe now that $q_1^X = q_1^Y$ and $q_2^X=q_2^Y$.
			}
			
			\item \new{$a\in\mathit{Act}_o^S$, $l_1 = (l^T,l_1^S)$, $l_2 = l_u$, $\varphi = \neg G_S$, $c = \emptyset$, $l^T\in \mathit{Loc}^T$ and $G_S = \bigvee \{\varphi^S \wedge \mathit{Inv}(l_2^S)[r\mapsto 0]_{r\in c^S} \mid (l_1^S,a,\varphi^S,c^S,l_2^S) \in E^S\}$. Since $v_1\models \varphi$ and $\mathit{Clk}^S\cap\mathit{Clk}^T=\emptyset$, it holds that $v_1^S\models \neg G_S$. Therefore, $v_1^S\not\models G_S$, which indicates that $\forall (l_1^S,a,\varphi^S,c^S,l_2^S) \in E^S$: $v_1^S\not\models \varphi^S \wedge \mathit{Inv}(l_2^S)[r\mapsto 0]_{r\in c^S}$. This means that $v_1^S\not\models \varphi^S$ or $v_1^S\not\models \mathit{Inv}(l_2^S)[r\mapsto 0]_{r\in c^S}$ or both, where the second option is equivalent to $v_1^S[r\mapsto 0]_{r\in c^S} \not\models \mathit{Inv}(l_2^S)$. Following Definition~\ref{def:semanticTIOA}, we can conclude that $(l_1^S,v_1^S) \arrownot\xlongrightarrow{\ a\ }$ in $\sem{S}$. From Definition~\ref{def:semanticTIOA} it also follows that $(l^T, v_1^T) \in Q^{\sem{T}}$. Now, following Definition~\ref{def:quotientTIOTS}, we have transition $((l^T,v_1^T),(l_1^S,v_1^S)) = (l^T, l_1^S, v_1) = q_1^Y \xlongrightarrow{a}{}^{\!\! Y} u = q_2^Y$ in $Y$. Thus we can simulate a transition in $Y$. Also, observe now that $q_1^X = q_1^Y$ and $q_2^X=q_2^Y$ (where $(l_u, v_2)$ is mapped into $u$ by $f$ from Lemma~\ref{lemma:quotientTSandAsamestateset}). 
			}
			
			\item \new{$a \in \mathit{Act}^S\cup\mathit{Act}^T$, $l_1 = (l^T,l^S)$, $l_2 = l_u$, $\varphi = \neg \mathit{Inv}(l^S)$, $c = \emptyset$, $l^T\in \mathit{Loc}^T$, and $l^S\in\mathit{Loc}^S$. Since $v_1\models \varphi$ and $\mathit{Clk}^S\cap\mathit{Clk}^T=\emptyset$, it holds that $v_1^S\models \neg \mathit{Inv}(l^S)$. Therefore, $v_1^S\not\models \mathit{Inv}(l^S)$. Since $(q_1^X,q_1^Y) \in R$ and $v_1^S\not\models \mathit{Inv}(l^S)$, it follows that $q_1^Y = u$. From Definition~\ref{def:quotientTIOTS} it follows that $u = q_1^Y \xlongrightarrow{a}{}^{\!\! Y} u = q_2^Y$ in $Y$. Thus we can simulate a transition in $Y$. Also, observe now that $q_2^X=q_2^Y$ (where $(l_u, v_2)$ is mapped into $u$ by $f$ from Lemma~\ref{lemma:quotientTSandAsamestateset}).
			}
			
			\item \new{$a\in\mathit{Act}_o^S\cap\mathit{Act}_o^T$, $l_1 = (l_1^T,l_1^S)$, $l_2 = l_e$, $\varphi = \varphi^S \wedge \mathit{Inv}(l_1^S) \wedge \mathit{Inv}(l_2^S)[r\mapsto 0]_{r\in c^S} \wedge \neg G_T$, $c = \{x_{\mathit{new}}\}$, $(l_1^S, a, \varphi^S, c^S, l_2^S)\in E^S$, and $G_T = \bigvee \{\varphi^T \wedge \mathit{Inv}(l_2^T)[r\mapsto 0]_{r\in c^T} \mid (l_1^T,a,\varphi^T,c^T, l_2^T) \in E^T\}$. Since the target location is the error location, it holds that $q_2^X\notin\mathrm{cons}^X$. Thus this case is not feasible.
			}
			
			\item \new{$a = i_{\mathit{new}}$, $l_1 = (l^T,l^S)$, $l_2 = l_e$, $\varphi = \neg \mathit{Inv}(l^T) \wedge \mathit{Inv}(l^S)$, $c = \{x_{\mathit{new}}\}$, $l^T\in \mathit{Loc}^T$, and $q^S\in\mathit{Loc}^S$. This case is infeasible, since $i_{\mathit{new}} \notin \mathit{Act}^Y$, thus $i_{\mathit{new}}\notin \mathit{Act}^X\cap\mathit{Act}^Y$.
			}
			
			\item \new{$a = i_{\mathit{new}}$, $l_1 = l_2 = (l_1^T,l_1^S)$, $\varphi = \mathit{Inv}(l^T) \vee \neg\mathit{Inv}(l^S)$ and $c = \emptyset$. This case is infeasible, since $i_{\mathit{new}} \notin \mathit{Act}^Y$, thus $i_{\mathit{new}}\notin \mathit{Act}^X\cap\mathit{Act}^Y$.
			}
			
			\item \new{$a\in\mathit{Act}^T\setminus\mathit{Act}^S$, $l_1 = (l_1^T,l^S)$, $l_2 = (l_2^T,l^S)$, $\varphi = \varphi^T \wedge \mathit{Inv}(l_2^T)[r\mapsto 0]_{r\in c^T} \wedge \mathit{Inv}(l^S)$, $c = c^T$, $l^S\in \mathit{Loc}^S$, and $(l_1^T, a, \varphi^T,c^T, l_2^T)\in E^T$. Since $v_1\models \varphi$ and $\mathit{Clk}^S\cap\mathit{Clk}^T=\emptyset$, it holds that $v_1^T\models \varphi^T$ and $v_1^T \models \mathit{Inv}(l_2^T)[r\mapsto 0]_{r\in c^T}$. Since $v_2 = v_1[r\mapsto 0]_{r\in c}$ and $c = c^T$, it holds that $v_2^T = v_1^T[r\mapsto 0]_{r\in c^T}$, $v_2^S = v_1^S$, and $v_2^T\models\mathit{Inv}(l_2^T)$. Combining all information above about $T$, it follows from Definition~\ref{def:semanticTIOA} that $(l_1^T,v_1^T) \xlongrightarrow{a} (l_2^T,v_2^T)$ in $\sem{T}$. From Definition~\ref{def:semanticTIOA} it also follows that $(l^S, v_1^S) \in Q^{\sem{S}}$. Therefore, following Definition~\ref{def:quotientTIOTS} it follows that $((l_1^T,v_1^T),(l^S,v_1^S)) = (l_1^T, l^S, v_1) = q_1^Y \xlongrightarrow{a}{}^{\!\! Y} ((l_2^T,v_2^T),(l^S,v_1^S)) = (l_2^T, l^S, v_2) = q_2^Y$ in $Y$. Thus, we can simulate a transition in $Y$. Also, observe now that $q_1^X = q_1^Y$ and $q_2^X=q_2^Y$.
			}
			
			\item \new{$a\in\mathit{Act}^S\cup\mathit{Act}^T$, $l_1 = l_u$, $l_2 = l_u$, $\varphi = \mathbf{T}$, $c = \emptyset$. From the construction of the bisimulation relation $R$, we know that if $q_1^X = f((l_u,v_1)) = u$ for some valuation $v_1$, then $q_1^Y = u$. From Definition~\ref{def:quotientTIOTS} it follows directly that there exists a transition $q_1^Y = u \xlongrightarrow{a}{}^{\!\! Y} u = q_2^Y$ in $Y$. Thus, we can simulate a transition in $Y$. Also, observe now that $q_1^X = q_1^Y$ and $q_2^X=q_2^Y$.
			}
			
			\item \new{$a\in\mathit{Act}_i^S\cup\mathit{Act}_i^T$, $l_1 = l_e$, $l_2 = l_e$, $\varphi = x_{\mathit{new}}=0$, $c = \emptyset$. Since the source and target locations are the error location, it holds that $q_1^X, q_2^X\notin\mathrm{cons}^X$. Thus this case is not feasible.}
		\end{enumerate}
		
		\new{In all feasible cases we can show that $q_1^Y = q_1^X$ or $q_1^Y = u$ and $q_2^Y = q_2^X$. Since $q_1^X, q_2^X\in\mathrm{cons}^X$ and $u \in \mathrm{cons}^Y$ by construction of $u$, it follows from Lemma~\ref{lemma:quotientTSandAsamecons} that $q_1^Y, q_2^Y\in\mathrm{cons}^Y$. Therefore, we can conclude that $q_1^Y\xlongrightarrow{a}{}^{\!\! Y^{\Delta}} q_2^Y$. And from the construction of the bisimulation relation $R$ it follows that $(q_2^X,q_2^Y)\in R$.
		}
		
		\item \new{$q_1^X\xlongrightarrow{a}{}^{\!\! X^{\Delta}} q_2^X$, $q_2^X\in Q^X$, and $a = i_{\mathit{new}}$. From Definition~\ref{def:adversarialpruning} of adversarial pruning we have that $q_1^X\allowbreak {\xlongrightarrow{a}{}^{\!\! X}}\allowbreak q_2^X$ and $q_1^X, q_2^X\in\mathrm{cons}^X$. Following Definition~\ref{def:semanticTIOA} of the semantic, it follows that there exists an edge $(l_1,a,\varphi, c, l_2)\in E^{T\quotient S}$ with $q_1^X = (l_1,v_1)$, $q_2^X = (l_2,v_2)$, $l_1,l_2 \in \mathit{Loc}^{T\quotient S}$, $v_1, v_2 \in [\mathit{Clk} \mapsto \mathbb{R}_{\geq 0}]$, $v_1\models \varphi$, $v_2 = v_1[r\mapsto 0]_{r\in c}$, and $v_2\models \mathit{Inv}(l_2)$. There are three cases from Definition~\ref{def:quotientTIOA} of the quotient for TIOA that apply here.
		}
		
		\begin{itemize}
			\item \new{$l_1 = (l^T,l^S)$, $l_2 = l_e$, $\varphi = \neg \mathit{Inv}(l^T) \wedge \mathit{Inv}(l^S)$, $c = \{x_{\mathit{new}}\}$, $l^T\in \mathit{Loc}^T$, and $q^S\in\mathit{Loc}^S$. Since the target location is the error location, it holds that $q_2^X\notin\mathrm{cons}^X$. Thus this case is not feasible. 
			}
			
			\item \new{$l_1 = l_2 = (l_1^T,l_1^S)$, $\varphi = \mathit{Inv}(l^T) \vee \neg\mathit{Inv}(l^S)$ and $c = \emptyset$. Since $c = \emptyset$, it follows that $v_2 = v_1$. Therefore, $q_1^X = q_2^X$. Following the second case of Definition~\ref{def:bisimulation} and knowing that $(q_1^X, q_1^Y) \in R$, if follows immediately that $(q_2^X,q_1^Y) \in R$. Since $q_1^X\in\mathrm{cons}^X$, it follows from the construction of $R$ and Lemma~\ref{lemma:quotientTSandAsamecons} that $q_1^y = q_1^X$ and thus $q_1^Y\in\mathrm{cons}^Y$.
			}
			
			\item \new{$l_1 = l_2$, $l_2 = l_e$, $\varphi = x_{\mathit{new}}$, and $c = \emptyset$. Since the source and target locations are the error location, it holds that $q_1^X,q_2^X\notin\mathrm{cons}^X$. Thus this case is not feasible. }
		\end{itemize}
		
		\item \new{$q_1^Y\xlongrightarrow{a}{}^{\!\! Y^{\Delta}} q_2^Y$, $q_2^Y\in Q^Y$, and $a\in\mathit{Act}^Y\cap\mathit{Act}^X$. Combining Definitions~\ref{def:adversarialpruning},~\ref{def:quotientTIOTS} and~\ref{def:quotientTIOA} it follows that $a\in\mathit{Act}^S\cup\mathit{Act}^T$. From Definition~\ref{def:adversarialpruning} of adversarial pruning we have that $q_1^Y\xlongrightarrow{a}{}^{\!\! Y} q_2^Y$ and $q_1^Y, q_2^Y\in\mathrm{cons}^Y$. Now, consider the nine cases from Definition~\ref{def:quotientTIOTS} of the quotient of TIOTS. We have to show for each feasible case that we can simulate a transition in $X$, that the involved states in $X$ are consistent, and that the resulting state pair is again in the bisimulation relation $R$.
		}
		
		\begin{enumerate}
			\item \new{$a\in \mathit{Act}^S \cap\mathit{Act}^T$, $q_1^Y = (q_1^{\sem{T}},q_1^{\sem{S}})$, $q_2^Y = (q_2^{\sem{T}},q_2^{\sem{S}})$, $q_1^{\sem{T}} \xlongrightarrow{a}{}^{\!\! \sem{T}} q_2^{\sem{T}}$, and $q_1^{\sem{S}}\allowbreak {\xlongrightarrow{a}{}^{\!\! \sem{S}}} q_2^{\sem{S}}$. From Definition~\ref{def:semanticTIOA} of semantic it follows that there exists an edge $(l_1^T,a,\varphi^T, c^T, l_2^T)\in E^T$ with $q_1^{\sem{T}} = (l_1^T,v_1^T)$, $q_2^{\sem{T}} = (l_2^T,v_2^T)$, $l_1^T,l_2^T \in \mathit{Loc}^T$, $v_1^T, v_2^T \in [\mathit{Clk}^T \mapsto \mathbb{R}_{\geq 0}]$, $v_1^T\models \varphi^T$, $v_2^T = v_1^T[r\mapsto 0]_{r\in c^T}$, and $v_2^T\models \mathit{Inv}^T(l_2^T)$. Similarly, it follows from the same definition that there exists an edge $(l_1^S,a,\varphi^S, c^S, l_2^S)\in E^S$ with $q_1^{\sem{S}} = (l_1^S,v_1^S)$, $q_2^{\sem{S}} = (l_2^S,v_2^S)$, $l_1^S,l_2^S \in \mathit{Loc}^S$, $v_1^S, v_2^S \in [\mathit{Clk}^S \mapsto \mathbb{R}_{\geq 0}]$, $v_1^S\models \varphi^S$, $v_2^S = v_1^S[r\mapsto 0]_{r\in c^S}$, and $v_2^S\models \mathit{Inv}^S(l_2^S)$. Based on Definition~\ref{def:quotientTIOA} of the quotient for TIOA, we need to consider the following two cases.
			}
			
			\begin{itemize}
				\item \new{$v_1^S \models \mathit{Inv}(l_1^S)$. In this case, there exists an edge $((l_1^T, l_1^S),a,\varphi^T \wedge \mathit{Inv}(l_2^T)[r\mapsto 0]_{r\in c^T} \wedge \varphi^S \wedge \mathit{Inv}(l_1^S) \wedge \mathit{Inv}(l_2^S)[r\mapsto 0]_{r\in c^S}, c^T \cup c^S, (l_2^T,l_2^S))$ in $T\quotient S$. Let $v_i, i=1,2$ be the valuations that combines the one from $T$ with the one from $S$, i.e. $\forall r \in \mathit{Clk}^T: v_i(r) = v_i^T(r)$ and $\forall r \in \mathit{Clk}^S: v_i(r) = v_i^S(r)$. Because $\mathit{Clk}^T\cap\mathit{Clk}^S = \emptyset$, it holds that $v_1\models \varphi^T$, $v_1\models \varphi^S$, and $v_1^S \models \mathit{Inv}(l_1^S)$, thus $v_1\models \varphi^T\wedge\varphi^S \wedge \mathit{Inv}(l_1^S)$; $v_2 = v_1[r\mapsto 0]_{r\in c^T\cup c^S}$; and $v_2\models\mathit{Inv}^T(l_2^T)$ and $v_2\models\mathit{Inv}^S(l_2^S)$, thus $v_2\models\mathit{Inv}^T(l_2^T)\wedge\mathit{Inv}^S(l_2^S)$. 
				}
				
				\new{From Definition~\ref{def:semanticTIOA} it now follows that $((l_1^T,l_1^S), v_1) \xrightarrow{a} ((l_2^T,l_2^S), v_2)$ is a transition in $\sem{T\quotient S}$. Because $\mathit{Clk}^T\cap\mathit{Clk}^S = \emptyset$, we can rearrange the states into $((l_1^T,l_1^S),v_1) = ((l_1^T,v_1^S) , (l_1^T,v_1^S)) = q_1^Y$ and $((l_2^T,l_2^S),v_2) = ((l_2^T,v_2^T) , (l_2^S,v_2^S)) = q_2^Y$. Thus, $q_1^Y \xrightarrow{a} q_2^Y$ is a transition in $\sem{T\quotient S} = Y$. Also, observe now that $q_1^X = q_1^Y$ and $q_2^X=q_2^Y$.
				}
				
				\item \new{$v_1^S \not\models \mathit{Inv}(l_1^S)$. From the construction of $R$, it follows that $((l_1^T, l_1^S, v_1), u) \in R$, i.e. $q_1^Y = u$. This contradicts with the start of this case that $q_2^Y = (q_2^{\sem{T}},q_2^{\sem{S}})$. Thus this case is infeasible.}
			\end{itemize}
			
			\item \new{$a\in \mathit{Act}^S \setminus \mathit{Act}^T$, $q_1^Y = (q^{\sem{T}},q_1^{\sem{S}})$, $q_2^Y = (q^{\sem{T}},q_2^{\sem{S}})$, $q^{\sem{T}}\in Q^{\sem{T}}$, and $q_1^{\sem{S}}\xlongrightarrow{a}{}^{\!\! \sem{S}} q_2^{\sem{S}}$. From Definition~\ref{def:semanticTIOA} of semantic it follows that there exists an edge $(l_1^S,a,\varphi^S, c^S, l_2^S)\in E^S$ with $q_1^{\sem{S}} = (l_1^S,v_1^S)$, $q_2^{\sem{S}} = (l_2^S,v_2^S)$, $l_1^S,l_2^S \in \mathit{Loc}^S$, $v_1^S, v_2^S \in [\mathit{Clk}^S \mapsto \mathbb{R}_{\geq 0}]$, $v_1^S\models \varphi^S$, $v_2^S = v_1^S[r\mapsto 0]_{r\in c^S}$, and $v_2^S\models \mathit{Inv}^S(l_2^S)$. From the same definition, it follows that $q^{\sem{T}} = (l^T, v^T)$ for some $l^T\in\mathit{Loc}^T$ and $v^T \in [\mathit{Clk}^T \mapsto \mathbb{R}_{\geq 0}]$. Based on Definition~\ref{def:quotientTIOA} of the quotient for TIOA, we need to consider the following two cases.
			}
			
			\begin{itemize}
				\item \new{$v_1^S \models \mathit{Inv}(l_1^S)$. In this case, there exists an edge $((l^T, l_1^S),a, \varphi^S \wedge \mathit{Inv}(l_1^S) \wedge \mathit{Inv}(l_2^S)[r\mapsto 0]_{r\in c^S}, c^S,\allowbreak (l^T,l_2^S))$ in $T\quotient S$. Let $v_i, i=1,2$ be the valuations that combines the one from $T$ with the one from $S$, i.e. $\forall r \in \mathit{Clk}^T: v_i(r) = v_i^T(r)$ and $\forall r \in \mathit{Clk}^S: v_i(r) = v_i^S(r)$. Because $\mathit{Clk}^T\cap\mathit{Clk}^S = \emptyset$, it holds that $v_1\models \varphi^S$, and $v_1 \models \mathit{Inv}(l_1^S)$, thus $v_1\models \varphi^S \wedge \mathit{Inv}(l_1^S)$; $v_2 = v_1[r\mapsto 0]_{r\in c^S}$; and $v_2\models\mathit{Inv}^S(l_2^S)$.
				}
				
				\new{Since $\mathit{Inv}((l^T, l_2^S)) = \mathbf{T}$ by definition $T\quotient S$, we have that $v_2\models \mathit{Inv}((l^T, l_2^S))$. From Definition~\ref{def:semanticTIOA} it now follows that $((l^T,l_1^S), v_1) \xrightarrow{a} ((l^T,l_2^S), v_2)$ is a transition in $\sem{T\quotient S}$. Using Definition~\ref{def:quotientreducedquotient} of the reduced $\sim$-quotient of $\sem{T\quotient S}$ and Lemma~\ref{lemma:quotientTSandAsamestateset}, we can rearrange the states into $((l^T,l_1^S),v_1) = ((l^T,v_1^T) , (l_1^S,v_1^S)) = q_1^Y$ and $((l^T,l_2^S),v_2) = ((l^T,v_2^T) , (l_2^S,v_2^S)) = q_2^Y$, and we can show that $q_1^Y \xrightarrow{a} q_2^Y$ is a transition in $\sem{T\quotient S}^{\rho} = X$. Also, observe now that $q_1^X = q_1^Y$ and $q_2^X=q_2^Y$.
				}
				
				\item \new{$v_1^S \not\models \mathit{Inv}(l_1^S)$. From the construction of $R$, it follows that $((l_1^T, l_1^S, v_1), u) \in R$, i.e. $q_1^Y = u$. This contradicts with the start of this case that $q_2^Y = (q_2^{\sem{T}},q_2^{\sem{S}})$. Thus this case is infeasible.}
			\end{itemize}
			
			\item \new{$a\in \mathit{Act}^T \setminus \mathit{Act}^S$, $q_1^Y = (q_1^{\sem{T}},q^{\sem{S}})$, $q_2^Y = (q_2^{\sem{T}},q^{\sem{S}})$, $q^{\sem{S}}\in Q^{\sem{S}}$, and $q_1^{\sem{T}}\xlongrightarrow{a}{}^{\!\! \sem{T}} q_2^{\sem{T}}$. 
				From Definition~\ref{def:semanticTIOA} of semantic it follows that there exists an edge $(l_1^T,a,\varphi^T, c^T, l_2^T)\in E^T$ with $q_1^{\sem{T}} = (l_1^T,v_1^T)$, $q_2^{\sem{T}} = (l_2^T,v_2^T)$, $l_1^T,l_2^T \in \mathit{Loc}^T$, $v_1^T, v_2^T \in [\mathit{Clk}^T \mapsto \mathbb{R}_{\geq 0}]$, $v_1^T\models \varphi^T$, $v_2^T = v_1^T[r\mapsto 0]_{r\in c^T}$, and $v_2^T\models \mathit{Inv}^T(l_2^T)$. From the same definition, it follows that $q^{\sem{S}} = (l^S, v^S)$ for some $l^S\in\mathit{Loc}^S$ and $v^S \in [\mathit{Clk}^S \mapsto \mathbb{R}_{\geq 0}]$. Based on Definition~\ref{def:quotientTIOA} of the quotient for TIOA, we need to consider the following two cases.
			}
			
			\begin{itemize}
				\item \new{$v_1^S \models \mathit{Inv}(l_1^S)$. In this case, there exists an edge $((l_1^T, l^S),a, \varphi^T \wedge \mathit{Inv}(l_2^T)[r\mapsto 0]_{r\in c^T} \wedge \mathit{Inv}(l^S), c^T,\allowbreak (l_2^T,l^S))$ in $T\quotient S$. Let $v_i, i=1,2$ be the valuations that combines the one from $T$ with the one from $S$, i.e. $\forall r \in \mathit{Clk}^T: v_i(r) = v_i^T(r)$ and $\forall r \in \mathit{Clk}^S: v_i(r) = v_i^S(r)$. Because $\mathit{Clk}^T\cap\mathit{Clk}^S = \emptyset$, it holds that $v_1\models \varphi^T$, and $v_1 \models \mathit{Inv}(l^S)$, thus $v_1\models \varphi^T \wedge \mathit{Inv}(l^S)$; $v_2 = v_1[r\mapsto 0]_{r\in c^T}$; and $v_2\models\mathit{Inv}^T(l_2^T)$. 
				}
				
				\new{Since $\mathit{Inv}((l_2^T, l^S)) = \mathbf{T}$ by definition $T\quotient S$, we have that $v_2\models \mathit{Inv}((l_2^T, l^S))$. From Definition~\ref{def:semanticTIOA} it now follows that $((l_1^T,l^S), v_1) \xrightarrow{a} ((l_2^T,l^S), v_2)$ is a transition in $\sem{T\quotient S}$. Using Definition~\ref{def:quotientreducedquotient} of the reduced $\sim$-quotient of $\sem{T\quotient S}$ and Lemma~\ref{lemma:quotientTSandAsamestateset}, we can rearrange the states into $((l_1^T,l^S),v_1) = ((l_1^T,v_1^T) , (l^S,v_1^S)) = q_1^Y$ and $((l_2^T,l^S),v_2) = ((l_2^T,v_2^T) , (l^S,v_2^S)) = q_2^Y$, and we can show that $q_1^Y \xrightarrow{a} q_2^Y$ is a transition in $\sem{T\quotient S}^{\rho} = X$. Also, observe now that $q_1^X = q_1^Y$ and $q_2^X=q_2^Y$.
				}
				
				\item \new{$v_1^S \not\models \mathit{Inv}(l_1^S)$. From the construction of $R$, it follows that $((l_1^T, l^S, v_1), u) \in R$, i.e. $q_1^Y = u$. This contradicts with the start of this case that $q_2^Y = (q_2^{\sem{T}},q_2^{\sem{S}})$. Thus this case is infeasible.}
			\end{itemize}
			
			\item \new{$d\in\mathbb{R}_{\geq 0}$, $q_1^Y = (q_1^{\sem{T}},q_1^{\sem{S}})$, $q_2^Y = (q_2^{\sem{T}},q_2^{\sem{S}})$, $q_1^{\sem{T}} \xlongrightarrow{d}{}^{\!\! \sem{T}} q_2^{\sem{T}}$, and $q_1^{\sem{S}}\allowbreak {\xlongrightarrow{d}{}^{\!\! \sem{S}}} q_2^{\sem{S}}$. This case is infeasible, since $a \neq d$ (delays will be treated later in the proof).
			}
			
			\item \new{$a\in \mathit{Act}_o^S$, $q_1^Y = (q^{\sem{T}},q^{\sem{S}})$, $q_2^Y = u$, $q^{\sem{T}} \in Q^{\sem{T}}$, and $q^{\sem{S}}\arrownot \xlongrightarrow{a}{}^{\!\! \sem{S}}$. From Definition~\ref{def:semanticTIOA} of semantic it follows that $q^{\sem{T}} = (l^T,v^T)$ and $q^{\sem{S}} = (l^S,v^S)$. There are two reasons why $q^{\sem{S}}\arrownot \xlongrightarrow{a}{}^{\!\! \sem{S}}$: there might be no edge in $E^S$ labeled with action $a$ from location $l^S$ or none of the edges labeled with $a$ from $l^S$ are enabled. An edge $(l^S,a,\varphi, c, l^{S\prime}) \in E^S$ is not enabled if $v^S\not\models \varphi$ or $v^S[r\mapsto 0]_{r\in c} \not\models \mathit{Inv}(l^{S\prime})$ (or both), which can also be written as $v^S\not\models \varphi \wedge \mathit{Inv}(l^{S\prime})[r\mapsto 0]_{r\in c}$. Looking at the third rule in Definition~\ref{def:quotientTIOA} of the quotient for TIOA, we have that $((l^T,l^S),a,\neg G_S, \emptyset, l_u) \in E^{T\quotient S}$ and $v^S\not\models G_S$, or $v^S\models \neg G_S$. Because $\mathit{Clk}^T\cap\mathit{Clk}^S = \emptyset$, it holds that $v\models \neg G_S$. 
			}
			
			\new{Now, since $\mathit{Inv}(l_u) = \textbf{T}$ and no clocks are reset, it holds that $v[r\mapsto 0]_{r\in\emptyset} = v \models \mathit{Inv}(l_u)$. From Definition~\ref{def:semanticTIOA} it now follows that $((l^T,l^S), v) \xrightarrow{a} (l_u, v_2)$ is a transition in $\sem{T\quotient S}$. From the state label renaming function $f$ from Lemma~\ref{lemma:quotientTSandAsamestateset} we have that $q_2^X = f((l_u, v_2)) = u = q_2^Y$ and $q_1^X = q_1^Y$. And from Definition~\ref{def:quotientreducedquotient} of the reduced $\sim$-quotient of $\sem{T\quotient S}$ we have that $q_1^Y \xrightarrow{a} q_2^Y$ is a transition in $\sem{T\quotient S}^{\rho} = X$.
			}
			
			\item \new{$d\in\mathbb{R}_{\geq 0}$, $q_1^Y = (q^{\sem{T}},q^{\sem{S}})$, $q_2^Y = u$, $q^{\sem{T}} \in Q^{\sem{T}}$, and $q^{\sem{S}}\arrownot \xlongrightarrow{d}{}^{\!\! \sem{S}}$. This case is infeasible, since $a \neq d$ (delays will be treated later in the proof).
			}
			
			\item \new{$a\in \mathit{Act}_o^S \cap\mathit{Act}_o^T$, $q_1^Y = (q^{\sem{T}},q^{\sem{S}})$, $q_2^Y = e$, $q^{\sem{T}} \arrownot\xlongrightarrow{a}{}^{\!\! \sem{T}}$, and $q^{\sem{S}} \xlongrightarrow{a}{}^{\!\! \sem{S}}$. Since the target state is the error state, it holds that $q_2^Y\notin\mathrm{cons}^Y$. Thus this case is not feasible.
			}
			
			\item \new{$a\in \mathit{Act}^T \cup \mathit{Act}^S \cup \mathbb{R}_{\geq 0}$, $q_1^Y = u$, $q_2^Y = u$. From the construction of $R$ it follows that there are two options for $q_1^X$ for the pair $(q_1^X, u) \in R$.
			}
			\begin{itemize}
				\item \new{$q_1^X = u\ (= (l_u, v))$. In this case, it follows directly from Definition~\ref{def:quotientTIOA} that $(l_u, a, \mathbf{T}, \emptyset, l_u)\in E^{T\quotient S}$. Since any valuation satisfies a true guard and by definition of $T\quotient S$ that $\mathit{Inv}(l_u) = \mathbf{T}$, we have with Definition~\ref{def:semanticTIOA} of semantic that $(l_u, v) \xrightarrow{a} (l_u, v)$ is a transition in $\sem{T\quotient S}$. From the state label renaming function $f$ from Lemma~\ref{lemma:quotientTSandAsamestateset} we have that $q_1^X = q_1^Y$ and $q_2^X = f((l_u, v)) = u = q_2^Y$. And from Definition~\ref{def:quotientreducedquotient} of the reduced $\sim$-quotient of $\sem{T\quotient S}$ we have that $q_1^Y \xrightarrow{a} q_2^Y$ is a transition in $\sem{T\quotient S}^{\rho} = X$.
				}
				
				\item \new{$q_1^X = ((l^T, l^S), v) \in Q^{X^{\Delta}}$ with $v \not\models \mathit{Inv}(l^S)$. In this case, it follows directly from Definition~\ref{def:quotientTIOA} that $((l^T, l^S), a, \neg \mathit{Inv}(l^S), \emptyset, l_u)\in E^{T\quotient S}$. Since $v \not\models \mathit{Inv}(l^S)$, we have $v \models \neg \mathit{Inv}(l^S)$. By definition of $T\quotient S$ we have that $\mathit{Inv}(l_u) = \mathbf{T}$, thus $v[r\mapsto 0]_{r\in \emptyset} = v \models \mathit{Inv}(l_u)$.  Now, with Definition~\ref{def:semanticTIOA} of semantic we it follows that $(l_u, v) \xrightarrow{a} (l_u, v)$ is a transition in $\sem{T\quotient S}$. From the state label renaming function $f$ from Lemma~\ref{lemma:quotientTSandAsamestateset} we have that $q_2^X = f((l_u, v)) = u = q_2^Y$. And from Definition~\ref{def:quotientreducedquotient} of the reduced $\sim$-quotient of $\sem{T\quotient S}$ we have that $q_1^Y \xrightarrow{a} q_2^Y$ is a transition in $\sem{T\quotient S}^{\rho} = X$.}
			\end{itemize}
			
			\item \new{$a\in \mathit{Act}_i^T \cup \mathit{Act}_o^S$, $q_1^Y = e$, $q_2^Y = e$. Since the source and target states are the error state, it holds that $q_1^Y, q_2^Y\notin\mathrm{cons}^Y$. Thus this case is not feasible.}
		\end{enumerate}
		
		\new{In all feasible cases we can show that $q_1^X = q_1^Y$ or $q_1^X = ((l^T,l^S), v)$ with $v\not\models\mathit{Inv}(l^S)$ and $q_2^X = q_2^Y$. Since $q_1^Y, q_2^Y\in\mathrm{cons}^Y$ and $((l^T, l^S), v) \in Q^{X^{\Delta}}$ by construction of $R$, it follows from Lemma~\ref{lemma:quotientTSandAsamecons} that $q_1^X, q_2^X\in\mathrm{cons}^X$. Therefore, we can conclude that $q_1^X\xlongrightarrow{a}{}^{\!\! X^{\Delta}} q_2^X$. And from the construction of the bisimulation relation $R$ it follows that $(q_2^X,q_2^Y)\in R$.
		}
		
		\item \new{$q_1^Y\xlongrightarrow{a}{}^{\!\! Y^{\Delta}} q_2^Y$, $q_2^Y\in Q^Y$, and $a\in\mathit{Act}^Y\setminus\mathit{Act}^X$. This case is infeasible, as  $\mathit{Act}^X = \mathit{Act}^Y \cup \{i_{\mathit{new}}\}$. }
		
		\item \new{$q_1^X\xlongrightarrow{d}{}^{\!\! X^{\Delta}} q_2^X$, $q_2^X\in Q^X$, and $d\in\mathbb{R}_{\geq 0}$. From Definition~\ref{def:adversarialpruning} of adversarial pruning we have that $q_1^X\xlongrightarrow{d}{}^{\!\! X} q_2^X$ and $q_1^X, q_2^X\in\mathrm{cons}^X$. Following Definition~\ref{def:semanticTIOA} of the semantic and Definition~\ref{def:quotientreducedquotient} of the reduced $\sim$-quotient of $\sem{T\quotient S}$, it follows that $q_1^X = (l_1, v_1)$ and $q_2^X = (l_1, v_1 + d)$ with $l_1\in\mathit{Loc}^{T\quotient S}$, $v_1 \in [\mathit{Clk} \mapsto \mathbb{R}_{\geq 0}]$, $v_1 + d\models \mathit{Inv}(l_1)$, and $\forall d'\in\mathbb{R}_{\geq 0}, d' < d: v_1 + d'\models \mathit{Inv}(l_1)$. Since $q_1^X\in\mathrm{cons}^X$, it follows that $l_1 = (l_1^T, l_1^S)$ or $l_1 = l_u$. Therefore, from Definition~\ref{def:quotientTIOA} of the quotient for TIOA, we have that  and $\mathit{Inv}(l_1) = \mathbf{T}$. Note that we do not directly get information about whether the valuation $v_1 + d$ satisfy the location invariant in $T$ or $S$.
		}
		
		\new{Now consider first the simple case where $l_1 = l_u$. From Definition~\ref{def:quotientTIOTS} of the quotient for TIOTS, it follows directly that $u \xlongrightarrow{d}{}^{\!\! Y} u$. And note with Lemma~\ref{lemma:quotientTSandAsamestateset} that $q_2^X = f((l_u, v_1 + d)) = u = q_2^Y$ and thus $(q_2^X,q_2^Y) \in R$.
		}
		
		\new{Now consider the case where $l_1 = (l_1^T, l_1^S)$. We have to consider whether delays are possible in $\sem{T}$ and $\sem{S}$ in order to show that $Y$ can follow the delay and that the resulting state pair is in the bisimulation relation $R$.}
		\begin{itemize}
			\item \new{$q_1^{\sem{T}}\xlongrightarrow{d}{}^{\!\! \sem{T}} q_2^{\sem{T}}$ and $q_1^{\sem{S}}\xlongrightarrow{d}{}^{\!\! \sem{S}} q_2^{\sem{S}}$. In this case, it follows from Definition~\ref{def:semanticTIOA} of the semantic that $q_1^{\sem{T}} = (l_1^T, v_1^T)$, $\forall c \in \mathit{Clk}^T: v_1^T(c) = v_1(c)$, $q_2^{\sem{T}} = (l_1^T, v_1^T + d)$, $v_1^T + d \models \mathit{Inv}(l_1^T)$, and $\forall d'\in\mathbb{R}_{\geq 0}, d'<d: v_1^T + d' \models \mathit{Inv}(l_1^T)$; similarly we have that $q_1^{\sem{S}} = (l_1^S, v_1^S)$, $\forall c \in \mathit{Clk}^S: v_1^S(c) = v_1(c)$, $q_2^{\sem{S}} = (l_1^S, v_1^S + d)$, $v_1^S + d \models \mathit{Inv}(l_1^S)$, and $\forall d'\in\mathbb{R}_{\geq 0}, d'<d: v_1^S + d' \models \mathit{Inv}(l_1^S)$. From Definition~\ref{def:quotientTIOTS} of the quotient for TIOTS it follows that $(q_1^T, q_1^S) \xlongrightarrow{d}{}^{\!\! Y} (q_2^T, q_2^S)$. Observe with Lemma~\ref{lemma:quotientTSandAsamestateset} that $q_1^Y = (q_1^{\sem{T}}, q_1^{\sem{S}}) = (l_1^T, l_1^S,v_1) = q_1^X$ and $q_2^Y = (q_2^{\sem{T}}, q_2^{\sem{S}}) = (l_1^T, l_1^S,v_2) = q_2^X$. Thus $(q_2^X,q_2^Y) \in R$.
			}
			
			\item \new{$q_1^{\sem{T}}\xlongrightarrow{d}{}^{\!\! \sem{T}} q_2^{\sem{T}}$ and $q_1^{\sem{S}}\arrownot\xlongrightarrow{d}{}^{\!\! \sem{S}}$. In this case, it follows from Definition~\ref{def:semanticTIOA} of the semantic that $q_1^{\sem{T}} = (l_1^T, v_1^T)$, $\forall c \in \mathit{Clk}^T: v_1^T(c) = v_1(c)$, $q_2^{\sem{T}} = (l_1^T, v_1^T + d)$, $v_1^T + d \models \mathit{Inv}(l_1^T)$, and $\forall d'\in\mathbb{R}_{\geq 0}, d'<d: v_1^T + d' \models \mathit{Inv}(l_1^T)$; similarly we have that $q_1^{\sem{S}} = (l_1^S, v_1^S)$, $\forall c \in \mathit{Clk}^S: v_1^S(c) = v_1(c)$, and $\exists d'\in\mathbb{R}_{\geq 0}, d' \leq d: v_1^S + d' \not\models \mathit{Inv}(l_1^S)$. We have to consider two cases.
			}
			\begin{itemize}
				\item \new{$v_1^S\models \mathit{Inv}(l_1^S)$. Since $\mathit{Clk}^T\cap\mathit{Clk}^S=\emptyset$, $v_1 \models\mathit{Inv}(l_1^S)$. Since $(q_1^X, q_1^Y) \in R$ and $v_1^S \models \mathit{Inv}(l_1^S)$, we have that $q_1^Y = q_1^X$. From Definition~\ref{def:quotientTIOTS} of the quotient for TIOTS, it follows  that $q_1^Y = ((l_1^T, v_1^T), (l_1^S, v_1^S)) \xlongrightarrow{d}{}^{\!\! Y} u = q_2^Y$. From the construction of $R$ we have that $q_2^X \in A$, thus we can confirm that $(q_2^T, q_2^Y) \in R$.
				}
				
				\item \new{$v_1^S \not\models \mathit{Inv}(l_1^S)$. Again, since $\mathit{Clk}^T\cap\mathit{Clk}^S=\emptyset$, $v_1 \not\models\mathit{Inv}(l_1^S)$. Since $(q_1^X, q_1^Y) \in R$ and $v_1^S \not\models \mathit{Inv}(l_1^S)$, we have that $q_1^X \in A$, thus $q_1^Y = u$. From Definition~\ref{def:quotientTIOTS} of the quotient for TIOTS, it follows  that $u \xlongrightarrow{d}{}^{\!\! Y} u$. And by construction of $R$ it follows that $(q_2^X,q_2^Y) \in R$. }
			\end{itemize}
			
			\item \new{$q_1^{\sem{T}}\arrownot\xlongrightarrow{d}{}^{\!\! \sem{T}} q_2^{\sem{T}}$ and $q_1^{\sem{S}}\arrownot\xlongrightarrow{d}{}^{\!\! \sem{S}}$. This case follows the exact same reasoning as the one above, since Definition~\ref{def:quotientTIOTS} of the quotient for TIOTS does not care whether a delay $d$ is possible in $\sem{T}$ once it is not possible in $\sem{S}$.
			}
			
			\item \new{$q_1^{\sem{T}}\arrownot\xlongrightarrow{d}{}^{\!\! \sem{T}} q_2^{\sem{T}}$ and $q_1^{\sem{S}}\xlongrightarrow{d}{}^{\!\! \sem{S}} q_2^{\sem{S}}$. In this case, it follows directly from Definition~\ref{def:quotientTIOTS} of the quotient for TIOTS that there is no delay possible in $Y$, i.e., $(q_1^{\sem{T}}, q_1^{\sem{S}}) \arrownot\xlongrightarrow{d}{}^{\!\! \sem{T}\quotient \sem{S}}$. It follows from Definition~\ref{def:semanticTIOA} of the semantic that $q_1^{\sem{T}} = (l_1^T, v_1^T)$, $\forall c \in \mathit{Clk}^T: v_1^T(c) = v_1(c)$, and $\exists d'\in\mathbb{R}_{\geq 0}, d' \leq d: v_1^T + d' \not\models \mathit{Inv}(l_1^T)$; similarly we have that $q_1^{\sem{S}} = (l_1^S, v_1^S)$, $\forall c \in \mathit{Clk}^S: v_1^S(c) = v_1(c)$, $q_2^{\sem{S}} = (l_1^S, v_1^S + d)$, $v_1^S + d \models \mathit{Inv}(l_1^S)$, and $\forall d'\in\mathbb{R}_{\geq 0}, d'<d: v_1^S + d' \models \mathit{Inv}(l_1^S)$. Without loss of generality, we can assume that $v_1^T + 0 \not\models \mathit{Inv}(l_1^T)$\footnote{\new{In case there would be a $d' < d$ such that $v_1^T + d' \models \mathit{Inv}(l_1^T)$, we can use the first case to simulate the delay $d'$ in$Y$.}}, which simplifies to $v_1^T \not\models \mathit{Inv}(l_1^T)$. Combining this information, we have that $v_1\models \neg\mathit{Inv}(l_1^T) \wedge \mathit{Inv}(l_1^S)$, where we used the fact that $\mathit{Clk}^T\cap\mathit{Clk}^S=\emptyset$. Now, using Definition~\ref{def:quotientTIOA} of the quotient for TIOA and Definition~\ref{def:semanticTIOA} of the semantics, we have that $(l_1^T, l_1^S, v_1)\xlongrightarrow{i_{\mathit{new}}}{}^{\!\! \sem{T\quotient S}} (l_e, v_1)$. Since $(l_e, v_1) \notin\mathrm{cons}^X$ and $i_{\mathit{new}}$ is an input, it follows that $(l_1^T, l_1^S, v_1) = q_1^X \notin \mathrm{cons}^X$. This contradicts with our assumption that $q_1^X \in \mathrm{cons}^X$. Therefore, this case is infeasible.}
		\end{itemize}
		
		\new{In all feasible cases we can show that $(q_2^X,q_2^Y)\in R$. Since $q_1^X, q_2^X\in\mathrm{cons}^X$ and $A\subseteq Q^{X^{\Delta}}$ by construction of $R$, it follows from Lemma~\ref{lemma:quotientTSandAsamecons} that $q_1^Y, q_2^Y\in\mathrm{cons}^Y$. Therefore, we can conclude that $q_1^Y\xlongrightarrow{d}{}^{\!\! Y^{\Delta}} q_2^Y$.
		}
		
		\item \new{$q_1^Y\xlongrightarrow{d}{}^{\!\! Y^{\Delta}} q_2^Y$, $q_2^Y\in Q^Y$, and $d\in\mathbb{R}_{\geq 0}$. From Definition~\ref{def:adversarialpruning} of adversarial pruning we have that $q_1^Y\xlongrightarrow{d}{}^{\!\! Y} q_2^Y$ and $q_1^Y, q_2^Y\in\mathrm{cons}^Y$. Consider the following three cases from Definition~\ref{def:quotientTIOTS} of the quotient for TIOTS. 
		}
		\begin{itemize}
			\item \new{$q_1^Y = (q_1^{\sem{T}}, q_1^{\sem{S}})$, $q_2^Y = (q_2^{\sem{T}}, q_2^{\sem{S}})$, $q_1^{\sem{T}}\xlongrightarrow{d}{}^{\!\! \sem{T}} q_2^{\sem{T}}$, and $q_1^{\sem{S}}\xlongrightarrow{d}{}^{\!\! \sem{S}} q_2^{\sem{S}}$. From Definition~\ref{def:semanticTIOA} of the semantic it follows that $q_1^{\sem{T}} = (l_1^T, v_1^T)$, $q_2^{\sem{T}} = (l_1^T, v_1^T + d)$, $v_1^T + d \models \mathit{Inv}(l_1^T)$, $\forall d'\in\mathbb{R}_{\geq 0}, d' < d: v_1^T + d' \models \mathit{Inv}(l_1^T)$, $q_1^{\sem{S}} = (l_1^S, v_1^S)$, $q_2^{\sem{S}} = (l_1^S, v_1^S + d)$, $v_1^S + d \models \mathit{Inv}(l_1^S)$, and $\forall d'\in\mathbb{R}_{\geq 0}, d' < d: v_1^S + d' \models \mathit{Inv}(l_1^S)$. Now, from Definition~\ref{def:quotientTIOA} of the quotient for TIOA we have that $\mathit{Inv}((l_1^S, l_1^T)) = \mathbf{T}$ in $T\quotient S$, thus using Definitions~\ref{def:semanticTIOA} and~\ref{def:quotientreducedquotient} we have $q_1^X = (l_1^S, l_1^T, v_1)\xlongrightarrow{d}{}^{\!\! X} (l_1^S, l_1^T, v_1 + d) = q_2^X$. Observe that $q_1^X = q_1^Y$ and $q_2^X = q_2^Y$, thus $q_2^X, q_2^Y\in R$.
			}
			
			\item \new{$q_1^Y = (q_1^{\sem{T}}, q_1^{\sem{S}})$, $q_2^Y = u$, and $q_1^{\sem{S}}\arrownot\xlongrightarrow{d}{}^{\!\! \sem{S}} q_2^{\sem{S}}$. From Definition~\ref{def:semanticTIOA} of the semantic it follows that $q_1^{\sem{T}} = (l_1^T, v_1^T)$, $q_1^{\sem{S}} = (l_1^S, v_1^S)$, and $\exists d'\in\mathbb{R}_{\geq 0}, d' \leq d: v_1^S + d' \not\models \mathit{Inv}(l_1^S)$. Now, from Definition~\ref{def:quotientTIOA} of the quotient for TIOA we have that $\mathit{Inv}((l_1^S, l_1^T)) = \mathbf{T}$ in $T\quotient S$, thus using Definitions~\ref{def:semanticTIOA} and~\ref{def:quotientreducedquotient} we have $q_1^X = (l_1^S, l_1^T, v_1)\xlongrightarrow{d}{}^{\!\! X} (l_1^S, l_1^T, v_1 + d) = q_2^X$. We have to consider two cases to show that $(q_2^X,q_2^Y) \in R$.
			}
			\begin{itemize}
				\item \new{$v_1^S \models \mathit{Inv}(l_1^S)$. In this case $q_1^X \notin A$ and $q_2^X \in A$. Therefore, $(q_2^X,q_2^Y) \in R$.}
				
				\item \new{$v_1^S \not\models \mathit{Inv}(l_1^S)$. In this case $q_1^X, q_2^X \in A$. From the construction of $R$ it follows that any state from $A$ can only be related to state $u$ in $Y$, but $q_1^Y = (q_1^{\sem{T}}, q_1^{\sem{S}})$. This contradiction renders this case infeasible. }
			\end{itemize}
			
			\item \new{$q_1^Y = u$ and $q_2^Y = u$. From Definition~\ref{def:quotientTIOA} of the quotient for TIOA, it follows directly that $(l_u, v)\ {\xlongrightarrow{d}{}^{\!\! X}}\allowbreak (l_u, v)$ for any $v \in [\mathit{Clk} \mapsto \mathbb{R}_{\geq 0}]$. And note with Lemma~\ref{lemma:quotientTSandAsamestateset} that $q_1^X = q_2^X = f((l_u, v)) = u = q_1^Y = q_2^Y$ and thus $(q_2^X,q_2^Y) \in R$.}
		\end{itemize}
		
		\new{In all feasible cases we can show that $(q_2^X,q_2^Y)\in R$. Since $q_1^Y, q_2^Y\in\mathrm{cons}^Y$ and $A\subseteq Q^{X^{\Delta}}$ by construction of $R$, it follows from Lemma~\ref{lemma:quotientTSandAsamecons} that $q_1^X, q_2^X\in\mathrm{cons}^X$. Therefore, we can conclude that $q_1^X\xlongrightarrow{d}{}^{\!\! X^{\Delta}} q_2^X$.}
	\end{itemize}
	
	\new{We have show for state pair $(q_1^X, q_1^Y) \in R$ that all the six cases of bisimulation hold. Since we have chosen an arbitrary state pair from $R$, it holds for all state pairs in $R$. This concludes the proof.}
\end{proof}

\end{appendices}


\bibliography{references}

\end{document}